\newif\ifabstract
\abstracttrue
\newif\iffull
\ifabstract \fullfalse \else \fulltrue \fi

\documentclass[11pt]{article}
\usepackage{amsfonts}
\usepackage{amssymb}
\usepackage{amstext}
\usepackage{amsmath}
\usepackage{xspace}
\usepackage{ntheorem}
\usepackage{graphicx}
\usepackage{url}
\usepackage{graphics}
\usepackage{colordvi}
\usepackage{hyperref}
\usepackage{subfigure}
\usepackage{xcolor}
\usepackage{cleveref}

\advance \oddsidemargin by -0.8in
\newcommand{\myparskip}{3pt}
\parskip \myparskip

\newcommand{\td}{\tilde d}
\newcommand{\BC}{BC}

\newcommand{\constructexpander}{{\textsc{ConstructExpander}}\xspace}
\newcommand{\tk}{\tilde k}

\newcommand{\Otilde}{\widetilde O}

\newcommand{\HSS}{Hierarchical Support Structure\xspace}

\renewcommand{\cong}{\eta}
\newcommand{\transcript}{\mathbb{T}}

\newcommand{\matchorcut}{{\textsc{MatchOrCut}}\xspace}

\newcommand{\MMF}{\mbox{\sf{Maximum Multicommodity Flow}}\xspace}
\newcommand{\MM}{\mbox{\sf{Minimum Multicut}}\xspace}
\newcommand{\SSSP}{\mbox{\sf{SSSP}}\xspace}
\newcommand{\APSP}{\mbox{\sf{APSP}}\xspace}

\newcommand{\deledge}{\ensuremath{\operatorname{DeleteEdge}}\xspace}
\newcommand{\delexpedge}{\ensuremath{\operatorname{DeleteExpanderEdge}}\xspace}

\newcommand{\reducedegree}{{\textsc{ReduceDegree}}\xspace}

\newcommand{\ceil}[1]{\ensuremath{\left\lceil#1\right\rceil}}
\newcommand{\floor}[1]{\ensuremath{\left\lfloor#1\right\rfloor}}

\newcommand{\opt}{\mathsf{OPT}}
\newcommand{\hG}{\hat G}
\newcommand{\makecanonical}{{\textsc{MakeCanonical}}\xspace}
\newcommand{\set}[1]{\left\{ #1 \right\}}

\newcommand{\tset}{{\mathcal T}}
\newcommand{\iset}{{\mathcal{I}}}
\newcommand{\pset}{{\mathcal{P}}}
\newcommand{\qset}{{\mathcal{Q}}}

\newcommand{\cset}{{\mathcal{C}}}

\newcommand{\mset}{{ M}}
\newcommand{\jset}{{\mathcal{J}}}

\newcommand{\rset}{{\mathcal{R}}}

\newcommand{\hset}{{\mathcal{H}}}
\newcommand{\gset}{{\mathcal{G}}}
\newcommand{\sset}{{\mathcal{S}}}

\newcommand{\dset}{{\mathcal{D}}}

\newcommand{\be}{\begin{enumerate}}
\newcommand{\ee}{\end{enumerate}}
\newcommand{\bd}{\begin{description}}
\newcommand{\ed}{\end{description}}
\newcommand{\bi}{\begin{itemize}}
\newcommand{\ei}{\end{itemize}}

\newtheorem{theorem}{Theorem}[section]
\newtheorem{lemma}[theorem]{Lemma}
\newtheorem{observation}[theorem]{Observation}
\newtheorem{corollary}[theorem]{Corollary}
\newtheorem{claim}[theorem]{Claim}

\newtheorem{definition}{Definition}[section]
\newenvironment{proof}{\par \smallskip{\bf Proof:}}{\hfill\stopproof}
\def\stopproof{\square}
\def\square{\vbox{\hrule height.2pt\hbox{\vrule width.2pt height5pt \kern5pt
\vrule width.2pt} \hrule height.2pt}}

\newenvironment{proofof}[1]{\noindent{\bf Proof of #1.}}%
        {\hfill\stopproof}

\newenvironment{prog}[1]{
\begin{minipage}{5.8 in}
\begin{center}
{\sc #1}
\end{center}
}
{
\end{minipage}
}

\newcommand{\program}[3]{\begin{figure} \fbox{\vspace{2mm}\begin{prog}{#1} #3 \end{prog}\vspace{2mm}} 
			\caption{#1 \label{#2}} \end{figure}}

\renewcommand{\phi}{\varphi}
\newcommand{\eps}{\epsilon}

\newcommand{\half}{\ensuremath{\frac{1}{2}}}
\global\long\def\Otil{\tilde{O}}

\newcommand{\poly}{\operatorname{poly}}
\newcommand{\dist}{\mbox{\sf dist}}
\newcommand{\diam}{\mbox{\sf diam}}

\newenvironment{properties}[2][0]
{
\begin{enumerate} \setcounter{enumi}{#1}}{\end{enumerate}}

\newcommand{\mynote}[2][red]{\textcolor{red}{\sc\bf{[#2]}}}

\newcommand{\vol}{\operatorname{Vol}}

\newcommand{\proccut}{\ensuremath{\mathsf{ProcCut}}\xspace}
\newcommand{\extractexpander}{\ensuremath{\mathsf{AlgExtractExpander}}\xspace}
\newcommand{\procsep}{\ensuremath{\mathsf{ProcSeparate}}\xspace}
\newcommand{\procpathpeel}{\ensuremath{\mathsf{ProcPathPeel}}\xspace}
\newcommand{\procpart}{\ensuremath{\mathsf{ProcPartition}}\xspace}

\newcommand{\distquery}{\mbox{\sf{dist-query}}\xspace}

\newcommand{\EST}{\mbox{\sf{ES-Tree}}\xspace}

\newcommand{\CMG}{{\sf{Cut-Matching Game}}\xspace}
\newcommand{\DMG}{{\sf{Distanced Matching Game}}\xspace}

\newcommand{\MBSC}{{\sf{Most Balanced Sparse Cut}}\xspace}
\newcommand{\sparsest}{{\sf{Sparsest Cut}}\xspace}
\newcommand{\mbc}{{\sf{Minimum Balanced Cut}}\xspace}
\newcommand{\mbcwc}{{\sf{Minimum Balanced Cut with Certificate}}\xspace}
\newcommand{\lowcond}{{\sf{Lowest Conductance Cut}}\xspace}
\newcommand{\expanderdec}{{\sf{Expander Decomposition}}\xspace}

\begin{document}

\begin{titlepage}
	
	\title{A Distanced Matching Game, Decremental APSP in Expanders, and Faster Deterministic Algorithms for Graph Cut Problems\footnote{SODA 2023, to appear.}}
	\author{Julia Chuzhoy\thanks{Toyota Technological Institute at Chicago. Email: {\tt cjulia@ttic.edu}. Supported in part by NSF grant CCF-2006464.}}
	\maketitle
\pagenumbering{gobble}
	
	\thispagestyle{empty}
	\begin{abstract}
		Expander graphs play a central role in graph theory and algorithms. With a number of powerful algorithmic tools developed around them, such as the Cut-Matching game, expander pruning, expander decomposition, and algorithms for decremental All-Pairs Shortest Paths (APSP) in expanders, to name just a few, the use of expanders in the design of graph algorithms has become ubiquitous. Specific applications of interest to us are fast deterministic algorithms for cut problems in static graphs, and algorithms for dynamic distance-based graph problems, such as APSP.
		
		Unfortunately, the use of expanders in these settings incurs a number of drawbacks. For example, the best currently known algorithm for decremental APSP in constant-degree expanders can only achieve a $(\log n)^{O(1/\eps^2)}$-approximation with  $n^{1+O(\eps)}$ total update time for any $\eps$.
		All currently known algorithms for the Cut Player in the Cut-Matching game are either randomized, or provide rather weak guarantees: expansion  $1/(\log n)^{1/\eps}$ with running time $n^{1+O(\eps)}$. This, in turn, leads to somewhat weak  algorithmic guarantees for several central cut problems:  the best current almost linear time deterministic algorithms for Sparsest Cut, Lowest Conductance Cut, and Balanced Cut can only achieve approximation factor $(\log n)^{\omega(1)}$. Lastly, when relying on expanders in distance-based problems, such as dynamic APSP, via current methods, it seems inevitable that one has to settle for approximation factors that are at least $\Omega(\log n)$. In contrast, we do not have any negative results that rule out any super-constant approximation with almost linear total update time.
		
		In this paper we propose the use of well-connected graphs, and introduce a new algorithmic toolkit for such graphs that, in a sense, mirrors the above mentioned algorithmic tools for expanders. 
		One of these new tools is the Distanced Matching game, an analogue of the Cut-Matching game for well-connected graphs. 
		We demonstrate the power of these new tools by obtaining  better results for several of the problems mentioned above.
		First, we design an algorithm for decremental APSP in expanders with significantly better guarantees: in a constant-degree expander, the algorithm achieves 
		$(\log n)^{1+o(1)}$-approximation, with total update time  $n^{1+o(1)}$.
We also obtain a deterministic algorithm for the Cut Player in the Cut-Matching game that achieves expansion $\frac{1}{(\log n)^{5+o(1)}}$ in time $n^{1+o(1)}$,  deterministic almost linear-time algorithms for Sparsest Cut, Lowest-Conductance Cut, and Minimum Balanced Cut with approximation factors $O(\poly\log n)$, as well as an improved deterministic algorithm for Expander Decomposition. We believe that the use of well-connected graphs instead of expanders in various dynamic distance-based problems (such as APSP in general graphs) has the potential of providing much stronger guarantees, since we are no longer necessarily restricted to superlogarithmic approximation factors.

	\end{abstract}
\end{titlepage}

\tableofcontents{}

\newpage

\pagenumbering{arabic}

\section{Introduction}

Expander graphs are a central graph theoretic object that has been studied extensively, and they are frequently used in the design of graph algorithms.  In recent years, a number of powerful algorithmic tools have been developed around expanders, including, for example, the \CMG of \cite{KRV}, expander pruning of \cite{expander-pruning}, \expanderdec, and algorithms for decremental All-Pairs Shortest Paths (\APSP) in expanders (see e.g.  \cite{APSP-old,APSP-previous}), to name just a few. This powerful algorithmic toolkit has led to many new algorithms for optimization problems in graphs. In this paper we are most interested in applications to fast deterministic algorithms for classical cut problems, such as  \sparsest, \lowcond, and \mbc, and to dynamic algorithms, especially for distance-based graph problems, such as \APSP. 

Unfortunately, the use of expander graphs with the currently available expander-related tools incurs a number of drawbacks in these settings. While the results of \cite{KRV} provide a near-linear time randomized algorithm for the Cut Player in the \CMG, which guarantees that the resulting graph is an $\Omega(1)$-expander, in the regime of deterministic algorithms, the best currently known results are significantly weaker: the algorithm of \cite{detbalanced} provides an implementation of the Cut Player in time $O\left (n^{1+O(\eps)}\right ) $, but only guarantees expansion of $\frac{1}{(\log n)^{1/\eps}}$ for a sufficiently large $\eps$.  As a result, despite having randomized  algorithms for \sparsest, \lowcond, and \mbc, that achieve $O(\log^2 n)$-approximation in $m^{1+o(1)}$ time, the best current deterministic algorithms for these problems with running time $m^{1+o(1)}$ can only achieve approximation factor $(\log n)^{\omega(1)}$. Similar issues, that we discuss in more detail below, arise in algorithms for \expanderdec. We note that  all these cut problems are used extensively, including in settings (such as, e.g. dynamic algorithms), where deterministic algorithms are especially desirable. In recent years, the above mentioned expander-based tools have also found many applications in dynamic algorithms; for now we focus on decremental All-Pairs Shortest Paths (\APSP). In this context, the use of expanders with the currently available algorithmic tools also has several drawbacks. First, in a typical use of expanders in dynamic \APSP, one cannot obtain a better than $\Theta(\log n)$ approximation factor; we discuss this in more detail below. Second, the use of expanders usually involves either the \CMG, or \expanderdec, which suffer from the drawbacks mentioned above. Lastly, in most uses of expander graphs for dynamic \APSP, one needs to rely on algorithms for decremental \APSP in expanders. However, the best such current algorithm only provides a rather weak tradeoff between the approximation factor and total update time: for bounded-degree expanders, the results of \cite{APSP-old,APSP-previous} achieve  $(\log n)^{O(1/\eps^2)}$-approximation with total update time $ n^{1+O(\eps)} $, where $\eps$ is a given precision parameter.

In this paper, we propose to study a different kind of graphs, that we call \emph{well-connected} graphs. Given an $n$-vertex graph $G$,  a set $S$ of its vertices called \emph{supported vertices}, and parameters $\eta,d>0$, we say that graph $G$ is $(\eta,d)$-well-connected with respect to $S$, if, for every pair $A,B\subseteq S$ of disjoint equal-cardinality subsets of supported vertices, there is a collection $\pset$ of paths in  $G$, that connect every vertex of $A$ to a distinct vertex of $B$, such that the length of each path in $\pset$ is at most $d$, and every edge of $G$ participates in at most $\eta$ paths in $\pset$. For intuition, it would be convenient to think of $d=2^{\poly(1/\eps)}$, $\eta=n^{O(\eps)}$, and $|S|\geq |V(G)|-n^{1-\eps}$, for some parameter $0<\eps<1$. In the discussion below, we will informally refer to a graph $G$ that is $(\eta,d)$-well-connected with respect to a set $S$ of its vertices, with the above setting of parameters, as a \emph{well-connected graph}. We develop an algorithmic toolkit for well-connected graphs that is, in a sense, analogous to some of the tools that are known for expander graphs. We then show that using well-connected graphs, together with these new algorithmic tools, allows us to overcome many of the            hurdles mentioned above. For example, we obtain deterministic almost linear-time $O(\poly\log n)$-approximation algorithms for \sparsest, \lowcond, and \mbc, as well as a better deterministic algorithm for the Cut Player in the \CMG, and better algorithms for \expanderdec and decremental \APSP in expanders. 
While our algorithmic toolkit immediately leads to strengthening the currently known algorithmic tools for expander graphs, it is our hope that it will eventually lead to replacing expanders with well-connected graphs in some of the applications mentioned above.

We note that a very recent paper of
\cite{haeupler2022hop} introduced a graph-theoretic object, called \emph{$h$-hop expander}, that appears closely related to the notion of well-connected graphs. Informally, a graph $G$ is an $h$-hop $\phi$-expander, if any unit demand between pairs of its vertices that are within distance at most $h$ from each other, can be routed via paths of length comparable to $h$, with congestion $\tilde O(1/\phi)$. Like with standard expanders, $h$-hop $\phi$-expanders can also be defined via a dual notion of cuts (called \emph{moving cuts} in this case), and  \cite{haeupler2022hop}, building on the work of \cite{haeupler2020network}, established a tight relationship between the two definitions, somewhat similar to the standard notion of the flow-cut gap. They also provide an algorithm to compute an $h$-hop expander decomposition of a given graph. While the notions of $h$-hop expanders and well-connected graphs seem similar in spirit, there are some technical differences. At the same time, an $h$-hop expander with small diameter is essentially a special case of a well-connected graph.
We also note that in this paper, our main focus is on routing via very short paths, of sublogarithmic length, and on deterministic algorithms, while \cite{haeupler2022hop} consider randomized algorithms for (oblivious) routing on paths of length $O(h\poly\log n)$.

The remainder of the Introduction is organized as follows. We start by describing the algorithmic tools that we developed for well-connected graphs. We then provide an overview of several applications of these results to static graph cut problems, and then discuss new results and potential future uses of these new algorithmic tools in dynamic \APSP.

Before we continue, we need to introduce some notation. Given a graph $G$, a \emph{cut} is a partition $(A,B)$ of its vertices into non-empty subsets. The \emph{sparsity} of the cut $(A,B)$ is $\frac{|E(A,B)|}{\min\set{|A|,|B|}}$. 
Throughout this paper, we say that a graph $G$ is a $\phi$-expander, if every cut in $G$ has sparsity at least $\phi$. 
All graphs discussed in this paper are unweighted and undirected, unless stated otherwise. We will informally say that the running time or a total update time of an algorithm is \emph{almost linear}, if it is bounded by $O(m^{1+o(1)})$, where $m$ is the number of edges in the input graph (or the initial number of edges, if the graph is decremental).
Given a graph $G$ and a subset $S$ of its vertices, the \emph{volume} of $S$, denoted by $\vol_G(S)$, is the sum of degrees of all vertices in $S$. The \emph{volume of the graph} $G$, denoted by $\vol(G)$, is $\vol(G)=\vol_G(V)=2|E(G)|$. Given a collection $\mset$ of pairs of vertices in graph $G$, and a collection $\pset$ of paths in $G$, we say that the paths in $\pset$ \emph{route} the pairs in $\mset$ if, for every pair $(u,v)\in \mset$ of vertices, there is a path in $\pset$ whose endpoints are $u$ and $v$. The \emph{congestion} of a collection $\pset$ of paths in a graph $G$ is the largest number of paths that contain the same edge.

\subsection{Algorithmic Tools for Well-Connected Graphs.}
 
\paragraph{The Distanced Matching Game.} 
The first tool that we develop is a \emph{\DMG}, that can be viewed as an analogue of the \CMG of \cite{KRV} for well-connected graphs. The input to the \DMG consists of an integer $n$, a distance parameter $d$, and another parameter $0<\delta<1$. The game uses the notion of a \emph{distancing}. Given an $n$-vertex graph $G$, a $(\delta, d)$-distancing for $G$ is a triple $(A,B,E')$, where $A$ and $B$ are disjoint equal-cardinality subsets of vertices of $G$ with $|A|\geq n^{1-\delta}$, and $E'$ is a subset of edges with $|E'|\leq |A|/16$. We require that the length of the shortest path connecting a vertex of $A$ to a vertex of $B$ in $G\setminus E'$ is at least $d$.

The \DMG is played between a \emph{Distancing Player} and a \emph{Matching Player}. Similarly to the \CMG, we start with a graph $H$ that contains $n$ vertices and no edges, and then iteratively add edges to $H$. In each iteration $i$, the Distancing Player needs to compute a $(\delta,d)$-distancing $(A_i,B_i,E'_i)$ in the current graph $H$, if it exists. The matching player then needs to compute an arbitrary matching $M_i\subseteq A_i\times B_i$, of cardinality at least $|A_i|/4$. The edges of $M_i$ are added to graph $H$, and we continue to the next iteration. The game terminates when no $(\delta,d)$-distancing exists in $H$ (alternatively, we may terminate it earlier, if we establish that graph $H$ has some desired properties, such as, e.g. it is well-connected).

Our first result shows that, if $d\geq 2^{4/\delta}$, then the number of iterations in the \DMG is bounded by $n^{8\delta}$, regardless of the strategies of the two players. We note that this is significantly larger than the number of iterations in the \CMG, which is typically bounded by $O(\poly\log n)$ (see, e.g. \cite{KRV,KhandekarKOV2007cut}). However, as we show later, we gain an advantage that, under some specific strategy of the Distancing Player that we provide below, the resulting graph $H$ is well-connected, and the distances between the vertices lying in the set $S$ of supported vertices are quite low, as opposed to the (super)-logarithmic distances that expander graphs guarantee.
We also note that, while it appears that the bound on the number of iterations is close to being tight, it is not clear whether the exponential dependence of the distance parameter $d$ on $1/\delta$ in our result is necessary.

We note that, in a very recent independent work, \cite{hop-constr-CMG} suggested and analyzed a variant of the \CMG for constructing $h$-hop expanders, whose diameter is small. The resulting expanders are similar to well-connected graphs that we considere here.

\paragraph{Hierarchical Support Structure and an algorithm for the Distancing Player.}
Our next result provides a deterministic algorithm for the Distancing Player. Suppose we are given a parameter $n$, and a precision parameter $\eps$. Let $\delta=4\eps^3$ and $d=2^{32/\eps^4}$, and let $H$ be an $n$-vertex graph, that we can think of as arising during the execution of the \DMG. The algorithm either computes a $(\delta,d)$-distancing in graph $H$, or it computes a large set $S\subseteq V(H)$ of its vertices, such that $H$ is $(\eta,\tilde d)$-well-connected for $S$, where $\eta=n^{O(\eps)}$ and $\tilde d=2^{O(1/\eps^5)}$. The running time of the algorithm is $O(|E(H)|^{1+O(\eps)})$.

In fact the above algorithm provides  stronger guarantees. We define the notion of a Hierarchical Support Structure for an $n$-vertex graph $H$. Informally, the structure consists of a collection $\set{H_1,\ldots,H_r}$ of $r=\Omega(n^{\eps})$ graphs, where $|V(H_i)|=\ceil{n^{1-\eps}}$ for all $i$, and an embedding of graph $\bigcup_{i=1}^rH_i$ into $H$ via short paths that cause low congestion. For each one of the graphs $H_i$, we must in turn be given a Hierarchical Support Structure, that can be used to define a set $S(H_i)$ of supported vertices for $H_i$; graph $H_i$ must be well-connected with respect to $S(H_i)$. We then let $S(H)=\bigcup_iS(H_i)$ be the set of supported vertices for graph $H$.
We note that a somewhat similar notion was (implicitly) used in \cite{detbalanced,APSP-old} in the context of expander graphs: that is, all graphs in the Hierarchical Support Structure were required to be expanders. 

Our deterministic algorithm for the Distancing Player either provides the desired $(\delta,d)$-distancing in graph $H$, or it constructs a \HSS for $H$, so that $H$ is well-connected with respect to the resulting set $S(H)$ of vertices. In all our subsequent results, we use the \DMG with this implementation of the Distancing Player. This ensures that, once the algorithm terminates, we obtain a \HSS for graph $H$, and a guarantee that $H$ is well-connected with respect to $S(H)$.

\paragraph{APSP in Well-Connected Graphs.}
The \HSS provides a convenient basis for decremental \APSP in well-connected graphs obtained via the \DMG. Indeed, suppose we are given such a graph $H$, that undergoes an online sequence of edge deletions. As edges are deleted from $H$, the graph may no longer be well-connected with respect to $S(H)$. We design a deterministic algorithm for decremental \APSP that, given such a graph $H$, can withstand up to $|V(H)|^{1-\Theta(\eps)}$ edge deletions, as it maintains a set $S'(H)\subseteq S(H)$ of supported vertices. Over the course of the algorithm, vertices may leave $S'(H)$ but they may not join it, and $|S'(H)|\geq |V(H)|/2^{O(1/\eps)}$ holds at all times. The algorithm supports short-path queries between pairs of vertices in $S'(H)$: given a pair $x,y\in S'(H)$ of such vertices, it must return a path $P$ of length at most $2^{O(1/\eps^6)}$ connecting them in the current graph $H$, in time $O(|E(P)|)$. The total update time of the algorithm is $O( |E(H)|^{1+O(\eps)})$. 
This algorithm can be thought of as mirroring similar algorithms for \APSP in expanders of \cite{APSP-old,APSP-previous} that we discuss below. For comparison, the best previous algorithm for \APSP in bounded-degree expanders could only report paths whose length is bounded by $(\log n)^{O(1/\eps^2)}$ in response to short-path queries, with total update time is $O\left (n^{1+O(\eps)}\right )$, though it could withstand a longer sequence of edge deletions. Interestingly, the algorithmic tools presented so far can also be used to obtain better algorithms for \APSP in expanders, as we discuss below. We note that it is currently not clear to us whether the exponential dependence of the lengths of the paths returned in response to short-path queries on $\poly(1/\eps)$ in our results is necessary. Improving this dependence may lead to further improvements in  applications discussed below.

\subsection{Applications to Other Algorithmic Tools.}

We  use the machinery that we have developed in order to design better implementations of existing algorithmic tools. The first such tool that we discuss is the \CMG of \cite{KRV}.

{\bf Cut-Matching Game.}
The \CMG was initially introduced by \cite{KRV}, as part of their fast approximation algorithms for \mbc and \sparsest. The game has a single parameter $n$, that is an even integer, and it is played between two players: a Cut Player and a Matching Player. The purpose of the game is to construct an $n$-vertex expander. The game starts with a graph $H$ containing $n$ vertices and no edges, and in every iteration edges are added to $H$, until we can certify that it becomes an expander. 
In the original version of the game suggested by \cite{KRV}, in every iteration $i$, the cut player needs to compute a partition of $V(H)$ into two equal-cardinality subsets $A_i$ and $B_i$, and the matching player needs to return an arbitrary perfect matching $\mset_i\subseteq A_i\times B_i$. The edges of $\mset_i$ are then added to graph $H$, and the algorithm continues to the next iteration. In their paper, \cite{KRV} provided an algorithm for the Cut Player, that computes, in every iteration, a partition $(A_i,B_i)$ of $V(H)$, so that, regardless of the strategy of the Matching Player, the algorithm is guaranteed to terminate with an $\Omega(1)$-expander after $O(\log^2n)$ iterations. 
While the algorithm of \cite{KRV} for the Cut Player is very efficient -- its running time is $\tilde O(n)$, it is unfortunately randomized, and it is unclear how to derandomize it, if our goal is to obtain an algorithm with an almost linear running time. 

Instead, we consider a variant of the \CMG due to \cite{KhandekarKOV2007cut}, that was also studied in \cite{detbalanced}. In this variant, in every iteration $i$, the Cut Player must either compute a partition $(A'_i,B'_i)$ of $V(H)$ with $|A'_i|\geq |B'_i|\geq n/4$ and $|E_H(A'_i,B'_i)|\leq n/10$; or it must compute a subset $X\subseteq V(H)$ of at least $n/2$ vertices, so that graph $H[X]$ is a $\phi$-expander (and we would like $\phi$ to be as close to $1$ as possible). In the former case, 
the matching player must compute an arbitrary partition $(A_i,B_i)$ of $V(H)$ with $|A_i|=|B_i|$ and $B'_i\subseteq B_i$, together with a perfect matching $\mset_i\subseteq A_i\times B_i$. The edges of $\mset_i$ are added to $H$, and we continue to the next iteration. In the latter case, the matching player must compute an arbitrary matching $\mset_i\subseteq X\times V(H)\setminus X$, of cardinality $|V(H)\setminus X|$. The edges of $\mset_i$ are then added to graph $H$, which is now guaranteed to be a $\phi/2$-expander, and the algorithm terminates. As shown by \cite{KhandekarKOV2007cut}, this variant of the game must terminate after $O(\log n)$ iterations\footnote{In fact the game presented here is a slight modification of the game suggested by \cite{KhandekarKOV2007cut}. A formal proof that the number of iterations in this variations is still bounded by $O(\log n)$ appears in \cite{detbalanced}.}. Note that the question of obtaining a fast deterministic algorithm for the Cut Player in this variation of the game leads to a sort of chicken and egg situation: the Cut Player essentially needs to solve a \mbc problem on the current graph $H$, and solving this problem efficiently typically requires running the \CMG, which in turn requires an efficient implementation of the Cut Player. 

In \cite{detbalanced}, a deterministic algorithm for the Cut Player was presented for the above setting. For a given precision parameter $\eps\geq \frac{\log\log n}{(\log n)^{1/2}}$, the algorithm has running time 
 $O\left (n^{1+O(\eps)}\right )$, and it ensures that the resulting graph has expansion $\phi\geq 1/(\log n)^{1/\eps}$. The algorithm proceeds by constructing a hierarchical system of expanders that are embedded into $H$, similarly to our \HSS, except that expanders are used instead of well-connected graphs. Unfortunately, it appears that the use of expanders forces one to lose a polylogarithmic in $n$ factor in the expansion (or alternatively, in the length of the embedding paths) with every recursive level, which eventually leads to a rather weak expansion guarantee. In this paper we design a deterministic algorithm for the Cut Player
 with running time $n^{1+o(1)}$, that ensures that the resulting graph has expansion $1/(\log n)^{5+o(1)}$.
Instead of using the approach of \cite{detbalanced}, we rely on another algorithmic tool introduced in this paper, that we call \emph{advanced path peeling}. We believe that this algorithmic tool is of independent interest, and in fact it can be used in order to either embed an expander graph into a given input graph $G$, or compute a sparse cut in $G$, thereby completely bypassing the \CMG. Before we provide more details on advanced path peeling, we briefly mention a typical implementation of the Matching Player in the \CMG, since it is related to the path peeling technique.

\paragraph{Path Peeling Algorithms.}
Typically, a \CMG is used in order to either embed a large expander graph $H$ into a given input graph $G$, or to compute a sparse cut in $G$. The game is played on graph $H$, that initially contains the set $V(G)$ of vertices and no edges. The Cut Player is implemented using one of the algorithms described above. In order to implement a Matching Player, in each iteration $i$, we try to construct a large matching $\mset_i\subseteq A_i\times B_i$, and its routing $\pset_i$ in $G$ via short paths that cause low congestion. Matching $\mset_i$ is then used in order to compute the response of the Matching Player. Paths in $\pset_i$ also define an embedding of the edges of $\mset_i$ into $G$, and so, as the algorithm progresses and new edges are added to  $H$, we  maintain an embedding of $H$ into $G$. Once graph $H$ becomes an expander, the algorithm terminates. Alternatively, if we are unable to compute a routing of a large matching $\mset_i\subseteq A_i\times B_i$, we would like to compute a sparse cut in $G$, and in this case the algorithm terminates with this cut. Therefore, in order to implement the Matching Player, we need an algorithm that, given a pair $A_i,B_i$ of disjoint sets of vertices in a graph $G$, either computes a routing $\pset_i$ of a large enough matching $\mset_i\subseteq A_i\times B_i$ via short paths that cause low  congestion, or returns a sparse cut in $G$.  In the original paper of \cite{KRV}, the Matching Player was implemented via an algorithm for computing maximum flow and minimum cut. However, as was later observed in \cite{fast-vertex-sparsest}, it is sufficient to compute a \emph{maximal} collection of such paths in $G$, by a simple greedy algorithm, that can be implemented very efficiently. We informally refer to such algorithms as \emph{basic path peeling}. The algorithm keeps greedily adding short paths connecting vertices of $A_i$ to vertices of $B_i$ to set $\pset_i$, while deleting from $G$ edges that already participate in many paths. If, at the end of the algorithm, $|\pset_i|$ is not sufficiently large, a simple application of the classical Ball Growing technique is used to compute a sparse cut in $G$. 
As an example, in one implementation of such an algorithm (Theorem 3.2 of \cite{detbalanced}), given a graph $G$, equal-cardinality sets $A_i,B_i$ of vertices of $G$, and parameters $z>0$ and $0<\phi<1$, the algorithm either computes a cut $(X,Y)$ of sparsity at most $\phi$ in $G$ with $|X|,|Y|\geq z/2$; or it constructs a routing $\pset_i$ of a matching $\mset_i\subseteq A_i\times B_i$, with $|\mset_i|\geq |A_i|-z$, such that paths in $\pset_i$ have length at most $O\left(\frac{\Delta\log n}{\phi}\right )$ each, and cause congestion at most $O\left(\frac{\Delta^2\log^2n}{\phi^2}\right )$, where $\Delta$ is  maximum vertex degree of $G$. The running time of the algorithm is $\tilde O(|E(G)|/\phi^3)$. Using different methods (that rely on algorithms for approximate Maximum Flow), \cite{detbalanced} design algorithms with better guarantees: the congestion is only bounded by $O\left(\Delta(\log n)/\phi\right )$, and the running time is $m^{1+o(1)}$. We note that a recent result of \cite{max-flow} provides a near-linear time randomized algorithm for the exact maximum flow and min-cost flow problems, with running time $m^{1+o(1)}$, where $m$ is the number of graph edges, if all edge capacities and costs are polynomially bounded. It is likely that one can obtain a randomized almost linear time algorithm for basic path peeling using this work as well.

Consider now a more general setting (that we refer to as \emph{advanced path peeling}), where we are given two sets $A$, $B$ of vertices, and we need to route a \emph{specific} matching $\mset\subseteq A\times B$. In this case, the input consists of an $m$-edge graph $G$, a collection $\mset=\set{(s_1,t_1),\ldots,(s_k,t_k)}$ of pairs of its vertices, and parameters $0<\phi<1$ and $z>0$. The goal is to either route a collection $\mset'\subseteq \mset$ of at least $k-z$ pairs of vertices via paths that are short and cause low congestion (compared to $1/\phi$), or return a cut $(X,Y)$ of sparsity at most $\phi$ in $G$, with $|X|,|Y|\geq \Omega(z)$. One could use the techniques employed in basic path peeling in order to design such an algorithm, but it seems inevitable that the running time of the algorithm can only be bounded by $O(mk)$ (unless we have a very efficient algorithm for decremental \APSP with low approximation factor). 

In this paper, we provide a deterministic algorithm for advanced path peeling. For a given precision parameter $\eps$, the algorithm has running time $O\left (\frac{m^{1+O(\eps)}}{\phi^3}\right )$, and, in case a routing $\pset$ is returned, the length of every path  is  bounded by 
$\frac{2^{O(1/\eps^6)}\cdot \Delta\cdot \log n}{\phi}$, with congestion bounded by $\frac{2^{O(1/\eps^6)}\cdot \Delta^2\cdot \log^2 n}{ \phi^2}$, where $\Delta$ is maximum vertex degree in $G$. 
The fact that we can pre-specify pairs of vertices to be routed makes advanced path peeling a much more powerful algorithmic tool than basic path peeling. For example, we can use it directly in order to either  embed some fixed low-degree expander $H$ into a given graph $G$, or to compute a sparse cut in $G$. This approach allows us to completely bypass the \CMG, and we use it in some of our algorithms. 

\subsection{Applications to Static Graphs.}

\paragraph{Sparsest Cut and Lowest-Conductance Cut.} Recall that the sparsity of a cut $(A,B)$ in a graph $G$ is $\frac{|E_G(A,B)|}{\min\set{|A|,|B|}}$. In the \sparsest problem, the goal is to compute a cut of minimum sparsity in the input graph $G$. 
A closely related notion is that of \emph{conductance}: the conductance of a cut $(A,B)$ in $G$ is $\frac{|E_G(A,B)|}{\min\set{\vol_G(A),\vol_G(B)}}$. In the \lowcond problem, given a graph $G$, the goal is to compute a cut of smallest conductance. We define the \emph{conductance} of a graph $G$, $\Psi(G)$, to be the smallest conductance value of any cut in $G$. Both \sparsest and \lowcond problems are among the most fundamental optimization problems, and are routinely used in the design of graph algorithms, for example, when divide-and-conquer paradigm is involved. The best current approximation algorithms for these problems, due to \cite{ARV}, achieve a factor-$O(\sqrt{\log n})$-approximation. Unfortunately, these algorithms are rather slow (though their running times are polynomially bounded), as they need to solve an SDP. 
The work of \cite{KRV}, that introduced the \CMG, provided a randomized $O(\log^2n)$-approximation algorithm for both problems, with running time $\tilde O(m+n^{3/2})$. The super-linear running time in this algorithm is mostly due to the rather slow state of the art algorithms for (approximate) maximum flow that were available at the time. With the more recent improvements in such algorithms, the running time of the algorithm of \cite{KRV} becomes almost linear. 
Since the \sparsest and \lowcond problems are so ubiquitous, it is however highly desirable to obtain fast \emph{deterministic} algorithms for them.
A number of deterministic algorithms, that are based on the Multiplicative Weights Update framework of \cite{GK98,Fleischer00,Karakostas08}, achieve a factor $O(\log n)$-approximation for both problems, in time $\Otil(m^2)$. 
Additionally, several algorithms in which the approximation factor is roughly $O(\phi^{1/2})$, where $\phi$ is the value of the optimal solution, are known (see e.g. \cite{Alon86, AndersenCL07,GaoLNPSY19}; the algorithms achieve running times $\tilde O(n^{\omega}),\tilde O(mn)$, and $O(m^{1.5+o(1)})$, respectively). In \cite{detbalanced}, deterministic algorithms for both \sparsest and \lowcond problems were presented, that, for a parameter $\frac{\log\log n}{(\log n)^{1/2}}\leq \eps<1$, achieve a factor $(\log n)^{O(1/\eps^2)}$-approximation in time $O(m^{1+\eps+o(1)})$.
Unfortunately, in time $m^{1+o(1)}$, these algorithms cannot achieve a polylogarithmic approximation factor. Our improved algorithm for the Cut Player in the \CMG immediately leads to deterministic algorithms for both \sparsest and \lowcond problems, with approximation factor $O(\log^7n \log\log n )$, and running time $m^{1+o(1)}$.

\paragraph{Minimum Balanced Cut.}
\mbc is another classical graph partitioning problem that is extensively used in algorithm design.
Given a graph $G$, we say that a cut $(A,B)$ in $G$ is $\beta$-balanced, if $\vol_G(A),\vol_G(B)\geq \vol(G)/\beta$. We say that the cut is \emph{balanced}, if it is $\beta$-balanced for $\beta=1/3$, and we say that it is \emph{almost balanced}, if it is $\beta$-balanced for some absolute constant $\beta$. In the \mbc problem, given a graph $G$, the goal is to compute a balanced cut $(A,B)$ minimizing $|E_G(A,B)|$. It is quite common to use bicriteria approximation algorithms for the problem: a factor-$\alpha$ bicriteria approximation algorithm must return an almost balanced cut $(A,B)$ with $|E_G(A,B)|\leq \alpha\cdot \opt$, where $\opt$ is the lowest possible value of $|E_G(A',B')|$ for any balanced cut $(A',B')$.
The seminal work of \cite{ARV} provides the best currently known bicriteria approximation algorithm for the problem, whose approximation factor is $O(\sqrt{\log n})$, though the algorithm is somewhat slow due to the need to solve an SDP.
As with the \sparsest and \lowcond problems, the randomized algorithm of \cite{KRV} can be used to obtain a factor-$O(\log^2n)$ bicriteria approximation in time $O(m^{1+o(1)})$. 
The best current deterministic algorithm for the problem, due to \cite{detbalanced}, obtained, for any 
$\frac{\log\log n}{(\log n)^{1/2}}\leq \eps<1$, a factor-$(\log n)^{O(1/\eps^2)}$ bicriteria approximation, in time $O(m^{1+\eps+o(1)})$. As in \sparsest and \lowcond problems, this algorithm can only achieve $(\log n)^{\omega(1)}$-approximation in time $m^{1+o(1)}$. In this paper we provide a deterministic bicriteria factor-$(\log n)^{8+o(1)}$ approximation algorithm, with running time $m^{1+o(1)}$.

We also consider an important variant of the problem, that we call \mbcwc. In this problem, we are given a graph $G$, and a target parameter $\psi$. The goal is to  compute a cut $(A,B)$ in $G$ with $|E_G(A,B)|\leq \alpha \psi\cdot \vol(G)$ (where $\alpha$ is the \emph{approximation factor} of the algorithm), such that either $\vol_G(A),\vol_G(B)\geq \vol(G)/3$; or $\vol_G(A)\geq 2\vol(G)/3$, and graph $G[A]$ has conductance at least $\psi$. In the latter outcome, if $\vol_G(B)\leq \vol(G)/4$, we can view graph $G[A]$ as a certificate that the value of the \mbc in $G$ is $\Omega(\psi\cdot |E(G)|)$. A factor-$\alpha$ approximation algorithm for \mbcwc can be easily converted into a factor-$O(\alpha)$ approximation algorithm for the \mbc problem, with running time that increases by at most factor $O(\log n)$ (this was shown in \cite{detbalanced}; we also provide more details in Section \ref{subsec: MBC and ED}). Algorithms for \mbcwc however appear to be significantly more powerful than those for \mbc, as they can be used in order to compute expander decomposition of a given graph efficiently. 
In \cite{detbalanced}, a deterministic algorithm for \mbcwc  that achieves approximation factor $\alpha=(\log n)^{O(1/\eps^2)}$, in time $O(m^{1+\eps+o(1)})$, for any $\frac{\log\log n}{(\log n)^{1/2}}\leq \eps<1$, was presented. We provide a deterministic algorithm for \mbcwc with approximation factor $(\log n)^{8+o(1)}$, whose running time is $O(m^{1+o(1)}/\psi)$. For $\psi\geq 1/m^{o(1)}$, which is a common setting used in algorithms for expander decomposition, the running time becomes $O(m^{1+o(1)})$. We provide another algorithm, whose running time is $O(m^{1+o(1)})$ for any value of $\psi$, and approximation factor remains the same, but it provides a somewhat weaker certificate: in case where $\vol_G(B)<\vol(G)/3$, it only guarantees that $G[A]$ contains a large subgraph with conductance at least $\psi$, but it does not compute such a graph.

\paragraph{Expander Decomposition.}
An \emph{$(\delta,\psi)$-expander decomposition} of a graph $G=(V,E)$
is a partition $\Pi=\{V_{1},\dots,V_{k}\}$ of the set $V$ of vertices, such that for all $1\leq i\leq k$, the conductance of graph $G[V_i]$ is at least $\psi$, and $\sum_{i=1}^k\delta_{G}(V_{i})\le\delta\cdot\vol(G)$.
Algorithms for expander decomposition are used extensively in the design of graph algorithms, in both static and dynamic settings. A long line of research  
 \cite{SpielmanT04,NanongkaiS17,Wulff-Nilsen17,expander-pruning,agassy2022expander} culminated in a randomized algorithm that computes a  $(\delta,\delta/ \poly(\log n))$-expander decomposition in  time $\Otil(m/\delta)$. 
The best previous deterministic algorithm, due to \cite{detbalanced}, computes a $(\delta,\phi)$-expander decomposition with  $\phi=\Omega(\delta/(\log m)^{O(1/\eps^2)})$,
	in time  $O\left ( m^{1+\eps+o(1)}  \right )$. This algorithm was in turn used  by \cite{detbalanced} in order to obtain the first deterministic algorithms for Dynamic Connectivity and Dynamic Minimum Spanning Forest,
	with $n^{o(1)}$ worst-case update time, making a significant progress on a major open question in the area of dynamic algorithms.
We provide a deterministic algorithm for computing a $\left(\delta,\psi \right)$-expander decomposition of $G$ with $\psi=\Omega\left (\frac{\delta}{(\log n)^{9+o(1)}}\right )$,
in time $O(m^{1+o(1)}/\delta)$.

\subsection{Applications to Dynamic Graphs: All-Pairs Shortest Paths (in Expanders).}

In the decremental All-Pairs Shortest Paths (\APSP) problem, the input is an $n$-vertex graph $G$ with non-negative length on edges, that undergoes an online sequence of edge deletions. The goal is to support (approximate) shortest path queries: given a pair $x,y$ of vertices, return  an (approximate) shortest path connecting $x$ to $y$ in the current graph $G$. We say that the algorithm achieves a factor-$\alpha$ approximation, if, in response to a shortest path query between $x$ and $y$, it is guaranteed to return a path 
of length at most $\alpha\cdot \dist_G(x,y)$. 
Decremental, and, more generally, fully dynamic \APSP is one of the most basic problem in the area of dynamic algorithms. It also has important connections to designing algorithms for classical cut and flow problems in the static graph model. For example,
by combining the standard primal-dual technique-based algorithm of \cite{GK98, Fleischer00} 
with an algorithm for a special case of decremental \APSP, called Single-Source Shortest Paths (\SSSP), one can obtain fast approximation algorithms for maximum $s$-$t$ flow, minimum-cost $s$-$t$ flow, minimum $s$-$t$ cut, and so on, in both edge- and vertex-capacitated settings (see e.g. \cite{fast-vertex-sparsest, APSP-previous}). Until recently, some of these algorithms provided the best available guarantees.  Additionally, by combining the same techniques with the ideas of \cite{Madry10_stoc} and the standard Ball Growing technique of \cite{LR,GVY}, we can essentially reduce the \MMF and \MM problems in unit-capacity graphs to decremental \APSP. Indeed, a recent algorithm for \APSP by \cite{APSP-previous} has led to fast deterministic algorithms for  \MMF and \MM in unit-capacity graphs:
 for any $\Theta(1/\log\log n)<\eps<1$, the algorithms achieve approximation factor $(\log m)^{2^{O(1/\eps)}}$, with running time $ O\left (m^{1+O(\eps)}(\log m)^{2^{O(1/\eps)}}+k/\eps\right )$, where $m$ is the number of edges in the input graph and $k$ is the number of the demand pairs. Most likely any further improvements in the current guarantees for  decremental \APSP will immediately lead to improved bounds for both these problems.
For comparison, the fastest previous approximation algorithms for \MMF, achieving $(1+\eps)$-approximation, had running times $O(k^{O(1)}\cdot m^{4/3}/\eps^{O(1)})$ \cite{KelnerMP12} and $\Otilde(mn/\eps^2)$ \cite{Madry10_stoc}, and we are not aware of any algorithms that achieve a faster running time with possibly worse approximation factors. 

We now turn to discuss the \APSP problem in more detail.
In addition to the approximation factor that the algorithm achieves and its total update time (the time it takes to maintain its data structures), two other parameters of interest are query time (the time the algorithm takes to respond to shortest path query), and whether the algorithm can withstand an adaptive adversary. The latter means that the input sequence of edge deletions may depend on the responses to the queries that the algorithm returned so far, and even on the inner state of the algorithm. This is in contrast to the oblivious-adversary setting, where the input sequence of edge deletions is fixed in advance. We note that a deterministic algorithm by definition can withstand an adaptive adversary. With four different parameters of interest to optimize, there is a vast amount of research achieving different tradeoffs between them; we will not attempt to present them all. Instead we will focus on a specific setting, where the algorithm must withstand an adaptive adversary (and is ideally deterministic), and the query time for shortest path query is bounded by $\tilde O(|E(P)|)$, where $P$ is the path returned; note that this is close to the best possible query time. 
Subject to these two constrains, we are interested in optimizing the tradeoff between the approximation factor that the algorithm achieves and its total update time. 
To the best of our knowledge, the best current algorithm for \APSP in this setting, due to \cite{APSP-previous}, is a deterministic algorithm that  achieves approximation factor $(\log m)^{2^{O(1/\eps)}}$, with total update time $O\left (m^{1+O(\eps)}\cdot (\log m)^{O(1/\eps^2)}\cdot \log L\right )$, for any $\Omega(1/\log\log m)\leq \eps< 1$; here, $L$ is the ratio of longest to shortest edge length. 
Another recent work\footnote{To the best of our knowledge, the two works are independent} \cite{bernstein2022deterministic} mostly focused on a special case of \APSP called Single-Source Shortest Path, but they also obtained a deterministic algorithm for decremental \APSP with approximation factor $m^{o(1)}$ and total update time $O(m^{1+o(1)})$; unfortunately, the tradeoff between the approximation factor and the total update time is not stated explicitly, though it is mentioned that the approximation factor is super-logarithmic.
Until recently, the best negative results only ruled out obtaining a better than factor-$4$ approximation in time $O(n^{3-\delta})$ and query time $O(n^{1-\delta})$, for any constant $0<\delta<1$, under either the Boolean Matrix Multiplication, or the Online Boolean Matrix-Vertex Multiplication conjectures \cite{DorHZ00,HenzingerKNS15}. 
A very recent result of \cite{dalirooyfard2022approximation} showed that, for any integer $k\geq 1$, under the combinatorial $k$-clique hypothesis, there is no combinatorial algorithm for \APSP in static unweighted undirected graphs, that achieves an approximation ratio better than $(1+1/k-\eps)$, and has running time  $O(m^{2-2/(k+1)}\cdot n^{1/(k+1)-\eps})$, even if approximate shortest-path queries are restricted to a specific collection of $n$ vertex pairs. The paper provides several other lower bound results with constant approximation factors, that are based on various conjectures.
Another very recent result of \cite{abboud2022hardness} provided new lower bounds for the dynamic \APSP problem, in the regime where only approximate distance queries need to be supported, under either the 3-SUM conjecture or the APSP conjecture. Let $k\geq 4$ be an integer, let $\eps,\delta>0$ be parameters, and let $c=\frac{4}{3-\omega}$ and $d=\frac{2\omega-2}{3-\omega}$, where $\omega$ is the exponent of matrix multiplication. Then \cite{abboud2022hardness}  show that, assuming either the 3-SUM Conjecture or the APSP Conjecture, there is no $(k-\delta)$-approximation algorithm for decremental \APSP with total update time  $O(m^{1+\frac{1}{ck-d}-\eps})$ and query time for \distquery bounded by $O(m^{\frac{1}{ck-d}-\eps})$. They also show that there is no $(k-\delta)$-approximation algorithm for fully dynamic \APSP that has   $O(n^3)$ preprocessing time, and then supports (fully dynamic) updates and \distquery queries in $O(m^{\frac{1}{ck-d}-\eps})$ time. 
Lastly, a very recent result of \cite{abboud2022stronger} shows that, under the 3-SUM conjecture, for any integer constant $k\geq 2$, there is
no approximate distance oracle for sparse graphs (in which $m=O(n)$) with stretch $k$, preprocessing time $\tilde O(m^{1+p})$ and query
time $\tilde O(m^{1+q})$, for all $p,q$ with $kp + (k + 1)q < 1$.
To summarize, current negative results do not rule out any superconstant approximation for decremental \APSP with total update time $m^{1+o(1)}$.

A very interesting special case of decremental \APSP is decremental \APSP in expanders. In this problem, the input graph $G$ has unit edge lengths, and it is initially an expander. 
It is well known that, if an $n$-vertex graph $G$ with maximum vertex degree $\Delta$ is a $\phi$-expander, then for any pair $x,y$ of its vertices, there is a path of length at most $O \left (\frac{\Delta\log n}{\phi}\right )$ connecting them. 
Assume now that we are given an $n$-vertex $\phi$-expander $G$ as above, that undergoes a sequence of edge deletions. 
We would like to design an algorithm that supports   short-path queries: given a pair $x,y$ of vertices of $G$, return a path connecting $x$ to $y$ in the current graph $G$, whose length is at most  $\alpha\cdot \frac{\Delta\log n}{\phi}$, where we refer to $\alpha$ as the approximation factor of the algorithm\footnote{Note that $\alpha$ is not necessarily an approximation factor strictly speaking, as it is possible that for a pair $x,y$ of vertices queried, $\dist_G(x,y)\ll \frac{\Delta\log n}{\phi}$. However, the approximation factors that we discuss are quite high, making this difference insignificant, and it is convenient for us to use the "approximation factor" notion for brevity.}. As graph $G$ undergoes edge deletions, it may no longer remain an expander, and  distances between some of its vertices may grow significantly. In order to overcome this difficulty, we only ask that the algorithm maintains a large enough subset $S(G)\subseteq V(G)$ of \emph{supported vertices}, and we restrict short-path queries to pairs of vertices in $S(G)$. We additionally require that the set $S(G)$ of supported vertices is decremental, that is, vertices may leave $S(G)$ over the course of the algorithm, but they may not join it. In its typical applications, the problem needs to be solved on expander graphs that arise from the \CMG; these are usually expanders with maximum vertex degree $\Delta\leq O(\log^2n)$ and expansion $\phi=\Theta(1)$.

Decremental \APSP in expanders is especially interesting for several reasons. First, it seems to be a relatively simple special case of  \APSP, and, if our goal is to obtain better algorithms for general \APSP, solving the problem in expander graphs is a natural first step. 
 Second, this problem arises in various algorithms for \emph{static} cut and flow problems, and seems to be intimately connected to efficient implementations of the \CMG. Third, expander graphs are increasingly becoming a central tool for designing algorithms for various dynamic graph problems, and obtaining good algorithms for \APSP in expanders will likely become a powerful tool in the toolkit of algorithmic techniques in this area. 
As such, we feel that it is crucial to obtain a good understanding of this problem. 

The best previous algorithm for \APSP in expanders due to  \cite{APSP-old} (see also \cite{APSP-previous}), uses techniques  similar to those in \cite{detbalanced}, and provides the following guarantees\footnote{The algorithm from \cite{APSP-old} was only analyzed for a specific setting of the parameters; a proof for the whole range of the parameters was provided in \cite{APSP-previous}.}.
Suppose we are given an $n$-vertex and $m$-edge graph $G$ with maximum vertex degree $\Delta$ that is a $\phi$-expander, which undergoes a sequence of at most $O\left (\frac{\phi m}{\Delta}\right)$  edge deletions, and a precision parameter $\eps$. The algorithm maintains a set $U$ of vertices of $G$, that is incremental: that is, vertices may be added to $U$ but not deleted from it. For every integer $t$, after $t$ edges are deleted from $G$, we are guaranteed that $|U|\leq O(t\Delta/\phi)$ holds. Throughout the algorithm, we let $S=V(G)\setminus U$ be the set of supported vertices. The algorithm supports short-path queries between pairs of vertices in $S$. Given such a pair $x,y\in S$, it returns a  path $P$ in $G[S]$ of length at most $O\left( \Delta^2(\log n)^{O(1/\eps^2)}/\phi\right )$, with query time $O(|E(P)|)$. 
The total update time of the algorithm is $O\left (n^{1+O(\eps)}\Delta^7(\log n)^{O(1/\eps^2)}/\phi^5\right ) $.	Assuming a typical setting where $\Delta\leq \poly\log n$, $\phi\geq (1/\poly\log n)$, and $\eps\geq 1/(\log n)^{1/3}$, the algorithm achieves approximation factor $(\log n)^{O(1/\eps^2)}$, with total update time $O\left (n^{1+O(\eps)}\right ) $. For example, if we wish to achieve a polylogarithmic approximation factor, then the total update time of the algorithm is only bounded by $O(n^{1+\delta})$ for some constant $\delta$, and if we would like the total update time of the algorithm to be bounded by $n^{1+o(1)}$, then the approximation factor must be super-polylogarithmic, that is, $(\log n)^{\omega(1)}$.

We use the algorithmic tools that we developed for well-connected graphs in order to design a deterministic algorithm for \APSP in expanders that achieves a better tradeoff between the approximation factor and total update time: our algorithm, when responding to short-path query between a pair $x,y\in S$ of vertices is guaranteed to return a path $P$ connecting $x$ to $y$, of length at most 
$\frac{2^{O(1/\eps^6)}\cdot \Delta^2\cdot\log n}{\phi}$, with query time $O(|E(P)|)$. 
The total update time of the algorithm is  $O\left(\frac{m^{1+O(\eps)}\cdot \Delta^5}{\phi^2}\right )$. If we consider again the setting where $\Delta\leq \poly\log n$ and $\phi\geq 1/\poly\log n$, the algorithm achieves approximation factor $O(\poly\log n)$ with total update time $n^{1+o(1)}$. On the negative side, our algorithm can only withstand a somewhat shorter sequence of edge deletions: at most $O\left (n\cdot \phi^2/\Delta^4\right )$ edges may be deleted, and it only guarantees that, after $t$ edge deletions from $G$, $|U|\leq O\left (\Delta^4 t/\phi^2\right )$. However, since the algorithm is typically used in the setting where $\Delta,1/\phi\leq O(\poly\log n)$, these drawbacks are generally insignificant. Another difference is that the path returned in response to a query by our algorithm is guaranteed to be contained in the current graph $G$, while the paths returned by the algorithm of \cite{APSP-old,APSP-previous} is contained in $G\setminus U$. We are not aware of any negative implications of this difference.

Returning to the \APSP problem in general graphs, the best current deterministic algorithm of \cite{APSP-previous} uses \APSP in expanders as its building block. As mentioned already, the algorithm of \cite{APSP-previous} achieves approximation factor  $(\log m)^{2^{O(1/\eps)}}$, with total update time $O\left (m^{1+O(\eps)}\cdot (\log m)^{O(1/\eps^2)}\cdot \log L\right )$. It seems conceivable that the techniques from \cite{APSP-previous} can be used in order to improve the approximation factor to $(\log m)^{O(1/\poly(\eps))}$ with similar total update time, but there are significant obstacles to further improvements. The first such obstacle is that the algorithm relies on the best previous algorithm for \APSP in expanders, in the setting where $\Delta\leq \poly\log n$ and $\phi\geq 1/\poly\log n$, whose approximation factor is $(\log n)^{O(1/\eps^2)}$  with total update time $O\left (n^{1+O(\eps)}\right ) $. Our new algorithm removes this obstacle. The second obstacle is that the algorithm from \cite{APSP-previous} is recursive. The number of recursive levels is $O(1/\eps)$, and in each recursive level, a factor $(\log n)^{O(1/\eps^2)}$ is lost in approximation (due to the algorithm for \APSP in expanders). Even when using our new algorithm for \APSP in expanders, at least a $\poly\log n$ factor must be lost in each recursive level, resulting in a $(\log n)^{\Theta(1/\eps)}$ approximation factor. One of the main reasons for this polylogarithmic loss  in each recursive level is that, when we rely on \APSP in expanders, we are committing ourselves to at least a logarithmic loss in the approximation factor. The reason is that we typically only require that, in response to short-path query, the algorithm returns a path whose length is within factor $\alpha$ of $\frac{\Delta\log n}{ \phi}$, a quantity that bounds the diameter of the expander, even if the two queried vertices are very close to each other. Furthermore, one of the typical ways to exploit expanders is to first embed a large expander into the given input graph $G$, for example, using the \CMG, and such an embedding typically does not preserve distances between vertices, except with a $\poly\log n$ distortion. We note that well-connected graphs do not suffer from this drawback.

A very recent follow-up work of \cite{APSP-new} provides a deterministic algorithm for fully-dynamic \APSP, that, for a given precision parameter $\eps$, achieves approximation $(\log\log n)^{2^{O(1/\eps^3)}}$, and has amortized update time $O(n^{\eps}\log L)$ per operation, where $L$ is the ratio of longest to shortest edge length, if the initial graph has no edges. Their work improves the algorithm of \cite{APSP-previous}, by first improving the approximation factor to $(\log\log n)^{2^{O(1/\eps^3)}}$, and then extending the result to fully dynamic graphs. The improvement in the approximation factor requires overcoming several major obstacles, one of which arises from the use of expander graphs, as described above. By replacing expander graphs with well-connected graphs, and exploiting the machinery introduced in this work, \cite{APSP-new} overcome this obstacle.

\paragraph{Organization.}
For convenience, we provide a formal statement of our main results, and a high-level overview of our techniques in \Cref{sec: overview}. We provide preliminaries in \Cref{sec: prelims}. In \Cref{sec: distanced matching game} we formally define the \DMG, and prove the upper bound on its number of iterations. 
We formally define Hierarchical Support Structure in \Cref{sec: HSS}, and provide an algorithm for the Distancing Player in \Cref{sec: HSS algorithm}. 
 We provide algorithms for decremental \APSP in well-connected graphs and in expanders in Sections \ref{sec: APSP} and \ref{sec: APSP on expanders} respectively. 
 We present our algorithm for advanced path peeling in \Cref{sec: advanced path peeling}. 
 Our algorithm for the Cut Player in the \CMG is presented in \Cref{sec: cut player}. Finally, we present our algorithms for \sparsest, \lowcond, \mbc, and \expanderdec in \Cref{sec: applications}.

\section{Overview of Our Results and Techniques}
\label{sec: overview}
\subsection{The Distanced Matching Game and Related Algorithmic Toolkit.}

Let $G$ be an $n$-vertex graph, and let $d>0$ and $1<\delta<1$ be parameters. A $(\delta,d)$-\emph{distancing} in  $G$ is a triple $(A,B,E')$, where $A,B$ are disjoint equal-cardinality subsets of vertices of $G$ with $|A|\geq n^{1-\delta}$, and $E'$ is a subset of edges of $G$ with $|E'|\leq |A|/16$. We require that, in graph $G\setminus E'$, the smallest distance between a vertex of $A$ and a vertex of $B$ is at least $d$.

We introduce the \DMG, that is played between two players: a Distancing Player and a Matching Player. The game can be thought of as an analogue of the \CMG for well-connected graphs. The input to the game consists of an integral parameter $n$, and two additional parameters, $0<\delta<1$ and $d$. The game starts with a graph $H$ that contains $n$ vertices and no edges, and then proceeds in iterations. In every iteration some edges are inserted into $H$. In order to execute the $i$th iteration, the Distancing Player must provide a $(\delta,d)$-distancing $(A_i,B_i,E'_i)$ in the current graph $H$. The matching player must return a matching $M_i\subseteq A_i\times B_i$ of cardinality at least $|A_i|/8$. The matching cannot contain any pairs of vertices $(x,y)$ for which an edge $(x,y)$ lies in $E_i'$. We then add the edges of $M_i$ to $H$, and continue to the next iteration. The game terminates when the distancing player can no longer compute a $(\delta,d)$-distancing, though we may choose to terminate it earlier, if graph $H$ has desired properties. Our first result bounds the number of iterations in the \DMG:

\begin{theorem}\label{thm: intro distancing-matching game - number of iterations}
	Consider a \DMG with parameters $n>0,0<\delta<1/4$  and $d$, such that $d\geq 2^{4/\delta}$ and $n^{\delta}\geq \frac{2^{14}\log n}{\delta^2}$.
	Then the number of iterations in the game is at most $n^{8\delta}$.
\end{theorem}

This theorem can be thought of as an analogue of similar results of \cite{KRV, KhandekarKOV2007cut}, that bound the number of iterations in the \CMG. The bound on the number of iteration that we obtain here is significantly higher that those for the \CMG, which are typically bounded by $O(\poly\log n)$. However, as we show below, we gain in other aspects -- specifically, by constructing a large enough subset $S$ of vertices of $H$, so that $H$ is well-connected with respect to $S$. This in turn allows us to achieve significantly shorter distances between the vertices of $S$ in $H$, than those guaranteed in expander graphs. 

Our proof of \Cref{thm: intro distancing-matching game - number of iterations} is completely different from the types of arguments that were used in order to bound the number of iterations in the \CMG by \cite{KRV, KhandekarKOV2007cut}. For all $i>0$, let $H_i$ denote graph $H$ at the beginning of iteration $i$. Let $E'=\bigcup_iE'_i$, and, for all $i>0$, let $H'_i=H_i\setminus E'$. We observe how the graphs $H'_1,H'_2,\ldots$ evolve over the course of the execution of the game (note that the set $E'$ of edges is computed in hindsight, after the game terminates, so in a sense we ``replay'' the game to observe the evolution of these graphs). For all $i$, we define a partition $\cset^i$ of the vertices of $H_i$ into clusters. We ensure that the set $\cset^{i+1}$ of clusters can only be obtained from set $\cset^i$ by merging existing clusters. We say that a cluster $C\in \cset^i$ belongs to level $j$, if $n^{\delta j}<|V(C)|\leq  n^{\delta(j+1)}$. We also ensure that, for all $j$, the diameter of every level-$j$ cluster in $\cset_i$ is at most $2^{O(j)}$. If $C\in \cset^i$ is a level-$j$ cluster, then we say that all vertices of $C$ lie at level $j$. If a vertex of $H$ lies at level $j$ of clustering $\cset^i$, and at a level $j'>j$ of clustering $\cset^{i+1}$, then we say that vertex $v$ has been \emph{promoted} over the course of iteration $i$. The key in the proof is to show that, once $\ceil{n^{4\delta}}$ iterations pass, a large number of vertices are promoted. Since every vertex may only be promoted at most $O(1/\delta)$ times, this is sufficient in order to bound the total number of iterations in the game.

Next, we define a Hierarchical Support Structure. The structure uses two main parameters: the base parameter $N>0$, and a level parameter $j>0$. We also assume that we are given a precision parameter $0<\eps<1$. The notion of Hierarchical Support Structure is defined inductively, using the level parameter $j$. If $H$ is a graph containing $N$ vertices, then a level-1 support structure for $H$ simply consists of a set $S(H)$ of vertices of $H$, with $|V(H)\setminus S(H)|\leq N^{1-\eps^4}$. Assume now that we are given a graph $H$ containing exactly $N^j$ vertices. A level-$j$ Hierarchical Support Structure for $H$ consists of a collection $\hset=\set{H_1,\ldots,H_r}$ of $r=N-\ceil{2N^{1-\eps^4}}$ graphs, such that for all $1\leq i\leq r$, $V(H_i)\subseteq V(H)$, and $V(H_1),\ldots,V(H_r)$ are all mutually disjoint. Additionally, it must contain, for all $1\leq i\leq r$, a level-$(j-1)$ Hierarchical Support Structure for $H_i$, which in turn must define the set $S(H_i)$ of supported vertices for graph $H_i$. We require that each such graph $H_i$ is $(\eta_{j-1},\td_{j-1})$-well-connected with respect to $S(H_i)$, where $\tilde d_{j-1}=2^{O(j/\eps^4)}$ and $\eta_{j-1}=N^{6+O(j\eps^2)}$. Lastly, the Hierarchical Support Structure for graph $H$ must contain an embedding of graph $H'=\bigcup_{i=1}^rH_i$ into graph $H$, via path of length at most $2^{O(1/\eps^4)}$, that cause congestion at most $N^{O(\eps^2)}$.  We then set $S(H)=\bigcup_{i=1}^rS(H_i)$, and we view $S(H)$ as the set of supported vertices for graph $H$, that is defined by the \HSS.

We provide an algorithm for the Distancing Player of the \DMG, that either produces the desired $(\delta,d)$-distancing in the current graph $H$, or constructs a level-$\ceil{1/\eps}$ \HSS for $H$, together with a large set $S(H)$ of supported vertices, such that $H$ is well-connected with respect to $S(H)$.

\begin{theorem}\label{thm: construct HSS last level}
	There is a deterministic algorithm, whose input consists of a parameter $0<\eps<1/4$, such that $1/\eps$ is an integer, an integer $N>0$, and a graph $H$ with $|V(H)|=N^{1/\eps}$, such that $N$ is sufficiently large, so that $\frac{N^{\eps^4}}{\log N}\geq 2^{128/\eps^5}$ holds. The algorithm computes one of the following:
	
	\begin{itemize}
		\item either a $(\delta,d)$-distancing $(A,B,E')$ in  $H$, where $\delta=4\eps^3$, $d=2^{32/\eps^4}$ and $|E'|\leq \frac{|A|}{N^{\eps^3}}$; or
		\item a level-$(1/\eps)$ \HSS for $H$, such that  graph $H$ is $(\eta,\td)$-well-connected with respect to the set $S(H)$ of vertices defined by the support structure, where $\eta=N^{6+O(\eps)}$ and $\td=2^{O(1/\eps^5)}$.
	\end{itemize}
	The running time of the algorithm is bounded by: $O(|E(H)|^{1+O(\eps)})$.
\end{theorem}

We note that our definition of the \HSS ensures that $|S(H)|\geq |V(H)|\cdot \left(1-\frac{1}{N^{\Omega(\eps^4)}}\right )$.
The proof of \Cref{thm: construct HSS last level} is similar to some of the arguments from \cite{detbalanced}, and arguments used in previous algorithms for decremental \APSP in expanders by \cite{APSP-old,APSP-previous}. We prove by induction on the level $j$ that there is a deterministic algorithm, that, given a graph $H$ with $|V(H)|=N^j$, either computes a $(\delta_j,d_j)$-distancing in graph $H$ (for appropriately chosen parameters $\delta_j$ and $d_j$), or computes a level-$j$ \HSS for $H$, such that $H$ is $(\eta_j,\td_j)$-well-connected with respect to the set $S(H)$ of vertices defined by the support structure, for appropriately chosen parameters $\eta_j,\td_j$. The algorithm for level $j$ proceeds as follows. We partition the vertices of $H$ into $N$ subsets $V_1,\ldots,V_N$, containing $N^{j-1}$ vertices each. We then let $\hset'=\set{H_1,\ldots,H_N}$ be an initial collection of graphs, where for all $1\leq i\leq N$, $V(H_i)=V_i$ and $E(H_i)=\emptyset$. We run the \DMG on all of the graphs of $\hset'$ in parallel, with the level parameter $(j-1)$; the algorithm for the Distancing Player is obtained from the induction hypothesis for level $(j-1)$. The algorithm for the Matching Player performs a routing in graph $H$ via basic path peeling, and is very similar to the algorithm employed together with \CMG in numerous previous results, e.g. \cite{fast-vertex-sparsest,APSP-old,detbalanced,APSP-previous}. If we successfully complete the \DMG on at least $r'=\Omega(r)$ graphs of $\hset$, that we denote by $\hset'=\set{H_{i_1},\ldots,H_{i_{r'}}}$, then we simultaneously obtain an embedding of graph $\bigcup_{z=1}^{r'}H_{i_z}$ into $H$, and also a guarantee that each graph $H_{i_z}\in \hset'$ is well-connected with respect to the corresponding set $S(H_{i_z})$ of vertices that is defined by its \HSS that the algorithm constructed  . We then attempt to connect, for every pair $1\leq z<z'\leq r'$ of indices, the sets $S(H_{i_z}),S(H_{i_{z'}})$ of vertices by many paths in graph $H$, so that the paths are sufficiently short and cause a low congestion. If we manage to do so for many such pairs $z,z'$ of indices, then we will obtain a collection $\hset''\subseteq \hset'$ of $r$ graphs, and a certificate that graph $H$ is well-connected with respect to the set $S(H)=\bigcup_{H_i\in \hset''}^rS(H_{i})$ of vertices. Otherwise, we will compute the required $(\delta_j,d_j)$-distancing in graph $H$. Lastly, if we fail to complete the \DMG on many of the graphs in $\hset$, then we will also compute the required $(\delta_j,d_j)$-distancing in graph $H$.

In all our subsequent algorithms, we will employ the \DMG with the algorithm for the Distancing Player implemented by \Cref{thm: construct HSS last level}. Therefore, when the algorithm terminates, it outputs a level-$(1/\eps)$ \HSS for the input graph $H$, together with a large set $S(H)$ of supported vertices, so that graph $H$ is well-connected with respect to $S(H)$.

Lastly, we provide an algorithm for decremental \APSP in a well-connected graph with a given \HSS. Specifically, we assume that we are given a graph $H$ that is an outcome of the \DMG, in which the Distancing Player is implemented by the algorithm from \Cref{thm: construct HSS last level}. Therefore, we are given a level-$(1/\eps)$ \HSS for $H$, together with a large set $S(H)$ of its vertices, so that $H$ is well-connected with respect to $S(H)$. We then assume that graph $H$ undergoes a sequence of edge deletions. As edges are deleted from $H$, the well-connectedness property may no longer hold, and the \HSS may be partially destroyed. Therefore, we only require that the algorithm maintains a large enough subset $S'(H)\subseteq S(H)$ of supported vertices, and that it can respond to short-path queries between pairs of vertices in $S'(H)$: given a pair $x,y$ of such vertices, the algorithm needs to return a path of length at most $2^{O(1/\eps^6)}$ in the current graph $H$ connecting them. We also require that the set $S'(H)$ is \emph{decremental}, so vertices can leave this set but they may not join it. The result is summarized in the following theorem.

\begin{theorem}\label{thm: APSP in HSS full}
	There is a deterministic algorithm, whose input consists of:
	
	\begin{itemize}
		\item a parameter $0<\eps<1/400$, so that $1/\eps$ is an integer;
		\item an integral parameter $N$ that is sufficiently large, so that $\frac{N^{\eps^4}}{\log N}\geq 2^{128/\eps^6}$ holds;
		\item  a graph $H$ with $|V(H)|=N^{1/\eps}$; and
		\item a level-$(1/\eps)$ hierarchical support structure for  $H$, such that $H$ is $(\eta,\td)$-well-connected with respect to the set $S(H)$ of vertices defined by the \HSS, where $\eta$ and $\td$ are the parameters from \Cref{thm: construct HSS last level}.  
	\end{itemize} 
	Further, we assume that graph $H$ undergoes an online sequence of at most $\Lambda=|V(H)|^{1-10\eps}$ edge deletions. The algorithm maintains a set $S'(H)\subseteq S(H)$ of vertices of $H$, such that, at the beginning of the algorithm, $S'(H)=S(H)$, and over the course of the algorithm, vertices can leave $S'(H)$ but they may not join it. The algorithm ensures that $|S'(H)|\geq \frac{|V(H)|}{2^{4/\eps}}$ holds at all times, and it supports short-path queries between supported vertices: given a pair $x,y\in S'(H)$ of vertices, return a path $P$ connecting $x$ to $y$ in the current graph $H$, whose length is at most $2^{O(1/\eps^6)}$, in time $O(|E(P)|)$. The total update time of the algorithm is $O( m^{1+O(\eps)})$, where $m=\max\set{|E(H)|,|V(H)|}$. 
\end{theorem}

The algorithm for \Cref{thm: APSP in HSS full} is somewhat similar to the algorithm for \APSP in expanders from \cite{APSP-old}. Instead of proving \Cref{thm: APSP in HSS full} directly, we prove a more general theorem, that, for all $1\leq j\leq 1/\eps$, given a graph $H$ with $|V(H)|=N^j$ and a level-$j$ \HSS for $H$, such that $H$ is well-connected with respect to the set $S(H)$ of vertices defined by the \HSS, supports \APSP in $H$, as the graph undergoes a limited number of edge deletions. The proof of the theorem is by induction on $j$. In order to obtain an algorithm for  a fixed level $j$, we recursively maintain a data structure for \APSP in graphs $H_1,\ldots,H_r\in \hset$ that belong to the \HSS of graph $H$. We  also maintain, for all $1\leq i\leq r$, an Even-Shiloach Tree data structure in graph $H$, that is rooted at the vertices of $S'(H_i)$. These data structures allow us to maintain a large enough decremental set $S'(H)\subseteq \bigcup_iS'(H_i)$ of vertices, and to support short-path queries between pairs of vertices in $S'(H)$  efficiently.

We compare this algorithm to the best previous algorithm for \APSP in expanders, due to \cite{APSP-old,APSP-previous}. For  \APSP in expanders, we consider a typical setting where the maximum vertex degree is $\Delta=O(\poly\log n)$, and the expansion parameter is $\phi=\Omega(1/\poly\log n)$, where $n$ is the number of vertices in the input graph. For this setting, the algorithm of 
\cite{APSP-old,APSP-previous} could only return paths between pairs of vertices from the supported set of length at most $(\log n)^{O(1/\eps^2)}$, compared to path length $2^{O(1/\eps^6)}$ of the above algorithm. The running time of both algorithms in this setting (assuming that $\eps$ is not too small) is similar. On the negative side, our algorithm can only withstand  $n^{1-\Theta(\eps)}$ edge deletions, compared to the algorithm of \cite{APSP-old}, that can withstand up to $\Theta(m/\poly\log n)$ edge deletions. Also, the size $S(H)$ of supported vertices that the algorithm from \cite{APSP-old} is significantly larger: it is $\Omega(n)$, compared to our bound of $n/2^{O(1/\eps)}$. Interestingly, the tools that we developed here allow us to obtain better algorithms for the \APSP in expanders problem itself, as we show next.

\subsection{Decremental \APSP in Expanders.}

In the decremental \APSP in expanders problem, the input is a graph $G$, that is initially a $\phi$-expander. The graph undergoes an online sequence of edge deletions. The algorithm needs to maintain a partition $(S,U)$ of vertices of $G$ into a set $S$ of \emph{supported} vertices, and a set $U$ of \emph{unsupported} vertices. As the algorithm progresses, vertices may be moved from $S$ to $U$, but not in the opposite direction. The algorithm must support shorth-path query: given a pair $x,y\in S$ of supported vertices, return a short path $P$ connecting $x$ to $y$, in time $ O(|E(P)|)$. Ideally, we would like to ensure that the algorithm can withstand a long enough sequence of edge deletions, and that the set $S$ of supported vertices remains sufficiently large.
We prove the following theorem for decremental \APSP in expanders.

\begin{theorem}
	\label{thm: APSP on expanders main}
	There is a deterministic algorithm, whose input consists of
	an $n$-vertex graph $G$ with $|E(G)|=m$ that is a $\phi$-expander for some $0<\phi<1$, with maximum vertex degree at  most $\Delta$, and a parameter $\frac{2}{(\log n)^{1/12}}< \eps<\frac{1}{400}$, such that $1/\eps$ is an integer. We assume that graph $G$ undergoes an online sequence of at most $ \frac{n\cdot \phi^2}{2^{13} \Delta^4}$ edge deletions. The algorithm
	maintains a set $U\subseteq V(G)$ of vertices, such that, for every integer $t>0$, after $t$ edges are deleted from $G$, $|U|\leq \frac{2^{11}\Delta^4t}{\phi^2}$ holds. 
	Vertex set $U$ is incremental, so vertices may join it but they may not leave it.
	The algorithm also supports short-path query: given a pair of vertices $x,y\in V(G)\setminus U$, return an $x$-$y$  path $P$ in the current graph $G$, of length at most $\frac{2^{O(1/\eps^6)}\cdot \Delta^2\cdot\log n}{\phi}$, with query time $O(|E(P)|)$. 
	The total update time of the algorithm is  $O\left(\frac{m^{1+O(\eps)}\cdot \Delta^5}{\phi^2}\right )$.
\end{theorem}

For a typical setting where $\Delta,1/\phi=O(\poly\log n)$, the algorithm, in response to a short-path query, returns a path of length at most $2^{O(1/\eps^6)}\cdot \poly\log n$, with total update time $O(n^{1+O(\eps)})$. For the same setting, the best previous algorithm of \cite{APSP-old}, returned paths of length at most $(\log n)^{O(1/\eps^2)}$ in response to queries, and had similar total update time.  On the negative side, the algorithm of \cite{APSP-old,APSP-previous} could withstand a longer sequence of edge deletions, though in both cases it remains $\Omega(n/\poly\log n)$.  The cardinality of the set $U$ of unsupported vertices is somewhat lower in \cite{APSP-old}, though for this setting it remains in both cases $\Omega(t\cdot \poly\log n)$ after $t$ edge deletions. 
Note that, for constant-degree expanders, by letting $\eps=(1/\log\log\log n)$, we can ensure that the paths returned in response to short-path queries have length at most $(\log n)^{1+o(1)}$, and the total update time of the algorithm is $n^{1+o(1)}$.

\subsection{Advanced Path Peeling and Deterministic Algorithm for the Cut Player in the Cut-Matching Game.}

We prove the following theorem for advanced path peeling.

\begin{theorem}\label{thm: main main advanced path peeling}
	There is a deterministic algorithm, whose input consists of a connected $n$-vertex $m$-edge graph $G$,  a collection $\mset=\set{(s_1,t_1),\ldots,(s_k,t_k)}$ of pairs of vertices in $G$, such that $\mset$ is a matching, and  parameters $0<\alpha\leq 1/2$, $0<\phi<1$ and $\frac{4}{(\log n)^{1/24}}< \eps<\frac{1}{400}$. The algorithm computes one of the following:
	
	\begin{itemize}
		\item either a cut $(A,B)$ with $|E_G(A,B)|\leq \phi\cdot \min\set{|E_G(A)|,|E_G(B)|}$, and each of $A$, $B$ contains at least $\frac{\alpha k}{16}$ vertices of set $T=\set{s_1,t_1,\ldots,s_k,t_k}$; or

		\item a routing $\pset$ in $G$ of a subset $\mset'\subseteq \mset$ containing at least $(1-\alpha)k$ pairs of vertices, such that every path in $\pset$ has length at most $\frac{2^{O(1/\eps^6)}\cdot \log n}{\phi}$, and the total congestion caused by the paths in $\pset$ is at most $\frac{2^{O(1/\eps^6)}\cdot \log n}{ \phi^2}\cdot\min\set{\frac 1{\alpha},\log n}$.
	\end{itemize} 
	
	The running time of the algorithm is bounded by $O\left (\frac{m^{1+O(\eps)}}{\phi^3}\right )$.
\end{theorem}

The idea in the proof of the theorem is to attempt to embed a well-connected graph $H$, whose vertex set is $T$, into $G$, via the \DMG. We require that the embedding paths are short and cause low congestion. If we fail to do so, we will immediately obtain the desired sparse cut. Otherwise, we can rely on the algorithm for \APSP in well-connected graphs from \Cref{thm: APSP in HSS full}, together with an \EST in graph $G$ that is rooted at the set $S'(H)$ of supported vertices of $H$ that the algorithm from \Cref{thm: APSP in HSS full} maintains, in order to support approximate shortest path queries in graph $G$. We then greedily compute short paths routing pairs of vertices in $\mset$, while deleting edges that participate in too many paths from $G$. Once a large enough number of paths is routed (so the algorithm from \Cref{thm: APSP in HSS full} may no longer support short-path queries), we start the whole procedure from scratch.

Next, we provide the following deterministic algorithm for the Cut Player from the \CMG.

\begin{theorem}\label{thm: new cut player}
	There is a deterministic algorithm, that, given an $n$-vertex and $m$-edge graph $G=(V,E)$ with maximum vertex degree $\Delta$, and a parameter $\frac{2}{(\log n)^{1/25}}< \eps<\frac{1}{400}$, returns one of the following:
	
	\begin{itemize}
		\item either a cut $(A,B)$ in $G$ with $|A|,|B|\geq n/4$ and $|E_G(A,B)|\leq n/100$; or
		\item a subset $S\subseteq V$ of at least $n/2$ vertices, such that graph $G[S]$ is $\phi^*$-expander, for $\phi^*\geq \Omega\left ( \frac{1}{2^{O(1/\eps^6)}\cdot \Delta^3\cdot \log^2 n }\right )$.
	\end{itemize}
	
	The running time of the algorithm is $O\left(m^{1+O(\eps)}\cdot \Delta^7\right )$.
\end{theorem}

We note that, since the number of iterations in the \CMG is bounded by $O(\log n)$, we can assume that $\Delta\leq O(\log n)$. By setting $\eps=1/(\log\log\log n)^{1/6}$, we can then guarantee that $\phi^*\geq \frac{1}{(\log n)^{5+o(1)}}$, and the running time of the algorithm is bounded by $O(n^{1+o(1)})$. In contrast, the algorithm of \cite{detbalanced} could only achieve expansion $\phi^*\geq 1/(\log n)^{1/\eps}$ with running time $n^{1+O(\eps)}$, and so in time $n^{1+o(1)}$ it could only achieve expansion $1/(\log n)^{\omega(1)}$.

Our techniques are different from those of \cite{detbalanced}, who rely on a recursive application of the \CMG to smaller and smaller graphs. Instead, we compute a constant-degree $n$-vertex expander $H$, and then attempt to embed it into $G$ using the algorithm for advanced path peeling from \Cref{thm: main main advanced path peeling}. If we successfully embed most edges of $H$ into $G$, then, by invoking the expander pruning result of \cite{expander-pruning}, we can compute a large enough subset $X\subseteq V(G)$ of vertices, such that $G[X]$ is a $\phi^*$-expander. Otherwise, we obtain a sparse cut $(A,B)$ in $G$ with $|A|\geq |B|$. We then delete the vertices of $B$ from $G$, and repeat this procedure. The algorithm continues as long as $G$ contains at least $2n/3$ vertices. Once the number of vertices in $G$ falls below $2n/3$, if we did not successfully embed an expander into $G$ so far, then we obtain a sparse cut $(A',B')$ in $G$, where $A'$ contains all vertices that currently remain in $G$.

\subsection{Sparsest Cut and Lowest Conductance Cut.}

We prove the following result for the \sparsest and \lowcond problems.

\begin{theorem}\label{thm: sparsest and lowest cond}
	There are deterministic algorithms for the \sparsest and the \lowcond problems, that achieve a factor-$O(\log^7n \log\log n )$-approximation in time  $O\left (m^{1+o(1)}\right )$, where $n$ and $m$ are the number of vertices and edges, respectively, in the input graph.
\end{theorem}

The best previous deterministic algorithm for both problems, due to \cite{detbalanced}, achieved a factor $(\log n)^{1/\eps^2}$-approximation, in time $O(m^{1+\eps})$, for any  $\frac{\log\log n}{(\log n)^{1/2}}\leq \eps<1$.
Our algorithms for \sparsest and \lowcond are essentially identical to those of \cite{detbalanced}, except that we plug in our stronger algorithm for the Cut Player in the \CMG from \Cref{thm: new cut player} into their proof. 

As in \cite{detbalanced}, we first consider the \MBSC problem. The input to the problem is an $n$-vertex graph $G$, and a parameter $0<\phi\leq 1$. The goal is to compute a cut $(X,Y)$ in $G$ of sparsity at most $\phi$, while maximizing $\min\set{|X|,|Y|}$, that we refer to as the \emph{size of the cut}. 
An $(\alpha,\beta)$-bicriteria approximation algorithm for the problem, given parameters $0<\phi<1$ and $z\geq 1$, must either compute a cut $(X,Y)$ in $G$ of sparsity at most $\phi$ and size at least $z$; or correctly establish that every cut $(X',Y')$ whose sparsity is at most $\phi/\alpha$ has size at most $\beta\cdot z$.

The problem is a natural intermediate step for obtaining fast algorithms for \sparsest and \lowcond problems. It was first introduced independently by  \cite{NanongkaiS17} and \cite{Wulff-Nilsen17}, and has been studied extensively since (see e.g. \cite{fast-vertex-sparsest,ChangS19,detbalanced}). As observed in previous work, a fast bicriteria approximation algorithm for this problem can be obtained by employing the \CMG. In \cite{detbalanced} (see Lemma 7.3), an $(\alpha,\beta)$-bicriteria deterministic approximation algorithm was obtained for the \MBSC problem, with $\alpha=(\log n)^{O(1/\eps)}$ and $\beta=(\log n)^{O(1/\eps)}$, in time $O\left(m^{1+O(\eps)+o(1)}\cdot (\log n)^{O(1/\eps^2)}\right )$ for any $\frac 1 {c\log n}\leq \eps\leq 1$, for some fixed constant $c$.
We obtain a deterministic $(\alpha,\beta)$-bicriteria approximation algorithm with $\alpha=2^{O(1/\eps^6)}\cdot \log^7 n$ and $\beta= 2^{O(1/\eps^6)}\cdot \log^6 n$, with running time $O\left (m^{1+O(\eps)+o(1)}\right)$, for any $\eps>2/(\log n)^{1/24}$. For example, by setting $\eps=1/(\log\log\log n)^{1/6}$, we can obtain an $(\alpha,\beta)$-bicriteria approximation with $\alpha=(\log n)^{7+o(1)}$ and $\beta=(\log n)^{6+o(1)}$, and running time $m^{1+o(1)}$. In contrast, obtaining an $(\alpha,\beta)$-bicriteria approximation with $\alpha=O(\log^cn)$ and $\beta=O(\log^cn)$ for any constant $c$, using the algorithm of \cite{detbalanced} would result in a running time that can only be bounded by $m^{1+O(1/c)}$. 
Our algorithm for the \MBSC is essentially identical to that of \cite{detbalanced}, except that it uses our stronger algorithm for the Cut Player in the \CMG. Algorithms for \sparsest and \lowcond easily follow from the algorithm for \MBSC, as shown in \cite{detbalanced}.

\subsection{Minimum Balanced Cut and Expander Decomposition.}
\label{subsec: MBC and ED}

We provide a deterministic factor-$ (\log n)^{8+o(1)}$ approximation algorithm for \mbcwc problem, by proving the following theorem.

\begin{theorem}\label{thm: balanced cut high cond}
	There is a deterministic algorithm, that, given a graph $G$ with $n$ vertices and $m$ edges,
	and a parameter $0<\psi\leq 1$, computes a cut $(A,B)$ in $G$ with $|E_G(A,B)|\leq \psi\cdot (\log n)^{8+o(1)}\cdot \vol(G)$, such that one of the following holds:
	
	\begin{itemize}
		\item either $\vol_{G}(A),\vol_G(B)\ge \vol(G)/3$; or
		\item  $\vol_G(A)\geq 2\vol(G)/3$, and graph $G[A]$ has conductance at least $\psi$.
	\end{itemize}
	The running time of the algorithm is $O(m^{1+o(1)}/\psi)$.
\end{theorem}

For $\psi\geq 1/m^{o(1)}$, which is a common setting used in algorithms for expander decomposition, our running time becomes $O(m^{1+o(1)})$.
As mentioned already, \cite{detbalanced} presented a deterministic algorithm for \mbcwc, that achieves approximation factor $\alpha=(\log n)^{O(1/\eps^2)}$, in time $O(m^{1+\eps})$, for any $\frac{\log\log n}{(\log n)^{1/2}}\leq \eps<1$.  

We provide another algorithm, that can be used in low-conductance regime, whose running time does not depend on $\psi$. Unfortunately, this algorithm provides a somewhat weaker certiciate if the cut that it returns is not balanced.

\begin{theorem}\label{thm: balanced cut low cond}
	There is a deterministic algorithm, that, given a graph $G$ with  $n$ vertices and  $m$ edges,
	and a parameter $0<\psi\leq 1$, computes a cut $(A,B)$ in $G$ with $|E_G(A,B)|\leq \psi\cdot (\log n)^{8+o(1)}\cdot \vol(G)$, such that one of the following holds:
	
	\begin{itemize}
		\item either $\vol_{G}(A),\vol_G(B)\ge \vol(G)/3$; or
		\item  $\vol_G(A)\geq 2\vol(G)/3$, and for every partition $(Z,Z')$ of $A$ with $\vol_G(Z),\vol_G(Z')\geq \vol(G)/100$, $|E_G(Z,Z')|\geq \psi\cdot \vol(G)$.
	\end{itemize}
	The running time of the algorithm is $O(m^{1+o(1)})$.
\end{theorem}

The algorithm from \Cref{thm: balanced cut low cond} can be easily used to obtain a deterministic bicriteria factor-$(\log n)^{8+o(1)}$ approximation algorithm for the \mbc problem in time $O(m^{1+o(1)})$. 
Given an input graph $G$, we perform a binary search on the parameter $\psi$, until we find a value for which the 
algorithm from \Cref{thm: balanced cut low cond}, when applied to $G$ and $\psi$, returns a cut $(A,B)$ with $|E_G(A,B)|\leq \psi\cdot (\log n)^{8+o(1)}\cdot \vol(G)$ and $\vol_{G}(A),\vol_G(B)\ge \vol(G)/4$; while, if applied to $G$ and $\psi/2$, it returns a cut $(A',B')$ with $|E_G(A',B')|\leq \psi\cdot (\log n)^{8+o(1)}\cdot \vol(G)$ and $\vol_G(B')< \vol(G)/4$. 
Note that $(A,B)$ is an almost balanced cut, with $|E_G(A,B)|\leq \alpha \psi\cdot \vol(G)$, where $\alpha= (\log n)^{8+o(1)}$. Let $(A^*,B^*)$ be the optimal balanced cut, so $\vol_G(A^*),\vol_G(B^*)\geq \vol(G)/3$. We claim that $|E_G(A^*,B^*)|\geq \frac{\psi}{2}\cdot \vol(G)$. This is since cut $(A^*,B^*)$ defines a partition of the set $A'$ of vertices, that we denote by $(Z,Z')$, for which $\vol_G(Z),\vol_G(Z')\geq \vol(G)/100$ must hold. Therefore, $|E_G(A^*,B^*)|\geq |E_G(Z,Z')|\geq \frac{\psi}{2}\cdot \vol(G)$. We conclude that cut $(A,B)$ is a factor-$(2\alpha)$ bicriteria solution to instance $G$ of \mbc.

Our proofs of \Cref{thm: balanced cut high cond} and \Cref{thm: balanced cut low cond} depart from that of \cite{detbalanced}, who iteratively used the algorithm for \MBSC.  The reason is that, while we obtain significantly better guarantees for the \MBSC problem, the approximation factor is still at least polylogarithmic  in $n$. Therefore, if we follow the framework of \cite{detbalanced}, who apply the algorithm for the \MBSC over the course of $O(1/\eps)$ iterations, we will still accumulate an approximation factor that is at least as high as $(\log n)^{\Theta(1/\eps)}$. Instead, we employ the \CMG directly and iteratively. In every iteration, we either cut off a large enough subgraph of $G$ via a low-conductance cut, or we (implicitly) embed a large expander into $G$.

Lastly, we consider expander decomposition.
Recall that
an \emph{$(\delta,\psi)$-expander decomposition} of a graph $G=(V,E)$
is a partition $\Pi=\{V_{1},\dots,V_{k}\}$ of the set $V$ of vertices, such that for all $1\leq i\leq k$, the conductance of graph $G[V_i]$ is at least $\psi$, and $\sum_{i=1}^k\delta_{G}(V_{i})\le\delta\cdot\vol(G)$. We prove the following theorem. 

\begin{theorem}
	\label{thm:expander decomp} There is a deterministic algorithm, that, given a graph
	$G$ with  $n$ vertices and $m$ edges, and a parameter $0<\delta<1$, where $c$ is a large enough constant,
	computes a $\left(\delta,\psi \right)$-expander decomposition of $G$ with $\psi=\Omega\left (\frac{\delta}{(\log n)^{9+o(1)}}\right )$,
	in time $O(m^{1+o(1)}/\delta)$.
\end{theorem}

The best previous deterministic algorithm, due to \cite{detbalanced}, computes a $(\delta,\phi)$-expander decomposition with  $\phi=\Omega(\delta/(\log m)^{O(1/\eps^2)})$,
in time  $O\left ( m^{1+O(\eps)+o(1)}  \right )$.
Our algorithm  is very similar to the algorithm of \cite{detbalanced}, except that we use the algorithm from \Cref{thm: balanced cut high cond} for the \mbc problem, instead of its counterpart from \cite{detbalanced}.

\section{Preliminaries}
\label{sec: prelims}

All logarithms in this paper are to the base of $2$. All graphs are simple, undirected and unweighted, unless stated otherwise. Graphs with parallel edges are explicitly referred to as multigraph.
Throughout the paper, we use a $\tilde O(\cdot)$ notation to hide multiplicative factors that are polynomial in $\log n$, where $n$ is the number of vertices in the input graph.

We follow standard graph-theoretic notation.   Given a graph $G$ and two disjoint subsets $A,B$ of its vertices, we denote by $E_G(A,B)$ the set of all edges with one endpoint in $A$ and another in $B$, and by $E_G(A)$ the set of all edges with both endpoints in $A$. We also denote by $\delta_G(A)$ the set of all edges with exactly one endpoint in $A$. For a vertex $v\in V(G)$, we denote by $\delta_G(v)$ the set of all edges incident to $v$ in $G$, and by $\deg_G(v)$ the degree of $v$ in $G$.
We may omit the subscript $G$ when clear from context. Given a subset $S$ of vertices of $G$, we denote by $G[S]$ the subgraph of $G$ induced by $S$. We say that a subgraph $C$ of $G$ is a \emph{cluster}, if $C$ is a connected vertex-induced subgraph of $G$.

\paragraph{Matchings and Routings.}
If $G$ is a graph, and $\pset$ is a collection of paths in $G$, we say that the paths in $\pset$ cause congestion $\eta$ in $G$ if every edge $e\in E(G)$ participates in at most $\eta$ paths in $\pset$, and some edge $e\in E(G)$ participates in exactly $\eta$ such paths.

Let $G$ be a graph, and let $\mset=\set{(s_1,t_1),\ldots,(s_k,t_k)}$ be a collection of pairs of vertices of $G$. We say that $\mset$ is a \emph{matching} if every vertex $v\in V(G)$ participates in at most one pair in $\mset$, and for every pair $(s_i,t_i)\in \mset$, $s_i\neq t_i$. Note that we do not require that the pairs $(s_i,t_i)\in \mset$ correspond to edges of $G$. We say that a collection $\pset$ of paths is a \emph{routing} of the pairs in $\mset$ in graph $G$, if $|\pset|=k$, the paths in $\pset$ are simple paths that are contained in $G$, and for every pair $(s_i,t_i)\in \mset$ of vertices, there is a path $P_i\in \pset$ whose endpoints are $s_i$ and $t_i$.

Assume now that we are given a graph $G$, two disjoint sets $S,T$ of its vertices, and a collection $\pset$ of paths. We say that the paths in $\pset$ \emph{route} vertices of $S$ to vertices of $T$, or that $\pset$ is a \emph{routing} of $S$ to $T$, if $\pset=\set{P(s)\mid s\in S}$, and, for all $s\in S$, path $P(s)$ originates at vertex $s$ and terminates at some vertex of $T$. We say that $\pset$ is a \emph{one-to-one routing} of $S$ to $T$, if the endpoints of all paths in $\pset$ are distinct.


\paragraph{Embeddings of Graphs.}
Let $G$ and $X$ be two graphs with $V(X)\subseteq V(G)$. An \emph{embedding} of $X$ into $G$ is a collection $\pset=\set{P(e)\mid e\in E(X)}$ of paths in graph $G$, such that, for every edge $e=(x,y)\in E(X)$, path $P(e)$ connects vertex $x$ to vertex $y$. The \emph{congestion} of the embedding is the maximum, over all edges $e'\in E(G)$, of the number of paths in $\pset$ containing $e'$. 

Given graphs $G$ and $X$ as above, and a subset $E'\subseteq E(X)$ of edges of $X$, an embedding of $E'$ into $G$
is defined similarly: it is simply an embedding of the subgraph of $X$ induced by $E'$.

We will sometimes use a more general setting, where $V(X)\cap V(G)=\emptyset$. In this case, an embedding of $X$ into $G$ must include a mapping $\pi: V(X)\rightarrow V(G)$, where every vertex of $X$ is mapped to a distinct vertex of $G$. Additionally, it must include a collection $\pset=\set{P(e)\mid e\in E(X)}$ of paths in graph $G$, such that, for every edge $e=(x,y)\in E(X)$, path $P(e)$ connects vertex $\pi(x)$ to vertex $\pi(y)$. The congestion of this embedding is defined as before.

We will use the following easy observation.

\begin{observation}\label{obs: paths from embedding}
	There is a deterministic algorithm, whose input consists of a pair $H,G$ of graphs with $V(H)\subseteq V(G)$, an embedding $\pset$ of $H$ into $G$, so that the paths in $\pset$ have length at most $d$ each and cause congestion at most $\eta$, a collection $\Pi$ of pairs of vertices of $H$, and a collection $\qset=\set{Q(u,v)\mid (u,v)\in \Pi}$ of simple paths in $H$, such that, for every pair $(u,v)\in \Pi$ of vertices, path $Q(u,v)$ connects $u$ to $v$, the paths in $\qset$ have length at most $d'$ each, and cause congestion at most $\eta'$ in $H$. The algorithm computes a collection $\qset'=\set{Q'(u,v)\mid (u,v)\in \Pi}$ of paths in graph $G$, such that, for every pair $(u,v)\in \Pi$ of vertices, path $Q'(u,v)$ connects $u$ to $v$, the paths in $\qset'$ have length at most $d\cdot d'$ each, and cause congestion at most $\eta\cdot \eta'$ in $G$. The running time of the algorithm is at most $O\left (\min\set{|\Pi|\cdot d\cdot d', |E(G)|\cdot \eta\cdot \eta'}\right )$.
\end{observation}

\begin{proof}
	We process every pair $(u,v)\in \Pi$ one by one. When pair $(u,v)$ is processed, we consider the path $Q(u,v)\in \qset$, and we denote the sequence of edges on $Q(u,v)$ by $(e_1,e_2,\ldots,e_r)$, where $r\leq d'$. For all $1\leq i\leq r$, let $P(e_i)\in \pset$ be the path that serves as the embedding of edge $e_i$, whose length must be at most $d$. We obtain path $Q'(u,v)$ connecting $u$ to $v$ in $G$ by concatenating the paths $P(e_1),P(e_2),\ldots,P(e_r)$. It is immediate to verify that the length of path $Q'(u,v)$ is at most $d\cdot d'$. Let $\qset'=\set{Q'(u,v)\mid (u,v)\in \Pi}$  be the resulting set of paths. Consider any edge $e\in E(G)$, and let $S(e)$ be the collection of all edges $e'\in E(H)$ with $e\in P(e')$, where $P(e')\in \pset$ is the embedding path of $e'$. Then $|S(e)|\leq \eta$, and every edge $e'\in S(e)$ participates in at most $\eta'$ paths in $\qset$. Therefore, edge $e$ may participate in at most $\eta\cdot \eta'$ paths in $\qset'$, and so the congestion that the paths in $\qset'$ cause in $G$ is at most $\eta\cdot \eta'$. It is immediate to verify that every pair $(u,v)\in P$ of vertices can be processed in time $O(d\cdot d')$, and so the total running time of the algorithm is at most $O(|\Pi|\cdot d\cdot d')$. Since every edge of $E(G)$ belongs to at most $\eta\cdot \eta'$ paths in $\qset'$, it is also easy to verify that the running time is bounded by $O(\eta\cdot \eta'\cdot |E(G)|)$.
\end{proof}

\paragraph{Distances and Balls.}
Given a graph $G$, for a pair $u,v\in V(G)$ of its vertices, we denote by $\dist_G(u,v)$ the \emph{distance} between $u$ and $v$ in $G$, that is, the length of the shortest path between $u$ and $v$.  For a pair $S,T$ of subsets of vertices of $G$, we define the distance between $S$ and $T$ to be $\dist_G(S,T)=\min_{s\in S,t\in T}\set{\dist_G(s,t)}$.
For a vertex $v\in V(G)$, and a vertex subset $S\subseteq V(G)$, we also define the distance between $v$ and $S$ as $\dist_G(v,S)=\min_{u\in S}\set{\dist_G(v,u)}$.
The \emph{diameter} of the graph $G$, denoted by $\diam(G)$, is the maximum distance between any pair of vertices in $G$. 
For a vertex $v\in V(G)$ and a distance parameter $D\geq 0$, we denote by $B_G(v,D)=\set{u\in V(G)\mid \dist_G(u,v)\leq D}$ the \emph{ball of radius $D$ around $v$}.
Similarly, for a subset $S\subseteq V(G)$ of vertices, we let the ball of radius $D$ around $S$ be $B_G(S,D)=\set{u\in V(G)\mid \dist_G(u,S)\leq D}$.
We will sometimes omit the subscript $G$ when clear from context.

\subsection{Dynamic Algorithms}

\paragraph{Dynamic Graphs.}
Consider a graph $G$ that undergoes an online sequence $\Sigma=(\sigma_1,\sigma_2,\ldots)$ of edge deletions, that we may also refer to as \emph{updates}. After each update operation, the algorithm will perform some updates to the data structures that it maintains. 
We refer to different ``times'' during the algorithm's execution. The algorithm starts at time $0$. For each $t\geq 0$, we refer to ``time $t$ in the algorithm's execution'' as the time immediately after all updates to the data structures maintained by the algorithm following the $t$th edge deletion $\sigma_t\in \Sigma$ are completed. When we say that some property holds at every time during the algorithm's execution, we mean that the property holds at all times $t$ of the algorithm's execution, but it may not hold, for example, during the procedure that updates the data structures maintained by the algorithm, following some edge deletion $\sigma_t\in \Sigma$.
For $t\geq 0$, we denote by $G^{(t)}$ the graph $G$ at time $t$; that is, $G^{(0)}$ is the original graph, and for $t\geq 0$, $G^{(t)}$ is the graph obtained from $G$ after the first $t$ edge deletions $\sigma_1,\ldots,\sigma_t$.

We say that a set $S$ of elements is \emph{decremental} if, once it is initialized, elements can be deleted from $S$ but they may not be added to $S$. Similarly, we say that $S$ is \emph{incremental} if elements can be added to $S$ as the time progresses, but not deleted from $S$.

\paragraph{Even-Shiloach Trees~\cite{EvenS,Dinitz,HenzingerKing}.}
Suppose we are given a graph $G=(V,E)$ with integral lengths $\ell(e)\geq 1$ on its edges $e\in E$, a source vertex $s$, and a distance bound $D\geq 1$. Even-Shiloach Tree (\EST) algorithm maintains, for every vertex $v$ with $\dist_G(s,v)\leq D$, the distance $\dist_G(s,v)$, under the deletion of edges from $G$. Moreover, it maintains a shortest-path tree $\tau$ rooted at vertex $s$, that includes all vertices $v$ with $\dist_G(s,v)\leq D$. We denote the corresponding data structure by $\EST(G,s,D)$, or just \EST when clear from context. 
The total update time of the algorithm, including the initialization and all edge deletions, is $O(m\cdot D\log n)$, where $m$ is the initial number of edges in $G$ and $n=|V|$. 

\subsection{Cuts, Flows, Sparsity, Conductance and Expanders.}

Even though all graphs that we deal with are undirected, it will sometimes be useful to assign directions to paths in such graphs. In order to do so, for a path $P$ in an undirected graph $G$, we designate one of its endpoints (say $u$) as the \emph{first endpoint} of $P$, and the other endpoint (say $v$) as its \emph{last endpoint}. We may then say that path $P$ is \emph{directed from $u$ to $v$}, or that it originates at $u$ and terminates at $v$. If $\pset$ is a collection of path in an undirected graph $G$, and we have assigned a direction to each of the paths, we may refer to $\pset$ as a \emph{collection of directed paths}, even though graph $G$ is undirected.

\paragraph{Flows.}
Let $G$ be a graph, and let $\pset$ be a collection of directed paths in graph $G$. A \emph{flow} over the set $\pset$ of paths is an assignment of non-negative values $f(P)\geq 0$, called \emph{flow-values}, to every path $P\in \pset$. We sometimes refer to paths in $\pset$ as \emph{flow-paths for flow $f$}. For each edge $e\in E(G)$, let $\pset(e)\subseteq \pset$ be the set of all paths whose first edge is $e$, and let $\pset'(e)\subseteq \pset$ be the set of all paths whose last edge is $e$. We say that edge $e$ \emph{sends $z$ flow units} in $f$ if $\sum_{P\in \pset(e)}f(e)=z$, and we say that edge $e$ \emph{receives $z$ flow units} in $f$ if $\sum_{P\in \pset'(e)}f(P)=z$. Similarly, for a vertex $v\in V(G)$, we say that $v$ sends $z$ flow units in $f$ if the sum of flow-values of all paths $P\in \pset$ that originate at $v$ is $z$. We say that $v$ receives $z$ flow units in $f$ if the sum of the flow-values of all paths $P\in \pset$ that terminate at $v$ is $z$. The \emph{congestion} that flow $f$ causes on an edge $e$ is $\sum_{\stackrel{P\in \pset:}{e\in E(P)}}f(P)$, and the \emph{total congestion} of the flow $f$ is the maximum congestion that it causes on any edge $e\in E(G)$.

\paragraph{Cuts and Expansion.}
Given a graph $G=(V,E)$, a \emph{cut} in $G$ is a bipartition $(A,B)$ of the set $V$ of its vertices, with $A,B\neq \emptyset$. The \emph{sparsity} of the cut $(A,B)$ is $\phi_G(A,B)=\frac{|E_G(A,B)|}{\min \set{|A|,|B|}}$. We denote by $\Phi(G)$ the smallest sparsity of any cut in $G$, and we refer to $\Phi(G)$ as the \emph{expansion} of $G$.

\paragraph{Expanders.}
We define the notion of expanders using graph expansion.

\begin{definition}[Expander]
	We say that a graph $G$ is a $\phi$-expander, for a parameter $0<\phi<1$, if $\Phi(G)\geq \phi$.
\end{definition}

We will sometimes informally say that graph $G$ is an \emph{expander} if $\Phi(G)$ is a constant independent of $|V(G)|$. 
We use the following immediate observation, that was also used in previous works, (see e.g. Observation 2.3 in \cite{detbalanced}). 

\begin{observation}\label{obs: exp plus matching is exp}
	Let $G=(V,E)$ be an $n$-vertex graph that is a $\phi$-expander, and let $G'$ be another graph that is obtained from $G$ by adding to it a new set $V'$ of at most $n$ vertices, and a matching $M$, connecting every vertex of $V'$ to a distinct vertex of $G$. Then $G'$ is a $\phi/2$-expander.
\end{observation}

We also use the following theorem that provides a fast algorithm for an explicit construction of an expander, that is based on the results of Margulis \cite{Margulis} and Gabber and Galil \cite{GabberG81}. The proof was shown in \cite{detbalanced}.
\begin{theorem}[Theorem 2.4 in \cite{detbalanced}]
	\label{thm:explicit expander}
	There is a constant $\alpha_0 > 0$ and a deterministic algorithm, that we call \constructexpander, that, given an integer $n>1$, in time $O(n)$ constructs a graph $H_n$ with $|V(H_n)|=n$, such that $H_n$ is an $\alpha_0$-expander, and every vertex in $H_n$ has degree at most $9$.
\end{theorem}

\paragraph{Expander Pruning.}

We use an algorithm for expander pruning by \cite{expander-pruning}. We slightly rephrase it so it is defined in terms of graph expansion, instead of conductance that was used in the original paper. This variation of the original expander pruning theorem of~\cite{expander-pruning} was proved explicitly in \cite{APSP-previous} (see Theorem 2.2 in full version of the paper).

\begin{theorem}[Adaptation of Theorem 1.3 in~\cite{expander-pruning}; see Theorem 2.2 in \cite{APSP-previous}]\label{thm: expander pruning}
	There is a deterministic algorithm, that, given an access to the adjacency list of a graph $G$ that is a $\phi$-expander, for some parameter $0<\phi<1$, such that the maximum vertex degree in $G$ is at most $\Delta$, and a sequence $\Sigma=(e_1,e_2,\ldots,e_k)$ of $k\leq \frac{\phi |E(G)|}{10\Delta}$ online edge deletions from $G$, maintains a set $\tilde U\subseteq V(G)$ of vertices, with the following properties. Let $G^{(i)}$ denote the graph $G\setminus\set{e_1,\ldots,e_i}$; let $\tilde U_0=\emptyset$ be the set $\tilde U$ at the beginning of the algorithm, and for all $0<i\leq k$, let $\tilde U_i$ be the set $\tilde U$ after the deletion of the edges of $e_1,\ldots,e_i$ from graph $G$. Then, for all $1\leq i\leq k$: $\tilde U_{i-1}\subseteq \tilde U_i$;
		 $ |\tilde U_i|\leq \frac{8i\Delta}{\phi}$; and
		 graph $G^{(i)}\setminus\tilde U_i$ is a $\frac{\phi}{6\Delta}$-expander.
	The total running time of the algorithm is $\Otilde(k\Delta^2/\phi^2)$.
\end{theorem}

\paragraph{Graph Conductance.}
For a graph $G=(V,E)$ and a subset $S\subseteq V$ of its vertices, the \emph{volume} of $S$ is $\vol_G(S)=\sum_{v\in S}\deg_G(v)$. We denote by $\vol(G)=\vol_G(V)$.
The \emph{conductance} of a cut $(A,B)$ in $G$ is: $\psi_G(A,B)=\frac{|E_G(A,B)|}{\min\set{\vol_G(A),\vol_G(B)}}$. We denote by $\Psi(G)$ the smallest conductance of any cut in $G$, and we refer to $\Psi(G)$ as the \emph{conductance} of $G$.

\subsection{Embeddings with Fake Edges and Expansion.}

Typically, when using the Cut-Matching game, we either embed an expander graph $H$ with $V(H)=V(G)$ into the given graph $G$, or compute a sparse cut $(A,B)$ in $G$. Unfortunately, it is possible that one side of the cut, say $A$, is quite small in the latter case. This often poses challenges in applications of the Cut-Matching game where the goal is to obtain very efficient algorithms. This is since we essentially spend time  $\Omega(|E(G)|)$ in order to execute the Cut-Matching game, and end up computing a sparse cut whose one side may be very small. Ideally, for efficient algorithms, it is desirable that the sparse cut that we compute is as balanced as possible. A standard way to overcome this issue, that was suggested in the original paper of \cite{KRV} that introduced the Cut-Matching game, is to use \emph{fake edges}. Intuitively, we will augment the graph $G$ with a small number of edges, that we refer to as fake edges, to indicate that they do not actually lie in $G$. If $F$ is the set of fake edges, we will denote by $G+F$ the graph obtained by adding the edges of $F$ into $G$. We will use the Cut-Matching game to either compute a sparse cut in $G$, whose both sides are relatively large; or to compute an embedding of some expander graph $H$ into $G+F$. In the latter case, both the embedding and the set $F$ of fake edges are constructed during the Cut-Matching game, and we will then extract a large expander graph from $G$. 
%
%
%
The following lemma from \cite{detbalanced} provides an algorithm to extract a large expander graph from $G$ efficiently.

\begin{lemma}[Lemma 2.9 from \cite{detbalanced}]\label{lem: embedding expander w fake edges gives expander}
	Let $G$ be an $n$-vertex graph, and let $H$ be another graph with $V(H)= V(G)$, with maximum vertex degree $\Delta_H$, such that $H$ is a $\psi$-expander, for some $0<\psi<1$. Let $F$ be any set of $k$ fake edges for $G$, and let $\Delta_G$ be the maximum vertex degree in $G+F$. Assume that there exists an embedding $\pset=\set{P(e)\mid e\in E(H)}$ of $H$ into $G+F$, that causes congestion at most $\cong$, for some $\cong\geq 1$. Assume further that $k\leq \frac{\psi n}{32\Delta_G\cong}$. Then there is a subgraph $G'\subseteq G$ that is a $\psi'$-expander, for $\psi'\geq \frac{\psi}{6\Delta_G\cdot\cong}$, such that, if we denote by $A=V(G')$ and $B=V(G)\setminus A$, then  $|A|\geq n-\frac{4k\cong}{\psi}$ and $|E_G(A,B)|\leq 4k$. Moreover, there is a deterministic algorithm, that we call \extractexpander, that, given $G,H,\pset$ and $F$, computes such a graph $G'$ in time $\tilde O(|E(G)|\Delta_G\cdot\cong/\psi)$.
\end{lemma}

\subsection{The Cut-Matching Game.}
\label{subsec:KKOV}

The \CMG was introduced by Khandekar, Rao, and Vazirani \cite{KRV} as a tool for obtaining fast approximation algorithms for the \sparsest and \mbc problems. We describe here a variant of this game, that was introduced by Khandekar et al. \cite{KhandekarKOV2007cut}, and later slightly modified by \cite{detbalanced}. The game is played between two players, the \emph{Cut Player}, and the \emph{Matching Player}. The game uses a parameter $n$, which is an even integer. The purpose of the game is to construct an $n$-vertex expander graph $H$. At the beginning of the game, graph $H$ contains $n$ vertices and no edges, and then in every iteration some edges are added to $H$. For intuition, it may be convenient to think of the Cut Player's goal being to construct the expander in as few iterations as possible, and the Matching Player's goal as trying to delay the construction of the expander. 

The game starts with graph $H$ containing $n$ vertices and no edges. The $i$th iteration is played as follows. The Cut Player either computes a partition $(A_i,B_i)$ of $V(H)$ with $|A_i|,|B_i|\geq n/4$ and $|E_H(A_i,B_i)|\leq n/100$; or it computes a set $X\subseteq V(H)$ of vertices with $|X|\geq n/2$, such that $H[X]$ is a $\phi$-expander, for some expansion parameter $0<\phi<1$. Assume first that the former happens, and  assume without loss of generality that $|A_i|\leq |B_i|$. The Matching Player must compute any partition $(A_i',B_i')$ of $V(H)$ with $|A'_i|=|B'_i|$, such that $A_i\subseteq A_i'$, and then it must compute an arbitrary perfect matching $M_i$ between $A_i'$ and $B_i'$. The edges of $M_i$ are then added to the graph $H$, and the algorithm continues to the next iteration. 
If the latter case happens, that is, the Cut Player returns a set $X\subseteq V(H)$ of at least $n/2$ vertices, so that $H[X]$ is a $\phi$-expander, denote $Y=V(H)\setminus X$. The Matching Player must then compute a matching $\mset_i\subseteq X\times Y$ with $|\mset_i|=Y$. The edges of $\mset_i$ are added to graph $H$, and the algorithm terminates. In this case, from \Cref{obs: exp plus matching is exp}, we are guaranteed that $H$ is a $\phi/2$-expander. The next theorem follows directly from the result of \cite{KhandekarKOV2007cut}, and was proved explicitly in \cite{detbalanced} (see Theorem 2.5 in the full version).

\begin{theorem}
	\label{thm:KKOV-new} There is a constant $c$, such that the algorithm described above terminates after at most $c\log n$ iterations.
\end{theorem}

\subsection{Graph Cutting and Partitioning.}

We use several graph cutting and partitioning procedures, that exploit standard tools. 
In all these procedures, the input is a graph $G$, with a subset $T$ of vertices of $G$ called terminals. The goal is to either compute a single cluster, or a collection of clusters in $G$ with some specific properties. Throughout, a subgraph $C\subseteq G$, we denote by $T_C=T\cap V(C)$ the set of all terminals contained in $C$.
We start with procedure \proccut, which is a variation of Leighton and Rao's ball growing technique \cite{LR}.

\subsubsection{Procedure \proccut.}

The input to the procedure is an $n$-vertex graph $G$, a set $T\subseteq V(G)$ of $k$ vertices called terminals, a specific terminal $t_C\in T$, and distance parameters $d$ and $\Delta$.

 The procedure returns a cluster $C\subseteq G$, and a subset $\hat T_C\subseteq T$ of terminals, for which the following properties hold:

\begin{properties}{C}
	\item $T_C\subseteq \hat T_C$; \label{prop: terminal subset}
	
	\item $|\hat T_C|\leq |T_C|\cdot k^{64/\Delta}$; \label{prop: few discarded terminals}
	
	\item $V(C)\subseteq B_G(t_C, \Delta\cdot d)$; \label{prop: small diam of cluster}
	
	\item $\hat T_C\subseteq B_G(t_C,\Delta\cdot d)$; and \label{prop: small dist between terminals}
	
	\item for every pair $x\in V(C)$, $t'\in T\setminus \hat T_C$ of vertices, $\dist_G(v,t')\geq 4d$. \label{prop: separation}
\end{properties}

The following lemma summarizes Procedure \proccut.

\begin{lemma}\label{lem: proccut} There is a deterministic algorithm called \proccut, that, given an $n$-vertex graph $G$, a subset $T\subseteq V(G)$ of $k$ vertices called terminals, a specific terminal $t_C\in T$, and parameters $d,\Delta>0$, computes a cluster $C\subseteq G$ together with a set $\hat T_C\subseteq T$ of terminals, for which properties (\ref{prop: terminal subset}) -- (\ref{prop: separation}) hold. The running time of the algorithm is $O(|E(C)|\cdot n^{64/\Delta})$.
\end{lemma}

\begin{proof}	
We assume that we are given as input an $n$-vertex graph $G$, a set $T\subseteq V(G)$ of $k$ vertices called terminals, together with a specific terminal $t_C\in T$  that we denote by $t$ for simplicity, and distance parameters $d$ and $\Delta$. 
The procedure performs a breadth-first-search (BFS) from vertex $t$ in graph $G$, up to a certain depth, that will be determined later.

For all $i\geq 1$, we denote by $L_i$ the set of all vertices of $G$ that lie at distance $4(i-1)d+1$ to $4id$ from $t$ in $G$. In other words:

\[L_i=B_G(t,4id)\setminus B_G(t,4(i-1)d).   \]

We refer to the vertices of $L_i$ as \emph{layer $i$ of the BFS}. We denote by $k_i$ the number of terminals lying in $L_1\cup\cdots\cup L_i$.
We also denote by $m_i$ the total number of edges of $G$ whose both endpoints lie in $L_1\cup \cdots\cup L_i$. 
 The following definition is crucial for the description of Procedure \proccut.

\begin{definition}[Eligible Layer]
	For an integer $i>1$, we say that layer $L_i$ of the BFS is \emph{eligible} if both of the following two conditions hold:
	
	\begin{properties}{L}
		\item  $m_i\leq m_{i-1}\cdot n^{64/\Delta}$;  and \label{cut condition 1: number of edges}
		\item  $k_i\leq k_{i-1}\cdot k^{64/\Delta}$  \label{cut condition 2: terminals}
	\end{properties}
\end{definition}

We need the following claim, whose proof uses standard arguments.
\begin{claim}\label{claim: eligible layer}
	There exists an index $1<i\leq \Delta/8$, such that layer $L_i$ is eligible.
\end{claim}

\begin{proof}
Assume otherwise. Then for all $1< i\leq \Delta/8$, layer $L_i$ is ineligible. For each such index $i$, we say that layer $L_i$ is \emph{type-1 ineligible} if it violates Condition (\ref{cut condition 1: number of edges}). Otherwise, we say that it is \emph{type-2 ineligible}, in which case it must violate Condition (\ref{cut condition 2: terminals}). Since every layer $L_i$ with $1< i\leq \Delta/8$ is ineligible, either there are at least $\Delta/32$ type-1 ineligible layers $L_i$ with $1<i\leq \Delta/8$, or there are at least $\Delta/32$ type-2 ineligible layers $L_i$ with $1<i\leq \Delta/8$. We now consider each of the two cases and prove that they are impossible.

Assume first that there are at least $\Delta/32$ type-1 ineligible layers $L_i$ with $1<i\leq \Delta/8$, and denote their indices by $i_1,i_2,\ldots,i_z$, where $1<i_1\leq i_2\leq\cdots\leq i_z\leq \Delta/8$, and $z\geq \Delta/32$. But then, for all $1\leq a<z$, $m_{i_{a+1}}>m_{i_a}\cdot n^{64/\Delta}$. Therefore, $m_{i_z}>n^{64z/\Delta}\geq n^2$, a contradiction.

Assume now that  there are at least $\Delta/32$ type-2 ineligible layers $L_i$ with $1<i\leq \Delta/8$, and denote their indices by $j_1,j_2,\ldots,j_z$, where $1<j_1\leq j_2\leq\cdots\leq j_z\leq \Delta/8$, and $z\geq \Delta/32$. But then, for all $1\leq a<z$, $k_{j_{a+1}}>k_{j_a}\cdot k^{64/\Delta}$. Therefore, $k_{j_z}>k^{64z/\Delta}>k$, a contradiction.
\end{proof}

We are now ready to describe the algorithm for \proccut. The algorithm performs a BFS from the input terminal $t$ in graph $G$, until it encounters the first index $i>1$, such that layer $L_i$ is eligible. 
The algorithm then returns cluster $C$, which is a subgraph of $G$ induced by $L_1\cup L_2\cup\cdots\cup L_{i-1}$, and the set $\hat T_C$ of terminals, containing all terminals in $L_1\cup\cdots\cup L_i$.

We now show that properties (\ref{prop: terminal subset}) -- (\ref{prop: separation}) hold for this output. Properties (\ref{prop: terminal subset}), (\ref{prop: small diam of cluster}) and (\ref{prop: small dist between terminals})  follow immediately from the definition of cluster $C$ and set $\hat T_C$ of terminals, and from the fact that, from Claim \ref{claim: eligible layer}, $i\leq \Delta/8$.
Property (\ref{prop: few discarded terminals}) follows immediately from Condition (\ref{cut condition 2: terminals}) in the definition of an eligible layer, since $|T_C|=k_{i-1}$ and $|\hat T_C|=k_i$. Lastly, property (\ref{prop: separation}) follows immediately from the fact that the vertices of $C$ and the terminals of $T\setminus\hat T_C$ are separated by layer $L_i$, so every path connecting a vertex of $C$ and a terminal of $T\setminus \hat T_C$ must contain at least $4d$ edges.

Notice that the running time of the algorithm is $O(m_i)$. Since $|E(C)|=m_{i-1}$, from Condition (\ref{cut condition 1: number of edges}) of an eligible layer, we get that the running time of the algorithm is bounded by $O(m_{i-1}\cdot n^{64/\Delta})=O(|E(C)|\cdot n^{64/\Delta})$.
\end{proof}

Next, we describe a procedure called \procpart, that exploits Procedure \proccut in order to compute a number of clusters in the input graph $G$, that contain a large fraction of terminals, such that the diameter of every cluster is relatively small, but pairs of vertices lying in different clusters are sufficiently far away from each other. 

\subsubsection{Procedure \procpart.}
The input to Procedure \procpart consists of an $n$-vertex graph $G$,
a set $T\subseteq V(G)$ of $k$ vertices called terminals, and distance parameters $d$ and $\Delta$.

The output of the procedure is a collection $\cset$ of disjoint clusters of $G$, and, for every cluster $C\in \cset$, a center terminal $t_C\in T_C$, and two sets $T'_C$, $\hat T_C$ of terminals, such that the following properties hold.

\begin{properties}{R}
	\item for every cluster $C\in \cset$, $t_C\in T'_C$; $T'_C\subseteq V(C)$, and $T'_C\subseteq \hat T_C$; \label{prop: terminal subset new}
	
	\item for every cluster $C\in \cset$, $|\hat T_C|\leq |T'_C|\cdot k^{64/\Delta}$; \label{prop: few discarded terminals new}
	
	\item for every cluster $C\in \cset$,  $V(C)\subseteq B_G(t_C, \Delta\cdot d)$; \label{prop: small diam of cluster new}
	
	\item for every cluster $C\in \cset$, $\hat T_C\subseteq B_G(t_C,\Delta\cdot d)$; \label{prop: small dist between terminals new}
	
	\item for every pair $C,C'\in \cset$ of distinct clusters, for every pair $t'\in T'_C,t''\in T'_{C'}$ of terminals, $\dist_G(t',t'')\geq d$; and\label{prop: separated clusters}
	\item $\bigcup_{C\in \cset}\hat T_C=T$. \label{prop: cover all terminals}
\end{properties}

The following lemma summarizes Procedure \proccut.

\begin{lemma}\label{lem: procpar}
	There is a deterministic algorithm, called Procedure \procpart, whose input is an $n$-vertex graph $G$, a set $T\subseteq V(G)$ of $k$ vertices called terminals, and distance parameters $d$ and $\Delta$. 
	The output of the procedure is a collection $\cset$ of disjoint clusters of $G$, and, for every cluster $C\in \cset$, a center terminal $t_C\in T_C$ and sets $T'_C,\hat T_C$ of terminals, for which Properties( \ref{prop: terminal subset new})--(\ref{prop: cover all terminals}) hold. The running time of the procedure is $O(|E(G)|\cdot n^{64/\Delta})$.
\end{lemma}

\begin{proof} 
Throughout the algorithm, we maintain a set $\cset$ of disjoint clusters of $G$, and, for every cluster $C\in \cset$, we maintain a center terminal $t_C\in T_C$, and sets $T'_C,\hat T_C$ of terminals. We ensure that, throughout the algorithm, properties (\ref{prop: terminal subset new})--(\ref{prop: separated clusters}) hold. The algorithm terminates once we achieve Property (\ref{prop: cover all terminals}).

At the beginning of the algorithm, we set $\cset=\emptyset$ and $G_0=G$. We then iterate.
In iteration $j$, we add a new cluster $C_j$ to set $\cset$, and define the corresponding terminal $t_{C_j}$ and sets $T'_{C_j}$, $\hat T_{C_j}$ of terminals. We denote $G_j=G\setminus\left (V(C_1)\cup V(C_2)\cup\cdots\cup V(C_j)\right )$. As the algorithm progresses, we will also delete some terminals from the set $T$. Specifically, we set $T^{(0)}=T$, and, for all $j\geq 1$, we set $T^{(j)}=T\setminus \left(\hat T_{C_1}\cup \hat T_{C_2}\cup \cdots\cup \hat T_{C_j}\right )$. We will ensure that the following additional invariants hold at the end of iteration $j$:

\begin{properties}{I}
	\item $T^{(j)}\subseteq V(G_j)$, and for every pair $t,t'\in T^{(j)}$ of terminals, if $\dist_{G_j}(t,t')\geq 4d$, then $\dist_G(t,t')\geq 4d$.\label{invariant}
	
	\item for every pair $t,t'$ of terminals with $t\in \bigcup_{j'=1}^jT'_{C_{j'}}$ and $t'\in T^{(j)}$, $\dist_G(t,t')\geq d$.\label{invariant: separating clusters from rest}
\end{properties}

At the beginning of the algorithm, $\cset=\emptyset$, $G_0=G$, and $T^{(0)}=T$. Clearly, Properties (\ref{prop: terminal subset new})--(\ref{prop: separated clusters}) and invariants (\ref{invariant}) and (\ref{invariant: separating clusters from rest}) are satisfied. We perform iterations until Property (\ref{prop: cover all terminals}) holds. We now describe the execution of the $j$th iteration.

\paragraph{Execution of the $j$th iteration.}
We assume that we are given a set $\cset=\set{C_1,\ldots,C_{j-1}}$ of disjoint clusters, and, for every cluster $C_{j'}\in \cset$, a terminal $t_{C_{j'}}$, and sets $T'_{C_{j'}},\hat T_{C_{j'}}$ of terminals, such that  Properties (\ref{prop: terminal subset new})--(\ref{prop: separated clusters}), and Invariants (\ref{invariant}) and (\ref{invariant: separating clusters from rest}) hold. Recall that $G_{j-1}=G\setminus \left (V(C_1)\cup V(C_2)\cup\cdots\cup V(C_{j-1})\right )$ and $T^{(j-1)}=T\setminus \left(\hat T_{C_1}\cup \hat T_{C_2}\cup \cdots\cup \hat T_{C_{j-1}}\right )$. We assume that Property (\ref{prop: cover all terminals}) does not hold, so there must be at least one terminal in $T^{(j-1)}$; we let $t\in T^{(j-1)}$ be any such terminal.  

In order to execute the $j$th iteration, we apply Procedure \proccut from \Cref{lem: proccut}  to graph $G_{j-1}$, set $T^{(j-1)}$ of terminals, and terminal $t$, keeping the parameters $d$ and $\Delta$ unchanged. We denote by $C_j$ the cluster of $G_{j-1}$ that the algorithm returns, by $\hat T_{C_j}\subseteq T^{(j-1)}$ the resulting set of terminals, and by $T'_{C_j}=T^{(j-1)}\cap V(C)$. Notice that, equivalently, $T'_{C_j}=T_{C_j}\setminus \left(\hat T_{C_1}\cup \hat T_{C_2}\cup\cdots\cup \hat T_{C_{j-1}}\right )$. We also denote $t_{C_{j}}=t$, and we add $C_j$ to $\cset$, completing the iteration.

\paragraph{Analysis of the $j$th iteration.}
We now verify that all required properties hold at the end of iteration $j$, assuming that they held at the beginning of the iteration. 
Consider first Properties (\ref{prop: terminal subset new})--(\ref{prop: small dist between terminals new}). Let $C\in \cset$ be any cluster. If $C\neq C_j$, then these properties clearly continue to hold for $C$. Assume now that $C=C_j$. Properties (\ref{prop: terminal subset new}) and (\ref{prop: few discarded terminals new}) immediately follow from Properties (\ref{prop: terminal subset}) and (\ref{prop: few discarded terminals}) that are guaranteed by Procedure \proccut, and from the definition of the set $T'_C$ of terminals. Properties (\ref{prop: small diam of cluster new}) and (\ref{prop: small dist between terminals new}) similarly follow from Properties (\ref{prop: small diam of cluster}) and (\ref{prop: small dist between terminals}), since, for every pair $x,y\in V(G_j)$, $\dist_{G_{j-1}}(x,y)\geq \dist_G(x,y)$, and so $B_{G_{j-1}}(t_C,\Delta\cdot d)\subseteq B_{G}(t_C,\Delta\cdot d)$.


Consider now any pair $C_{j'},C_{j''}\in \cset$ of clusters, and a pair $t'\in T'_{C_{j'}}$, $t''\in T'_{C_{j''}}$ of terminals. If both $j',j''<j$, then we are guaranteed that $\dist_G(t',t'')\geq d$ from the fact that Property (\ref{prop: separated clusters}) held at the beginning of the iteration. Assume now w.l.o.g. that $j'<j$ and $j''=j$. In this case, $t''\in T^{(j-1)}$, and, since we have assumed that Invariant (\ref{invariant: separating clusters from rest}) held at the beginning of the iteration, we get that $\dist_G(t',t'')\geq d$, establishing Property (\ref{prop: separated clusters}).

Next, we establish Invariant (\ref{invariant}). Recall that $T^{(j)}=T\setminus \left(\hat T_1\cup\cdots\cup \hat T_{j}\right )=T^{(j-1)}\setminus \hat T_j$. Consider some terminal $t'\in T^{(j)}$, and assume for contradiction that $t'\not\in V(G_j)$. Recall that $V(G_j)=V(G)\setminus\left(V(C_1)\cup \cdots\cup V(C_j)\right )=V(G_{j-1})\setminus V(C_j)$. From Invariant (\ref{invariant}), since $t'\in T^{(j-1)}$, $t'\in V(G_{j-1})$ must hold. Since $t'\not\in V(G_j)$, it must be the case that $t'\in V(C_j)$. However, set $T'_{C_j}$ contains every terminal that lies in $V(C_j)\setminus \left(\hat T_{C_1}\cup \cdots\cup \hat T_{C_{j-1}}\right )$. Since $t'\not\in \hat T_{C_1}\cup \cdots\cup \hat T_{C_{j-1}}$, it must be the case that $t'\in T'_{C_j}$, and so $t'\in \hat T_{C_j}$. But then $t'\not\in T^{(j)}$, a contradiction.


Consider now some pair $t,t'\in T^{(j)}$ of terminals, and assume that $\dist_{G_j}(t,t')\geq 4d$. Note that, if $\dist_{G_{j-1}}(t,t')\geq 4d$, then, from the fact that Invariant (\ref{invariant}) held at the beginning of the iteration, we get that $\dist_G(t,t')\geq 4d$ must hold. We now show that 
$\dist_{G_{j-1}}(t,t')\geq 4d$ must hold. Indeed, assume otherwise. Let $P$ be any path of length less than $4d$ in graph $G_{j-1}$ connecting $t$ to $t'$. Since this path does not lie in graph $G_j$, at least one vertex $v\in V(P)$ must lie in $V(G_{j-1})\setminus V(G_j)=V(C_j)$. Therefore, graph $G_{j-1}$ contains a path $P'\subseteq P$ of length less than $4d$ between a vertex $v\in V(C_j)$ and a terminal $t\in T^{(j-1)}\setminus \hat T_{C_j}$. This is impossible from Property (\ref{prop: separation}) of Procedure \proccut. Therefore, Invariant (\ref{invariant}) continues to hold at the end of iteration $j$.

Finally, we establish Invariant (\ref{invariant: separating clusters from rest}). Consider a pair $t,t'$ of terminals with $t\in \bigcup_{j'=1}^jT'_{C_{j'}}$ and $t'\in T^{(j)}$. If $t\in  \bigcup_{j'=1}^{j-1}T'_{C_{j'}}$, then, since Invariant (\ref{invariant: separating clusters from rest}) held at the beginning of the current iteration, we get that $\dist_G(t,t')\geq d$. Therefore we assume that $t\in T'_{C_j}$. In this case, $t\in V(C_j)$, and $t'\in T^{(j-1)}\setminus \hat T_{C_j}$ holds, so from Property (\ref{prop: separation}) of Procedure \proccut, $\dist_{G_{j-1}}(t,t')\geq 4d$. From the fact that Invariant (\ref{invariant}) held at the beginning of the iteration, we then get that $\dist_G(t,t')\geq 4d$. 

We conclude that at the end of iteration $j$, Properties (\ref{prop: terminal subset new})--(\ref{prop: separated clusters}), and Invariants (\ref{invariant}) and (\ref{invariant: separating clusters from rest}) continue to hold.

We terminate Procedure \procpart once Property (\ref{prop: cover all terminals}) holds. Since, in every iteration, the number of vertices in the current graph $G_j$ decreases, we are guaranteed that the algorithm indeed terminates.

\paragraph{Running time analysis.}

Recall that, from \Cref{lem: proccut}, the running time of Procedure \proccut is bounded by $O(|E(C)|\cdot n^{64/\Delta})$, where $C$ is the cluster that the procedure returns.
Therefore, for all $j$, the running time of iteration $j$ is at most $O(|E(C_j)|\cdot n^{64/\Delta})$. The total running time of the algorithm is then bounded by $\sum_{C_j\in \cset} O(|E(C_j)|\cdot n^{64/\Delta})$. Since the clusters in $\cset$ are disjoint, we get that the total running time of Procedure \procpart is bounded by $O(|E(G)|\cdot n^{64/\Delta})$.
\end{proof}

\subsubsection{Procedure \procsep.}
Lastly, we provide Procedure \procsep, whose input is similar to that of Procedure \procpart. The goal of the procedure is to either produce two large subsets $T_1,T_2$ of terminals that are sufficiently far from each other in $G$, or to compute a single terminal $t\in T$ with set $B_G(t,\Delta\cdot d)$ containing many terminals.

\begin{lemma}\label{lem: find separation of terminals}
	There is a deterministic algorithm, that we call \procsep, whose input is an $n$-vertex graph $G$, a set $T\subseteq V(G)$ of $k$ vertices called terminals, distance parameters $d$ and $\Delta$, and an additional parameter $0<\alpha\leq 1$. The algorithm either computes a terminal  $t\in T$ with $|B_G(t,\Delta\cdot d)\cap T|>\alpha k$, or it computes two subsets $T_1,T_2$ of terminals, with $|T_1|=|T_2|$, such that $|T_1|\geq k^{1-64/\Delta}\cdot \min\set{(1-\alpha),\frac 1 3}$, and for every pair $t\in T_1$, $t'\in T_2$ of terminals, $\dist_G(t,t')\geq d$. The running time of the algorithm is $O(|E(G)|\cdot n^{64/\Delta})$.
\end{lemma}
\begin{proof}
	We use the following simple observation.
	
	\begin{observation}\label{obs: greedy partition}
		There is a deterministic algorithm, that, given a collection $\set{k_1,k_2,\ldots, k_r}$ of non-negative integers with $\sum_{j=1}^rk_j=k$ and $\max_j\set{k_j}\leq \alpha k$, computes a partition $(J_1,J_2)$ of the set $J=\set{1,\ldots,r}$ of indices, such that $\sum_{j\in J_1}k_j,\sum_{j\in J_2}k_j\geq k\cdot \min\set{(1-\alpha),1/3}$. The running time of the algorithm is $O(r)$.
	\end{observation}

\begin{proof}
	Assume w.l.o.g. that $k_1=\max_j\set{k_j}$. 
	Assume first that $k_1 \geq k/3$. In this case, we let $J_1=\set{k_1}$ and $J_2=J\setminus J_1$.  Clearly, $\sum_{j\in J_2}k_j\geq \sum_{j=1}^rk_j-k_1\geq (1-\alpha)k$, while $\sum_{j\in J_1}k_j\geq k/3$. 
	
	Assume now that $k_1<k/3$. In this case, we start with $J_1=J_2=\emptyset$, and process the indices of $J$ one by one. When index $j$ is processed, we add it to $J_1$ if $\sum_{j'\in J_1}k_{j'}\leq \sum_{j'\in J_1}k_{j'}$ currently holds, and we add it to $J_2$ otherwise. It is easy to verify that, at the end of this algorithm, $ |\sum_{j'\in J_1}k_{j'}-\sum_{j'\in J_2}k_{j'}|\leq \max_{j'\in J}\set{k_{j'}}\leq k/3$. Therefore, $\sum_{j\in J_1}k_j,\sum_{j\in J_2}k_j\geq k/3$.
\end{proof}

We are now ready to describe the algorithm for the proof of \Cref{lem: find separation of terminals}. We start by applying Procedure \procpart from \Cref{lem: procpar} to graph $G$, the set $T$ of terminals, and parameters $d$ and $\Delta$. Let $\left (\cset=\set{C_1,\ldots,C_r},\set{t_{C_j}}_{j=1}^r,\set{\hat T'_{C_j}}_{j=1}^r,\set{\hat T_{C_j}}_{j=1}^r\right )$ be the outcome of the procedure. For $1\leq j\leq r$, denote $k_j=|\hat T_j|$. If there is a cluster $C_j\in \cset$ with $k_j> \alpha k$, then we return terminal $t_{C_j}$; from Property (\ref{prop: small dist between terminals new}), $\hat T_{C_j}\subseteq B_G(t_{C_j},\Delta\cdot d)$, and so $|B_G(t_C,\Delta\cdot d)\cap T|\geq |\hat T_{C_j}|=k_j> \alpha k$ must hold. 

Assume now that, for all $1\leq j\leq r$, $k_j\leq \alpha k$. Using the algorithm from \Cref{obs: greedy partition}, we compute a partition $(J_1,J_2)$ of the set $J=\set{1,\ldots,r}$ of indices, such that $\sum_{j\in J_1}k_j,\sum_{j\in J_2}k_j\geq k\cdot \min\set{(1-\alpha),1/3}$.

Denote $\hat T^1=\bigcup_{j\in J_1}\hat T_{C_j}$, and $\hat T^2=\bigcup_{j\in J_2}\hat T_{C_j}$. Then $|\hat T^1|,|\hat T^2|\geq k\cdot \min\set{(1-\alpha),1/3}$.
Lastly, we let $T_1=\bigcup_{j\in J_1} T'_{C_j}$ and $T_2=\bigcup_{j\in J_2} T'_{C_j}$. From Property (\ref{prop: separated clusters}) of \procpart, 
	for every pair $t\in T_1$, $t'\in T_2$ of terminals, $\dist_G(t,t')\geq d$. From Property (\ref{prop: few discarded terminals new}), $|T_1|\geq |\hat T^1|/k^{64/\Delta}$, and $|T_2|\geq |\hat T^2|/k^{64/\Delta}$. We conclude that $|T_1|,|T_2|\geq k^{1-64/\Delta}\cdot \min\set{(1-\alpha),\frac 1 3}$. We discard terminals from the larger of the sets $T_1,T_2$ as needed, until the cardinalities of both sets become equal.
	
The running time of the algorithm is dominated by the running time of the algorithm from \Cref{lem: procpar}, and is bounded by $O(|E(G)|\cdot n^{64/\Delta})$.
\end{proof}

\subsection{Basic Path Peeling.}
\label{subsec: path peeling}

In this subsection we present an algorithm, that we refer to as \procpathpeel. 
This is a simple greedy algorithm for connecting pre-specified pairs of subsets of vertices to each other with short paths. Similar algorithms were used numerous times before (see e.g. Lemma 6.2 in \cite{fast-vertex-sparsest}, as well as \cite{APSP-old, APSP-previous}, and Theorems 3.2 and 3.8 in \cite{detbalanced}).

The input to Procedure \procpathpeel consists of a graph $G$, collections $A_1,B_1,\ldots,A_k,B_k$ of subsets of its vertices, and parameters $d,\eta>0$. The output of the procedure is collections $\pset_1,\pset_2,\ldots,\pset_k$ of paths in graph $G$, for which the following properties hold:

\begin{properties}{P}
	\item for all $1\leq i\leq k$, every path in $\pset_i$ connects a vertex of $A_i$ to a vertex of $B_i$, and the endpoints of all paths in $\pset_i$ are distinct; \label{pp: path endpoints}
	
	\item every path in set $\pset=\bigcup_{i=1}^k\pset_i$ has length at most $d$, and the paths in $\pset$ cause congestion at most $\eta$; and \label{pp: length and congestion}
	
	\item let $E'$ be the set of all edges $e\in E(G)$, such that exactly $\eta$ paths of $\pset$ use $e$. For all $1\leq i\leq k$, let $A'_i\subseteq A_i$, $B'_i\subseteq B_i$ be the sets of vertices that do not serve as endpoints of the paths in $\pset_i$. Then for all $1\leq i\leq k$, $\dist_{G\setminus E'}(A'_i,B'_i)> d$. 
	\label{pp: distancing}
	\end{properties}

We note that we allow a path of $\pset_i$ to contain vertices of $A_i\cup B_i$, and also vertices from other sets $A_j\cup B_j$ as inner vertices.
The following simple lemma summarizes our algorithm for \procpathpeel.


\begin{lemma}\label{lem: path peel}
	There is a deterministic algorithm, called \procpathpeel, that, given a graph $G$, collections $A_1,B_1,\ldots,A_k,B_k$ of subsets of its vertices, and parameters $d,\eta>0$, outputs sets $\pset_1,\pset_2,\ldots,\pset_k$ of simple paths in $G$, for which properties (\ref{pp: path endpoints}) -- (\ref{pp: distancing}) hold. The running time of the algorithm is $O(m\eta+mdk\log n)$, where $m=|E(G)|$ and $n=V(G)$.
\end{lemma}
\begin{proof}
	We use a simple greedy algorithm combined with the \EST data structure. The algorithm consists of $k$ phases, where for all $1\leq i\leq k$, we construct the set $\pset_i$ of paths in phase $i$. At the beginning of the algorithm, we set, for all $1\leq i\leq k$, $\pset_i=\emptyset$. Throughout the algorithm, we maintain the set $\pset=\bigcup_{i=1}^k\pset_i$ of paths (that is set to $\emptyset$ at the beginning), and, for every edge $e\in E(G)$, we maintain a counter $n(e)$, whose value is equal to the number of paths in $\pset$ containing $e$. At the beginning of the algorithm, we initialize $n(e)=0$ for every edge $e$.
	
	We now describe the execution of the $i$th phase, for some $1\leq i\leq k$. We start by constructing a graph $G_i$. Initially, we let $G_i=G$. We then delete from $G_i$ every edge $e\in E(G)$ with $n(e)=\eta$. Additionally, we add a source vertex $s_i$ to $G_i$, that connects with an edge to every vertex of $A_i$, and a destination vertex $t_i$, that connects with an edge to every vertex of $B_i$. We also initialize an \EST data structure $\tset_i$ in graph $G_i$, with source vertex $s_i$, and distance bound $d+2$.
	
	We then perform iterations, as long as the distance between $s_i$ and $t_i$ in $\tset_i$ is bounded by $d+2$. In every iteration, we use the \EST data structure $\tset_i$ to compute the shortest $s_i$-$t_i$ path $P$ in the current graph $G_i$, whose length must be at most $d+2$. We delete the first and the last edges of $P$, obtaining a path $P'$ of length at most $d$, that connects some vertex $x\in A_i$ to some vertex $y\in B_i$. We add path $P'$ to $\pset_i$, and increase the counter $n(e)$ for every edge $e\in E(P')$. We then update the current graph $G_i$, by deleting the edges $(s_i,x),(y,t_i)$ from it, as well as every edge $e$ whose counter $n(e)$ has reached $\eta$. The \EST data structure $\tset_i$ is also updated with these deletions. Once the \EST data structure $\tset_i$ reports that the distance between $s_i$ and $t_i$ in the current graph $G_i$ is greater than $d+2$, the phase terminates.
	This completes the description of a phase, and of the algorithm.
	
	From the description of the algorithm, it is immediate to verify that properties (\ref{pp: path endpoints}) and (\ref{pp: length and congestion}) hold for the resulting sets $\pset_1,\ldots,\pset_k$ of paths. Property (\ref{pp: distancing}) is also easy to establish. Indeed, assume for contradiction that, for some index $1\leq i\leq k$, there is a path $P'$ of length at most $d$, connecting a vertex of $ A'_i$ to a vertex of $ B'_i$ in graph $G\setminus E'$. Then this path must have existed in graph $G_i$ at the end of the $i$th phase, and so the phase should not have terminated when it did.
	
	Lastly, we bound the running time of the algorithm.  Let $m=|E(G)|$.
	The time that is required to maintain a single \EST $\tset_i$ is bounded by $O(md\log n)$. Additionally, whenever a path $P$ is added to set $\pset$, the algorithm spends $O(|E(P)|)$ time on processing this path, and on increasing the counters $n(e)$ of edges $e\in E(G)$. Since the counter of an edge may be increased at most $\eta$ times, the total running time of the algorithm is bounded by $O(m\eta+mkd\log n)$.
\end{proof}
\section{The Distanced Matching Game}
\label{sec: distanced matching game}

The goal of the \DMG is to construct a well-connected graph, that we define next.

\begin{definition}[Well-connected graph]
	Let $G=(V,E)$ be a graph, let $S(G)\subseteq V$ be a subset of its vertices called \emph{supported vertices}, and let $\eta,D>0$ be parameters. We say that graph $G$ is \emph{$(\eta,D)$-well-connected  with respect to the set $S(G)$ of supported vertices} if, for every pair $A,B\subseteq S(G)$ of disjoint equal-cardinality subsets of supported vertices, there is a collection $\pset(A,B)$ of paths in graph $G$, routing every vertex of $A$ to a distinct vertex of $B$ (that is, $\pset(A,B)$ is a one-to-one routing of $A$ to $B$), such that the paths in $\pset(A,B)$ cause congestion at most $\eta$, and the length of every path is at most $D$.
\end{definition}

Unlike the \CMG, whose goal is to construct an expander graph, the goal of the \DMG is to construct a graph that is $(\eta,D)$-well-connected with respect to a set $S(G)$ of supported vertices; typically, for an $n$-vertex graph $G$, we will require that $|S(G)|\geq n-n^{1-\Theta(\epsilon)}$, $\eta\leq n^{O(\epsilon)}$, and $D=2^{O(1/\eps)}$, for a given parameter $0<\eps<1$.

The main component of the \DMG is a \emph{distancing}, that is defined next. Distancings play a role similar to that of cuts in the \CMG.

\begin{definition}[Distancing]
	Let $G$ be an $n$-vertex graph, and let $0<\delta<1,d>0$ be parameters. A $(\delta,d)$-distancing for $G$ is a triple $(A,B,E')$, where $A,B$ are disjoint subsets of $V(G)$ of cardinality at least $n^{1-\delta}$ each, with $|A|=|B|$, and $E'\subseteq E(G)$ is a subset of edges of cardinality at most $|A|/16$. We require that $\dist_{G\setminus E'}(A,B)\geq d$.
\end{definition}

While the notion of distancing, to the best of our knowledge, was never formally defined before, it is a well-known and widely used fact that one can efficiently obtain a sparse cut in a graph from a distancing. The following lemma, whose analogues have been widely used before, summarizes such an algorithm. The proof uses standard ball-growing technique, and appears in Section \ref{sec:ball growning} of Appendix.

\begin{lemma}\label{lem: distancing to sparse cut}
	There is a deterministic algorithm, whose input consists of a connected graph $G$ with $|V(G)|=n$ and $|E(G)|=m$, a parameter $0<\phi<1/2$, and a $(\delta,d)$-distancing $(X,Y,E')$ in $G$, where  $0<\delta<1$ is any parameter, $d\geq (32\log m)/\phi$, and $|E'|\leq \phi |X|/4$. The algorithm computes a cut $(X',Y')$ in graph $G$, with $X\subseteq X'$ and $Y\subseteq Y'$, such that $|E_G(X',Y')|\leq \phi\cdot \min\set{|E_G(X')|,|E_G(Y')|}$. 
	The running time of the algorithm is bounded by $O(n+\min\set{|E_G(X')|,|E_G(Y')|})$.
\end{lemma}

We note that, if the maximum vertex degree in $G$ is bounded by $\Delta$, then $|E_G(X')|\leq \Delta\cdot |X'|$, and similarly $|E_G(Y')|\leq \Delta\cdot |Y'|$. Therefore, the algorithm guarantees that $|E_G(X',Y')|\leq \phi\cdot \Delta\cdot \min\set{|X'|,|Y'|}$.


A \DMG receives as input an integral parameter $n$, and two other parameters $0<\delta<1$ and $d\geq 2^{4/\delta}$. The game is played between a \emph{distancing player} and a \emph{matching player}, in iterations. Over the course of the game, a graph $G$ is constructed. Initially, graph $G$ contains a set $V$ of $n$ vertices and no edges. In every iteration, some edges are added to $G$.

The $i$th iteration is executed as follows. First, the distancing player  either computes a $(\delta,d)$-distancing $(A_i,B_i,E'_i)$ in the current graph $G$, or returns ``END''. If the distancing player returned ``END'', then the game terminates, and we say that the game lasted for $(i-1)$ iterations. Otherwise, the matching player computes a (possibly partial) matching $M_i$ between vertices of $A_i$ and vertices of $B_i$, of cardinality at least $|A_i|/8$. We require that $M_i$ does not contain pairs $(u,v)$ of vertices for which edge $(u,v)\in E'$. We note that $M_i$ is not a subset of edges of $G$; it is just a collection of pairs of vertices from $A_i\times B_i$, with every vertex of $A_i\cup B_i$ appearing in at most one pair in $M_i$. We add the edges of $M_i$ to graph $G$, completing iteration $i$, and proceed to the next iteration.

 We note that, once the current graph $G$ contains no $(\delta,d)$-distancing, the game must terminate. From the above description, graph $G$ remains a simple graph (that is, we never add parallel edges to it).

The main technical result of this section is the proof of \Cref{thm: intro distancing-matching game - number of iterations}, that bounds the number of iterations in the \DMG.  We restate the theorem here for convenience.

\begin{theorem}[Restatement of \Cref{thm: intro distancing-matching game - number of iterations}]\label{thm: distancing-matching game - number of iterations}
	Consider a \DMG with parameters $n>0,0<\delta<1/4$ and $d$, such that $d\geq 2^{4/\delta}$ and $n^{\delta}\geq \frac{2^{14}\log n}{\delta^2}$.
Then the number of iterations in the game is at most $n^{8\delta}$.
\end{theorem}

We note that we do not currently know whether the bounds in this theorem are tight, and in particular whether the requirement that $d\geq 2^{4/\delta}$ is necessary. It would be interesting to establish whether a similar theorem can be proved for values of $d$ that have a lower dependence on $1/\delta$. 

We now turn to prove \Cref{thm: distancing-matching game - number of iterations}. We assume that the paramters $n,\delta$ and $d$ are fixed. We let $V$ be a set of $n$ vertices, over which the game is played. 

Consider a \DMG that lasts for $z$ iterations. We can summarize the game via a \emph{transcript} $\transcript=\left ((A_1,B_1,E_1'),M_1,\ldots,(A_z,B_z,E'_z),M_z\right )$, where for $1\leq i\leq z$, $(A_i,B_i,E_i')$ is the distancing computed by the distancing player, and $M_i$ is the matching returned by the matching player in iteration $i$. For all $1\leq i\leq z$, we denote by $G_i$ the graph obtained after $i$ iterations of the game, so $V(G_i)=V$ and $E(G_i)=\bigcup_{i'=1}^iM_{i'}$. We also let $G_0$ be the graph with $V(G_0)=V$ and $E(G_0)=\emptyset$.

Consider now any subset $I\subseteq \set{1,\ldots,z}$ of indices, and assume that $I=\set{i_1,i_2,\ldots,i_q}$ with $i_1<i_2<\cdots<i_q$. Consider now the following sequence, that, intuitively corresponds to only executing the iterations of the \DMG whose indices lie in $I$: 

$$\transcript'=\left ((A_{i_1},B_{i_1},E_{i_1}''),M_{i_1},\ldots,(A_{i_q},B_{i_q},E''_{i_q}),M_{i_q}\right).$$

Here for all $1\leq j\leq q$, $E''_{i_j}$ is defined to be $E'_{i_j}\cap \left(\bigcup_{j'=1}^{j-1}M_{i_{j'}}\right )$.

We claim that $\transcript'$ is a valid transcript of a \DMG. In order to show this, for all $1\leq j\leq q$, let $G'_j$ be the graph obtained after $j$ iterations of the game, that is, $V(G'_j)=V$, and $E(G'_j)=\bigcup_{j'=1}^jM_{i_{j'}}$. We also let $G_0'$ be the graph containing the set $V$ of vertices and no edges. It is enough to show that, for all $1\leq j\leq q$, $(A_{i_j},B_{i_j},E_{i_j}'')$ is a valid $(\delta,d)$-distancing in graph $G'_{j-1}$. We now prove that this is indeed the case.

Since $(A_{i_j},B_{i_j},E_{i_j}')$ is a valid $(\delta,d)$-distancing in the graph $G_{i_j-1}$ obtained after $(i_j-1)$ iterations of the original \DMG, it is enough to show that there is no path of length less than $d$ connecting a vertex of $A_{i_j}$ to a vertex of $B_{i_j}$ in graph $G'_{i_{j-1}}\setminus E''_{i_j}$. Assume for contradiction that such a path  $P$  exists. Recall that $G'_{i_{j-1}}\subseteq G_{i_j-1}$, and that every edge $e\in E'_{i_j}\cap E(G'_{i_{j-1}})$ lies in $E''_{i_j}$. Therefore, path $P$ also lies in graph $G_{i_j-1}\setminus E_{i_j}'$, contradicting the fact that $(A_{i_j},B_{i_j},E_{i_j}')$ is a valid $(\delta,d)$-distancing in  graph $G_{i_j-1}$. 
We conclude that $(A_{i_j},B_{i_j},E_{i_j}'')$ is a valid $(\delta,d)$-distancing in graph $G'_{j-1}$, and $\transcript'$ is a valid transcript of a \DMG.

To summarize, we can select a subset of iterations from the transcript of the \DMG, and obtain a valid transcript of a \DMG, induced by these iterations. We say that the \DMG associated with transcript $\transcript'$ is \emph{defined by the set $I$ of indices}.

For the sake of the proof of \Cref{thm: distancing-matching game - number of iterations}, it would be convenient for us to assume that the cardinalities of the matchings $M_i$ returned in every iteration of the \DMG are roughly the same. In order to do so, we partition the set $I^*=\set{1,\ldots,z}$ of indices into at most $r=\ceil{16\delta{\log n}}$ subsets, as follows. Recall that for all $1\leq i\leq z$, we are guaranteed that $|M_i|\geq |A_i|/8\geq n^{1-\delta}/8$, and clearly $|M_i|\leq n$ must hold. For all $1\leq j\leq r$, we let $I_j\subseteq I^*$ contain all indices $i$, for which $\frac{n}{2^{j}}< |M_i|\leq \frac{n}{2^{j-1}}$. Clearly, there must be an index $j$, such that $|I_j|\geq \frac{z}{r}=\frac{z}{\ceil{16\delta\log n}}$. From now on we will focus on the \DMG that is defined by the set $I_j$ of indices, and we will bound the number of iterations in this game, that we denote by $z'\geq \frac{z}{\ceil{16\delta\log n}}$. For simplicity of notation, for all $1\leq i\leq z'$, we denote the distancing associated with the $i$th iteration of the game by $(A_i,B_i,E'_i)$ and the matching associated with iteration $i$ by $M_i$. As before, we denote by $G_0$ the graph whose vertex set is $V$ and edge set is empty, and for $1\leq i\leq z'$, we let $G_i$ be the graph obtained after $i$ iterations of the game, that is, $V(G_i)=V$ and $E(G_i)=M_1\cup\cdots\cup M_i$.
We also denote $I=\set{1,\ldots,z'}$ and $G=G_{z'}$.

Let $E^*=\bigcup_{i=1}^{z'}E'_i$, and for all $1\leq i\leq z'$, let $M'_i=M_i\setminus E^*$. We also denote $H_0=G_0$, and for all $1\leq i\leq z'$,  we let $H_i$ be a graph whose vertex set is $V$, and edge set is $E(G_i)\setminus E^*=\bigcup_{i'=1}^iM'_{i'}$. We also denote $H=H_{z'}$. The following observation is immediate from the definitions.

\begin{observation}\label{obs: distancing}
For all $1\leq i\leq z'$, $(A_i,B_i,\emptyset)$ is a $(\delta,d)$-distancing in graph $H_i$; in other words, there is no path of length less than $d$ connecting $A_i$ to $B_i$ in $H_i$.
\end{observation}

Note that:

\begin{equation}\label{eq: many edges after deletion}
|E(H)|=|E(G)\setminus E^*|=\sum_{i=1}^{z'}|M_i|-\sum_{i=1}^{z'}|E'_i|\geq\sum_{i=1}^{z'}(|M_i|-|E'_i|)\geq \sum_{i=1}^{z'}\frac{|M_i|}{2} \geq \frac{|E(G)|} 2.
\end{equation}

We have used the fact that, from the definition of the \DMG, for all $i$, $|M_i|\geq |A_i|/8$, while $|E'_i|\leq |A_i|/16$, so $|E'_i|\leq |M_i|/2$ must hold.

 We say that an iteration $i\in I$ is \emph{bad} if $|M'_i|<|M_i|/16$; otherwise, iteration $i$ is good. We let $I^b\subseteq I$ be the set of all indices $i$, such that the $i$th iteration is bad, and we let $I^g=I\setminus I^b$ be the set of all indices of good iterations. We use the following simple observation.

\begin{observation}\label{obs: bad iterations}
	$|I^b|\leq 7z'/8$. 
\end{observation}

\begin{proof}
	Denote $\frac{|I^b|}{z'}=\beta$, and assume for contradiction that $\beta> 7/8$.  	Recall that, for all $i\in I$, $\frac{n}{2^{j}}< |M_i|\leq \frac{n}{2^{j-1}}$. 
	Therefore:
	
	\[ \begin{split}
	|E(H)|&= \sum_{i\in I^b}|M'_i|+\sum_{i\in I^g}|M'_i|\\
	&\leq \beta\cdot z'\cdot\frac{n}{16\cdot 2^{j-1}}+(1-\beta)\cdot z'\cdot \frac{n}{2^{j-1}}\\
	&=\frac{2z'n}{2^j}- \frac{\beta z' n}{2^j}\cdot \frac{15}{8}\\
	&\leq \frac{2z'n}{2^j}- \frac{z' n}{2^j}\cdot \frac 7 8\cdot \frac{15}{8}\\
	&\leq \frac{23}{64}\cdot \frac{z'n}{2^j}\\
	&<\frac{z'n}{2^{j+1}}.
	\end{split}\]

	On the other hand, from Equation \ref{eq: many edges after deletion}:
	
	 \[|E(H)|\geq \sum_{i=1}^{z'}\frac{|M_i|}{2}\geq \frac{z'n}{2^{j+1}},\]

	 a contradiction. We conclude that $\beta\leq 7z'/8$ holds.
\end{proof}

To summarize, so far we have shown that:

\begin{equation}\label{eq: many good iteration}
|I^g|\geq \frac{z'}{8}\geq \frac{z}{256\delta\log n}.
\end{equation}

In order to bound $z$, it is now enough to bound the number of good iterations in the game. In order to do so, we partition the game into \emph{phases}, each of which (except for, possibly, the last one), contains exactly $\ceil{n^{4\delta}}$ good iterations. 
It is now enough to prove the following lemma.

\begin{lemma}\label{claim: bound number of phases}
	The number of phases is bounded by $\frac{2n^{2\delta}}{\delta}$.
\end{lemma}

Indeed, assume that \Cref{claim: bound number of phases} holds. Then $|I_g|\leq \frac{2n^{2\delta}}{\delta}\cdot \ceil{n^{4\delta}}\leq \frac{4n^{6\delta}}{\delta}$, and, from Equation \ref{eq: many good iteration}, we get that:

\[z\leq |I^g|\cdot 256\delta\log n\leq 1024n^{6\delta}\log n\leq  n^{7\delta},\]

since $n^{\delta}\geq 1024\log n$.

In order to complete the proof of \Cref{thm: distancing-matching game - number of iterations}, it is now enough to prove \Cref{claim: bound number of phases}, which we do next. We denote the number of phases in the game by $\hat z$.

For every integer $1\leq k\leq \hat z$, we denote by $\hat H^{(k)}$ the graph $H$ at the beginning of the $k$th phase. In other words, if the last iteration of phase $(k-1)$ is $i$, then $\hat H^{(k)}=H_i$. We will define, for every phase $k$, a collection $\cset^{(k)}$ of disjoint subgraphs of $H^{(k)}$, that we refer to as \emph{clusters}; we also refer to $\cset^{(k)}$ as a \emph{clustering} of graph $H^{(k)}$. For all integers $1\leq s\leq \ceil{1/\delta}$, we let $\cset^{(k)}_s\subseteq \cset^{(k)}$ be the set of all clusters $C\in \cset^{(k)}$, with $n^{(s-1)\delta}< |V(C)|\leq n^{s\delta}$, and let $\cset^{(k)}_0\subseteq \cset^{(k)}$ be the set of all clusters $C$ containing a single vertex. For $0\leq s\leq \ceil{1/\delta}$, we say that a cluster $C\in \cset^{(k)}_s$ \emph{lies at level $s$ of the clustering $\cset^{(k)}$}. If cluster $C$ lies at level $s$ of $\cset^{(k)}$, then we say that every vertex of $C$ lies at level $s$ of $\cset^{(k)}$.

We will ensure that the following invariants hold for every integer $1\le k\leq \hat z$:

\begin{properties}{I}
\item $V=\bigcup_{C\in \cset^{(k)}}V(C)$; \label{inv: every vertex covered by clustering}
	
\item for all $0 \leq s\leq \ceil{1/\delta}$, if $C\in \cset^{(k)}_s$ is a level-$s$ cluster of $\cset^{(k)}$, then for every pair $x,y\in V(C)$ of vertices, $\dist_C(x,y)\leq 4^s$ (and in particular $\dist_C(x,y)<d$ must hold, as $d\geq 2^{4/\delta}$); \label{inv: distances in clusters} and

\item if $C$ is a cluster of $\cset^{(k)}$, and $C'$ is a cluster of $\cset^{(k+1)}$, then either $V(C)\cap V(C')=\emptyset$, or $C\subseteq C'$. \label{inv: containment of clusters}
\end{properties}

Note that Invariant \ref{inv: containment of clusters} ensures that, if $C'$ is a cluster of $\cset^{(k+1)}$, then either $C'$ is a cluster of $\cset^{(k)}$, or $C'$ is obtained by taking the union of several clusters of $\cset^{(k)}$, and possibly adding some edges to the resulting graph. In particular, the level of any given vertex $v\in V$ may only grow from phase to phase. Consider now some vertex $v\in V$. If vertex $v$ lies at level $s$ of $\cset^{(k)}$, and at level $s'>s$ of $\cset^{(k+1)}$, then we say that vertex $v$ was \emph{promoted} by $\cset^{(k+1)}$, or that it is promoted during phase $k$.

Initially, we let $\cset^{(1)}$ contain, for every vertex $v\in V$, a separate cluster $C(v)$, consisting of only vertex $v$ iteslf, so $\cset^{(1)}=\set{C(v)\mid v\in V}$. Therefore, every vertex lies at level $0$ of $\cset^{(1)}$. Clearly, as the algorithm progresses, the level of a vertex may be at most $\ceil{1/\delta}$. The key in bounding the number of phases is to show that the clusterings $\cset^{(k)}$ can be constructed so that a large number of vertices are promoted in every phase. Since every vertex may only be promoted at most $\ceil{1/\delta}$ times, this will be sufficient in order to bound the number of phases.
In order to complete the proof of \Cref{claim: bound number of phases}, it is enough to prove the following claim.

\begin{claim}\label{claim: construct clusterings for a phase}
	Consider some integer $k\geq 1$, such that phase $k$ is not the last phase, and assume that we are given a clustering $\cset^{(k)}$ of graph $\hat H^{(k)}$, for which Invariants \ref{inv: every vertex covered by clustering} and \ref{inv: distances in clusters} hold. Then there is a clustering $\cset^{(k+1)}$ of graph $\hat H^{(k+1)}$, for which Invariants \ref{inv: every vertex covered by clustering} and \ref{inv: distances in clusters} hold. Additionally, for every pair $C\in \cset^{(k)}$, $C'\in \cset^{(k+1)}$ of clusters, either $V(C)\cap V(C')=\emptyset$ or $C\subseteq C'$ hold. Lastly, the number of vertices that are promoted in $\cset^{(k+1)}$ is at least $n^{1-2\delta}$.
\end{claim}

Note that \Cref{claim: bound number of phases} follows immediately from \Cref{claim: construct clusterings for a phase}. Since every vertex may be promoted at most $\ceil{1/\delta}$ times, and every phase promotes at least $n^{1-2\delta}$ vertices, the total number of phases must be bounded by $\frac{n\cdot\ceil{1/\delta}}{n^{1-2\delta}}\leq \frac{2n^{2\delta}}{\delta}$. We now prove \Cref{claim: construct clusterings for a phase}.

\begin{proofof}{\Cref{claim: construct clusterings for a phase}}
Let $I_k\subseteq I$ be the collection of indices $i\in I$, such that iteration $i$ belongs to phase $k$, and let $I'_k=I_k\cap I^g$ be the set of indices corresponding to good iterations of phase $k$. Recall that $|I'_k|= \ceil{n^{4\delta}}$. We will use the following simple observation.

\begin{observation}\label{obs: separation in an iteration}
	Consider an iteration $i\in I_k$, and the corresponding distancing $(A_i,B_i,E'_i)$. Then for every cluster $C\in \cset^{(k)}$, either $A_i\cap V(C)=\emptyset$, or $B_i\cap V(C)=\emptyset$. Moreover, if there is a pair $C,C'\in \cset^{(k)}$ of clusters, such that some edge $e\in M'_i$ connects a vertex of $C$ to a vertex of $C'$, then for every subsequent iteration $i'>i$, either $A_{i'}\cap (V(C)\cup V(C'))=\emptyset$, or $B_{i'}\cap (V(C)\cup V(C'))=\emptyset$ must hold.
\end{observation}

\begin{proof}
	We fix an index $i\in I_k$, and consider the corresponding distancing $(A_i,B_i,E'_i)$. From \Cref{obs: distancing}, $(A_i,B_i,\emptyset)$ is a $(\delta,d)$-distancing in graph $H_i$. Therefore, if $P$ is any path in $H_i$ connecting a vertex of $A_i$ to a vertex of $B_i$, then the length of $P$ is at least $d$. Consider now some cluster $C\in \cset^{(k)}$, and assume that $C\in \cset^{(k)}_s$, for some $0\leq s\leq \ceil{1/\delta}$. Assume for contradiction that there is a pair $u,v\in V(C)$ of vertices, with $u\in A_i$ and $v\in B_i$. From Invariant \ref{inv: distances in clusters}, $\dist_C(u,v)\leq 4^s\leq 4^{\ceil{1/\delta}}<d$ (since $d\geq 2^{4/\delta}$). Since $C\subseteq \hat H^{(k)}\subseteq H_i$, this is a contradiction.

	Assume now that there is a pair  $C,C'\in \cset^{(k)}$ of clusters, such that some edge $e\in M'_i$ connects a vertex of $C$ to a vertex of $C'$. Using the same reasoning as above, and since edge $e$ is added to graph $H_{i+1}$, we get that, for every pair $u\in V(C),v\in V(C')$ of vertices, $\dist_{H_{i+1}}(u,v)\leq 1+2\cdot 4^{\ceil{1/\delta}}<d$. Therefore, for all $i'>i$, $\dist_{H_{i'}}(u,v)<d$ holds as well. Since, from \Cref{obs: distancing}, $(A_{i'},B_{i'},\emptyset)$ is a $(\delta,d)$-distancing in graph $H_{i'}$, we conclude that at most one of the sets $A_{i'},B_{i'}$ may contain a vertex of $V(C)\cup V(C')$.
\end{proof}

In the remainder of the proof of \Cref{claim: construct clusterings for a phase}, we consider  good iterations $i\in I'_k$. For each such iteration $i$, we will select a large enough subset $M''_i\subseteq M'_i$ of edges, and  integers $0\leq s', s\leq \ceil{1/\delta}$, so that every edge in $M''_i$ connects a vertex of $\bigcup_{C\in \cset^{(k)}_s}V(C)$ to a vertex of  $\bigcup_{C'\in \cset^{(k)}_{s'}}V(C')$.
We then say that iteration $i$ \emph{belongs to class $(s,s')$}. Since the number of possible classes is small, we can find a class $(s,s')$ to which many iterations of $I'_k$ belong. Assume w.l.o.g. that $s\geq s'$. We will then use the iterations of $I'_k$ from class $(s,s')$ in order to identify a large number of level-$s$ clusters that can be merged together, so that, on the one hand, Invariants \ref{inv: every vertex covered by clustering} and \ref{inv: distances in clusters} continue to hold, while, on the other hand, a large number of vertices are promoted. The remainder of the proof of \Cref{claim: construct clusterings for a phase} consists of three steps. In the first step, we define a subset $M''_i\subseteq M'_i$ of edges for each iteration $i\in I'_k$, and classify the iteration into some class $(s,s')$. Then in the second step we define a new contracted graph $J$, representing the clusters of $\cset^{(k)}_s\cup \cset^{(k)}_{s'}$, where $(s,s')$ is the most common class among the iterations of $I_k'$. Graph $J$ is then used in Step 3 in order to define the new  clustering $\cset^{(k+1)}$. We now present each of the three steps in turn.

\paragraph{Step 1: Iteration Classification.}
Consider an iteration $i\in I'_k$. Since iteration $i$ is good, $|M'_i|\geq \frac{|M_i|}{16}\geq \frac{|A_i|}{128}\geq \frac{n^{1-\delta}}{128}$.
Consider now an edge $e=(u,v)\in M'_i$, with $u\in A_i,v\in B_i$. 
From Invariant \ref{inv: every vertex covered by clustering}, there must be clusters $C,C'\in \cset^{(k)}$ with $u\in V(C)$ and $v\in V(C')$. Moreover, from \Cref{obs: separation in an iteration}, $C\neq C'$ must hold.
We say that edge $e$ is of type $(s,s')$, for a pair $0\leq s,s'\leq \ceil{1/\delta}$ of indices, if $C\in \cset^{(k)}_s$ and $C'\in \cset^{(k)}_{s'}$; note that it is possible that $s=s'$. Clearly, there is a pair $0\leq s,s'\leq \ceil{1/\delta}$ of indices, such that the number of edges of type $(s,s')$ in $M'_i$ is at least $\frac{|M'_i|}{(1+\ceil{1/\delta})^2}\geq \frac{n^{1-\delta}\cdot \delta^2}{512}$. From now on we fix this pair of indices, and we denote by $M''_i\subseteq M'_i$ the set of all edges of type $(s,s')$ in $M'_i$, so $| M''_i|\geq \frac{n^{1-\delta}\cdot \delta^2}{512}$. We say that iteration $i$ is an iteration of type $(s,s')$.

Note that there must be an ordered pair $(s,s')$ of indices, with 
$0\leq s,s'\leq \ceil{1/\delta}$, such the number of good iterations in $I'_k$ that are of type $(s,s')$ is at least: $\frac{|I'_k|}{(1+\ceil{1/\delta})^2}\geq \frac{\delta^2}{6}\cdot \ceil{n^{4\delta}}$.
From now on we fix this pair $(s,s')$ of indices, and we denote by $I''_k\subseteq I'_k$ the set of all good iterations $i$ from phase $k$ that are of type $(s,s')$, so $|I''_k| \geq \frac{\delta^2}{6}\cdot \ceil{n^{4\delta}}$. We assume w.l.o.g. that $s\geq s'$.

\paragraph{Step 2: Contracted Graph.}
In this step we construct a weighted contracted graph $J$, as follows. For every cluster $C\in \cset^{(k)}_s\cup \cset^{(k)}_{s'}$, we add a vertex $v(C)$ to graph $J$; we refer to vertices of $J$ as \emph{supernodes}, to distinguish them from vertices of $V$. The set of edges of $J$ is the union of the sets $\set{E_i\mid i\in I''_k}$ of edges that we define below. We refer to the edges of $J$ as \emph{meta-edges}. For every iteration $i\in I''_k$, every meta-edge $\hat e\in E_i$ represents some collection $S(\hat e)\subseteq M''_i$ of edges.

Consider some iteration $i\in I''_k$. For every pair $C\in \cset^{(k)}_s,C'\in \cset^{(k)}_{s'}$ of clusters, such that at least one edge of $M''_i$ connects a vertex of $C$ to a vertex of $C'$, we add a meta-edge $\hat e =(v(C),v(C'))$ to $E_i$. We let $S(\hat e)$ be the set of all edges of $M''_i$ that connect vertices of $C$ to vertices of $C'$, and we set the weight of the meta-edge $\hat e$ to be $w(\hat e)=|S(\hat e)|$. We note that, since $M''_i$ is a matching of cardinality at least $\frac{n^{1-\delta}\cdot \delta^2}{512}$, we get that $\sum_{\hat e\in E_i}w(e)\geq \frac{n^{1-\delta}\cdot \delta^2}{512}$. Moreover, for every cluster $C\in \cset^{(k)}_s\cup \cset^{(k)}_{s'}$, the total weight of all meta-edges of $E_i$ incident to $C$ is at most $|V(C)|$.

Consider a pair $i,i'\in I''_k$ of indices with $i<i'$. Note that, if a meta-edge $(v(C),v(C'))$ belongs to $E_i$, then, from \Cref{obs: separation in an iteration}, meta-edge $(v(C),v(C'))$ may not lie in $E_{i'}$.
We set $E(J)=\bigcup_{i\in I''_k}E_i$. From the above discussion, graph $J$ contains no parallel edges. For every supernode $v(C)\in V(J)$, the total weight of all meta-edges incident to $v(C)$ is bounded by $|V(C)|\cdot |I''_k|$. The total weight of all meta-edges in graph $J$ is at least $\frac{n^{1-\delta}\cdot \delta^2}{512}\cdot |I''_k|$.

\paragraph{Step 3: Constructing the Clustering $\cset^{(k+1)}$.}

In order to construct clustering $\cset^{(k+1)}$, we start with $\cset^{(k+1)}=\cset^{(k)}$, and then iteratively merge some clusters of $\cset^{(k+1)}$. In every iteration, we consider the graph $J$. Assume that there is a cluster $C'\in \cset^{(k)}_{s'}$, such that supernode $v(C')$ has at least $n^{\delta}$ neighbor vertices in graph $J$. We denote the neighbor vertices of $v(C')$ by $v(C_1),\ldots,v(C_q)$. Note that for all $1\leq a\leq q$, $C_a\in \cset^{(k)}_s$ must hold. We delete clusters $C',C_1,\ldots,C_q$ from $\cset^{(k+1)}$, and instead add a single cluster $C^*$, whose vertex set is $V(C')\cup V(C_1)\cup\cdots\cup V(C_q)$, and edge set is the union of $E(C')\cup E(C_1)\cup\cdots\cup E(C_q)$ with the set $\bigcup_{a=1}^qS(v(C'),v(C_a))$ of edges. 
Note that $C^*\subseteq \hat H^{(k+1)}$ holds. We then delete vertices $v(C'),v(C_1),\ldots,v(C_q)$ from graph $J$.

Observe that, for every vertex $x\in V(C^*)$, the level of $x$ in $\cset^{(k)}$ was either $s$ or $s'\leq s$. Moreover, if $C\in \cset^{(k)}$ is the cluster that contained $x$, then $|V(C^*)|>n^{\delta}\cdot |V(C)|$. Therefore, the level of $x$ in $\cset^{(k+1)}$ is strictly greater than that in $\cset^{(k)}$. We conclude that every vertex in $V(C^*)$ is promoted in the current phase, and the level of cluster $C^*$ is at least $(s+1)$. We use the following simple observation, that will allow us to establish Invariant \ref{inv: distances in clusters}.

\begin{observation}\label{obs: distance grows slowly}
	For every pair $x,y\in V(C^*)$ of vertices, $\dist_{C^*}(x,y)\leq 4^{s+1}$.
\end{observation}
\begin{proof}
	If both $x$ and $y$ belong to a single cluster of $\set{C',C_1,\ldots,C_q}$, then, since Invariant \ref{inv: distances in clusters} held for $\cset^{(k)}$, and since each such cluster is contained in $C^*$, $\dist_{C^*}(x,y)\leq 4^{s}$ must hold. Assume now that $x$ and $y$ belong to different clusters. We assume w.l.o.g. that $x\in V(C_1)$ and $y\in V(C_2)$; the other cases are treated similarly.
	
	From the definition of cluster $C^*$, meta-edges $\hat e_1=(v(C'),v(C_1))$, $\hat e_2=(v(C'),v(C_2))$ lie in $J$. Consider any real edge $e_1\in S(\hat e_1)$ and $e_2\in S(\hat e_2)$. Both edges must lie in $\hat H^{(k+1)}$, and in $C^*$. We denote $e_1=(x_1,y_1)$ with $x_1\in V(C_1)$, and $e_2=(x_2,y_2)$ with $y_2\in V(C_2)$. In particular, $y_1,x_2\in V(C')$ must hold.  From Invariant \ref{inv: distances in clusters}, there is a path $P_1$ of length at most $4^s$ connecting $x$ to $x_1$ in $C_1$; a path $P'$ of length at most $4^{s'}\leq 4^s$ connecting $y_1$ to $x_2$ in $C'$; and a path $P_3$ of length at most $4^s$ connecting $y_2$ to $y$ in $C_2$. By combining these three paths with edges $e_1$ and $e_2$, we obtain a path in cluster $C^*$, connecting $x$ to $y$, whose length is at most $3\cdot 4^s+1\leq 4^{s+1}$. Therefore, $\dist_{C^*}(x,y)\leq 4^{s+1}$.
\end{proof}

The algorithm terminates once,  every supernode $v(C')$ of $J$ with $C'\in \cset^{(k)}_{s'}$, the number of meta-edges incident to $v(C')$ in the current graph $J$ is less than $n^{\delta}$.

From the above discussion, once the algorithm terminates, Invariant \ref{inv: distances in clusters} holds for the final clustering $\cset^{(k+1)}$, as every newly added cluster to $\cset^{(k+1)}$ belongs to level $(s+1)$ or higher. It is immediate to verify that Invariant \ref{inv: every vertex covered by clustering} holds for the final set $\cset^{(k+1)}$ of clusters. It now only remains to prove that sufficiently many vertices are promoted in the current iteration.

Consider the graph $J$ that is obtained at the end of the algorithm. In this graph, for every cluster $C'
\in \cset^{(k)}_{s'}$, its corresponding supernode $v(C')$ has fewer than $n^{\delta}$ meta-edges incident to $v(C')$. Since, for every meta-edge $\hat e\in E(J)$ that is incident to a supernode $v(C)$, $w(\hat e)\leq |V(C)|$ must hold, we get that the total weight of all edges remaining in graph $J$ at the end of the algorithm is bounded by:

\[ \sum_{C'\in \cset^{(k)}_{s'}}|V(C')|\cdot n^{\delta}\leq n^{1+\delta}. \]

Recall that the total weight of all meta-edges of $J$ at the beginning of the algorithm was at least: 

\[\frac{n^{1-\delta}\cdot \delta^2}{512}\cdot |I''_k|\geq\frac{n^{1+3\delta}\cdot \delta^4}{2^{12}}>2n^{1+\delta},\]

since $|I''_k| \geq \frac{\delta^2}{6}\cdot \ceil{n^{4\delta}}$ and $n^{\delta}\geq 2^{14}/\delta^2$.

Therefore, the total weight of the meta-edges that remain at the end of the algorithm in $J$ is less than half the original total weight. In other words, the total weight of all meta-edges that were deleted from $J$ is at least $\frac{n^{1-\delta}\cdot \delta^2}{1024}\cdot |I''_k|$. Each of the deleted meta-edges is incident to some supernode $v(C)$, with $C\in \cset^{(k)}_s$ that was deleted from $J$. Recall that every vertex of such a cluster $C$ is promoted in the current phase.

Consider now some supernode $v(C)$ that was deleted from $J$ in the current phase. The total weight of all meta-edges incident to $v(C)$ in the original graph $J$ was at most $|V(C)|\cdot |I''_k|$, and each of the vertices of $C$ was promoted in phase $k$. Therefore, if we denote by $U$ is the set of all vertices of $V$ that were promoted in phase $k$, then the total weight of all meta-edges that were deleted from $J$ is at most $|U|\cdot |I''_k|$. Since, as shown above, the total weight of all such edges is at least $\frac{n^{1-\delta}\cdot \delta^2}{1024}\cdot |I''_k|$, and since we have assumed that $n^{\delta}\geq 2^{14}/\delta^2$, 
 we get that $|U|\geq \frac{n^{1-\delta}\cdot \delta^2}{1024} \geq n^{1-2\delta}$.
\end{proofof}

This concludes the proof of \Cref{thm: distancing-matching game - number of iterations}.
We obtain the following immediate corollary of the theorem.

\begin{corollary}\label{cor: distancing matching game number of edges}
	Consider a \DMG with parameters $n>0,0<\delta<1/4$ and $d$, such that $d\geq 2^{4/\delta}$, $n^{\delta}\geq \frac{2^{14}\log n}{\delta^2}$. Let $G$ be the graph that is obtained at the end of the game. Then $|E(G)|\leq n^{1+8\delta}$ holds, and every vertex of $G$ has degree at most $n^{8\delta}$.
	\end{corollary}

The corollary follows from the fact that the set $E(G)$ of edges is partitioned into at most $n^{8\delta}$ matchings -- the responses of the matching player in the game.

\section{Hierarchical Support Structure}
\label{sec: HSS}

A \HSS uses two parameters, an integer $N>0$, and another parameter $0<\eps\leq 1/4$. Throughout, we denote $\eps'=\eps^4$. For an integer $j\geq 1$, we also let $\eta_j=N^{6+256j\eps^2}$ and $\td_j=2^{cj/\eps^4}$, where $c$ is a sufficiently large constant.


For all $1\leq j\leq \ceil{1/\eps}$, we define a \emph{level}-$j$ \HSS for a graph containing $N^j$ vertices. The definition of the support structure is recursive.

\paragraph{Level-$1$ \HSS.} Given a graph $H$ with $|V(H)|=N$, level-$1$ \HSS for $H$ consists of a subset $S(H)$ of vertices of $H$, such that $|V(H)\setminus S(H)|\leq N^{1-\eps^4}$.

\paragraph{Level-$j$ \HSS.}
Consider now some integer $1<j\leq \ceil{1/\eps}$. Let $H$ be a graph with $|V(H)|=N^j$. A \emph{level-$j$ \HSS} for graph $H$ consists of the following:

\begin{itemize}
	\item a collection $\hset=\set{H_1,\ldots,H_{r}}$ of $r=N-\ceil{2N^{1-\eps^4}}$ graphs,  such that all vertices in sets $V(H_1),V(H_2),\ldots,V(H_{r})$ are mutually disjoint, and additionally, for all $1\leq i\leq r$:
 $V(H_i)\subseteq V(H)$; $|V(H_i)|=N^{j-1}$; and $|E(H_i)|\leq N^{j-1+32\eps^2}$ hold;

\item an embedding of the graph $H'=\bigcup_{i=1}^{r}H_i$ into graph $H$ via paths of length at most $2^{64/\eps^4}$, that causes congestion at most $N^{128\eps^2}$; 
\item for all $1\leq i\leq r$, a level-$(j-1)$ \HSS for graph $H_i$; and
\item a set $S(H)\subseteq V(H)$ of vertices, where $S(H)=\bigcup_{H_i\in \hset}S(H_i)$, and, for all $H_i\in \hset$, $S(H_i)$ is the set of vertices that is given as part of the level-$(j-1)$ \HSS for $H_i$.
\end{itemize}

Additionally, we require every graph $H_i\in \hset$ is $(\eta_{j-1},\td_{j-1})$-well-connected with respect to the set $S(H_i)$ of vertices.
We say that $\hset$ is the set of graphs associated with the level-$j$ \HSS for graph $H$.

This completes the definition of a \HSS. We will need to use the following simple claim.

\begin{claim}\label{claim size of supported}
	Let $1\leq j\leq \ceil{1/\eps}$ be an integer, and let $H$ be a graph with $|V(H)|=N^j$, together with a level-$j$ \HSS. Then $|V(H)\setminus S(H)|\leq |V(H)|\cdot \frac{4j}{N^{\eps^4}}$.
\end{claim}

\begin{proof}
	The proof is by induction on $j$. When $j=1$, then $|V(H)|=N$, and, from the definition of level-$1$ \HSS, $|V(H)\setminus S(H)|\leq N^{1-\eps^4}=\frac{|V(H)|}{N^{\eps^4}}\leq|V(H)|\cdot \frac{4j}{N^{\eps^4}} $.
	
	Consider now some integer $j>1$, and assume that the claim holds for $j-1$. Let $H$ be a graph with $|V(H)|=N^j$, for which a level-$j$ \HSS is given, and let $\hset$ be the collection of graphs associated with the structure.  Let $V_1=\bigcup_{H_i\in \hset}V(H_i)$ and $V_2=V(H)\setminus V_1$. From the definition of level-$j$ \HSS, $|\hset|=N-\ceil{2N^{1-\eps^4}}\geq N-4N^{1-\eps^4}$. Therefore, $|V_1|\geq N^j-\frac{4N^j}{N^{\eps^4}}$, and $V_2\leq \frac{4N^j}{N^{\eps^4}}=\frac{4|V(H)|}{N^{\eps^4}}$.
	
	Let $V'_1=\bigcup_{H_i\in \hset}S(H_i)$ and $V''_1=V_1\setminus V'_1$. Clearly, $V(H)\setminus S(H)=V''_1\cup V_2$. We now bound $|V''_1|$.
	
	Since, for every graph $H_i\in \hset$, by the induction hypothesis, $|V(H_i)\setminus S(H_i)|\leq |V(H_i)|\cdot \frac{4(j-1)}{N^{\eps^4}}=\frac{4(j-1)\cdot N^{j-1}}{N^{\eps^4}}$, and since $|\hset|<N$, we get that $|V''_1|\leq \frac{4(j-1)\cdot N^{j}}{N^{\eps^4}}=\frac{4(j-1)\cdot |V(H)|}{N^{\eps^4}}$.
	
	Altogether, we get that:
	
	\[|V(H)\setminus S(H)| =|V_1''|+|V_2|\leq \frac{4(j-1)\cdot |V(H)|}{N^{\eps^4}}+ \frac{4|V(H)|}{N^{\eps^4}}=\frac{4j\cdot |V(H)|}{N^{\eps^4}}\]
\end{proof}

The following theorem provides an algorithm for the Distancing Player in the Distanced Matching game. 

\begin{theorem}\label{thm: construct HSS}
	There is a large enough constant $c$, and a deterministic algorithm, whose input consists of a parameter $0<\eps<1/4$, a pair $N$, $1\leq j\leq \ceil{1/\eps}$ of integers, and a graph $H$ with $|V(H)|=N^j$, such that $N$ is sufficiently large, so that $\frac{N^{\eps^4}}{\log N}\geq 2^{128/\eps^5}$ holds. The algorithm computes one of the following:
	
	\begin{itemize}
		\item either a $(\delta_j,d)$-distancing $(A,B,E')$ in graph $H$, where $\delta_j=4j\eps^4$, $d=2^{32/\eps^4}$ and $|E'|\leq \frac{|A|}{N^{j\eps^4}}$; or
		\item a level-$j$ \HSS for $H$, such that  graph $H$ is $(\eta_j,\td_j)$-well-connected with respect to the set $S(H)$ of vertices defined by the support structure, where $\eta_j=N^{6+256j\eps^2}$ and $\td_j=2^{cj/\eps^4}$.
	\end{itemize}
	The running time of the algorithm is bounded by: \[cj\cdot N^{j(1+64\eps^2)+7}+c|E(H)|\cdot N^6. \]
\end{theorem}

We prove \Cref{thm: construct HSS} in \Cref{sec: HSS algorithm}. Note that \Cref{thm: construct HSS last level} follows from the theorem directly, by setting $j=1/\eps$. 

The following immediate corollary of the theorem can be used in order to either embed a large graph $H$ into an input graph $G$, and construct a \HSS for $H$, so that graph $H$ is well-connected with respect to the resulting set $S(H)$ of vertices given by the \HSS; or to compute a distancing in graph $G$. The latter can in turn be used in order to compute a sparse cut in $G$, via \Cref{lem: distancing to sparse cut}. We state the corollary in a slightly more general form that will be helpful for us later: we assume that, together with graph $G$, we are given a subset $T$ of its vertices called terminals, and that we are interested in embedding a large graph $H$ into $G$ with $V(H)\subseteq T$.

\begin{corollary}\label{cor: HSS witness}
	There is a deterministic algorithm, whose input consists of an $n$-vertex graph $G$, a set $T$ of $k$ vertices of $G$ called terminals, and parameters $\frac{2}{(\log k)^{1/12}}< \eps<\frac{1}{400}$, $d>1$ and $\eta>1$, such that $1/\eps$ is an integer. The algorithm computes one of the following:

	\begin{itemize}
		\item either a pair $T_1,T_2\subseteq T$ of disjoint subsets of terminals, and a set $E'$ of edges of $G$, such that:
		
		\begin{itemize}
			\item $|T_1|=|T_2|$ and $|T_1|\geq \frac{k^{1-4\eps^3}}{4}$;
			\item $|E'|\leq \frac{d\cdot |T_1|}{\eta}$; and
			\item for every pair $t\in T_1,t'\in T_2$ of terminals, $\dist_{G\setminus E'}(t,t')>d$;
		\end{itemize}

	\item or a graph $H$ with $V(H)\subseteq T$, $|V(H)|=N^{1/\eps}\geq k-k^{1-\eps/2}$, where $N=\floor{k^{\eps}}$,  and maximum vertex degree at most   $k^{32\eps^3}$, together with an embedding $\pset$ of $H$ into $G$ via paths of length at most $d$ that cause congestion at most $\eta\cdot k^{32\eps^3}$, and a level-$(1/\eps)$ \HSS for $H$, such that $H$ is $(\eta',\td)$-well-connected with respect to the set $S(H)$ of vertices defined by the support structure, where $\eta'=N^{6+256\eps}$, and $\td=2^{c/ \eps^5}$, with $c$ being the constant used in the definition of the \HSS.
	\end{itemize}
The running time of the algorithm is  $O\left (k^{1+O(\eps)}+|E(G)|\cdot k^{O(\eps^3)}\cdot(\eta+d\log n)\right )$.
\end{corollary}

\begin{proof}
	Let $q=1/\eps$, let $N=\floor{k^{\eps}}$, and let $T'\subseteq T$ be any subset of $N^q$ terminals. 
	Observe that:
	
\begin{equation}\label{eq: number of terminals}
N^q= \floor{k^{ \eps}}^{1/ \eps}\geq \left(k^{ \eps}-1\right )^{1/ \eps}=k\cdot \left(1-\frac{1}{k^{ \eps}}\right )^{1/ \eps}\geq k\cdot \left (1-\frac{1}{k^{ \eps}\cdot \eps}\right)\geq k-\frac{k^{1- \eps}}{ \eps}\geq k-k^{1-\eps/2} 
\end{equation}
	
		(we have used the fact that for all $0<\delta<1$ and $a>2$, $(1-\delta)^a\geq 1-\delta a$).

	We will attempt to construct a graph $H$ with $V(H)=T'$, together with an embedding $\pset$ of $H$ into $G$ with congestion at most $\eta\cdot k^{32\eps^3}$ and path lengths at most $d$, and a  level-$q$ \HSS for $H$, such that $H$ is $(\eta',\td)$-well-connected with respect to the set $S(H)$ of vertices defined by the support structure. If we fail to construct such a graph $H$, we will compute the sets $T_1,T_2\subseteq T$ of terminals and the set $E'\subseteq E(G)$ of edges as required.

	We will employ the distanced-matching game on graph $H$ with parameters $n=|T'|$, $\delta=4\eps^3$, and distance parameter $d'=2^{32/\eps^4}$.

	 We will use the algorithm from 
	\Cref{thm: construct HSS} for the distancing player, with parameter $j=q$, and parameters $N$ and $\eps$ remaining unchanged. In order to be able to use the algorithm, we need to verify that
	$\frac{N^{\eps^4}}{\log N}\geq 2^{128/\eps^5}$ holds. 
Recall that, from the conditions of \Cref{cor: HSS witness}, $\frac{2}{(\log k)^{1/12}}< \eps<\frac{1}{400}$. Therefore, $\log k>(2/\eps)^{12}$
	and $k>2^{(2/\eps)^{12}}$. Moreover, from the above calculations, $\frac{\log k}{\log\log k}>\frac{1}{\eps^8}$ holds, and so $k^{\eps^8}>\log k$ must hold.
	Altogether, we get that:

\begin{equation}\label{eq: param math}
 \frac{N^{\eps^4}}{\log N}\geq \frac{k^{ \eps^5}}{2\cdot \log k}\geq 
k^{ \eps^6}\geq 2^{128/\eps^5}. 
\end{equation}

We will bound the number of iterations of the Distanced Matching game via \Cref{thm: distancing-matching game - number of iterations}. In order to use \Cref{thm: distancing-matching game - number of iterations}, we need to verify that $\frac{|T'|^{\delta}}{\log (|T'|)}\geq \frac{2^{14}}{\delta^2}$, or, equivalently: $\frac{|T'|^{4\eps^3}}{\log (|T'|)}\geq \frac{2^{10}}{\eps^6}$.
Recall that $|T'|=N^q=N^{1/\eps}$, and so 
$\frac{|T'|^{4\eps^3}}{\log (|T'|)}=\frac{\eps\cdot N^{4\eps^2}}{\log N}\geq \frac{N^{\eps^4}}{\log N}\geq 2^{128/\eps^5}\geq \frac{2^{10}}{\eps^6}$ from Equation \ref{eq: param math}, and since $\eps\leq 1/400$. We also need to verify that $d'\geq 2^{4/\delta}$. Since $d'=2^{32/\eps^4}$ and  $\delta=4\eps^3$, this is immediate to verify.
From \Cref{thm: distancing-matching game - number of iterations}, we can now conclude that the number of iterations in a Distanced Matching game with parameters $n=|T'|$, $\delta=4\eps^3$, and distance parameter $d'=2^{32/\eps^4}$ is bounded by $|T'|^{8\delta}\leq k^{32\eps^3}$.

We start with a graph $H$ whose vertex set is $V(H)=T'$, and edge set is $E(H)=\emptyset$, and then iterate. In each iteration $i$, we will add some set $E_i$ of edges to graph $H$, and we will define an embedding $P(e)$ for every edge $e\in E_i$. We now describe the execution of a single iteration.

\subsection*{Execution of Iteration $i$.}
We apply the algorithm from \Cref{thm: construct HSS} to the current graph $H$, with parameter $j=q$, and parameters $N,\eps$ remaining unchanged. Note that $|V(H)|=|T'|=N^q$, and, as we have established above, $\frac{N^{\eps^4}}{\log N}\geq 2^{128/\eps^5}$ holds.

We now consider two cases. The first case is when the algorithm from \Cref{thm: construct HSS} returns a $(\delta_q,d')$-distancing $(X_i,Y_i,E'_i)$ in graph $H$, where $\delta_q=4j\eps^4=4q\eps^4=4\eps^3=\delta$ (since $q=1/\eps$), and $d'=2^{32/\eps^4}$. In this case, we say that iteration $i$ is \emph{regular}. We view this distancing as the response of the distancing player in iteration $i$ of the Distanced Matching game that we play on graph $H$.

We then apply Procedure \procpathpeel from \Cref{lem: path peel} to graph $G$ and sets $A_1=X_i,B_1=Y_i$ of its vertices, together with parameters $d$ and $\eta$ from the statement of \Cref{cor: HSS witness}. Let $\qset_1$ denote the collection of paths that the algorithm returns.
Recall that every path in $\qset_1$ connects some vertex of $X_i$ to a vertex of $Y_i$, and that every vertex of $X_i\cup Y_i$ may serve as an endpoint of at most one such path. We let $M_i\subseteq X_i\times Y_i$ be the matching that is defined by the paths in $\qset_1$: a pair $(x,y)$ of vertices with $x\in X_i,y\in Y_i$ is added to $M_i$ iff some path $Q(x,y)\in \qset_1$ has endpoints $x,y$. 
 We again consider two cases. The first case happens if $|M_i|\geq |X_i|/2$. In this case, we say that iteration $i$ is successful. We obtain a collection $E_i\subseteq M_i$ of edges as follows: we start with $E_i=M_i$, and we delete from $E_i$ all pairs $(x,y)$ of vertices where edge $(x,y)$ lies in $E'_i$. Since, from the definition of distancing, $|E'_i|\leq |X_i|/16$, we get that $|E_i|\geq |X_i|/4$. We let $\pset_i=\set{Q(x,y)\mid (x,y)\in E_i}$ be the collection of paths that route the pairs of vertices in $E_i$. We add the edges of $E_i$ to graph $H$, and we view $E_i$ as the response of the matching player in iteration $i$. We also view the set $\pset_i$ of paths in graph $H$ as an embedding of the set $E_i$ of edges. Recall that each path in $\pset_i$ has length at most $d$, and the paths in $\pset_i$ cause congestion at most $\eta$. We then continue to the next iteration.
 
 The second case happens if $|M_i|<|X_i|/2$. In this case we say that iteration $i$ is unsuccessful. Let $E'$ be the set of all edges $e$ in graph $G$ that participate in exactly $\eta$ paths in $\qset_1$. Let $X'\subseteq X_i$ and $Y'\subseteq Y_i$ be the sets of vertices that do not serve as endpoints of the paths in $\qset_1$. Recall that \Cref{lem: path peel} guarantees (via Property \ref{pp: distancing}) that 
 the length of the shortest path connecting a vertex of $X'$ to a vertex of $Y'$ in $G\setminus E'$ is greater than $d$.
 
 Recall that $|X'|=|Y'|\geq \frac{|X_i|}{2}\geq \frac{|T'|^{1-\delta}}{2}\geq \frac{k^{1-\delta}}{4}\geq \frac{k^{1-4\eps^3}}{4}$. Since every path in $\qset_1$ has length at most $d$, we get that $\sum_{Q\in \qset_1}|E(Q)|\leq d\cdot \frac{|X_i|}{2}\leq d\cdot |X'|$. Since set $E'$ contains edges that participate in $\eta$ paths in $\qset_1$, we get that: $|E'|\leq \frac{d\cdot |X'|}{\eta}$. We return the set $E'$ of edges and the sets $T_1=X'$, $T_2=Y'$ of terminals. From the above discussion, $|T_1|=|T_2|$, $|T_1|\geq \frac{k^{1-4\eps^3}}{4}$, and $|E'|\leq \frac{d\cdot |T_1|}{\eta}$ hold as required. Moreover, for every pair $t\in T_1,t'\in T_2$ of terminals, $\dist_{G\setminus E'}(t,t')>d$.
 
 It remains to consider the second case, when the algorithm from \Cref{thm: construct HSS}  constructs a level-$q$ \HSS for $H$, such that  graph $H$ is $(\eta_q,\td_q)$-well-connected with respect to the set $S(H)$ of vertices defined by the support structure, where  $\td_q=2^{cq/\eps^4}=2^{c/\eps^5}=\td$, and:

\[\eta_q=N^{6+256q\eps^2}=N^{6+256\eps}=\eta'.\]

In this case, we say that iteration $i$ is irregular.
Recall that $|V(H)|=|T'|=N^q\geq k-k^{1-\eps/2}$ from Inequality \ref{eq: number of terminals}. As observed already, the number of iterations in the Distanced Matching game is bounded by $k^{32\eps^3}$, and so the maximum vertex degree in $H$ is bounded by  $k^{32\eps^3}$, and the number of edges in graph $H$ is bounded by $k^{1+32\eps^3}$ throughout the algorithm. If we denote by $z\leq k^{32\eps^3}$ the number of iterations in the Distanced Matching game, then for all $1\leq i\leq z$, we have constructed a set $\pset_i$ of paths in graph $H$ embedding the edges of $E_i$. The paths in $\pset_i$ have length at most $d$ each, and they cause congestion at most $\eta$ in $G$. By letting $\pset=\bigcup_{i=1}^z\pset_i$, we obtain an embedding of graph $H$ into $G$ via paths of length at most $d$, that cause congestion at most $\eta\cdot z\leq \eta\cdot k^{32\eps^3}$. We output graph $H$, its embedding $\pset$, and the level-$q$ \HSS for $H$.

This completes the description of a single iteration. It now remains to bound the running time of the algorithm.

\subsection*{Running Time Analysis}
Recall that, throughout the algorithm, $|E(H)|\leq k^{1+32\eps^3}$ holds, and that the number of regular successful iterations in the algorithm is at most $k^{32\eps^3}$. Additionally, there could be at most one irregular iteration, and at most one regular but unsuccessful iteration.

We now bound the running time of a single iteration. This running time is dominated by the running times of the algorithms from \Cref{thm: construct HSS} and \Cref{lem: path peel}.

The former is bounded by:

\[O(q\cdot N^{q(1+64\eps^2)+7}+|E(H)|\cdot N^6)\leq O\left (k^{1+O(\eps)}+k^{1+32\eps^3}\cdot k^{O(\eps)}\right )\leq O\left (k^{1+O(\eps)}\right ). \]

The latter is bounded by:
$O(|E(G)|(\eta+d\log n))$.

The total running time of the algorithm is then bounded by:

\[  O\left (k^{1+O(\eps)}+|E(G)|\cdot k^{O(\eps^3)}\cdot(\eta+d\log n)\right ).\]
\end{proof}

\section{Algorithm for the Distancing Player -- Proof of \Cref{thm: construct HSS}}
\label{sec: HSS algorithm}

This section is dedicated to proving \Cref{thm: construct HSS}. 
Throughout the proof we use the following three parameters: $\eps'=\eps^4$;  $\Delta=64/\eps'$; and $d'=2d\cdot \Delta$. Note that:

\begin{equation}\label{eq: bound on d'}
d'=2\Delta d=\frac{128d}{\eps'}=\frac{128\cdot 2^{32/\eps^4}}{\eps^4}\leq 2^{64/\eps^4}<\frac 1 4\cdot 2^{c/\eps^4}, 
\end{equation}

since $c$ is large enough.

The proof is by induction on $j$.
We start with the base case, where $j=1$.

\subsection*{Base Case: $j=1$}

We assume that we are given a graph $H$ on $N$ vertices. Our goal is to  either compute a $(\delta_1,d)$-distancing in graph $H$, or to construct a set $S(H)\subseteq V(H)$ of vertices, such that $|V(H)\setminus S(H)|\leq N^{1-\eps^4}$, and graph $H$ is $(\eta_1,\td_1)$-well-connected with respect to $S(H)$.

We apply Algorithm \procsep from \Cref{lem: find separation of terminals} to graph $H$, with the set $T=V(H)$ of terminal vertices, with distance parameters $d$ and $\Delta$ remaining the same, and parameter $\alpha=1-\frac{1}{N^{\eps'}}$. 

Assume first that the outcome of the algorithm is a pair $T_1,T_2\subseteq V(H)$ of subsets of vertices, with $|T_1|=|T_2|$, such that for every pair $t\in T_1$, $t'\in T_2$ of vertices, $\dist_H(t,t')\geq d$, and additionally:

\[|T_1|\geq N^{1-64/\Delta}\cdot \min\set{(1-\alpha),\frac 1 3}\geq N^{1-2\eps'}.\]

(recall that $\Delta=64/\eps'$ and $N^{\eps'}>3$).
In this case, since $\delta_1=4\eps'$, we obtain a $(\delta_1,d)$-distancing $(T_1,T_2,\emptyset)$ in graph $H$. We return this distancing as the outcome of the algorithm.

Otherwise, Procedure \procsep must return a vertex $v\in V(H)$, with $|B_H(v,\Delta\cdot d)|>\alpha\cdot N=N\cdot \left(1-\frac{1}{N^{\eps'}}\right )$.
In this case, we set $S(H)=B_H(v,\Delta\cdot d)$, and we report that graph $H$ is $(\eta_1,\td_1)$-well-connected with respect to $S(H)$. We also return $S(H)$ as level-$1$ \HSS for $H$. Note that $|V(H)\setminus S(H)|\leq N^{1-\eps'}=N^{1-\eps^4}$ as required. It remains to show that graph $H$ is indeed $(\eta_1,\td_1)$-well-connected with respect to $S(H)$. Recall that $\eta_1>N$ and $d'<\td_1$ from Inequality \ref{eq: bound on d'}, since $\td_1=2^{c/\eps^4}$. Let $A,B\subseteq S(H)$ be any pair of equal-cardinality subsets of vertices. We define an arbitrary perfect matching $M\subseteq A\times B$ between vertices of $A$ and vertices of $B$. Consider now any pair $(a,b)\in M$ of matched vertices. Since $a,b\in B_H(v,\Delta\cdot d)$, there is a path $P(a,b)$ connecting $a$ to $b$ in $H$ of length at most $2\Delta \cdot d=d'\leq \td_1$. We then let $\pset(A,B)=\set{P(a,b)\mid (a,b)\in M}$ be the resulting collection of paths, that routes every vertex of $A$ to a distinct vertex of $B$. The length of every path in $\pset(A,B)$ is at most $\td_1$, and the congestion caused by the paths in $\pset(A,B)$ is at most $|\pset(A,B)|\leq |V(H)|\leq N<\eta_1$.

The running time of the algorithm is dominated by the running time of Procedure \procsep, which is bounded by $O(|E(H)|\cdot N^{64/\Delta})\leq O(|E(H)|\cdot N^{\eps'})<c|E(H)|\cdot N^6$, if $c$ is large enough, since $\Delta=64/\eps'$.

\subsection*{Step: $1<j\leq h$}
We now assume that we are given some integer $1<j\leq \ceil{1/\eps}$, such that the statement of \Cref{thm: construct HSS} holds for $j-1$, and we prove the statement of the theorem for $j$. Recall that we are given as input a graph $H$ with $|V(H)|=N^j$.

We partition the set $V(H)$ of vertices into $N$ subsets $V_1,\ldots,V_N$, each of which contains exactly $N^{j-1}$ vertices. The algorithm consists of two phases. In the first phase, we run the \DMG in parallel on $N$ graphs $H_1,\ldots,H_N$, where for all $1\leq i\leq N$, $V(H_i)=V_i$. We will add edges to graphs $H_1,\ldots,H_N$ gradually via the \DMG, while computing an embedding of all resulting edges into graph $H$, so that the resulting embedding paths have sufficiently low length and cause sufficiently low congestion. In this phase, we will either compute the required $(\delta_j,d)$-distancing in  $H$, or we will be able to successfully complete the \DMG on a sufficiently large collection $\hset'\subseteq \set{H_1,\ldots,H_N}$ of the graphs. In the latter case, for each such graph $H_i\in \hset'$, we will also compute a level-$(j-1)$ \HSS for $H_i$. If the outcome of the first phase is a $(\delta_j,d)$-distancing for $H$, then we terminate the algorithm and return this distancing. Otherwise, we continue to the second phase. In the second phase we will exploit the graphs in $\hset'$ to either compute a $(\delta_j,d)$-distancing in graph $H$, or to compute a  
subset $\hset''\subseteq \hset'$ of $r$ graphs, so that, if we let $S(H)=\bigcup_{H_i\in \hset''}S(H_i)$, then graph $H$ is $(\delta_j,\td_j)$-well-connected with respect to the set $S(H)$ of vertices. 
We now describe each of the two phases in turn.

\subsection{Phase 1: Construction of Smaller Well-Connected Graphs}

In this phase, we gradually construct a collection $\hset=\set{H_1,\ldots,H_N}$ of graphs, over the course of at most $N^{64\eps^2}$ iterations, by running the \DMG over these graphs in parallel with parameters $\delta=\delta_{j-1}$ and $d=2^{32/\eps^4}$. Initially, for all $1\leq i\leq N$, we let $V(H_i)=V_i$ and $E(H_i)=\emptyset$. In every iteration $q\geq 1$, we will compute, for all $1\leq i\leq N$, a partial matching $E_i^q$ over the vertices of $V_i$. We will ensure that either $E_i^q=\emptyset$, or $|E_i^q|\geq N^{(j-1)(1-\delta_{j-1})}/8$. The edges of $E_i^q$ are then added to graph $H_i$. Additionally, in the $q$th iteration, we will compute an embedding $\pset^q$ of all edges in $\bigcup_{i=1}^NE_i^q$ into $H$. 
We now describe a single iteration $q$.

\subsubsection{Description of Iteration $q$}
Using the induction hypothesis, we apply the algorithm from \Cref{thm: construct HSS} to each of the graphs $H_i\in \hset$, with parameter $(j-1)$. We denote by $\hset^1\subseteq \hset$ the collection of all graphs $H_i$, for which the algorithm returned a $(\delta_{j-1},d)$-distancing in $H_i$, and we denote by $\hset^2=\hset\setminus \hset^1$ all remaining graphs. Recall that, for each graph $H_i\in \hset^2$, the algorithm computes a level-$(j-1)$ \HSS. 
We now consider two cases, depending on whether $|\hset^1|>N^{1-\eps'}$.

\paragraph{Case 1: $|\hset^1|>N^{1-\eps'}$.}
If this case happens, then we say that iteration $q$ is \emph{regular}. Recall that, for every graph $H_i\in \hset^1$, the algorithm from \Cref{thm: construct HSS} returned a $(\delta_{j-1},d)$- distancing $(A_i,B_i,E'_i)$, where $|A_i|=|B_i|\geq |V(H_i)|^{1-\delta_{j-1}}=N^{(j-1)(1-\delta_{j-1})}$, and $|E'_i|\leq \frac{|A_i|}{N^{(j-1)\eps'}}$.

Let $z'=\ceil{\frac{2\log N}{\eps}}$.
We further partition the collection $\hset^1$ of graphs into subsets $\hset^1_0,\ldots,\hset^1_{z'}$, where for all $0\leq z\leq z'$, class $\hset^1_z$ contains all graphs $H_i\in \hset^1$, for which $\frac{N^{j-1}}{2^{z+1}}\leq  |A_i|< \frac{N^{j-1}}{2^{z}}$ holds. Clearly, there must be an index $0\leq z\leq z'$, such that:

\begin{equation}
|\hset^1_z|\geq \frac{|\hset^1|\cdot \eps}{4\log N}\geq \frac{N^{1-\eps'}\cdot \eps}{4\log N}\geq N^{1-2\eps'}.\label{H1z bound}
\end{equation}

(since $\frac{N^{\eps'}}{\log N}\geq 2^{128/\eps^5}$ from the statement of \Cref{thm: construct HSS}).

From now on we fix this index $z$. 
Since, for all $H_i\in \hset^1$, $|A_i|\geq N^{(j-1)(1-\delta_{j-1})}\geq N^{j-j\delta_{j-1}-1}$, we get that:

\begin{equation}\label{eq: bound on z}
2^z\leq N^{j\delta_{j-1}}.
\end{equation}

For convenience, we denote by $I\subseteq \set{1,\ldots,N}$ the set of all indices $i$ with $H_i\in \hset^1_z$.
Next, we apply Algorithm \procpathpeel from \Cref{lem: path peel}, to graph $H$, collections of subsets $\set{(A_i,B_i)}_{i\in I}$ of its vertices, length parameter $d'$, and congestion parameter $\eta=N^{8\eps^3}$. Consider the resulting collections $\set{\pset^q_i}_{i\in I}$ of paths. Recall that, from  \Cref{lem: path peel}, for all $i\in I$, every path in the corresponding set $\pset^q_i$ has length at most $d'$, and it connects a vertex of $A_i$ to a vertex of $B_i$, so that the endpoints of all paths in $\pset_i$ are distinct. Therefore, we can use the paths in $\pset^q_i$ in order to define a partial matching $M_i^q$ between vertices of $A_i$ and vertices of $B_i$: a pair $(a,b)\in A_i\times B_i$ of vertices is added to $M_i^q$ iff some path of $\pset^q_i$ has endpoints $a,b$. We denote by $A'_i\subseteq A_i$ and $B'_i\subseteq B_i$ the subsets of vertices of $A_i$ and $B_i$ respectively, that do not serve as endpoints of the paths in $\pset_i^q$. Let $E^*\subseteq E(H)$ be the set of all edges $e$ that participate in exactly $\eta$ paths of $\bigcup_{i\in I}\pset^q_i$. Recall that \Cref{lem: path peel} further guarantees that the paths of $\bigcup_{i\in I}\pset^q_i$ cause congestion at most $\eta$ in $H$. Moreover, if we consider graph $H\setminus E^*$, then for all $i\in I$, the length of a shortest path connecting a vertex of $A'_i$ to a vertex of $B'_i$ is greater than $d'$.

Running time of Algorithm \procpathpeel is $O(|E(H)|\cdot (\eta+jNd'\log N))\leq O(|E(H)|\cdot (N^{8\eps^3}+jNd'\log N))\leq O(|E(H)|\cdot jNd'\log N)$, since $|\hset|=N$.

We say that a graph $H_i\in \hset^1_z$ is \emph{successful} if $|\pset_i^q|\geq |A_i|/4$, and it is unsuccessful otherwise. We let $\gset^s\subseteq \hset^1_z$ be the collection of all successful graphs, and $\gset^u=\hset^1_z\setminus \gset^s$ the collection of all unsuccessful graphs. For convenience, we also partition the collection $I$ of indices into a set $I^s$ containing all indices $i\in I$ where $H_i\in \gset^s$ and $I^u=I\setminus I^s$. 
We consider again two cases, depending on whether $|\gset^s|\geq  |\hset^1_z|/2$.

\paragraph{Case 1a: $|\gset^s|\geq |\hset^1_z|/2$.}
If Case 1a happens then we say that iteration $q$ is \emph{successful}. Consider some index $i\in I^s$. 
Recall that the matching $M_i^q$ that we have defined contains at least $|A_i|/4$ pairs of vertices (that we refer to as edges), while $|E'_i|\leq |A_i|/16$ by the definition of distancing. We discard from $M_i^q$ all edges that lie in $E'_i$; note that $|M_i^q|\geq |A_i|/8$ continues to hold. We add the edges of the resulting matching $M_i^q$ to graph $H_i$, and we say that graph $H_i$ \emph{received a matching} in iteration $q$. We view this matching as the response of the matching player in the \DMG. We let $\pset^q=\bigcup_{i\in I^s}\pset_i^q$. Notice that the paths in $\pset^q$ can be viewed as an embedding of the set $\bigcup_{i\in I^s}M_i^q$ of edges into graph $H$. The length of each path is at most $d'$, and the congestion of the embedding is at most $\eta$. Observe that, if an iteration is successful, then the number of graphs in $\hset$ that receive matchings is at least $|\gset^s|\geq |\hset^1_z|/2\geq N^{1-2\eps'}/2$ (from Equation \ref{H1z bound}). We terminate the current iteration, and proceed to the next iteration.

\paragraph{Case 1b: $|\gset^s|<|\hset^1_z|/2$.} If Case 1b happens, then we say that the current iteration is \emph{unsuccessful}.
In this case, we will construct a $(\delta_j,d)$-distancing $(A,B,E^*)$ in graph $H$, where $E^*$ is the set of edges that we have defined above. We will ensure that $|E^*|\leq |A|/N^{j\eps^4}$. The current iteration then becomes the last iteration of the algorithm, and we return $(A,B,E^*)$ as its output.
We need the following simple observation bounding $|E^*|$:

\begin{observation}\label{obs: bound deleted edges}
	$|E^*|\leq \frac{N^{j}\cdot d'}{\eta\cdot 2^z}$.
\end{observation}
\begin{proof}
Note that for every graph $H_i\in \hset^1_z$, $|\pset^q_i|\leq |A_i|\leq \frac{N^{j-1}}{2^z}$.  Since the length of each such path is at most $d'$, the total number of edges on all paths in $\bigcup_{i\in I}\pset^q_i$ is $\sum_{i\in I}\sum_{P\in \pset_i^q}|E(P)|\leq d'\cdot  |\hset^1_z|\cdot \frac{N^{j-1}}{2^z}\leq \frac{N^{j}\cdot d'}{2^z}$ (as $|\hset^1_z\leq N$). Since set $E^*$ only contains edges that lie on $\eta$ path of $\pset^q$, the observation follows.
\end{proof}

Recall that, if $H_i\in \gset^u$, then $|A'_i|\geq \frac{|A_i|}{2}\geq \frac{N^{j-1}}{2^{z+2}}$. Recall also that $|B'_i|=|A'_i|$, and $\dist_{H\setminus E^*}(A'_i,B'_i)>d'$. 

We denote $T=\bigcup_{i\in I^u}(A'_i\cup B'_i)$, and we call the vertices of $T$ \emph{terminals}. 
Since we have assumed that $|\gset^s|<|\hset^1_z|/2$, we get that:

\begin{equation}\label{eq: bound T}
|T|\geq |\gset^u|\cdot  \frac{N^{j-1}}{2^{z+1}}\geq  |\hset^1_z|\cdot \frac{N^{j-1}}{2^{z+2}}\geq \frac{N^{j-2\eps'}}{2^{z+2}}.
\end{equation}

(we have used Equation \ref{H1z bound}.)

We need the following simple observation.

\begin{observation}\label{obs: small ball}
	For every terminal $t\in T$, at most $|T|/2$ terminals may lie in $B_{H\setminus E^*}(t,d'/2)$.
\end{observation}
\begin{proof}
Consider any terminal $t\in T$, and denote $B^t=B_{H\setminus E^*}(t,d'/2)$. Recall that for all $H_{i'}\in \gset^u$, $\dist_{H\setminus E^*}(A'_{i'},B'_{i'})> d'$. Therefore, $B^t$ may not contain a vertex of $A'_{i'}$ and a vertex of $B'_{i'}$. We conclude that, for every terminal $t\in T$, $|B^t\cap T|\leq |T|/2$.
\end{proof}

We apply Algorithm \procsep from \Cref{lem: find separation of terminals} to graph $G=H\setminus E^*$, the set $T$ of terminals, distance parameters $d$ and $\Delta$ that remain unchanged, and $\alpha=1/2$. 
The running time of the algorithm is $O(|E(H)|\cdot |V(H)|^{64/\Delta})\leq O(|E(H)|\cdot N^{j\eps'})$.

From \Cref{obs: small ball}, the algorithm may not return a terminal $t\in T$ with $|B_{H\setminus E^*}(t,\Delta\cdot d)\cap T|=|B_{H\setminus E^*}(t,d'/2)\cap T|>\alpha \cdot |T|$. Therefore, it must compute two subsets $T_1,T_2$ of terminals, with $|T_1|=|T_2|$, such that for every pair $t\in T_1,t'\in T_2$ of terminals, $\dist_{H\setminus E^*}(t,t')\geq d$. Moreover:

\[\begin{split}
|T_1| \geq \frac{|T|^{1-64/\Delta}}3
\geq \frac{|T|^{1-\eps'}}3
\geq \frac{N^{(j-2\eps')(1-\eps')}}{3\cdot 2^{(z+2)(1-\eps')}}
\geq \frac{N^{j-2j\eps'}}{2^{z+4}}
\end{split} \]

(we have used the fact that $\Delta=64/\eps'$, and Equation \ref{eq: bound T}).

Recall  that, from \Cref{obs: bound deleted edges},
$|E^*|\leq \frac{N^{j}\cdot d'}{\eta\cdot 2^z}$.
Recall also that $\eta=N^{8\eps^3}$, $j\leq \ceil{1/\eps}$, and $d'\leq 2^{64/\eps^4\leq N^{\eps'}}$ from Inequality \ref{eq: bound on d'}, and since $\frac{N^{\eps'}}{\log N}\geq 2^{128/\eps^5}$ from the statement of \Cref{thm: construct HSS}. Therefore, we get that:

\[
\begin{split}
|E^*|\leq \frac{N^j}{2^{z}}\cdot\frac{N^{\eps^4}}{N^{8\eps^3}}
\leq \frac{N^j}{2^{z+4}}\cdot \frac{1}{N^{6\eps^3}} \leq \frac{N^j}{2^{z+4}}\cdot \frac{1}{N^{3\eps^4\cdot \ceil{1/\eps}}}
\leq \frac{N^j}{2^{z+4}}\cdot\frac{1}{N^{3\eps'j}}
\leq
\frac{|T_1|}{N^{j\eps'}}.
\end{split}\]

Lastly, recall that we have shown in Inequality \ref{eq: bound on z}, that $2^z\leq N^{j\delta_{j-1}}$. Therefore, we get that: 

\[|T_1|\geq \frac{N^j}{16N^{j(2\eps'+\delta_{j-1})}}\geq 
\frac{N^j}{16N^{j(4(j-1)\eps'+2\eps')}}\geq \frac{N^j}{N^{j(4j\eps')}}=N^{j(1-\delta_j)}.\]

 We conclude that $(T_1,T_2,E^*)$ is a $(\delta_j,d)$-distancing, with $|E^*|\leq |T_1|/N^{j\eps'}$ as required. We return this distancing and terminate the algorithm.

\paragraph{Case 2: $|\hset^1|\leq N^{1-\eps'}$.}
If this case happens, then we say that the current iteration is \emph{irregular}. In this case, we terminate Phase 1. The outcome of the phase is the collection $\hset^2\subseteq \hset$ of at least $N-N^{1-\eps'}$ graphs.
 Recall that, for each graph $H_i\in \hset^2$, we computed a level-$(j-1)$
 \HSS, that includes a subset $S(H_i)\subseteq V(H_i)$ of vertices, such that $H_i$ is $(\eta_{j-1},\td_{j-1})$-well-connected with respect to $S(H_i)$.
Let $H'=\bigcup_{H_i\in \hset^2}H_i$. Notice that the sets $\pset^q$ of paths that are computed in each iteration provide an embedding of graph $H'$ into $H$. The length of each resulting path is bounded by $d'$. We use the following observation in order to both bound the congestion of this embedding, and the running time of the algorithm.

\begin{observation}\label{obs: num of iterations}
	At the end of Phase 1, for each graph $H_i\in \hset$, $|E(H_i)|\leq N^{j-1+32\eps^2}$. The number of iterations in Phase 1 is at most $N^{64\eps^2}$.
\end{observation}
\begin{proof}
Recall first that $\delta_{j-1}=4(j-1)\eps'\geq 4 \eps^4$, while $d=2^{32/\eps^4}$. Therefore, $d\geq 2^{4/\delta_{j-1}}$ holds. We can then view our algorithm from Phase 1 as running the \DMG simultaneously over the graphs in $\hset$. 
Note that for every graph $H_i\in \hset$:

\[|V(H_i)|^{\delta_{j-1}}=N^{(j-1)\delta_{j-1}}\geq N^{4\eps'}\geq \frac{2^{128}\log N}{\eps^{11}}\geq \frac{2^{124}\log N}{16(\eps')^2\eps^3}\geq \frac{2^{14}(j-1)\log N}{\delta_{j-1}^2}=\frac{2^{14}\log (|V(H_i)|)}{\delta_{j-1}^2},
\]

since
$\delta_{j-1}=4(j-1)\eps'$, $2\leq j\leq \ceil{1/\eps}$, and $\frac{N^{\eps^4}}{\log N}\geq 2^{128/\eps^5}$ from the statement of \Cref{thm: construct HSS}.
We conclude that $|V(H_i)|^{\delta_{j-1}}\geq \frac{2^{14}\log (|V(H_i)|)}{\delta_{j-1}^2}$ holds, satisfying the conditions of \Cref{thm: distancing-matching game - number of iterations}.
From \Cref{thm: distancing-matching game - number of iterations}, the number of iterations in a \DMG in a single graph $H_i\in \hset$ is bounded by:

\[|V(H_i)|^{8\delta_{j-1}}\leq N^{(j-1)\cdot 8\delta_{j-1}}\leq N^{32(j-1)^2\eps'}\leq N^{32\eps^2},\]

since $\eps'=\eps^4$,  $\delta_{j-1}=4(j-1)\eps'$, and $j\leq \ceil{1/\eps}$.
Therefore, a graph $H_i\in \hset$ may receive a matching in at most $N^{32\eps^2}$ iterations, and the cardinality of each such matching is at most $|V(H_i)|/2\leq N^{j-1}/2$. We conclude that at the end of Phase 1, for each graph $H_i\in \hset$, $|E(H_i)|\leq N^{j-1}\cdot N^{32\eps^2}$.

Next, we bound the number of iterations in Phase 1.
Note that at most one iteration of Phase 1 may be irregular, and at most one iteration may be regular and unsuccessful. It is now enough to bound the number of regular and successful iterations. In every regular and successful iteration, at least $N^{1-2\eps'}/2$ graphs in $\hset$ receive matchings. Therefore, every regular and successful iteration of Phase 1 results in the completion of a single iteration of the \DMG in at least $N^{1-2\eps'}/2$ graphs of $\hset$. 
At the same time, the number of pairs $(i,q)$, where graph $H_i\in \hset$ receives a matching in iteration $q$ must be bounded by $|\hset|\cdot 	 N^{32\eps^2}\leq N^{1+32\eps^2}$. Since in every regular and successful iteration at least $N^{1-2\eps'}/2$ graphs in $\hset$ receive matchings, we get that the number of iterations is bounded by $2N^{32\eps^2+2\eps'}\leq N^{64\eps^2}$.
\end{proof}

Since every set $\pset^q$ of paths causes congestion at most $\eta=N^{8\eps^3}$, and the number of iterations is bounded by $N^{64\eps^2}$, we obtain an embedding of graph $H'$ into $H$ via paths of length at most $d'\leq 2^{64/\eps^4}$ (from Equation \ref{eq: bound on d'}), that cause total congestion at most $N^{64\eps^2}\cdot N^{8\eps^3}\leq N^{128\eps^2}$. We denote this embedding by $\pset$.

Note that the Phase 1 can either terminate with a regular unsuccessful iteration, in which case we terminate the algorithm and return the resulting distancing for graph $H$, or with an irregular iteration, in which case we proceed to Phase 2 of the algorithm. Before we describe Phase 2 of the algorithm, we analyze the running time of Phase 1.

\subsubsection*{Running Time Analysis of Phase 1.}

We start by bounding the running time of a single iteration. Over the course of the iteration, we apply the algorithm from \Cref{thm: construct HSS} to each of the graphs $H_i\in \hset$, with parameter $(j-1)$. Recall that, from \Cref{obs: num of iterations}, for each graph $H_i\in \hset$, $|E(H_i)|\leq N^{j-1+32\eps^2}$. 
From the induction hypothesis, the running time of the algorithm from  \Cref{thm: construct HSS} on a single graph $H_i\in \hset$ is bounded by:

\[
\begin{split}
c(j-1)\cdot N^{(j-1)(1+64\eps^2)+7}&+c|E(H_i)|\cdot N^6\\
&\leq c(j-1)\cdot N^{j+64(j-1)\eps^2+6}+cN^{j-1+32\eps^2}\cdot N^6\\
&=  c(j-1)\cdot N^{j+64(j-1)\eps^2+6}+cN^{j+5+32\eps^2}. \\
&\leq cj\cdot N^{j+64(j-1)\eps^2+6}.
\end{split}
\]

Since $|\hset|=N$, the running time of this part of the algorithm is bounded by:

\[  cj\cdot N^{j+64(j-1)\eps^2+7}\]

If the iteration is regular, then  we apply Algorithm \procpathpeel from \Cref{lem: path peel}, to graph $H$, collections of subsets $\set{(A_i,B_i)}_{i\in I}$ of its vertices, length parameter $d'$, and congestion parameter $\eta=N^{8\eps^3}$. As observed above, the running time of Algorithm \procpathpeel is $O(|E(H)|\cdot jNd'\log N)=O(|E(H)|\cdot N\cdot 2^{64/\eps^4}\cdot \log N/\eps)\leq O(|E(H)|\cdot N^2)$, since $d'\leq 2^{64/\eps^4}$ from Equation \ref{eq: bound on d'} and  $N^{\eps'}>2^{64/\eps^5}\cdot \log N$ from the statement of \Cref{thm: construct HSS}.

If the iteration is regular and unsuccessful, then
we apply Algorithm \procsep from \Cref{lem: find separation of terminals}, whose running time, as observed above, is bounded by $O(|E(H)|\cdot N^{j\eps'})\leq O(|E(H)|\cdot N^{2\eps^3})$, as $j\leq \ceil{1/\eps}$ and $\eps'=\eps^4$.

Overall, the running time of a single iteration is bounded by:

\[ cj\cdot N^{j+64(j-1)\eps^2+7}+ O\left (|E(H)|\cdot N^2\right ).\]

Since, from \Cref{obs: num of iterations}, the number of iterations in Phase 1 is at most  $N^{64\eps^2}$,
we get that the total running time of Phase 1 is bounded by:

\[ cj\cdot N^{j+64j\eps^2+7}+ O\left (|E(H)|\cdot N^3\right ).\]

\subsection{Phase 2: Distancing or Well-Connectedness}

The starting point of Phase 2 is the collection $\hset^2\subseteq \hset$ of at least $N-N^{1-\eps'}$ graphs that was computed in Phase 1.
Recall that, for each graph $H_i\in \hset_2$, we computed a level-$(j-1)$
\HSS, that includes a subset $S(H_i)\subseteq V(H_i)$ of vertices, such that $H_i$ is $(\eta_{j-1},\td_{j-1})$-well-connected with respect to $S(H_i)$. Additionally, we have computed an embedding $\pset$ of graph $H'=\bigcup_{H_i\in \hset^2}H_i$, so that the length of each path in $\pset$ is at most $d'\leq 2^{64/\eps^4}$, and the paths in $\pset$ cause congestion at most  $N^{128\eps^2}$. 

In this phase we will either compute a subset $\hset'\subseteq \hset^2$ of $r$ graphs, such that graph $H$ is $(\eta_j,\td_j)$-well-connected with respect to the set $S(H)=\bigcup_{H_i\in \hset'}S(H_i)$ of vertices; or we compute a $(\delta_j,d)$-distancing $(A,B,E')$ in graph $H$, with $|E'|\leq \frac{|A|}{N^{j\eps^4}}$.

We consider every pair $H_i,H_{i'}\in \hset^2$ of graphs, with $i<i'$ one by one. When the pair $(H_i,H_{i'})$ of graphs is considered, we apply Procedure \procpathpeel from \Cref{lem: path peel} to graph $H$, and two sets $A_1=S(H_i),B_1=S(H_{i'})$ of its vertices, with distance parameter $d'$ and congestion parameter $\eta=N^4$. Recall that the running time of the procedure is $O(|E(H)|(N^4+jd'\log N))\leq O(|E(H)|\cdot N^4)$. We denote by $\qset_{i,i'}$ the set of paths that the algorithm returns, and by $E'_{i,i'}$ the set of all edges of $H$ that participate in $N^4$ paths of $\qset_{i,i'}$. We also denote by $A'_{i,i'}\subseteq S(H_i)$ and $B'_{i,i'}\subseteq S(H_{i'})$ the sets of vertices that do not serve as endpoints of paths in $\qset_{i,i'}$. Recall that the paths in $\qset_{i,i'}$ have length at most $d'$ each, and they cause congestion at most $N^4$. Every path in $\qset_{i,i'}$ connects a vertex of $S(H_i)$ to a vertex of $S(H_{i'})$, and every vertex of $S(H_i)\cup S(H_{i'})$ may serve as an endpoint of at most one path in $\qset_{i,i'}$. Moreover, the length of the shortest path in $H\setminus E'_{i,i'}$ connecting a vertex of $A'_{i,i'}$ to a vertex of $B'_{i,i'}$ is greater than $d'$. Observe also that, since $|\qset_{i,i'}|\leq |S(H_i)|\leq |V(H_i)|=N^{j-1}$, and since the length of every path in $\qset_{i,i'}$ is at most $d'$, we get that $\sum_{Q\in \qset_{i,i'}}|E(Q)|\leq d'\cdot N^{j-1}$. Since $E'_{i,i'}$ only contains edges that participate in $N^4$ paths in $\qset$, we get that $|E'_{i,i'}|\leq d'\cdot N^{j-5}$.

Let $E'$ be the union of all sets $E'_{i,i'}$ of edges, over all pairs $H_i,H_{i'}\in \hset^2$ of graphs with $i<i'$. Clearly, $|E'|\leq d'\cdot N^{j-3}$. Moreover, for every pair $H_i,H_{i'}\in \hset^2$ of graphs with $i<i'$, the length of the shortest path in $H\setminus E'$ connecting a vertex of $A'_{i,i'}$ to a vertex of $B'_{i,i'}$ is greater than $d'$.

Next, we apply procedure \procsep from \Cref{lem: find separation of terminals} to graph $\tilde H=H\setminus E'$, with the set $T=V(H)$ of terminal vertices, distance parameters $d$ and $\Delta$ that remain unchanged, and parameter $\alpha=1-\frac{1}{8N^{\eps'}}$. Recall that the running time of the algorithm is 
$O(|E(H)|\cdot N^{64j/\Delta})\leq O(|E(H)|\cdot N^{2\eps^3})$ (since $\Delta=64/\eps'$, $j\leq \ceil{1/\eps}$, and $\eps'=\eps^4$).

We now consider two cases. The first case is that Procedure \procsep returns two subsets $A,B\subseteq V(H)$ of vertices, such that $|A|=|B|$, and for every pair $v\in A, u\in B$ of vertices, $\dist_{\tilde H}(u,v)\geq d$. Recall that in this case, the algorithm also ensures that:

\begin{equation}\label{eq: bound deleted edges}
\begin{split}
 |A|&\geq |V(H)|^{1-64/\Delta}\cdot \min\set{(1-\alpha),\frac 1 3} \\
 &\geq N^{j(1-\eps')}\cdot \frac{1}{8N^{\eps'}}\\
 &\geq N^{j(1-\eps')-2\eps'}\\
 &\geq N^{j(1-4j\eps')}\\
 &=|V(H)|^{1-\delta_j}.
 \end{split}
\end{equation}

(We have used the fact that $\Delta=64/\eps'$ and $\delta_j=4j\eps'$).

Recall that:

\[|E'|\leq d'\cdot N^{j-3}\leq  N^{j-3+\eps'}.\]

(since $d'\leq 2^{64/\eps^3}$ from Equation \ref{eq: bound on d'} and $N^{\eps^4}\geq 2^{128/\eps^5}$ from the statement of \Cref{thm: construct HSS}.)
Since, from Equation \ref{eq: bound deleted edges}, $|A|\geq N^{j-j\eps'-2\eps'}\geq N^{j-3\eps^3}$ (as $\eps'=\eps^4$), we get that $|E'|\leq |A|/N^{j\eps^4}$.

We conclude that $(A,B,E')$ is a valid $(\delta_j,d)$-distancing in graph $H$, with $|E'|\leq |A|/N^{j\eps^4}$. We return this distancing and terminate the algorithm.

From now on we assume that Procedure \procsep computed a vertex $v\in V(H)$, such that $|B_{\tilde H}(v,\Delta\cdot d)|>\alpha\cdot |V(H)|=N^{j}\cdot \left (1-\frac{1}{8N^{\eps'}}\right )$.
For convenience, we denote $B^*=B_{\tilde H}(v,\Delta\cdot d)$.

In this case, we construct a set $\hset'\subseteq \hset^2$ of graphs as follows: we add graph $H_i$ to $\hset'$ iff $|B^*\cap V(H_i)|\geq \frac{7|V(H_i)|}{8}$.
We need the following observation.

\begin{observation}\label{obs: many good graphs}
	$|\hset'|\geq N-2N^{1-\eps'}$.
\end{observation}

\begin{proof}
	Recall that $|\hset^2|\geq N-N^{1-\eps'}$. 
	Notice that, if $H_i\in \hset^2\setminus \hset'$, then $|V(H_i)\setminus B^*|\geq |V(H_i)|/8= N^{j-1}/8$.
	Since $|V(H)\setminus B^*|\leq N^{j-\eps'}/8$, we get that $|\hset^2\setminus \hset'|\leq \frac{N^{j-\eps'}/8}{N^{j-1}/8}=N^{1-\eps'}$.
	Therefore, $|\hset'|\geq |\hset^2|-N^{1-\eps'}\geq N-2N^{1-\eps'}$.
\end{proof}

We discard arbitrary graphs from $\hset'$, until $|\hset'|=N-\ceil{2N^{1-\eps'}}$ holds. Let $S(H)=\bigcup_{H_i\in \hset'}S(H_i)$. Note that we have now obtained a level-$j$ \HSS for graph $H$, whose associated collection of graphs is $\hset'$, with $|\hset'|=N-\ceil{2N^{1-\eps'}}=r$. We use the embedding $\pset$ of the graph $\bigcup_{H_i\in \hset^2}H_i$ that we have computed in Phase 1. By discarding paths that are no longer needed from $\pset$, we obtain an embedding $\pset$ of graph $\bigcup_{H_i\in \hset'}H_i$ into $H$, such that every path in the embedding has length at most $d'\leq 2^{64/\eps^4}$, and the paths in $\pset$ cause congestion at most $N^{128/\eps^3}$.
We prove the following lemma in \Cref{subsec: proof of certificate}.

\begin{lemma}\label{lem: certificate}
	Graph $H$ is $(\eta_j,\td_j)$-well-connected with respect to $S(H)$.
\end{lemma}

In order to complete the proof of \Cref{thm: construct HSS}, it is now enough to show that the running time of the algorithm is suitably bounded.

\subsubsection*{Running Time Analysis}

The algorithm performs at most $N^2$ calls to Procedure \procpathpeel. As observed above, the running time of each such call is at most $O(|E(H)|\cdot N^4)$. The running time of Procedure \procsep, as shown above, is at most $ O(|E(H)|\cdot N^{2\eps^3})$. Overall, the running time of Phase 2 is $O(|E(H)|\cdot N^6)$.

Altogether, the running time of the whole algorithm is bounded by:

\[ cj\cdot N^{j+64j\eps^2+7}+ O\left (|E(H)|\cdot N^6\right )\leq cj\cdot N^{j(1+64\eps^2)+7}+c|E(H)|\cdot N^6,\]

if $c$ is a large enough constant.

In order to complete the proof of \Cref{thm: construct HSS}, it is now enough to prove \Cref{lem: certificate}, which we do next.
\subsection{Proof of \Cref{lem: certificate}}
\label{subsec: proof of certificate}

We will use the following simple observation.

\begin{observation}\label{obs: good graphs can route}
For every pair $H_i,H_{i'}\in\hset'$ with $i<i'$, $|\qset_{i,i'}|\geq N^{j-1}/4$. (Here, $\qset_{i,i'}$ is the set of paths that was computed in Phase 2 of the algorithm.)
\end{observation}
\begin{proof}
	Recall that, from \Cref{claim size of supported}, $|V(H_i)\setminus S(H_i)|\leq |V(H_i)|\cdot \frac{4(j-1)}{N^{\eps^4}}<\frac{|V(H_i)|}{64}$ (since $j\leq \ceil{\frac{1}{\eps}}$). Therefore, $|S(H_i)|\geq \frac{63|V(H_i)|}{64}=\frac{63\cdot N^{j-1}}{64}$, and similarly, $|S(H_{i'})|\geq \frac{63\cdot N^{j-1}}{64}$.
	
	Assume now for contradiction that $|\qset_{i,i'}|< \frac{N^{j-1}}{4}$. Recall that we have defined sets $A'_{i,i'}\subseteq S(H_i)$ and $B'_{i,i'}\subseteq S(H_{i'})$ of vertices that do not serve as endpoints of paths in $\qset_{i,i'}$. If $|\qset_{i,i'}|< \frac{N^{j-1}}{4}$, then:
	
	 \[|A'_{i,i'}|\geq |S(H_i)|-|\qset_{i,i'}|\geq \frac{63\cdot N^{j-1}}{64}-\frac{N^{j-1}}{4}> \frac{2\cdot N^{j-1}}{3}.\]
	 
	 Similarly, $|B'_{i,i'}|> \frac{2\cdot N^{j-1}}{3}$. Since $H_i$ was added to $\hset'$, $|B^*\cap V(H_i)|\geq \frac{7\cdot N^{j-1}}{8}$. So at least one vertex $u\in A'_{i,i'}$ must lie in $B^*$. For similar reasons, at least one vertex $u'\in B'_{i,i'}$ must lie in $B^*$. From the definition of $B^*$, $\dist_{\tilde H}(u,u')\leq 2\Delta\cdot d=d'$. Recall however that $\tilde H=H\setminus E'$, and $E'_{i,i'}\subseteq E'$. Procedure \procpathpeel guarantees that the shortest path connecting $u$ to $u'$ in graph $H\setminus E'_{i,i'}$ has length greater than $d'$. So $\dist_{\tilde H}(u,u')>d'$ must hold, contradicting our previous claim that $\dist_{\tilde H}(u,u')\leq d'$.
\end{proof}

We now turn to the proof of \Cref{lem: certificate}.
We assume that we are given two disjoint equal-cardinality subsets $A,B$ of vertices of $S(H)$. Our goal is to prove that there exists a set $\pset^*$ of paths in graph $H$, routing every vertex of $A$ to a distinct vertex of $B$, such that the paths in $\pset^*$ cause congestion at most $\eta_j$ in $H$, and the length of every path is at most $\td_j$. We will prove that such a set of paths exists by exploiting the fact that, every graph $H_i\in \hset'$ is $(\eta_{j-1},\td_{j-1})$-well-connected with respect to the set $S(H_i)$ of its vertices, together with the embedding of these graphs into $H$, and the sets $\qset_{i,i'}$ of paths that we have computed in Phase 2 of the algorithm for every pair $H_i,H_{i'}\in \hset'$ of graphs with $i<i'$.

The remainder of the proof of \Cref{lem: certificate} consists of three steps. In the first step, we route some pairs in $A\times B$ within the graphs $H_i\in \hset'$. After the completion of this step, for every graph $H_i\in \hset'$, either all vertices of $S(H_i)$ that remain to be routed lie in $A$, or all such vertices lie in $B$. In the remaining two steps we complete the routing of these remaining vertices. Specifically, in Step 2 we define a ``meta-graph'' $\hat G$, whose vertices represent the graphs $H_i\in \hset'$, with weights on its edges. Intuitively, if an edge connecting two vertices that represent graphs $H_i$ and $H_{i'}$ has weight $w(e)$, then we intend to construct $w(e)$ paths that connect vertices of $S(H_i)\cap A$ to vertices of $S(H_{i'})\cap B$. In this step, we also perform a ``global routing'': for every pair $H_i,H_{i'}\in \hset'$ of graphs, whose corresponding edge in $\hat G$ has weight $w$,   we connect vertices of $S(H_i)$ to vertices of $S(H_{i'})$ with $w$ paths. In the third and the final step, we complete the construction of the set $\pset^*$ of paths by using ``local routing'', in which some pairs of vertex subsets are routed within each graph $H_i\in \hset'$. We now describe each of the three steps in turn.

\subsubsection*{Step 1: Initial Routing within the Graphs of $\hset'$}
We process every graph $H_i\in \hset'$ one by one. When graph $H_i$ is processed, we denote $N^A_i=|A\cap S(H_i)|$ and $N^B_i=|B\cap S(H_i)|$. Denote $\beta_i=\min\set{N^A_i,N^B_i}$. Next, we select two arbitrary subsets $X_i\subseteq A\cap S(H_i)$ and $Y_i\subseteq B\cap S(H_i)$, each of which contains exactly $\beta_i$ vertices.
Since our algorithm ensures that graph $H_i$ is $(\eta_{j-1},\td_{j-1})$-well-connected with respect to $S(H_i)$, there exists a set $\rset_i$ of paths in graph $H_i$, which is a one-to-one routing of vertices of $X_i$ to vertices of $Y_i$. Every path in $\rset_i$ has length at most $\td_{j-1}$, and the paths in $\rset_i$ cause congestion at most $\eta_{j-1}$ in $H_i$.

Let $H'=\bigcup_{H_i\in \hset'}H_i$. Recall that we have computed, in Phase 1 of the algorithm, an embedding $\pset$ of $H'$ into $H$, where every path in the embedding has length at most $d'$, and the paths in $\pset$ cause congestion at most $N^{128\eps^2}$ in $H$.

Consider now the set $\bigcup_{H_i\in \hset'}\rset_i$ of paths in graph $H'$. This set of paths defines a one-to-one routing of vertex set $X=\bigcup_{H_i\in \hset'}X_i$ to vertex set $Y=\bigcup_{H_i\in \hset'}Y_i$, where the length of every path is at most $\td_{j-1}$, and the paths in $\rset$ cause congestion at most $\eta_{j-1}$ in $H'$. 
We now use the embedding $\pset$ of $H'$ into $H$ in order to compute a set $\pset'$ of paths in graph $H$, that route every vertex of $X$ to a distinct vertex of $Y$, via the algorithm from \Cref{obs: paths from embedding}. 
We are then guaranteed that the length of every path in $\pset'$ is at most $\td_{j-1}\cdot d'$. Recall that $\td_{j-1}=2^{c(j-1)/\eps^4}$, while $d'<\frac 1 4\cdot 2^{c/\eps^4}$ from Inequality \ref{eq: bound on d'}. Therefore, the length of every path in $\pset'$ is bounded by:

\[\td_{j-1}\cdot d'\leq  2^{c(j-1)/\eps^4}\cdot 2^{c/\eps^4}=2^{cj/\eps^4}=\td_j.\]



 The algorithm from \Cref{obs: paths from embedding} also guarantees that the congestion caused by the paths in $\pset'$ in $H$ is at most $\eta_{j-1}\cdot N^{128\eps^2}$. Since $\eta_{j-1}=N^{6+256(j-1)\eps^2}$, we get that the congestion caused by the paths in $\pset'$ is bounded by:
 
 \[ \eta_{j-1}\cdot N^{128\eps^2}\leq  N^{6+256(j-1)\eps^2}\cdot N^{128\eps^2}
 \leq \frac{N^{6+256j\eps^2}}2=\frac{\eta_j} 2.\]

We have now obtained a set $\pset'$ of paths in graph $H$, that routes every vertex of $X$ to a distinct vertex of $Y$, so that the length of every path is at most $\td_j$, and the congestion caused by the paths in $\pset'$ is bounded by $\eta_j/2$.

We partition the  graphs of $\hset'$ into three subsets. Set $\hset^N$ contains all graphs $H_i\in \hset'$, in which $N_i^A=N_i^B$. We no longer need to route any vertices in such graphs, as for each such graph $\beta_i=N^A_i=N^B_i$; $X_i=A\cap S(H_i)$; and $Y_i=B\cap S(H_i)$ must hold. Set $\hset^A$ contains graphs $H_i\in \hset'$ with $N_i^A>N_i^B$. For each such graph $H_i$, we denote by $D(H_i)=N_i^A-N_i^B$, and by $X'_i=( A \cap S(H_i))\setminus X_i$ -- the set of vertices of $S(H_i)$ that remain to be routed. Clearly, $|X'_i|=D(H_i)$. Similarly, set $\hset^B$ contains graphs 
$H_i\in \hset'$ with $N_i^B>N_i^A$. For each such graph $H_i$, we denote by $D(H_i)=N_i^B-N_i^A$, and by $X'_i=( B \cap S(H_i))\setminus Y_i$ -- the set of vertices that remain to be routed. As before, $|X'_i|=D(H_i)$ holds. Notice also that, since $|A|=|B|$, $\sum_{H_i\in \hset^A}D(H_i)=\sum_{H_i\in \hset^B}D(H_i)$ must hold.  Denote $\hat A=\bigcup_{H_i\in \hset^A}X'_i$ and $\hat B=\bigcup_{H_i\in \hset^B}X'_i$. 
It is now enough to prove the following lemma.

\begin{lemma}\label{lem: good routing}
	There is a set $\pset''$ of paths in graph $H$, routing every vertex of $\hat A$ to a distinct vertex of $\hat B$, so that the length of every path in $\pset''$ is at most $\td_j$, and the paths in $\pset''$ cause congestion at most $\eta_j/2$ in $H$.
\end{lemma}

Indeed, by letting $\pset^*=\pset'\cup \pset''$, we obtain a set of paths in graph $H$ that defines a one-to-one routing of the set $A$ of vertices to the set $B$ of vertices via paths of length at most $\td_j$, so that the paths in $\pset^*$ cause congestion at most $\eta_j$. In order to complete the proof of \Cref{lem: certificate}, it is now enough to prove \Cref{lem: good routing}. We focus on the proof of \Cref{lem: good routing} in the remainder of this section.

We will start by constructing a ``meta-graph'' representing the graphs of $\hset^A\cup \hset^B$, that will guide the construction of global routing.

\subsection*{Step 2: Meta-Graph and Global Routing} 


Abusing the notation, for simplicity, in the remainder of this proof we denote $\hset^A=\set{H_1,H_2,\ldots,H_q}$, and for all $1\leq i\leq q$, we denote $D(H_i)$ by $D_i$. 
We also denote $\hset^B=\set{H'_1,H'_2,\ldots,H'_{q'}}$, and for $1\leq i'\leq q'$, we denote $D(H'_{i'})$ by $D'_{i'}$. 
For all $1\leq i\leq q$, we denote the set $X'_i\subseteq S(H_i)$ of $D_i$ vertices that remains to be routed by $Y_i$, and for all $1\leq i'\leq q'$, we denote the corresponding subset of $D'_i$ vertices of $S(H'_{i'})$ by $Y'_{i'}$.
We now define a \emph{routing meta-graph}, that will be used in order to guide the construction of the paths in $\pset^*$, and show that such a graph exists.

\subsection*{Routing Meta-Graph}

We start by defining a routing meta-graph.
\begin{definition}[Routing Meta-Graph]
A bipartite graph  $\hat G=(U,U',\hat E)$ with integral weights $w(e)\geq 0$ on its edges $e\in \hat E$ is a \emph{routing meta-graph} if:

	\begin{itemize}
		\item  $U=\set{v_1,\ldots,v_q}$;
		\item $U'=\set{v'_1,\ldots,v'_{q'}}$;
		\item for every vertex $v_i\in U$, $\sum_{e\in \delta_{\hat G}(v_i)}w(e)=D_i$; and 
		
		\item for every vertex $v_{i'}'\in U'$, $\sum_{e\in \delta_{\hat G}(v'_{i'})}w(e)=D'_{i'}$.
		\end{itemize}
		We refer to vertices of $\hat G$ as supernodes and edges of $\hat G$ as meta-edges.
\end{definition}

We use the following claim to show that a routing meta-graph exists.

\begin{claim}\label{claim: routing meta-graph}
	There exists a routing meta-graph.
\end{claim}
\begin{proof}
	We start with the graph  $\hat G=(U,U',\hat E)$, where  $U=\set{v_1,\ldots,v_q}$,  $U'=\set{v'_1,\ldots,v'_{q'}}$, and $\hat E=\emptyset$, and then iterate, as long as there exist indices $1\leq i\leq q$ and $1\leq i'\leq q'$, such that $D_i>0$ and $D'_{i'}>0$ holds.
	
	In order to execute an iteration, we consider any pair of such indices $(i,i')$. Let $\hat \Delta=\min\set{D_i,D'_{i'}}$. We add an edge $(v_i,v'_{i'})$ to $\hat E$, whose weight is $\hat \Delta$, and we decrease $D_i$ and $D'_{i'}$ by $\hat \Delta$. Once the algorithm terminates, since $\sum_{i=1}^qD_i=\sum_{i'=1}^{q'}D'_{i'}$, it is immediate to verify that the resulting graph $\hat G$ is a valid routing meta-graph. 
\end{proof}

\subsection*{Global Routing}

Consider some pair $v_i\in U$, $v'_{i'}\in U'$ of supernodes in graph $\hat G$. From \Cref{obs: good graphs can route}, there exists a collection of paths in graph $H$, that we denote, abusing the notation, by $\qset_{i,i'}$, such that the following hold:

\begin{itemize}
	\item every path in $\qset_{i,i'}$ originates at a vertex of $S(H_i)$ and terminates at a vertex of $S(H'_{i'})$; 
	\item every vertex of $S(H_i)\cup S(H'_{i'})$ is an endpoint of at most one path in $\qset_{i,i'}$;
	\item $|\qset_{i,i'}|=\ceil{N^{j-1}/4}$;
	\item each path in $\qset_{i,i'}$ has length at most $d'$; and
	\item the paths in $\qset_{i,i'}$ cause congestion at most $N^4$ in $H$.
\end{itemize}

The set $\qset_{i,i'}$ of paths naturally defines  a matching $M_{i,i'}\subseteq S(H_i)\times S(H'_{i'})$: we include a pair $(x,y)$ of vertices in $M_{i,i'}$ if $x\in S(H_i)$, $y\in S(H'_{i'})$, and some path in $\qset_{i,i'}$ has $x$ and $y$ as its endpoints. Clearly, $|M_{i,i'}|=\ceil{N^{j-1}/4}$. Notice that for every meta-edge $e=(v_i,v'_{i'})$ in graph $\hat G$, $w(e)\leq D_i\leq |S(H_i)|\leq N^{j-1}$ must hold. We will select, for every edge $e=(v_i,v'_{i'})\in \hat G$, a multi-set $M'_{i,i'}$ of pairs $(x,y)\in M_{i,i'}$ of vertices, of cardinality $w(e)$. (We note that a pair $(x,y)\in M_{i,i'}$ of vertices may be added to $M'_{i,i'}$ multiple times).  We will ensure that, for every supernode $v_i\in U$, a vertex $x\in S(H_i)$ may participate in at most four pairs in $\bigcup_{e=(v_i,v'_{i'})\in \delta_{\hat G}(v_i)}M'_{i,i'}$, and the same holds for supernodes of $U'$. For every meta-edge $e=(v_i,v'_{i'})\in \hat E$, we will then use the paths of $\qset_{i,i'}$ whose endpoints lie in $M'_{i,i'}$ to define a \emph{global routing}. Let $\qset^0$ denote the resulting collection (multiset) of all such paths. So for every meta-edge $(v_i,v'_{i'})\in \hat E$, for every pair $(x,y)\in M'_{i,i'}$ of vertices, $\qset^0$ contains the path of $\qset_{i,i'}$ whose endpoints are $x$ and $y$. If $(x,y)$ appears multiple times in $M'_{i,i'}$, then $\qset^0$ contains multiple copies of this path. For every graph $\tilde H\in \hset^A\cap \hset^B$, for every vertex $x\in S(\tilde H)$, we denote by $k(x)$ the number of paths of $\qset^0$, for which $x$ serves as an endpoint. Our construction will guarantee that, for all $1\leq i\leq q$, $\sum_{x\in S(H_i)}k(x)=D_i$, and similarly, for all $1\leq i'\leq q'$, $\sum_{x\in S(H'_{i'})}k(x)=D'_{i'}$. As mentioned already, we will ensure that, for every graph $\tilde H\in \hset'$, for every vertex $x\in S(\tilde H)$, $k(x)\in \set{0,\ldots,4}$. In our last step, we will perform \emph{local routing}, in which, for all $1\leq i\leq q$, we connect every vertex of $Y_i$ to some vertex of $S(H_i)$ by a path, such that every vertex $x\in S(H_i)$ is the last endpoint of exactly $k(x)$ such paths. We perform a similar routing in graphs of $\hset^B$. This local routing exploits the fact that every graph $\tilde H\in \hset'$ is $(\eta_{j-1},\td_{j-1})$-well-connected, together with the embedding $\pset$ of the graph $H'=\bigcup_{\tilde H\in \hset'}\tilde H$ into $H$ that we have computed.

In order to simplify the notation, for all $1\leq i\leq q$, we denote by $\hat E_i\subseteq\hat E$ the set of all meta-edges of $\hat G$ that are incident to supernode $v_i$ in $\hat G$. Similarly, for all $1\leq i'\leq q'$, we denote by $\hat E_{i'}'\subseteq\hat E$ the set of all meta-edges of $\hat G$ that are incident to supernode $v'_{i'}$ in $\hat G$.
We prove the following lemma, that allows us to perform global routing.

\begin{lemma}\label{lem: global}
	For every meta-edge $e=(v_i,v'_{i'})\in \hat E$, there is a multiset $M'_{i,i'}$ of pairs of vertices of $S(H_i)\times S(H'_{i'})$, for which the following hold. For all $1\leq i\leq q$, for every vertex $x\in S(H_i)$, let $k(x)$ be the total number of pairs in $\bigcup_{(v_i,v'_{i'})\in \hat E_i}M'_{i,i'}$, in which vertex $x$ participates. Similarly, for all $1\leq i'\leq q'$, for every vertex $x\in S(H'_{i'})$, let $k(x)$ be the total number of pairs in $\bigcup_{(v_{i},v'_{i'})\in \hat E'_{i'}}M'_{i,i'}$, in which vertex $x$ participates. Then:
	
	\begin{itemize}
		\item for every meta-edge $e=(v_i,v'_{i'})\in E(\hat G)$, a pair $(x,y)$ of vertices may only belong to $M'_{i,i'}$ if $(x,y)\in M_{i,i'}$ (but it may be added to $M'_{i,i'}$ multiple times);
		
		\item for every vertex $x\in \left (\bigcup_{i=1}^qS(H_i)\right )\cup \left (\bigcup_{i'=1}^{q'}S(H'_{i'})\right )$, $k(x)\in \set{0,\ldots,4}$;
		
		\item for all $1\leq i\leq q$, $\sum_{x\in S(H_i)}k(x)=D_i$; and
		
		\item for all $1\leq i'\leq q'$, $\sum_{x\in S(H'_{i'})}k(x)=D'_{i'}$.
	\end{itemize}
\end{lemma}

\begin{proof}
	We construct the following directed flow network. We start with a bipartite graph $\tilde G=(\tilde X,\tilde Y,\tilde E)$, where $\tilde X=\bigcup_{i=1}^qS(H_i)$, $\tilde Y=\bigcup_{i'=1}^{q'}S(H'_{i'})$, and  $\tilde E =\bigcup_{(v_i,v'_{i'})\in E(\hat G)}M_{i,i'}$. All edges are directed from vertices of $\tilde X$ towards vertices of $\tilde Y$, and they have capacity $4$ each. For all $1\leq i\leq q$, we add a vertex $s_i$, that connects to every vertex in $S(H_i)$ with an edge of capacity $4$. For all $1\leq i'\leq q'$, we add a vertex $t_{i'}$, to which every vertex of $S(H'_{i'})$ connects with an edge of capacity $4$. Lastly, we add a source vertex $s$, and a destination vertex $t$. For all $1\leq i\leq q$, we add an edge $(s,s_i)$ of capacity $D_i$, and for all $1\leq i'\leq q'$, we add an edge $(t_i,t)$ of capacity $D'_{i'}$.
	
	We claim that this network as a valid $s$-$t$ flow $f$ of value $D=\sum_{i=1}^qD_i=\sum_{i'=1}^{q'}D'_{i'}$. We obtain this flow as follows. Consider a meta-edge $e=(v_i,v'_{i'})\in E(\hat G)$. Recall that we are given an integral weight $w(e)\leq N^{j-1}$, and a matching $M_{i,i'}\subseteq S(H_i)\times S(H'_{i'})$ of cardinality $\ceil{N^{j-1}/4}$. For every edge $e'=(x,y)\in M_{i,i'}$, we set the flow $f(e')=\frac{w(e)}{|M_{i,i'}|}=\frac{w(e)}{\ceil{N^{j-1}/4}}$. Notice that this ensures that the total flow on all edges of $M_{i,i'}$ is precisely $w(e)$, and for every edge $e'\in M_{i,i'}$, $f(e')\leq 4$. Once we process every meta-edge of $\hat G$, we finalize the flow values $f(e')$ for all edges $e'\in \tilde E$.
	
	Consider now some index $1\leq i\leq q$, and some vertex $x\in S(H_i)$. We claim that the total flow on all edges of $\tilde E$ that are incident to $x$ in the flow network is at most $4$. Indeed, recall that $\hat E_i$ is the set of all meta-edges of $\hat G$ that are incident to supernode $v_i$. From the definition of a routing meta-graph, we are guaranteed that $\sum_{e\in \hat E_i}w(e)=D_i$. For every meta-edge $e=(v_i,v_{i'}')\in \hat E_i$, if some edge of $M_{i,i'}$ is incident to $x$, then the flow on this edge is $\frac{w(e)}{\ceil{N^{j-1}/4}}$. Therefore, the total flow on all edges of $\tilde E$ that are incident to $x$ is bounded by:
	
	\[\sum_{e\in \hat E_i} \frac{w(e)}{\ceil{N^{j-1}/4}}= \frac{D_i}{\ceil{N^{j-1}/4}}\leq 4,\]
	
	since $D_i\leq |S(H_i)|\leq N^{j-1}$ must hold.
	
	We set the flow on edge $(s_i,x)$ to be equal to the total amount of flow on all edges of $\tilde E$ that are incident to $x$ in the flow network, which, from the above discussion, is bounded by $4$.
	
	From similar arguments, for every index $1\leq i'\leq q'$, and every vertex $y\in S(H'_{i'})$, the total flow on all edges of $\tilde E$ that are incident to $y$ in the flow network is at most $4$. We set the flow on the edge $(y,t_i)$ to be the total flow on all edges of $\tilde E$ that are incident to $y$ in the flow network.

	Next, we consider an index $1\leq i\leq q$. We set the flow on edge $(s,s_i)$ to be $D_i$. We claim that $\sum_{x\in S(H_i)}f(s_i,x)=D_i$. Indeed, from our construction:

	\[ \sum_{x\in S(H_i)}f(s_i,x)=\sum_{(v_i,v'_{i'})\in \hat E_i}\sum_{(x,y)\in M_{i,i'}}f(x,y)= \sum_{e=(v_i,v'_{i'})\in \hat E_i}\sum_{(x,y)\in M_{i,i'}}\frac{w(e)}{\ceil{N^{j-1}/4}}=\sum_{e=(v_i,v'_{i'})\in \hat E_i}w(e)=D_i.  \]
	
	(we have used the fact that $|M_{i,i'}|=\ceil{N^{j-1}/4}$ for every meta-edge $(v_i,v'_{i'})\in E(\hat G)$).
	
	Similarly, we consider an index $1\leq i'\leq q'$. We set the flow on edge $(t_{i'},t)$ to be $D'_{i'}$. Using the same reasoning as above, $\sum_{y\in S(H'_{i'})}f(y,t_i)=D'_{i'}$.
		We conclude that we have obtained a valid $s$-$t$ flow in the above flow network, whose value is $D$. Since all edge capacities in the flow network are integral, from the integrality of flow, there is an integral $s$-$t$ flow $f'$ of value $D$ in this flow network. 
		 
		We are now ready to define the multisets $M'_{i,i'}$ of pairs of vertices from $S(H_i)\times S(H'_{i'})$, for all $(v_i,v'_{i'})\in E(\hat G)$. Consider any meta-edge $(v_i,v'_{i'})\in E(\hat G)$, and some pair $(x,y)\in M_{i,i'}$ of vertices. If $f'(x,y)>0$, then we include $f'(x,y)$ copies of the pair $(x,y)$ to $M'_{i,i'}$. 
		
		We now verify that all requirements hold for this definition of the multisets $M'_{i,i'}$ for all $(v_i,v'_{i'})\in E(\hat G)$. Clearly, a pair $(x,y)$ of vertices may only be added to $M'_{i,i'}$ if $(x,y)\in M_{i,i'}$.

		Consider now some vertex $x\in \bigcup_{i=1}^qS(H_i)$. Recall that $k(x)$ is the total number of pairs in $\bigcup_{(v_i,v'_{i'})\in \hat E}M'_{i,i'}$ in which $x$ participates. This is equal to the total flow leaving vertex $x$ in $f'$, which, in turn, is equal to the flow on edge $(s_i,x)$. From our definition, the capacity of this edge is $4$, so $k(x)\in\set{0,\ldots,4}$. If $x\in 	\bigcup_{i'=1}^{q'}S(H'_{i'})$, then $k(x)\in \set{0,\ldots,4}$ for similar reasons.
		
		Consider now some index $1\leq i\leq q$. From the above discussion,  $\sum_{x\in S(H_i)}k(x)=\sum_{x\in S(H_i)}f'(s_i,x)$. In other words, $\sum_{x\in S(H_i)}k(x)$ is the total amount of flow leaving vertex $s_i$ in $f'$. From conservation of flow this must be equal to the total amount of flow entering $s_i$. Since we send $D=\sum_{i=1}^qD_i$ flow units from $s$ to $t$, and since, for all $1\leq i\leq q$, the capacity of the edge $(s,s_i)$ is $D_i$, we must send $D_i$ flow units on edge $(s,s_i)$. In other words,  for all $1\leq i\leq q$, $\sum_{x\in S(H_i)}k(x)=D_i$ must hold. From similar arguments,  for all $1\leq i'\leq q'$, $\sum_{x\in S(H'_{i'})}k(x)=D'_{i'}$ must hold.
\end{proof}

We are now ready to define the global routing.
For every meta-edge $e=(v_i,v'_{i'})\in E(\hat G)$, we consider the resulting collection $M'_{i,i'}\subseteq S(H_i)\times S(H_{i'})$ of pairs of vertices. We define a (multi)-set $\qset'_{i,i'}$ of paths, as follows. Consider any pair $(x,y)\in M'_{i,i'}$, and assume that the number of times that it appears in $M'_{i,i'}$ is $N(x,y)$. Recall that $(x,y)\in M_{i,i'}$ must hold, so there must be a path $Q(x,y)\in \qset_{i,i'}$ connecting $x$ to $y$. We add $N(x,y)$ copies of this path to $\qset'_{i,i'}$. Note that, from \Cref{lem: global}, $N(x,y)\leq k(x)\leq 4$ must hold.

We then let $\qset^0=\bigcup_{(v_i,v'_{i'})\in E(\hat G)}\qset'_{i,i'}$ (again, set $\qset^0$ is a multiset, so if some path appears several times in some set $\qset'_{i,i'}$, then it will appear several times in $\qset^0$).

Recall that, for every pair $H_i,H'_{i'}\in \hset'$ of graphs, the paths in $\qset_{i,i'}$ have length at most $d'$ each, and they cause congestion at most $N^4$ in $H$. Since $\qset'_{i,i'}$ contains at most four copies of each path in $\qset_{i,i'}$, and since $|E(\hat G)|\leq |\hset'|^2\leq N^2$, the paths in $\qset^0$ cause congestion at most $4N^6$ in $H$, and every path has length at most $d'$ as before. 
For every vertex $x\in \left (\bigcup_{i=1}^qS(H_i)\right )\cup \left (\bigcup_{i'=1}^{q'}S(H'_{i'})\right )$, we use the definition of the value $k(x)$ from \Cref{lem: global}. The number of paths in $\qset^0$, in which $x$ serves as an endpoint is then precisely $k(x)$, and, from \Cref{lem: global}, $k(x)\in \set{0,\ldots,4}$. Recall also that,
for all $1\leq i\leq q$, $\sum_{x\in S(H_i)}k(x)=D_i$, and
 for all $1\leq i'\leq q'$, $\sum_{x\in S(H'_{i'})}k(x)=D'_{i'}$.

\subsection*{Step 3: Local Routing}

Consider some graph $H_i\in \hset^A$. We have defined a set $Y_i\subseteq S(H_i)$ of $D_i$ vertices of $H_i$ that need to be routed. For every vertex $x\in S(H_i)$, we are also now given a value $k(x)\in \set{0,\ldots,4}$, which is exactly the number of paths in $\qset^0$ for which vertex $x$ serves as an endpoint. We are also guaranteed that $\sum_{x\in S(H_i)}k(x)=D_i$. We can then construct four sets $Z_i^1,Z_i^2,Z_i^3,Z_i^4$ of vertices of $S(X)$, such that for every vertex $x\in S(X)$, each of the four sets contains at most one copy of $x$; and the number of sets in $\set{Z_i^1,\ldots,Z_i^4}$ containing $x$ is exactly $k(x)$. Clearly, $\sum_{a=1}^4|Z_i^a|=D_i$. We also partition the set $Y_i$ of vertices into four subsets $Y_i^1,\ldots,Y_i^4$ arbitrarily, so that, for all $1\leq a\leq 4$, $|Y_i^a|=|Z_i^a|$. Since $Y_i\subseteq S(H_i)$, from the fact that graph $H_i$ is $(\eta_{j-1},\td_{j-1})$-well-connected with respect to $S(H_i)$, for all $1\leq a\leq 4$, there is a set $\hat\pset_i^a$ of $|Z_i^a|$ paths in graph $H_i$, routing every vertex of $Y_i^a$ to a distinct vertex of $Z_i^a$, such that the length of every path is at most $\td_{j-1}$, and the paths cause congestion at most $\eta_{j-1}$ in graph $H_i$. 
Let $\hat \pset_i=\bigcup_{a=1}^4\hat \pset_i^a$. We think of the paths in $\hat \pset_i$ as being directed away from vertices of $Y_i$. Notice that the paths in $\hat \pset_i$ route every vertex of $Y_i$ to some vertex of $S(H_i)$, such that, for every vertex $x\in S(H_i)$, exactly $k(x)$ paths of $\hat \pset_i$ terminate at $x$. The paths in $\hat \pset_i$ cause congestion at most $4\eta_{j-1}$ in graph $H_i$, and have length at most $\td_{j-1}$ each.

Consider the graph $H'=\bigcup_{H_i\in \hset'}H_i$. Recall that we have computed, in Phase 1 of the algorithm, an embedding $\pset$ of $H'$ into $H$, where every path in the embedding has length at most $d'$, and the paths in $\pset$ cause congestion at most $N^{128\eps^2}$ in $H$.

Consider now the set $\hat \pset=\bigcup_{1\leq i\leq q}\hat \pset_i$ of paths in graph $H'$. This set of paths routes every vertex of $\bigcup_{1\leq i\leq q}Y_i$ to some vertex of $\bigcup_{1\leq i\leq q}S(H_i)$, such that, for every vertex $x\in \bigcup_{1\leq i\leq q}S(H_i)$, exactly $k(x)$ paths of $\hat \pset$ terminate at $x$. Additionally, the length of every path in $\hat \pset$ is is at most $\td_{j-1}$, and the paths in $\hat \pset$ cause congestion at most $4\eta_{j-1}$ in $H'$. 

We now use the algorithm from \Cref{obs: paths from embedding}
with the collection $\hat \pset$ of paths, and
 the embedding $\pset$ of graph $H'=\bigcup_{H_i\in \hset'}H_i$ into $H$, in order to compute a set $\qset^1$ of paths in graph $H$, routing every vertex $\bigcup_{1\leq i\leq q}Y_i$ to some vertex of $\bigcup_{1\leq i\leq q}S(H_i)$, such that, for every vertex $x\in \bigcup_{1\leq i\leq q}S(H_i)$, exactly $k(x)$ paths of $\hat \pset$ terminate at $x$. 
The algorithm ensures that the length of every path in $\hat \pset$ is at most $\td_{j-1}\cdot d'$, and the paths in $\hat \pset$ cause congestion at most $4\eta_{j-1}\cdot N^{128\eps^2}$. 

For every index $1\leq i'\leq q'$, we similarly compute a set $\hat \pset'_{i'}$ of paths in graph $H'_{i'}$, that route some vertices of $S(H_{i'})$ to vertices of $Y'_{i'}$, so that for every vertex $y\in Y'_{i'}$, exactly one path in $\hat \pset'_{i'}$ terminates at $y$, and for every vertex $x\in S(H'_{i'})$, exactly $k(x)$ paths of $\pset'_{i'}$ originate at $x$.
We use the embedding $\pset$ of graph $H'$ into $H$ exactly as before,  in order to compute a set $\qset^2$ of paths in graph $H$, 
routing vertices of $\bigcup_{1\leq i'\leq q'}S(H'_{i'})$ to vertices of $\bigcup_{1\leq i'\leq q'}Y'_{i'}$, such that, for every vertex $x\in \bigcup_{1\leq i\leq q}S(H_i)$, exactly $k(x)$ paths of $\qset^2$ originate at $x$, and for every vertex $y\in \bigcup_{1\leq i'\leq q'}Y'_{i'}$, exactly one path of $\qset^2$ terminates at $y$. As before, we can ensure that the length of each every path in $\qset^2$ is at most $\td_{j-1}\cdot d'$, and the paths in $\qset^2$ cause congestion at most $4\eta_{j-1}\cdot N^{128\eps^2}$ in $H$.

By concatenating the paths of $\qset^1,\qset^0$ and $\qset^2$, we obtain the final set $\pset''$ of paths, that defines a one-to-one routing between vertices of $\bigcup_{i=1}^qY_i$ and vertices of $\bigcup_{i'=1}^{q'}Y'_{i'}$. The length of every path in $\pset''$ is bounded by $2\td_{j-1}\cdot d'+d'\leq 3\td_{j-1}\cdot d'$.

Recall that $\td_{j-1}=2^{c(j-1)/\eps^4}$, while $d'<\frac 1 4\cdot 2^{c/\eps^4}$ from Inequality \ref{eq: bound on d'}. Therefore, we get that the length of every path in $\pset''$ is at most:

\[3\td_{j-1}\cdot d'\leq 2^{c(j-1)/\eps^4}\cdot 2^{c/\eps^4}=2^{cj/\eps^4}=\td_j.\]

 The total congestion caused by the paths in $\pset''$ is bounded by:

\[ 4N^6+8\eta_{j-1}\cdot N^{128\eps^2}=4N^6+8\cdot N^{6+256(j-1)\eps^2}\cdot N^{128\eps^2}
\leq \frac{N^{6+256j\eps^2}} 2=\frac{\eta_j} 2, \]

since $\eta_{j-1}=N^{6+256(j-1)\eps^2}$, and $N$ is sufficiently large.

This completes the proof of \Cref{lem: good routing}

\section{APSP in Well-Connected Graphs -- Proof of  \Cref{thm: APSP in HSS full}}
\label{sec: APSP}

The goal of this section is to prove \Cref{thm: APSP in HSS full}. We do so using the following theorem. 

\begin{theorem}\label{thm: APSP in HSS}
	There are large enough constants $c',c''$, and a deterministic algorithm, whose input consists of:
	
	\begin{itemize}
		\item a parameter $0<\eps<1/400$;
		\item a pair $N$, $1\leq j\leq \ceil{1/\eps}$ of integers,  such that $N$ is sufficiently large, so that $\frac{N^{\eps^4}}{\log N}\geq 2^{128/\eps^6}$ holds;
		\item  a graph $H$ with $|V(H)|=N^j$; and
		\item a level-$j$ \HSS for graph $H$, such that graph $H$ is $(\eta_j,\td_j)$-well-connected with respect to the set $S(H)$ of vertices defined by the \HSS.  
	\end{itemize} 
Further, we assume that graph $H$ undergoes an online sequence of less than $\Lambda_j=N^{j-8-300j\eps^2}$ edge deletions. The algorithm maintains a set $S'(H)\subseteq S(H)$ of vertices of $H$, called \emph{supported vertices}, such that, at the beginning of the algorithm, $S'(H)=S(H)$, and over the course of the algorithm, vertices can leave $S'(H)$ but they may not join it. The algorithm ensures that $|S'(H)|\geq \frac{N^j}{16^j}$ holds over the course of the algorithm, and it supports short-path queries between supported vertices: given a pair $x,y\in S'(H)$ of vertices, return a path $P$ connecting $x$ to $y$ in the current graph $H$, whose length is at most $d^*_j=2^{c'j/\eps^5}$, in time $O(|E(P)|)$. The total update time of the algorithm is bounded by $2c''jN^{j+3}\cdot 2^{4c'/\eps^6}+c''m\cdot N^2\cdot 2^{4c'/\eps^6}$, where $m$ is the number of edges in graph $H$ at the beginning of the algorithm. (If $\Lambda_j\leq 1$, then the algorithm only needs to support short-path queries until the first edge deletion).
\end{theorem}
	
It is immediate to verify that \Cref{thm: APSP in HSS full} follows from \Cref{thm: APSP in HSS}, by substituting $j=1/\eps$.
In the remainder of this section we prove \Cref{thm: APSP in HSS}.
	The proof is by induction on $j$.

\subsection{Base Case: $j\leq 8$}
	We first consider the base case, where $j\leq 8$. In this case, $\Lambda_j\leq 1$ holds, and so we only need to support short-path queries until the first edge deletion.
	 
	In this case, the level-$j$ Hierarchical Support Structure for graph $H$ defines a set $S(H)$ of vertices, and, from \Cref{claim size of supported}, $|V(H)\setminus S(H)|\leq |V(H)|\cdot \frac{4j}{N^{\eps^4}}$.
	Since $j\leq \ceil{1/\eps}$, and $N^{\eps^4}\geq 2^{128/\eps^5}$ from the statement of \Cref{thm: APSP in HSS}, we get that $|S(H)|\geq |V(H)|/2=N^j/2$. We set $S'(H)=S(H)$, and this set remains unchanged throughout the algorithm.
	
	 Recall that are guaranteed that graph $H$ is $(\eta_j,\td_j)$-well-connected with respect to $S(H)$, where $\td_j=2^{cj/\eps^4}\leq \frac{2^{c'j/\eps^5}}{2}= \frac{d^*_j} 2$ (if $c'>c$), and $\eta_j=N^{6+256j\eps^2}$. In particular, for every pair $x,y\in S'(H)$ of supported vertices, there is a path of length at most $d^*_j/2$ connecting $x$ to $y$ in $H$.
	 We let $s\in S(H)$ be an arbitrary vertex, and we construct a BFS tree $\tau$ rooted at vertex $s$, whose depth is bounded by $d^*_j/2$. Computing the tree takes time $O(|E(H)|)$. In order to respond to a short-path query between a pair $x,y$ of vertices of $H$, we simply compute the unique simple path $P$ connecting $x$ to $y$ in the tree $\tau$, which can be done in time $O(|E(P)|)$. Since the depth of the tree is bounded by $d^*_j/2$, the length of the path is at most $d^*_j$.

\subsection{Step: $j>8$}
We assume that we are given a graph $H$ with $|V(H)|=N^j$, together
 level-$j$ hierarchical support structure for graph $H$, whose associated collection of graphs is $\hset=\set{H_1,\ldots,H_r}$. Recall that $r=N-\ceil{2N^{1-\eps^4}}$, and we are guaranteed that graph $H$ is $(\eta_j,\td_j)$-well-connected with respect to the set $S(H)=\bigcup_{H_i\in \hset}S(H_i)$ of vertices. Additionally, the hierarchical support structure contains an embedding $\pset$ of the graph $H'=\bigcup_{i=1}^{r}H_i$ into graph $H$ via paths of length at most $2^{64/\eps^4}$, that causes congestion at most $N^{128\eps^2}$. For every edge $e\in E(H')$, we denote by $P(e)\in \pset$ the unique path that serves as the embedding of $e$ in $G$.
 
 Our algorithm will maintain a set $S'(H)\subseteq S(H)$ of supported vertices, where initially $S'(H)=S(H)$. While vertices may leave set $S'(H)$  over the course of the algorithm, set $S(H)$ remains unchanged. 
 
 Our algorithm  recursively applies the algorithm from \Cref{thm: APSP in HSS} to each of the graphs in $\hset$. When an edge $e\in E(H)$ is deleted, then for all $1\leq i\leq r$, for every edge $e'\in E(H_i)$ whose corresponding embedding path $P(e')$ contains $e$, we delete edge $e'$ from graph $H_i$. Since the paths in $\pset$ cause congestion at most $N^{128\eps^2}$, the deletion of a single edge from graph $H$ may trigger the deletion of up to $N^{128\eps^2}$ edges from graphs $H_1,\ldots,H_r$ overall. As the result of these edge deletions, the supported sets of vertices $S'(H_i)$ that the algorithm from \Cref{thm: APSP in HSS} maintains recursively for each of the graphs $H_i\in \hset$ may need to be updated. Once a graph $H_i\in \hset$ undergoes $\ceil{\Lambda_{j-1}}$ edge deletions, we say that it is \emph{destroyed}. Once graph $H_i$ is destroyed, the corresponding set $S'(H_i)$ of vertices is set to $\emptyset$.

 Our algorithm maintains a partition of the set $\hset$ of graphs into three subsets: set $\hset^D$ of \emph{destroyed graphs}, set $\hset^I$ of \emph{inactive graphs}, and set $\hset^A$ of active graphs.
 Set $\hset^D$ contains all graphs that have been destroyed so far. Once a graph joins set $\hset^D$, it remains in $\hset^D$ for the remainder of the algorithm.
 We define the sets $\hset^I$ and $\hset^A$ of graphs later. We will ensure that the set $\hset^A$ of active graphs is decremental, and we will set $S'(H)=\bigcup_{H_i\in \hset^A}S'(H_i)$ throughout the algorithm. Intuitively, we will maintain, for every active graph $H_i\in \hset^A$, an \EST data structure that is rooted at the vertices of $S'(H_i)$. 
 We need the following simple observation bounding the number of graphs in $\hset^D$.
 
 \begin{observation}\label{obs: few destroyed graphs}
 	Over the course of the algorithm, $|\hset^D|< N/32$ always holds.
 \end{observation}
 
 \begin{proof}
 	Assume otherwise, and consider the first time $t$ during the algorithm when $|\hset^D|\geq N/32$ held. 
 	Recall that a graph $H_i\in \hset$ is destroyed once it undergoes $\ceil{\Lambda_{j-1}}$ edge deletions. Therefore, at least $\frac{N}{32}\cdot \ceil{\Lambda_{j-1}}$ edges have been deleted from $\bigcup_{H_i\in \hset}E(H_i)$ by time $t$.
 	On the other hand, the deletion of a single edge from $H$ may trigger the deletion of at most $N^{128\eps^2}$ edges from graphs $H_1,\ldots,H_r$. Therefore, the number of edges that have been deleted from $H$ by time $t$ is at least:
 	
 	\[ \frac{N}{32N^{128\eps^2}}\cdot \ceil{\Lambda_{j-1}}=\frac{N^{1-128\eps^2}}{32}\cdot \ceil{N^{j-9-300(j-1)\eps^2}}>N^{j-8-300j\eps^2}=\Lambda_j,  \] 
 	
 	a contradiction.
 \end{proof}
 
Consider now a graph $H_i\in \hset\setminus \hset^D$ at some time during the algorithm's execution. For every vertex $s\in S'(H_i)$, we denote by $\gset(s)$ the set of all graphs $H_{i'}\in \hset\setminus \hset^D$, such that $S'(H_{i'})\cap B_H(s,d^*_j/32)\neq \emptyset$.
In other words, set $\gset(s)$ contains all graphs $H_{i'}$ that have not been destroyed yet, such that some vertex in the current set $S'(H_{i'})$ of supported vertices is sufficiently close to $s$ in the current graph $H$. Recall that the set $S'(H_{i})$ of vertices is decremental. The following observation will be useful for us.

\begin{observation}\label{obs: critical set of graphs}
	Consider any time $t$ during the algorithm's execution. Let $H_i\in \hset\setminus\hset^D$ be a graph that has not been destroyed by time $t$, and let $s\in S'(H_i)$ be any vertex in the current set of supported vertices for $H_i$.  Then, since the beginning of the algorithm and until time $t$, the collection $\gset(s)$ of graphs has been decremental: that is, graphs may have left it, but no graph may have joined it since the beginning of the algorithm.
\end{observation}
\begin{proof}
	Assume for contradiction that $H_{i'}\in \hset$ is some graph that did not belong to set $\gset(s)$ at time $t'$, but belongs to set $\gset(s)$ at time $t''$, where $t'<t''\leq t$. 
	
	From the definition, at time $t''$, $H_{i'}\in \hset\setminus\hset^D$ held. Since graphs may leave set $\hset\setminus\hset^D$ (when they are destroyed) but they may never join $\hset\setminus\hset^D$ over the course of the algorithm, we get that at time $t'$, $H_{i'}\in \hset\setminus\hset^D$ held. Furthermore, from the definition, at time $t''$, some vertex $x\in S'(H_{i'})$ belonged to $B_H(s,d^*_j/32)$. Since set $S'(H_{i'})$ of vertices is decremental, $x\in S'(H_{i'})$ held at time $t'$. Since distances in graph $H$ may only grow over time, $x\in B_H(s,d^*_j/32)$ held at time $t'$. Therefore, $S'(H_{i'})\cap B_H(s,d^*_j/32)\neq \emptyset$ must have held at time $t'$, and graph $H_{i'}$ must have belonged to $\gset(s)$ at time $t'$, a contradiction. 
\end{proof}



Throughout the algorithm, the set $\hset^I$ of inactive graphs will only contain  graphs $H_i\in \hset$ that have not been destroyed yet, for which the following property holds:

\begin{properties}{P}
	\item For every vertex $s\in S'(H_i)$, $|\gset(s)|\leq 7N/8$.  \label{prop P}
\end{properties}

We note that it is possible that some graph $H_i\in \hset\setminus\hset^D$ has Property \ref{prop P} but is not added to $\hset^I$.
 
The following observation shows that once property $\ref{prop P}$ holds for some graph $H_i$, it will continue to hold until the algorithm terminates or the graph is destroyed.

\begin{observation}\label{obs: prop P}
	Let $H_i\in \hset$ be any graph, and assume that Property \ref{prop P} holds for $H_i$ at some time $t$ during the algorithm's execution. Then Property \ref{prop P} holds for $H_i$ from time $t$ and until the algorithm terminates or until $H_i$ is destroyed.
\end{observation}
\begin{proof}
	Assume that Property \ref{prop P} holds for graph $H_i\in \hset$ at some time $t$ during the algorithm's execution, and consider some time $t'>t$ during the algorithm's execution. We assume that graph $H_i$ is not destroyed at time $t'$, and prove that Property  \ref{prop P}  continues to hold for $H_i$ at time $t'$. Indeed, consider any vertex $s$ that lies in set $S'(H_i)$ at time $t'$. Since set $S'(H_i)$ is decremental, vertex $s$ lied in $S'(H_i)$ at time $t$. Since Property \ref{prop P}  held for graph $H_i$ at time $t$, $|\gset(s)|\leq 7N/8$ held at time $t$. From \Cref{obs: critical set of graphs}, set $\gset(s)$ is decremental, so $|\gset(s)|$ may not grow between time $t$ and time $t'$. Therefore, $|\gset(s)|\leq 7N/8$ holds at time $t'$. We conclude that   Property  \ref{prop P}  continues to hold for $H_i$ at time $t'$.
\end{proof}

To summarize, over the course of the algorithm, we maintain a partition of the collection $\hset$ of graphs into three subsets: the set $\hset^D$ of destroyed graphs; the set $\hset^I$ of inactive graphs; and the set $\hset^A$ of all remaining graphs, that are called \emph{active graphs}. We ensure that every graph in $\hset^I$ has Property \ref{prop P}. We also ensure that the set $\hset^A$ of active graphs is decremental -- graphs may leave it but they may not join it over time.
The following claim bounds the cardinality of the collection $\hset^I$ of graphs.

 \begin{claim}\label{claim: few inactive graphs}
 Over the course of the algorithm, $|\hset^I|\leq N/32$ always holds.
 \end{claim}

\begin{proof}
	Assume otherwise, and consider the first time $t$ during the algorithm when $|\hset^I|>N/32$ held. 
	We will construct two large sets $T_1,T_2$ of vertices, so that the distance between the vertices of $T_1$ and the vertices of $T_2$ is large in the current graph $H$. We will then reach a contradiction by using the facts that, at the beginning of the algorithm, graph $H$ was well-connected with respect to $S(H)$, and that the number of edges that were deleted from $H$ is relatively small.

	 Recall that $|\hset|=r=N-\ceil{2N^{1-\eps^4}}\geq 63N/64$, since $N^{\eps^4}\geq 2^{128/\eps^5}$ from the statement of \Cref{thm: APSP in HSS}.  Additionally, from \Cref{obs: few destroyed graphs}, $|\hset^D|< N/32$ holds at time $t$. Therefore, at time $t$, $ |\hset\setminus\hset^D|\geq \frac{63N}{64}-\frac{N}{32}\geq \frac{61N}{64}$.
Let $\tilde \hset\subseteq \hset\setminus\hset^D$ be a collection of graphs that is obtained as follows. We start with $\tilde \hset=\hset\setminus\hset^D$. We then consider the graphs of $\tilde \hset$ one by one, starting with the graphs of $\hset^A$. If, when graph $H_i$ is considered, $|\tilde \hset|>\ceil{\frac{61N}{64}}$ holds, then we discard $H_i$ from set $\tilde \hset$. Otherwise, we terminate the algorithm with the final collection $\tilde \hset$ of graphs, whose cardinality must be $\ceil{\frac{61N}{64}}$. Notice that, if any graph of $\hset^A$ lies in $\tilde \hset$, then $\hset^I\subseteq \tilde \hset$. 
	Recall that, from the induction hypothesis, for every graph $H_i\in \hset\setminus\hset^D$, $|S'(H_i)|\geq \frac{N^{j-1}}{16^{j-1}}$ holds throughout the algorithm.
	We construct a set $T$ of vertices as follows. For every graph $H_i\in \tilde \hset$, we let $S''(H_i)$ be an arbitrary  collection of  $\ceil{\frac{N^{j-1}}{16^{j-1}}}$ vertices that lie in set $S'(H_i)$ at time $t$. We then let $T=\bigcup_{H_i\in \tilde \hset}S''(H_i)$. Clearly:
		
	\[ |T|=\ceil{\frac{61N}{64}}\cdot\ceil{\frac{N^{j-1}}{16^{j-1}}}. \]

Intuitively, we would like to apply Procedure \procsep from \Cref{lem: find separation of terminals} to graph $H$, set $T$ of terminals, and distance parameters $\Delta=64/\eps^2$ and $d=\frac{d^*_j}{64\Delta}$, in order to compute two large subsets $T_1,T_2\subseteq T$ of vertices with $\dist_H(T_1,T_2)\geq d$. However, the procedure may instead return a single vertex $s\in T$, for which the ball $B_H(s,d^*_j/64)$ contains many vertices of $T$. In the next observation we show that this is impossible, that is, for every vertex $s\in T$, $|B_H(s,d^*_j/64)\cap T|$ is sufficiently small.
	
	\begin{observation}\label{obs: few terminals in ball}
		At time $t$, for every vertex $s\in T$, $|B_H(s,d^*_j/64)\cap T|<|T|\cdot \left(1-\frac{1}{256}\right )$ holds.
		\end{observation}
	\begin{proof}
		Consider some vertex $s\in T$, and denote $B=B_H(s,d^*_j/64)$. 
		Let $T'=\bigcup_{H_i\in \hset^I\cap \tilde \hset}S''(H_i)$. Assume first that $B$ does not contain any vertex of $T'$. Notice that in this case, at least one graph of $\hset^A$ lies in $\tilde \hset$, and so $\hset^I\subseteq \tilde \hset$. 
		Since we have assumed that $|\hset^I|>\frac{N}{32}$, we get that:
		
		\[|\tilde \hset\setminus \hset^I|\leq |\tilde \hset|-|\hset^I|\leq \ceil{\frac{61N}{64}}-\frac{N}{32}\leq \ceil{\frac{61N}{64}}\cdot \left(1-\frac 1 {61}\right )\leq \ceil{\frac{61N}{64}}\cdot \left(1-\frac 1 {256}\right ). \]
		
		Therefore:
		
		\[|B\cap T|\leq |T|-|T'|\leq |\tilde \hset\setminus \hset^I|\cdot\ceil{\frac{N^{j-1}}{16^{j-1}}}\leq \ceil{\frac{61N}{64}}\cdot \left(1-\frac 1 {256}\right )\cdot\ceil{\frac{N^{j-1}}{16^{j-1}}}\leq |T|\cdot \left(1-\frac 1 {256}\right ).\]
		
		

		Assume now that $B$ contains at least one vertex of $T'$, and let $s'$ be any such vertex. Note that $B=B_H(s,d^*_j/64)\subseteq B_H(s',d^*_j/32)$. Observe that $s'$ must lie in the current set $S'(H_i)$ of supported vertices of some graph $H_i$ that currently lies in $\hset^I$. From Property \ref{prop P}, $|\gset(s')|\leq \frac{7N}{8}$. Let $\gset'(s')=\hset\setminus \gset(s')$, so $|\gset'(s')|\geq \frac{N}{8}$. 
		
		 From our definitions, for every graph $H_{i'}\in \gset'(s')$, $S'(H_{i'})\cap B_H(s',d^*_j/32)=\emptyset$, and therefore, $S'(H_{i'})\cap B=\emptyset$. 
		 
		 Let $\tilde \hset'=\tilde \hset\cap \gset'(s')$. On the one hand, since $|\tilde \hset|=\ceil{\frac{61N}{64}}$, and $|\gset'(s')|\geq \frac{N}{8}$, while $|\hset|=N$, we get that $|\tilde \hset'|\geq \frac{5N}{64}\geq \ceil{\frac{61N}{64}}\cdot \frac{1}{256}$. On the other hand, the vertices of $\bigcup_{H_i\in \tilde \hset'}S''(H_i)$ may not lie in set $B$. Therefore:
		 
		 \[|B\cap T|\leq \left (|\tilde \hset|-|\tilde \hset'|\right )\cdot \ceil{\frac{N^{j-1}}{16^{j-1}}}\leq \ceil{\frac{61N}{64}}\cdot \left(1-\frac 1 {256}\right )\cdot\ceil{\frac{N^{j-1}}{16^{j-1}}}\leq |T|\cdot \left(1-\frac 1 {256}\right ).\]
%
%
	\end{proof}

We apply Algorithm \procsep from \Cref{lem: find separation of terminals} to graph $H$, set $T$ of terminals, distance parameters $\Delta=64/\eps^2$, and $d=\frac{d^*_j}{64\Delta}$, and parameter $\alpha=\left(1-\frac{1}{256}\right )$. From \Cref{obs: few terminals in ball}, the algorithm may not return a vertex $s\in T$ with $|B_H(s,\Delta\cdot d)\cap T|=|B_H(s,d^*_j/64)\cap T|>\alpha \cdot |T|$. Therefore, it must compute two subsets $T_1,T_2$ of vertices with $|T_1|=|T_2|$, such that $|T_1|\geq \frac{|T|^{1-64/\Delta}}{256}$, and for every pair $s\in T_1,s'\in T_2$ of terminals, $\dist_H(s,s')\geq d$. We will now exploit the facts that graph $H$ was well-connected with respect to set $S(H)$ of vertices at the beginning of the algorithm, and that relatively few edges were deleted from $H$, in order to reach a contradiction.

Observe first that:

\[|T_1|\geq \frac{|T|^{1-\eps^2}}{256}\geq \frac{1}{256}\cdot \left(\frac{61\cdot N^j}{4\cdot 16^j}\right )^{1-\eps^2}\geq \frac{N^{j(1-\eps^2)}}{1024\cdot 16^j}.   \]

On the other hand:

\[ d= \frac{d^*_j}{64\Delta}=\frac{2^{c'j/\eps^5}\cdot \eps^2}{2^{12}}>2^{cj/\eps^4}=\tilde d_j. \]

Let $H^{(0)}$ denote the graph $H$ at the beginning of the algorithm, and let $H^{(t)}$ denote graph $H$ at time $t$.
Recall that, at the beginning of the algorithm, graph $H^{(0)}$ was $(\eta_j,\tilde d_j)$-well-connected with respect to the set $S(H)=\bigcup_{H_i\in \hset}S(H_i)$ of vertices. Since the sets $S'(H_i)$ of vertices for graphs $H_i\in \hset$ are decremental, and since $S'(H_i)=S(H_i)$ at the beginning of the algorithm for each such graph, we get that $T\subseteq S(H)$ held at the beginning of the algorithm. Therefore, there was a collection $\pset(T_1,T_2)$ of paths in graph $H^{(0)}$, routing every vertex of $T_1$ to a distinct vertex of $T_2$, such that the paths in $\pset(T_1,T_2)$ cause congestion at most $\eta_j=N^{6+256j\eps^2}$, and the length of every path is at most $\tilde d_j$. 

Let $E'$ be the set of edges that have been deleted from graph $H$ by time $t$. Note that, in graph $H^{(t)}$, no path of length at most $\td_j$ connecting a vertex of $T_1$ to a vertex of $T_2$ exists. Therefore, set $E'$ must contain at least one edge from every path in $\pset(T_1,T_2)$. We conclude that:

\[|E'|\geq \frac{|T_1|}{\eta_j}\geq  \frac{N^{j(1-\eps^2)}}{1024\cdot 16^j\cdot N^{6+256j\eps^2}}=\frac{N^{j-257j\eps^2-6}}{1024\cdot 16^j} > N^{j-8-257j\eps^2}\geq \Lambda_j,\]

since $j\leq \ceil{1/\eps}$, $\eps\leq 1/400$, and $N^{\eps}>2^{8/\eps}$ from the statement of \Cref{thm: APSP in HSS}. This is a contradiction, since fewer than $\Lambda_j$ edges may be deleted from $H$.
\end{proof}

From \Cref{obs: few destroyed graphs} and \Cref{claim: few inactive graphs}, throughout the algorithm, $|\hset^I|+|\hset^D|\leq N/16$ holds. Since $|\hset|=r=N-\ceil{2N^{1-\eps^4}}\geq 63N/64$, we get that, throughout the algorithm:

\begin{equation}\label{eq: many active}
|\hset^A|\geq |\hset|-|\hset^I|-|\hset^D|\geq \frac{63N}{64}-\frac{N}{16}\geq \frac{59N}{64}.
\end{equation}

The set $S'(H)$ that we maintain throughout the algorithm is defined to be: $S'(H)=\bigcup_{H_i\in \hset^A}S'(H_i)$. Since, for every graph $H_i$, set $S'(H_i)$ of vertices is decremental, and since the collection $\hset^A$ of graphs is decremental, we get that the set $S'(H)$ of vertices is decremental as well. It is also easy to see that $S'(H)=S(H)=\bigcup_{H_i\in \hset}S(H)$ hods at the beginning of the algorithm. Since, from Inequality \ref{eq: many active}, $|\hset^A|\geq \frac{59N}{64}$ holds throughout the algorithm, and since, from the induction hypothesis, for every graph $H_i\in \hset^A$, $|S'(H_i)|\geq \frac{N^{j-1}}{16^{j-1}}$ holds, we get that, throughout the algorithm:

\[|S'(H)|\geq |\hset^A|\cdot  \frac{N^{j-1}}{16^{j-1}}\geq \frac{59N}{64}\cdot \frac{N^{j-1}}{16^{j-1}}\geq \frac{N^{j}}{16^{j}}, \]
 
 as required.

 \subsubsection{Data Structures and Initialization}
 Consider a graph $H_i\in \hset$. Note that, as part of the level-$j$ hierarchical support structure for $H$, we are given a level-$(j-1)$ hierarchical support structure for $H_i$, and we are guaranteed that $H_i$ is $(\eta_{j-1},\td_{j-1})$-well-connected with respect to $S(H_i)$.
 Therefore, from the induction hypothesis, we can apply the algorithm from \Cref{thm: APSP in HSS} with parameter $(j-1)$ to graph $H_i$, and the corresponding level-$(j-1)$ hierarchical support structure. Parameters $N$ and $\eps$ remain unchanged. We denote the corresponding data structure by $\dset_{j-1}(H_i)$. 
 
As part of the initialization procedure, for every graph $H_i\in \hset$, we initialize the corresponding data structure  $\dset_{j-1}(H_i)$.

Our algorithm also maintains, for every edge $e\in E(H)$, a list $L(e)$ of edges $e'\in \bigcup_{H_i\in \hset}E(H_i)$, such that the path $P(e')$,  contains $e$. Here, $P(e')$ is the path in the embedding $\pset$ of $\bigcup_{H\in \hset}H'$ into $H$, that serves as the embedding of edge $e'$.
Together with this list, we maintain a pointer from $e$ to each such edge $e'$ in its corresponding graph $H_i$, and a pointer in the opposite direction. We initialize the lists $L(e)$ for edges $e\in E(H)$ at the beginning of the algorithm.

We also initialize $\hset^A=\hset$, $\hset^I=\hset^D=\emptyset$, and $S'(H)=S(H)$.

Lastly, for every graph $H_i\in \hset$, we initialize an \EST data structure, whose corresponding tree is denoted by $\tau_i$, that, intuitively, is rooted at the set $S'(H_i)$ of vertices. Specifically, in order to construct data structure $\tau_i$, we let $\tilde G_i$ be a graph that is obtained from $H$, after we add a source vertex $s_i$ to it, which connects to every vertex in $S'(H_i)$ with an edge. We then let $\tau_i$ be an \EST data structure in graph $\tilde G_i$, rooted at vertex $s_i$, with depth bound $d^*_j/8+1$.
When an edge is deleted from $H$, we will also delete it from $\tilde G_i$, and update the \EST data structure $\tau_i$ accordingly. When a vertex $s$ is deleted from set $S'(H_i)$, we will delete the edge $(s_i,s)$ from graph $\tilde G_i$, and update $\tau_i$ accordingly.

We maintain, for every pair of graphs $H_i\in \hset^A$ and $H_{i'}\in \hset\setminus\hset^D$, a counter $n_{i,i'}$, that counts the number of vertices of $S'(H_{i'})$ that lie at distance at most $d^*_j/32+1$ from $s_i$ in tree $\tau_i$. Note that, if $n_{i,i'}=0$, then for every vertex $s\in S'(H_i)$,  $S'(H_{i'})\cap B_H(s,d^*_j/32)=\emptyset$. In other words, $H_{i'}\not\in \gset(s)$ holds for all $s\in S'(H_i)$. We also maintain a counter $\tilde n_i$, whose value is the number of graphs $H_{i'}\in \hset\setminus\hset^D$, with $n_{i,i'}>0$. 
We can initialize the values $n_{i,i'}$ for every pair $H_i,H_{i'}\in \hset$ of graphs, and the counters in $\set{\tilde n_i}_{H_i\in \hset}$ at the beginning of the algorithm. The time that is required in order to do so is subsumed by the time required to initialize the \EST data structures.
We need the following simple observation.

\begin{observation}\label{obs: move to Hi}
	Let $H_i\in \hset$ be a graph. Assume that at some time $t$ in the algorithm's execution, $H_i\not\in \hset^D$ holds, and $\tilde n_i\leq 7N/8$. Then graph $H_i$ has Property \ref{prop P} at time $t$.
	\end{observation}
\begin{proof}
	Assume otherwise. Then at time $t$, there is some vertex $s\in S'(H_i)$ with $|\gset(s)|>7N/8$. If $H_{i'}$ is a graph that lies in $\gset(s)$ at time $t$, then at least one vertex of $S'(H_{i'})$ must lie in $B_H(s,d^*_j/32)$ at time $t$, and so $n_{i,i'}>0$ must hold. But then $\tilde n_i>7N/8$ must hold at time $t$, a contradiction.
\end{proof}

Throughout the algorithm's execution, whenever, for some graph $H_i\in \hset^A$, the value of counter $\tilde n_i$ becomes at most $7N/8$, we move graph $H_i$ from $\hset^A$ to $\hset^I$. Lastly, we need the following simple observation.

\begin{observation}\label{obs: short distances between active}
	Let $H_i,H_{i'}$ be any pair of graphs that lie in set $\hset^A$ at some time $t$ during the algorithm's execution. Let $s$ be any vertex that lies in $S'(H_i)$ at time $t$, and let $s'$ be any vertex that lies in $S'(H_{i'})$ at time $t$. Then $\dist_H(s,s')\leq d^*_j/8$ at time $t$.
\end{observation}
\begin{proof}
	Let $\gset_i$ be the set of all graphs $H_{i''}\in \hset\setminus \hset^D$ with $n_{i,i''}>0$ at time $t$. Since $H_i\in \hset^A$ at time $t$, $|\gset_i|\geq 7N/8$ must hold. Similarly, let $\gset_{i'}$ be the set of all graphs $H_{i''}\in \hset\setminus \hset^D$ with $n_{i',i''}>0$ at time $t$. As before, $|\gset_{i'}|\geq 7N/8$. Therefore, there is some graph $H_{i''}\in \gset_i\cap \gset_{i'}$.
	
	In the remainder of this proof, whenever we refer to vertex sets $S'(H_i),S'(H_{i'}), S'(H_{i''})$, or to graphs $H_i,H_{i'},H_{i''}, H$, we mean the corresponding sets of vertices or the corresponding graphs at time $t$.
	
	Since, at time $t$, $n_{i,i''}>0$, there is some vertex $x\in S'(H_{i''})$, such that the distance from $s_i$ to $x$ in tree $\tau_i$ is at most $d^*_j/32+1$. Therefore, there is some vertex $x'\in S'(H_{i})$, such that $\dist_H(x,x')\leq d^*_j/32$. Let $Q$ be a path of length at most $d^*_j/32$ connecting $x$ to $x'$ in $H$. From a similar reasoning, there is a pair of vertices $y\in S'(H_{i''})$ and $y'\in S'(H_{i'})$, such that $\dist_H(y,y')\leq d^*_j/32$. Let $Q'$ be a path of length at most $d^*_j/32$ connecting $y$ to $y'$ in $H$.
	
	From the induction hypothesis, there is a path $P_1$ of length at most $d^*_{j-1}$ connecting $s$ to $x'$ in graph $H_i$. Recall that we are given an embedding $\pset$ of the edges of $\bigcup_{H_{a}\in \hset}E(H_a)$ into $H$, where the length of every path in $\pset$ is at most $2^{64/\eps^4}$. From \Cref{obs: paths from embedding}, there is a path $P_1'$ in graph $H$, connecting $s$ to $x'$, whose length is at most $2^{64/\eps^4}\cdot d^*_{j-1}$.
	
	Using similar reasonings, there is a path $P_2'$ in graph $H$ connecting $x$ to $y$ of length at most $2^{64/\eps^4}\cdot d^*_{j-1}$, and a path $P_3'$ in graph $H$ connecting $y'$ to $s'$, of length at most $2^{64/\eps^4}\cdot d^*_{j-1}$. By concatenating the paths $P_1',Q,P_2',Q',P_3'$, we obtain a path in graph $H$, connecting $s$ to $s'$, whose length is bounded by:
	
	\[ \frac{d^*_j}{16}+ 3\cdot 2^{64/\eps^4}\cdot d^*_{j-1}=\frac{d^*_j}{16}+3\cdot 2^{64/\eps^4}\cdot 2^{c'(j-1)/\eps^5} \leq \frac{d^*_j}{16}+\frac{2^{c'j/\eps^5}}{16}\leq \frac{d^*_j}{8}  \]
	
	(since $d^*_j=2^{c'j/\eps^5}$).
\end{proof}

Next, we describe an algorithm for updating the data structures after each edge deletion, followed by the analysis of the total update time of the algorithm. We then conclude with an algorithm for responding to queries.

\subsubsection{Maintaining the Data Structures}

For each graph $H_i\in \hset$, our algorithm will only maintain \EST data structure $\tau_i$, together with the corresponding counters $\tilde n_i$ and $n_{i,i'}$ for all $H_{i'}\in \hset\setminus\hset^D$, as long as graph $H_i$ lies in $\hset^A$. Once graph $H_i$ is removed from $\hset^A$, we no longer maintain these data structures.

We now describe an algorithm for updating the data structures following the deletion of an edge $e$ from graph $H$. The algorithm, called $\deledge(e)$, is straightforward and it is summarized in \Cref{alg: delete-edge}.

When an edge $e$ is deleted from graph $H$, we consider every edge $e'\in L(e)$ (that is, edges $e'\in \bigcup_{H_i\in \hset}E(H_i)$, whose embedding path $P(e')$ contains edge $e$). We assume that $e'\in E(H_i)$, and that $H_i\in \hset\setminus\hset^D$ (if $H_i\in \hset^D$, no further action is required for processing edge $e'$). We then delete edge $e'$ from graph $H_i$ and update the corresponding data structure $\dset_{j-1}(H_i)$. As a result, we obtain a set $X_i$ (that may be empty) of vertices that have been deleted from $S'(H_i)$. We also update lists $L(e'')$ of edges $e''\in E(H)$ that contain $e'$, to remove $e'$ from these lists.

If the number of edges deleted so far from $H_i$ reaches at least $\ceil{\Lambda_{j-1}}$, then graph $H_i$ is destroyed. We add the graph to set $\hset^D$, and we update all counters $n_{i',i}$ and $\tilde n_{i'}$ for graphs $H_{i'}\in \hset^A$ as needed. Otherwise, we process every vertex $s\in X_i$ one by one. If $H_i\in \hset^A$, then we delete edge $(s_i,s)$ from graph $\tilde G_i$, and the corresponding \EST data structure $\tau_i$. As the result, some distances in tree $\tau_i$ may have increased, and we need to update all counters in $\set{n_{i,i'}}_{H_{i'}\in \hset\setminus\hset^D}$, as well as counter $\tilde n_{i}$ accordingly.
Additionally, for every graph $H_{i'}\in \hset^A$, if the distance from $s_i$ to $s$ in tree $\tau_i$ was bounded by $d^*_j/32+1$, then we need to decrease $n_{i',i}$ by $1$ (as vertex $s$ no longer lies in $S'(H_i)$), and if needed, we need to update counter $\tilde n_{i'}$. 

Once we finish processing every edge $e'\in L(e)$, we also need to delete edge $e$ from every graph $\tilde G_{i'}$, where $H_{i'}\in \hset^A$, and the corresponding \EST data structure $\tau_{i'}$. As before, this may increase some distances in tree $\tau_{i'}$, and we may need to update counters $\set{n_{i',i''}}_{H_{i''}\in \hset\setminus\hset^D},\tilde n_{i'}$ accordingly. Lastly, we consider every graph  $H_{i'}\in \hset^A$ in turn. If $\tilde n_{i'}\leq 7N/8$ holds for any such graph, then we move it from  $\hset^A$ to $\hset^I$.

\program{Algorithm $\deledge(e)$}{alg: delete-edge}{
	
	\begin{enumerate}
		\item For every edge $e'\in L(e)$, such that the graph $H_i\in \hset$ containing $e'$ lies in $\hset\setminus\hset^D$ do:
		
		\begin{enumerate}
			\item Delete $e'$ from $H_i$ and update the corresponding data structure $\dset_{j-1}(H_i)$. Let $X_i$ be the set of vertices that were deleted from $S'(H_i)$ as the result of this update.
			
			\item For every edge $e''\in E(H)$ with $e'\in L(e'')$, delete $e'$ from $L(e'')$.
			
			\item If the number of edges deleted so far from $H_i$ becomes at least $\ceil{\Lambda_{j-1}}$:
			\begin{enumerate}
				\item Add graph $H_i$ to $\hset^D$ and remove it from the set $\hset^A$ or $\hset^I$ to which it belonged.
				\item For every graph $H_{i'}\in \hset^A$, if $n_{i',i}>0$, set $n_{i',i}$ to $0$ and decrease $\tilde n_{i'}$ by $1$.
			\end{enumerate}
		\item Otherwise: for every vertex $s\in X_i$ do:
		\begin{enumerate}
			\item If $H_i\in \hset^A$, delete edge $(s_i,s)$ from graph $\tilde G_i$ and update the \EST $\tau_i$, together with counters $\set{n_{i,i'}}_{H_{i'}\in \hset\setminus\hset^D},\tilde n_i$ accordingly. \label{step: update counters}
			
			\item For every graph $H_{i'}\in \hset^A$, if $\dist_{\tau_{i'}}(s_{i'},s)\leq d^*_j/32+1$, decrease $n_{i',i}$ by $1$. If $n_{i',i}$ decreases from $1$ to $0$, decrease $\tilde n_{i'}$ by $1$. \label{step: remove s}
		\end{enumerate}
		
		\end{enumerate}

		\item For every graph $H_{i'}\in \hset^A$, delete edge $e$ from graph $\tilde G_{i'}$, and update the  \EST $\tau_{i'}$, together with counters $\set{n_{i',i''}}_{H_{i''}\in \hset\setminus\hset^D},\tilde n_{i'}$ accordingly.
		
		\item For every graph $H_{i'}\in \hset^A$, if $\tilde n_{i'}\leq 7N/8$ holds, move $H_{i'}$ from $\hset^A$ to $\hset^I$.
	\end{enumerate}
}

It is easy to verify that the algoritm maintains all data structures correctly, and, from \Cref{obs: move to Hi}, when graph $H_i$ is added to $\hset^I$,  Property \ref{prop P} holds for it. From \Cref{obs: prop P}, once $H_i$ is added to $\hset^I$, Property \ref{prop P} continues to hold  for it until the end of the algorithm, or until $H_i$ is added to $\hset^D$. From the above discussion, we are guaranteed that, throughout the algorithm, $|\hset^A|\neq \emptyset$.

\subsubsection{Analysis of Total Update Time}

Let $m$ denote the number of edges in graph $H$ at the beginning of the algorithm.
Recall that every edge $e\in E(H)$ participates in at most $N^{128\eps^2}$ paths in $\pset$. Therefore, the length of the list $L(e)$ is bounded by $N^{128\eps^2}$. Every edge of $\bigcup_{H_i\in \hset}E(H_i)$ may be added at most once to list $L(e)$ when the data structure is initialized, and subsequently it may be deleted at most once from $L(e)$. Therefore, the total time required to maintain the lists $L(e)$ for all edges $e\in E(H)$ is at most $O(m\cdot N^{128\eps^2})$.

Consider now some graph $H_i\in \hset$, and denote by $m_i$ the number of edges in $H_i$ at the beginning of the algorithm.
Recall that, from the definition of the hierarchical support structure, 
$|E(H_i)|\leq N^{j-1+32\eps^2}$. From the induction hypothesis, maintaining data structure $\dset_{j-1}(H_i)$ recursively takes time at most:

\[\begin{split}
2c''(j-1)N^{j+2}\cdot 2^{4c'/\eps^6}+c''|E(H_i)|\cdot N^2\cdot 2^{4c'/\eps^6}&\leq 2c''(j-1)N^{j+2}\cdot 2^{4c'/\eps^6}+c''N^{j+1+32\eps^2}\cdot 2^{4c'/\eps^6}\\&\leq c''(2j-1) N^{j+2}\cdot 2^{4c'/\eps^6}.
\end{split}\]

Since $|\hset|=N$, the total update time needed in order to maintain these data structures for all graphs $H_i\in \hset$ is bounded by $c''(2j-1) N^{j+3}\cdot 2^{4c'/\eps^5}$.

Consider now some graph $H_i\in \hset$. The total update time that is needed in order to maintain \EST $\tau_i$ is bounded by:

\[O(jm\cdot d^*_j\cdot \log N)\leq O(jm\cdot 2^{c'j/\eps^5}\log N)\leq  O(m\cdot 2^{4c'/\eps^6}\log N),\]

since $d^*_j=2^{c'j/\eps^4}$ and $j\leq \ceil{1/\eps}$.
We can initialize the counters $n_{i,i'}$ for all graphs $H_{i'}\in \hset$ and $\tilde n_i$ without increasing this asymptotic running time. We can also perform updates to these counters in Step \ref{step: update counters} within the same asymptotic running time. Since $|\hset|=N$, this part of the algorithm takes total update time at most $O(N\cdot m\cdot 2^{4c'/\eps^6}\log N)$.

Whenever a vertex $s$ is deleted from a set $S'(H_i)$ for any graph $H_{i}\in \hset$, we may need to update the counters $n_{i',i}$ and $\tilde n_{i'}$ for some graphs $H_{i'}\in \hset^A$. This can be done in time $O(N)$ per vertex. Since a vertex may be deleted at most once from $\bigcup_{H_i\in \hset}S'(H_i)$, these updates can be done in time $O(N^{j+1})$.

Lastly, every graph $H_i\in \hset$ may be moved to set $\hset^D$ at most once over the course of the algorithm, at which time we may need to update counters $n_{i',i}$ and $\tilde n_{i'}$ for graphs $H_{i'}\in \hset_A$. This takes time $O(N)$ per graph $H_i\in \hset$, and $O(N^2)$ overall.

From the above discussion, the total update time of the algorithm is bounded by:

\[ c''(2j-1) N^{j+3}\cdot 2^{4c'/\eps^6}+O(N\cdot m\cdot 2^{4c'/\eps^6}\log N)+O(N^{j+1})\leq 2c''jN^{j+3}\cdot 2^{4c'/\eps^6}+c''m\cdot N^2\cdot 2^{4c'/\eps^6}.\]

\subsubsection{Response to Queries}

In this subsection we describe an algorithm for responding to a short-path query between a pair $x,y\in S'(H)$ of vertices. Recall that the goal is to return a path $P$ of length at most $d^*_j$ connecting $x$ to $y$ in $H$, in time $O(|E(P)|)$.

Recall that $S'(H)=\bigcup_{H_i\in \hset^A}S'(H_i)$. Therefore, there is a pair of graphs $H_i,H_{i'}\in \hset^A$ (where possibly $H_i=H_{i'}$), with $x\in S'(H_i)$ and $y\in S'(H_{i'})$. From \Cref{obs: short distances between active},  $\dist_H(x,y)\leq d^*_j/8$ must hold. In particular, if we consider the \EST data structure $\tau_i$, then $y\in V(\tau_i)$ must hold, and the distance from $y$ to $s_i$ in the tree must be at most $d^*_j/8+1$. Let $Q_1$ be the path connecting $y$ to $s_i$ in tree $\tau_i$. We delete the last vertex on the path, and let $x'\in S'(H_i)$ be the new last vertex of $Q_1$. Using the induction hypothesis, we can compute a path $P'$ in graph $H_i$ connecting $x$ to $x'$, whose length is at most $d^*_{j-1}$. Assume that the sequence of edges on path $Q'$ is $(e_1,e_2,\ldots,e_z)$. For all $1\leq z'\leq z$, consider the path $P(e_{z'})\in \pset$ that serves as the embedding of edge $z'$ into $H$. Recall that the length of the path is bounded by $2^{64/\eps^4}$. By concatenating the paths $P(e_1),\ldots,P(e_z)$, we obtain a path $Q_2$ in graph $H$, connecting $x$ to $x'$, whose length is at most $d^*_{j-1}\cdot 2^{64/\eps^4}$. Lastly, by concatenating paths $Q_1$ and $Q_2$, we obtain a path $P$ in graph $H$ connecting $x$ to $y$. The length of the path is bounded by:

\[d^*_{j-1}\cdot 2^{64/\eps^4}+\frac{d^*_j}{8}\leq 2^{c'(j-1)/\eps^5}\cdot 2^{64/\eps^4} +\frac{2^{c'j/\eps^5}}{8} \leq 2^{c'j/\eps^5}=d^*_j. \]

since $d^*_j=2^{c'j/\eps^5}$. It is easy to see that the algorithm can be implemented in time $O(|E(P)|)$.

\section{APSP in Expanders -- Proof of \Cref{thm: APSP on expanders main}}
\label{sec: APSP on expanders}

This section is dedicated to the proof of \Cref{thm: APSP on expanders main}. We first prove the following lemma, that can be viewed as a weaker variation of \Cref{thm: APSP on expanders main}, in the sense that it can only withstand a significantly shorter sequence of edge deletions.

\begin{lemma}
	\label{lem: APSP on expanders one phase} 
	There is a deterministic algorithm whose input consists of
	an $n$-vertex graph $G$ with $|E(G)|=m$ that is a $\phi^*$-expander, for some $0<\phi^*<1$, with maximum vertex degree at  most $\Delta$, and a parameter $\frac{2}{(\log n)^{1/12}}< \eps<\frac{1}{400}$, such that $1/\eps$ is an integer. We assume that graph $G$ undergoes an online sequence of at most $ \frac{n^{1-20\eps}(\phi^*)^2}{\Delta^2}$ edge deletions. The algorithm
	maintains a set $U\subseteq V(G)$ of vertices, such that, for every integer $t>0$, after $t$ edges are deleted from $G$, $|U|\leq \frac{4\Delta  t}{\phi^*} $ holds. 
	Vertex set $U$ is incremental, so vertices may join it but they may not leave it.
	The algorithm also supports short-path queries: given a pair of vertices $x,y\in V(G)\setminus U$, return an $x$-$y$  path $P$ in the current graph $G$, of length at most $\frac{2^{O(1/\eps^6)}\cdot \Delta\cdot\log n}{\phi^*}$, with query time $O(|E(P)|)$. 
	The total update time of the algorithm is $O\left(\frac{n^{1+O(\eps)}\cdot \Delta^3}{(\phi^*)^2}\right )$.
\end{lemma}

The proof of \Cref{lem: APSP on expanders one phase}  is deferred to \Cref{subsec: APSP in expanders one phase}. We now complete the proof of \Cref{thm: APSP on expanders main} using it.
We partition the execution of our algorithm into phases.  
Let $k'=\floor{\frac{n^{1-20\eps}\phi^2}{2^{11}\cdot \Delta^4}}$. 
The first phase lasts as long as the number of edges deleted from $G$ via the input update sequence is at most $k'$. Once $k'$ edges are deleted from graph $G$, the second phase begins. Each subsequent phase similarly lasts as long as at most $k'$ edges are deleted since the beginning of the phase, except for the last phase which may be shorter, if the input update sequence terminates before $k'$ edges are deleted from $G$ since the beginning of the phase.
For all $i\geq 1$, we denote by $\Sigma_i$ the sequence of edge deletions that graph $G$ undergoes as part of the online input sequence of edge deletions in phase $i$, and we denote by $E_i$ the set of edges that belong to $\Sigma_i$. 
Since the total number of edges in the input sequence of edge deletions is bounded by $ \frac{n\cdot \phi^2}{2^{13} \Delta^4}$, we get that the number of phases is bounded by $\frac{n\cdot \phi^2}{2^{13}\cdot \Delta^4\cdot k'}\leq \frac{n^{20\eps}}4$.

We define another dynamic graph $G'$. At the beginning of the algorithm, we set $G'=G$. As the algorithm progresses, we will delete some edges from $G'$. Specifically, at the end of every phase of the algorithm, we will define a set of edges to be deleted from graph $G'$; we do not delete any edges from $G'$ as long as a phase progresses. We will run the algorithm from  \Cref{thm: expander pruning} (expander pruning) on graph $G'$, and we let $\tilde U$ be the set of vertices of $G'$ that it maintains. We will ensure that the number of edges that are deleted from $G'$ over the course of the entire algorithm is bounded by $k=\frac{\phi n}{16\Delta}$. Clearly, $k\leq \frac{\phi\cdot |E(G)|}{10\Delta}$. From  \Cref{thm: expander pruning}, we are guaranteed that, throughout the algorithm, $|\tilde U|\leq \frac{8k\Delta}{\phi}\leq \frac{n}{2}$ holds.

In order to execute the first phase, we let $H_1=G'=G$, and we apply the algorithm from \Cref{lem: APSP on expanders one phase} to graph $H_1$, and the online sequence $\Sigma_1$ of edge deletions.
Recall that graph $G$ is a $\phi$-expander, and that $|\Sigma_1|\leq k'=\floor{\frac{n^{1-20\eps}\phi^2}{2^{11}\cdot \Delta^4}}$. Recall that the algorithm maintains an incremental set $U\subseteq V(G)$ of vertices, that we denote by $U_1$, such that, for every integer $t>0$, after $t$ edges are deleted from $G$, $|U_1|\leq \frac{4\Delta t}{\phi}$ holds. 
Throughout the first phase, we let the set $U$ of vertices that our algorithm maintains be $U_1$. Recall that the algorithm from \Cref{lem: APSP on expanders one phase} supports queries short-path queries: given a pair of vertices $x,y\in V(G)\setminus U_1$, return an $x$-$y$  path $P$ in the current graph $H_1=G$, of length at most $\frac{2^{O(1/\eps^6)}\cdot \Delta\cdot\log n}{\phi}$, with query time $O(|E(P)|)$.

From the above discussion, at the end of the first phase, $|U_1|\leq \frac{4\Delta k'}{\phi}\leq \frac{n^{1-20\eps}\cdot \phi}{2^{9}\cdot \Delta^3}$ holds. We denote by $\tilde E_1$ the set of all edges that are incident to the vertices of $U_1$ in the current graph $G$. We then update graph $G'$, by deleting the edges of $E_1\cup \tilde E_1$ from it. Therefore, at the end of the first phase, the number of edges that are deleted from $G'$ is bounded by:

\[|E_1|+|\tilde E_1|\leq  \frac{n^{1-20\eps}\cdot \phi}{2^{8}\cdot \Delta^2}<k. \]

We denote by $\tilde U_1$ the set of vertices that the algorithm from \Cref{thm: expander pruning} produces at the end of Phase 1, and the deletion of the edges of $E_1\cup \tilde E_1$ from $G'$. We also denote by $H_2=G'\setminus \tilde U_1$. From \Cref{thm: expander pruning}, graph $H_2$ is a $\phi/(6\Delta)$-expander. 
Notice also that, if $U_1$ is the set of vertices that the algorithm from \Cref{lem: APSP on expanders one phase}  maintains at the end of the phase, then $U_1\subseteq \tilde U_1$ must hold.

The remaining phases are executed similarly, with several minor differences. We let $\tilde U_i$ be the set $\tilde U$ that the algorithm from \Cref{thm: expander pruning} produces at the end of Phase $(i-1)$, and we denote $H_i=G'\setminus \tilde U_i$. We will ensure that the total number of edges that are deleted from graph $G'$ over the course of the first $(i-1)$ phases is bounded by:

\[(i-1)\cdot \frac{n^{1-20\eps}\cdot \phi}{32\cdot \Delta}. \]

Since the total number of phases is bounded by $ \frac{n^{20\eps}}4$, this ensures that the number of edges deleted so far from $G'$ is less than $\frac{n\cdot \phi}{16\Delta}=k$. Therefore, from \Cref{thm: expander pruning}, graph $H_i$ is a $\phi/(6\Delta)$-expander, and furthermore, $|\tilde U_i|\leq n/2$, so $|V(H_i)|\geq n/2$. Denote $\phi^*=\frac{\phi}{6\Delta}$, and note that:

\[k'=\floor{\frac{n^{1-20\eps}\phi^2}{2^{11}\cdot \Delta^4}}\leq \frac{|V(H_i)|^{1-20\eps}(\phi^*)^2}{\Delta^2}  \]

Therefore, we can apply the algorithm from \Cref{lem: APSP on expanders one phase} to graph $H_i$ with the sequence $\Sigma_i$ of edge deletions. We denote by $U_i$ the incremental set of vertices of $H_i$ that the algorithm maintains. The set $U$ of vertices of $G$ that our algorithm maintains over the course of the $i$th phase is defined to be $\tilde U_i\cup U_i$.

Recall that the total number of edges that are deleted over the course of the first $(i-1)$ phases of the algorithm from graph $G$ is $(i-1)\cdot k'\geq (i-1)\cdot \floor{\frac{n^{1-20\eps}\phi^2}{2^{11}\cdot \Delta^4}}$.

Since the total number of edges deleted from graph $G'$ over the course of the first $(i-1)$ phases is bounded by $(i-1)\cdot \frac{n^{1-20\eps}\cdot \phi}{32\cdot \Delta}$, from \Cref{thm: expander pruning}, we get that:

\[ |\tilde U_i|\leq \frac{n^{1-20\eps}\cdot (i-1)}{4}.\]

Consider some integer $0\leq t\leq k'$, and the time during the execution of phase $i$, immediately after the $t$-th edge of $\Sigma_i$ is deleted from graph $G$ in phase $i$. Then at this time, the total number of edges deleted from graph $G$ since the beginning of the algorithm is at least: 

\[m^i_t=(i-1)\cdot \floor{\frac{n^{1-20\eps}\phi^2}{2^{11}\cdot \Delta^4}}+t.\]

At the same time:

\[|U|=|\tilde U_i|+|U_i|\leq  (i-1)\cdot \frac{n^{1-20\eps}}{4}+\frac{4\Delta t}{\phi^*}= (i-1)\cdot \frac{n^{1-20\eps}}{4}+\frac{24\Delta^2 t}{\phi}\leq \frac{2^{11}\Delta^4}{\phi^2}\cdot m^i_t. \]

Notice that, over the course of the $i$th phase, $V(H_i)\setminus U_i=V(G)\setminus(U_i\cup \tilde U_i)=V(G)\setminus U$ holds. Therefore, when short-path query arrives for a pair $x,y\in V(G)\setminus U$ of vertices, it must be the case that $x,y\in V(H_i)\setminus U_i$. We can then perform 
short-path query in the data structure maintained by the algorithm from \Cref{lem: APSP on expanders one phase}, to obtain a path $P$ in the current graph $H_i\subseteq G$, of length at most $\frac{2^{O(1/\eps^6)}\cdot \Delta\cdot\log n}{\phi^*}\leq \frac{2^{O(1/\eps^6)}\cdot \Delta^2\cdot\log n}{\phi}$, in time $O(|E(P)|)$. We return this path as the response to the query.

From the above discussion, at every time $t$ during the execution of phase $i$, if $m_{t}^i$ is the number of edges deleted so far by the algorithm, and $U^{(t)}$ is the current set $U$, then:

\[ |U^{(t)}|\leq \frac{2^{11}\Delta^2}{\phi^2}\cdot m^i_t.  \]

Once all edges of $\Sigma_i$ edges are deleted from graph $G$, we let $\tilde E_i$ be the set of all edges that are incident to the vertices in the current set $U_i$. Observe that:

\[ |\tilde E_i|\leq \Delta\cdot |U_i|\leq \frac{4\Delta^2}{\phi^*}\cdot k'\leq  
\frac{24\Delta^3}{\phi}\cdot \frac{n^{1-20\eps}\phi^2}{2^{11}\cdot \Delta^4}\leq 
\frac{n^{1-20\eps}\cdot \phi}{64\cdot \Delta}.
\]

since $\phi^*=\frac{\phi}{6\Delta}$ and $k'=\floor{\frac{n^{1-20\eps}\phi^2}{2^{11}\cdot \Delta^4}}$.

We delete the edges of $E_i\cup \tilde E_i$ from graph $G'$, and update the data structure maintained by the algorithm from \Cref{thm: expander pruning}. Since $|E_i|\leq k'=\floor{\frac{n^{1-20\eps}\phi^2}{2^{11}\cdot \Delta^4}}$, we get that $|E_i\cup \tilde E_i|\leq \frac{n^{1-20\eps}\cdot \phi}{32\cdot \Delta}$. Recall that we have assumed that, over the course of the first $(i-1)$ phases, the total number of edges deleted from graph $G'$ is bounded by 
$(i-1)\cdot \frac{n^{1-20\eps}\cdot \phi}{32\cdot \Delta}$. We then get that, over the course of the first $i$ phases, the total number of edges deleted from graph $G'$ is bounded by $i\cdot \frac{n^{1-20\eps}\cdot \phi}{32\cdot \Delta}$. We let $\tilde U_{i+1}$ the set $\tilde U$ that the algorithm from \Cref{thm: expander pruning} maintains after the edges of $E_i\cup \tilde E_i$ are deleted from $G'$. Since the vertices of $U_i$ are isolated in graph $G'\setminus \tilde E_i$, while graph $G'\setminus \tilde U_{i+1}$ must be a $\phi/(6\Delta)$-expander, we get that $U_i\subseteq \tilde U_{i+1}$. We then let the set $U$ that the algorithm maintains be $\tilde U_{i+1}$, and we continue to the next phase, with the graph $H_{i+1}=G'\setminus \tilde U_{i+1}$. 

It is easy to verify that the set $U$ of vertices that the algorithm maintains is incremental. Indeed, consider some phase $i$ of the algorithm. Throughout the phase, we let $U=\tilde U_i\cup U_i$, where $\tilde U_i$ is fixed over the course of the phase, and $U_i$ is incremental. If we denote by $U_i'$ the set $U_i$ at the end of Phase $i$, then we are guaranteed that $U'_i\subseteq \tilde U_{i+1}$, and, from \Cref{thm: expander pruning}, $\tilde U_i\subseteq \tilde U_{i+1}$. At the beginning of Phase $(i+1)$, we set $U=\tilde U_{i+1}$, and so $\tilde U_i\cup U_i'\subseteq U$ at this point. Therefore, set $U$ is incremental
throughout the algorithm.

It now only remains to bound the total update time of the algorithm. The algorithm consists of at most $O(n^{20\eps})$ phases. In every phase, we run the algorithm from \Cref{lem: APSP on expanders one phase}, whose total update time is $O\left(\frac{m^{1+O(\eps)}\cdot \Delta^3}{(\phi^*)^2}\right )\leq O\left(\frac{m^{1+O(\eps)}\cdot \Delta^5}{\phi^2}\right )$.

Additionally, the total update time of the algorithm from 
\Cref{thm: expander pruning}, over the course of at most $m$ deletions of edges from $G'$, is bounded by  $\Otilde\left (\frac{m\Delta^2}{\phi^2}\right)$.

Altogether, the total update time of the algorithm is bounded by:

\[O(n^{20\eps})\cdot O\left(\frac{m^{1+O(\eps)}\cdot \Delta^5}{\phi^2}\right )+ \Otilde\left (\frac{m\Delta^2}{\phi^2}\right)\leq   O\left(\frac{m^{1+O(\eps)}\cdot \Delta^5}{\phi^2}\right ).\]

In order to complete the proof of \Cref{thm: APSP on expanders main}, it now remains to prove \Cref{lem: APSP on expanders one phase}.

\subsection{Proof of \Cref{lem: APSP on expanders one phase}}
\label{subsec: APSP in expanders one phase}

	We start by describing the data structures that the algorithm maintains, together with their initialization. We then describe an algorithm for maintaining the data structures under the deletion of edges from $G$. Finally, we describe an algorithm for responding to short-path query.

	Before we do so, we establish some bounds on the parameters that will be useful for us later. All of these bounds follow from the fact that  $\frac{2}{(\log n)^{1/12}}< \eps<\frac{1}{400}$ holds, from the statement of \Cref{lem: APSP on expanders one phase}.

First, since $\eps>\frac{2}{(\log n)^{1/12}}$, we get that $\log n\geq (2/\eps)^{12}$, and: 

\begin{equation}
n\geq 2^{(2/\eps)^{12}}\geq 2^{800^{12}}.\label{eq: bound 2}
\end{equation}	

Additionally:

\begin{equation}\label{eq: bound neps}
n^{\eps^8}>n^{(2/(\log n)^{1/12})^8}= n^{256/(\log n)^{2/3}} >2^{(\log n)^{1/3}}>\log n.
\end{equation}

	\subsubsection{Data Structures and Initialization}
	We start by
	applying the algorithm from \Cref{cor: HSS witness} to graph $G$, with the set $T=V(G)$ of terminals, parameter $\eps$ given by the statement of \Cref{lem: APSP on expanders one phase}, and parameters $d=\frac{64\Delta\log m}{\phi^*}$ and $\eta=\frac{512\Delta^2\log m}{(\phi^*)^2}$.

	We claim that the algorithm may not return a pair 
	$T_1,T_2\subseteq T$ of disjoint subsets of terminals, and a set $E'$ of edges of $G$, with $|T_1|=|T_2|$, $|T_1|\geq \frac{n^{1-4\eps^3}}{4}$ and $|E'|\leq \frac{d\cdot |T_1|}{\eta}$, such that for every pair $t\in T_1,t'\in T_2$ of terminals, $\dist_{G\setminus E'}(t,t')>d$. Indeed, assume for contradiction that the algorithm returns a pair $T_1,T_2$ of disjoint subsets of vertices of $G$ with the above properties. Denote $\phi=\frac{\phi^*}{\Delta}$, and note that $|E'|\leq \frac{d\cdot |T_1|}{\eta}=\frac{\phi^* \cdot |T_1|}{8\Delta}\leq \frac{\phi\cdot |T_1|}{4}$. 
	Note also that $d=\frac{64\Delta\log m}{\phi^*}\geq\frac{32\log m}{\phi}$.
	Clearly, we can view $(T_1,T_2,E')$ as a $(\delta,d)$-distancing in graph $G$, for some parameter $0<\delta<1$. From \Cref{lem: distancing to sparse cut}, there is a cut $(X,Y)$ in graph $G$, with $T_1\subseteq X$ and $T_2\subseteq Y$, such that $|E_G(X,Y)|\leq \phi\cdot \min\set{|E(X)|,|E(Y)|}$. Since the maximum vertex degree in $G$ is bounded by $\Delta$, $|E(X)|< \Delta\cdot |X|$, and similarly $|E(Y)|< \Delta\cdot |Y|$. Therefore, we are guaranteed that $|E(X,Y)|< \phi\cdot \Delta\cdot \min\set{|X|,|Y|}=\phi^*\cdot \min\set{|X|,|Y|}$, contradicting the fact that $G$ is a $\phi^*$-expander.

We conclude that the algorithm from \Cref{cor: HSS witness} must return a graph $H$ with $V(H)\subseteq V(G)$, $|V(H)|=N^{1/\eps}\geq n-n^{1- \eps/2}\geq n/2$, where $N=\floor{n^{\eps}}$, so that the maximum vertex degree in $H$ is at most  $n^{32\eps^3}$. The algorithm also must return  an embedding $\pset$ of $H$ into $G$ via paths of length at most $d$ that cause congestion at most $\eta\cdot n^{32\eps^3}$, and a level-$(1/\eps)$ hierarchical support structure for $H$, such that $H$ is $(\eta',\td)$-well-connected with respect to the set $S(H)$ of vertices defined by the support structure, where $\eta'=N^{6+256\eps}$, and $\td=2^{c/ \eps^5}$, with $c$ being the constant used in the definition of the \HSS.
Recall that the running time of the algorithm is:

\[O\left (n^{1+O(\eps)}+|E(G)|\cdot n^{O(\eps^3)}\cdot(\eta+d\log n)\right )\leq O\left(n^{1+O(\eps)}+ n^{1+O(\eps^3)}\cdot \frac{\Delta^3\log^2 n}{(\phi^*)^2}\right )\leq O\left (\frac{n^{1+O(\eps)}\cdot \Delta^3}{(\phi^*)^2}\right ).\]

(We have used the fact that $d=\frac{64\Delta\log m}{\phi^*}$, $\eta=\frac{512\Delta^2\log m}{(\phi^*)^2}$, and Inequality \ref{eq: bound neps}). Let $q=1/\eps$.
	
For every edge $e'\in E(H)$, let $P(e')\in \pset$ be the path embedding edge $e'$ into $G$. For every edge $e\in E(G)$, we will maintain a list $L(e)$ of all edges $e'\in E(H)$ with $e\in E(P(e'))$. The list also contains, for each edge $e'\in L(e)$, a pointer to edge $e'$ in graph $H$, and each edge $e'\in E(H)$ maintains a pointer to every edge in $E(P(e'))$. Whenever an edge $e\in E(G)$ is deleted from graph $G$, we will delete every edge $e'\in L(e)$ from graph $H$. Since the paths in $\pset$ cause congestion at most $\eta\cdot n^{32\eps^3}=\frac{512\Delta^2\cdot  n^{32\eps^3}\log m}{(\phi^*)^2}\leq \frac{\Delta^2\cdot n^{33\eps^3}}{(\phi^*)^2}$ in $H$, every deletion of an edge in $G$ may trigger the deletion of at most $\tilde \eta=\frac{\Delta^2\cdot n^{33\eps^3}}{(\phi^*)^2}$ edges from $H$. Let $q=1/\eps$. Then the total number of edge deletions from graph $H$ over the course of the algorithm is bounded by:

\[\frac{n^{1-20\eps}(\phi^*)^2}{\Delta^2}\cdot \tilde \eta= n^{1-20\eps+33\eps^2}\leq \frac{n^{1-16\eps}}{2}\leq N^{q-16q\eps}< N^{q-8-300q\eps^2},  \]

as  $q=1/\eps$ and $\frac n 2\leq N^q$.


 Note that $\eta'=N^{6+256\eps}=N^{6+256q\eps^2}=\eta_q$ and $\td=2^{c/ \eps^5}=\td_q$, where $\eta_q$ and $\td_q$ are the parameters from the definition of Hierarchical Support Structure. We will use the algorithm from \Cref{thm: APSP in HSS} in graph $H$, with parameter $j=q$, and parameters $\eps,N$ remaining unchanged.
In order to be able to use the theorem, we need to verify that $\frac{N^{\eps^4}}{\log N}\geq 2^{128/\eps^6}$ holds.
Since $N=\floor{n^{\eps}}$, we get that:

\[\frac{N^{\eps^4}}{\log N}\geq \frac{n^{\eps^6}}{\eps\cdot\log n}\geq n^{\eps^6/2}\geq 2^{128/\eps^6}. \]

(we have used the fact that, from Inequality \ref{eq: bound neps}, $\log n<n^{\eps^8}$, from Inequality \ref{eq: bound 2}, $n\geq 2^{(2/\eps)^{12}}$, and $\eps<1/400$).

As observed already, the number of edge deletions that graph $H$ undergoes over the course of the algorithm is bounded less than $N^{q-8-300q\eps^2}=\Lambda_q$. We will maintain a data structure from \Cref{thm: APSP in HSS} in graph $H$, with parameters $\eps,N$ and $q$ as defined above. We denote this data structure by $\dset(H)$. At the beginning of the algorithm, we initialize this data structure. 
Recall that data structure $\dset(H)$ maintains 
 a decremental set $S'(H)\subseteq S(H)$ of vertices of $H$, called {supported vertices}, such that, at the beginning of the algorithm, $S'(H)=S(H)$. The algorithm ensures that $|S'(H)|\geq \frac{N^q}{16^q}$ holds over the course of the algorithm, and it supports short-path queries between supported vertices: given a pair $x,y\in S'(H)$ of vertices, return a path $P$ connecting $x$ to $y$ in the current graph $H$, whose length is at most $d^*_q=2^{O(q/\eps^5)}=2^{O(1/\eps^6)}$, in time $O(|E(P)|)$. 
 If $m'\leq n^{1+32\eps^3}$ is the number of edges in $H$ at the beginning of the algorithm, then the total update time needed to maintain data structure $\dset(H)$ is bounded by:

 \[O\left( qN^{q+3}\cdot 2^{O(1/\eps^6)}+m'\cdot N^2\cdot 2^{O(1/\eps^6)}\right )\leq O\left (n^{1+O(\eps)}\cdot 2^{O(1/\eps^5)}\right )\leq O\left (n^{1+O(\eps)}\right ) \]

(We have used the fact that, from Inequality \ref{eq: bound 2}, $n^{\eps}\geq 2^{(2/\eps)^{12}\cdot\eps}\geq 2^{O(1/\eps^6)}$).

Lastly, we maintain an \EST data structure $\tau$ in graph $G$, rooted at the set $S'(H)$ of vertices, with depth parameter $d=\frac{64\Delta\log m}{\phi^*}$. Specifically, we maintain a graph $G'$, that is obtained from graph $G$ by adding a source vertex $s$, that connects to every vertex that lies in the current set $S'(H)$ of supported vertices with an edge. We then let $\tau$ be an \EST data structure in graph $G'$, rooted at $s$, with depth $d+1$. Whenever a vertex $x$ is deleted from set $S'(H)$, we will delete edge $(s,x)$ from $G'$, and update the data structure $\tau$ accordingly. Also, whenever an edge $e$ is deleted from graph  $G$, we will also delete $e$ from graph $G'$, and update the data structure $\tau$ accordingly. Throughout the algorithm, we let $U$ be the set of all vertices $v\in V(G)$, such that $v\not\in V(\tau)$. In other words, $\dist_G(S'(H),v)>d$. Clearly, the set $U$ of vertices is incremental: vertices can joint it but they cannot leave it. In the next claim we bound the cardinality of $U$.

\begin{claim}\label{claim: few deleted edges}
Let $t$ be any time during the algorithm's execution, let $E'_t$ be the set of edges that were deleted so far from $G$, and let $U_t$ be the current set $U$ of vertices. Then $|U_t|\leq \frac{4\Delta |E'_t|}{\phi^*}$.
	\end{claim}
	
\begin{proof}
	Assume otherwise, and let $t$ be some time during the algorithm's execution, when  $|U_t|> \frac{4\Delta |E'_t|}{\phi^*}$ holds. Denote $\phi=\frac{\phi^*}{\Delta}$, so that $|E'_t|<\phi\cdot |U_t|/4$.
	
	Let $S'_t$ be the set $S'(H)$ of vertices at time $t$. Recall that we are guaranteed that $|S'(H)|\geq \frac{N^q}{16^q}\geq \frac{n}{2^{4/\eps+1}}$ holds throughout the algorithm. Since the total number of the edges deleted from $G$ over the course of the algorithm is at most $\frac{n^{1-20\eps}(\phi^*)^2}{\Delta^2}=n^{1-20\eps}\phi^2\leq n^{1-20\eps}\cdot \phi$, we get that $|E'_t|\leq n^{1-20\eps}\cdot \phi$. On the other hand, from Inequality \ref{eq: bound 2}, $n^{20\eps}\geq 2^{4/\eps+4}$, and so:

	\[|S'_t|\geq  \frac{n}{2^{4/\eps+1}}\geq 4n^{1-20\eps}\geq \frac{4|E'_t|}{\phi}.  \]
	
	 We conclude that $|E'_t|\leq \frac{\phi}{4}\cdot\min\set{ |S'_t|,|U_t|}$.
	Denote $M=\min\set{|S'_t|,|U_t|}$. Let $X\subseteq S'_t$, $Y\subseteq U_t$ be arbitrary subsets of vertices of cardinality $M$ each. Let $G^0$ be the graph $G$ at the beginning of the algorithm. Clearly, $\dist_{G^0\setminus E'_t}(X,Y)>d$. Recall that $d=\frac{64\Delta\log m}{\phi^*}=\frac{64\log m}{\phi}$. From \Cref{lem: distancing to sparse cut}, there is a cut $(X',Y')$ in graph $G^0$, with 
	 with $X\subseteq X'$ and $Y\subseteq Y'$, such that $|E_{G^0}(X',Y')|\leq \phi\cdot \min\set{|E_{G^0}(X')|,|E_{G^0}(Y')|}$. Since $|E_{G^0}(X')|<\Delta\cdot |X'|$ and $|E_{G^0}(Y')|< \Delta\cdot |Y'|$, we get that:
	 
	 \[|E_{G^0}(X',Y')|< \phi\cdot \Delta \cdot \min\set{|X'|,|Y'|}=\phi^*\cdot  \min\set{|X'|,|Y'|}, \]
	
	contradicting the fact that graph $G^0$ is a $\phi^*$-expander.
\end{proof}

At the beginning of the algorithm, we initialize the \EST data structure $\tau$, and set $U=\emptyset$. The total update time needed in order to maintain $\tau$ is bounded by $O(md\log m)\leq O\left ( \frac{m\Delta\log^2m}{\phi^*}\right )$.

We now bound the running time that is needed to initialize all data structures, and to maintain data structures $\dset(H)$, $\tau$, and $\set{L(e)}_{e\in E(G)}$ over the course of the algorithm.
The running time of the algorithm from \Cref{cor: HSS witness}, from the above discussion, is bounded by  $O\left (\frac{n^{1+O(\eps)}\cdot \Delta^3}{(\phi^*)^2}\right )$. The time that is needed in order to initialize and maintain the lists $L(e)$ for edges $e\in E(G)$ is subsumed by this running time. Additionally, from the above discussion, the total update time of data structure $\dset(H)$ is bounded by $O\left(n^{1+O(\eps)}\right )$, and the total update time of data structure $\tau$ is bounded by $O\left ( \frac{m\Delta\log^2m}{\phi^*}\right )$. The set $U$ of vertices can be maintained within this asymptotic running time. Overall, the total time needed to initialize all data structures, and to maintain data structures $\dset(H)$, $\tau$, and $\set{L(e)}_{e\in E(G)}$ over the course of the algorithm is bounded by:

\[O\left ( \frac{m\Delta\log^2m}{\phi^*}\right )+O\left (\frac{n^{1+O(\eps)}\cdot \Delta^3}{(\phi^*)^2}\right )\leq O\left(\frac{n^{1+O(\eps)}\cdot \Delta^3}{(\phi^*)^2}\right ).\]

\subsubsection{Maintaining the Data Structures}
Maintaining the data structures under the deletion of edges from $G$ is now straightforward. In \Cref{alg: delete-expander-edge} we describe algorithm $\delexpedge(e)$ that is invoked whenever an edge $e$ is deleted from graph $G$. We start by considering every edge $e'\in E(H)$ whose embedding path $P(e')$ contains edge $e$ -- in other words, all edges of $L(e)$. We delete each such edge $e'$ from graph $H$, and update data structure $\dset(H)$ with this deletion. As a result, it is possible that some vertices are removed from set $S'(H)$. For each such vertex $x$, we delete edge $(s,x)$ from graph $G'$, and update the data structure $\tau$ accordingly. Finally, we delete edge $e$ from data structure $\tau$.
 Whenever a vertex leaves the tree $\tau$, we add it to $U$.

\program{Algorithm $\delexpedge(e)$}{alg: delete-expander-edge}{
	
	\begin{enumerate}
		\item For every edge $e'\in L(e)$ do:
		
		\begin{enumerate}
			\item Delete $e'$ from $H$ and update data structure $\dset(H)$. Let $X$ be the set of vertices that were deleted from $S'(H)$ as the result of this update.
			
			\item For every edge $e''\in E(G)$ with $e'\in L(e'')$, delete $e'$ from $L(e'')$.
			
			\item for every vertex $x\in X$, delete edge $(s,x)$ from graph $G'$ and update the \EST $\tau$ with this deletion. Add every vertex that is removed from $\tau$ to $U$.
	\end{enumerate}
		
		\item Delete edge $e$ from graph $G'$ and update data structure $\tau$ accordingly. Add every vertex that is removed from $\tau$ to $U$.
	\end{enumerate}
}

It is easy to verify that the total update time of the algorithm is dominated by the time needed to initialize the data structures, and to maintain 
 data structures $\dset(H)$, $\tau$, and $\set{L(e)}_{e\in E(G)}$. From the above discussion, the total update time of the algorithm is bounded by  $O\left(\frac{n^{1+O(\eps)}\cdot \Delta^3}{(\phi^*)^2}\right )$.

\subsubsection{Responding to Short-Path Queries}

We assume that we are given a pair of vertices $x,y\in V(G)\setminus U$, and describe an algorithm for responding to short-path query between $x$ and $y$. Recall that our goal is to return an $x$-$y$  path $P$ in the current graph $G$, of length at most $\frac{2^{O(1/\eps^6)}\cdot \Delta\cdot\log n}{\phi^*}$, with query time $O(|E(P)|)$.

Using the \EST $\tau$, we compute a path $Q$ connecting $x$ to some vertex $x'\in S'(H)$, and a path $Q'$ connecting vertex $y$ to some vertex $y'\in S'(H)$, so that the length of each path is bounded by $d$. Next, we query data structure $\dset(H)$ with the pair $x',y'\in S'(H)$ of vertices. The data structure must return a path $\tilde Q$ connecting $x'$ to $y'$ in $H$, whose length is at most $2^{O(1/\eps^6)}$, in time $O(|E(\tilde Q)|)$. Using the embedding $\pset$ of $H$ into $G$, in which the length of every path is bounded by $d$, we can compute a path $Q''$ in graph $G$, connecting $x'$ to $y'$, whose length is bounded by $|E(\tilde Q)|\cdot d\leq 2^{O(1/\eps^6)}\cdot d$. Lastly, by concatenating the paths $Q,Q''$ and $Q'$, we obtain a path $P$ connecting $x$ to $y$ in graph $G$, whose length is at most:

\[  2^{O(1/\eps^6)}\cdot d\leq \frac{2^{O(1/\eps^6)}\cdot \Delta\cdot\log m}{\phi^*}\leq \frac{2^{O(1/\eps^6)}\cdot \Delta\cdot\log n}{\phi^*}.  \]

It is easy to see that the running time of the algorithm is $O(|E(P)|)$.

\section{Advanced Path Peeling -- Proof of \Cref{thm: main main advanced path peeling}}
\label{sec: advanced path peeling}

In this section we prove \Cref{thm: main main advanced path peeling}. 
The main tool in the proof is the following theorem.

\begin{theorem}\label{thm: outer advanced path peeling}
	There is a large enough constant $c^*$, and a deterministic algorithm, whose input consists of a connected $n$-vertex $m$-edge graph $G$,  a collection $\mset=\set{(s_1,t_1),\ldots,(s_k,t_k)}$ of pairs of vertices of $G$, such that $\mset$ is a matching, and  parameters $d,\eta>0$, $0<\alpha\leq 1/2$ and $\frac{2}{(\log n)^{1/24}}< \eps<\frac{1}{400}$, such that $1/\eps$ is an integer 
	and $256d< \eta\leq \frac{d^2}{2^{c^*/\eps^6}\cdot \log m}$ holds. The algorithm computes one of the following:

	\begin{itemize}
		\item either a cut $(A,B)$ with $|E_G(A,B)|\leq \frac{1024d}{\eta}\cdot \min\set{|E_G(A)|,|E_G(B)|}$, and each of $A$, $B$ contains at least $\frac{\alpha k}{8}$ vertices of set $T=\set{s_1,t_1,\ldots,s_k,t_k}$; or
		\item a routing $\pset$ in $G$ of a subset $\mset'\subseteq \mset$ containing at least $(1-\alpha)k$ pairs of vertices, such that every path in $\pset$ has length at most $d$, and the total congestion caused by the paths in $\pset$ is at most $\frac{4\eta}{\alpha}$.
	\end{itemize} 
	
	The running time of the algorithm is bounded by $O\left (m^{1+O(\eps)}(d^2+\eta d)\right )$.
\end{theorem}


The proof of \Cref{thm: main main advanced path peeling} easily follows from \Cref{thm: outer advanced path peeling}.  
Let $z=\ceil{\frac{1}{\eps}}$, and let $\eps'=\frac{1}{z}$. Clearly, $\frac{1}{\eps}\leq z \leq \frac{2}{\eps}$, and so $\frac{\eps}{2}\leq \eps'\leq \eps$.
Since we have assumed that $\frac{4}{(\log n)^{1/24}}< \eps<\frac{1}{400}$, we get that $\frac{2}{(\log n)^{1/24}}< \eps'<\frac{1}{400}$.
For convenience, in the remainder of the proof we will denote $\eps'$ by $\eps$.

Let $c^*$ be the constant from \Cref{thm: outer advanced path peeling}.
We set $d=\frac{2^{c^*/\eps^6+10}\cdot \log m}{\phi}$ and $\eta=\frac{1024d}{\phi}=\frac{2^{c^*/\eps^6+20}\cdot \log m}{\phi^2}$. It is immediate to verify that $\eta>256d$. Observe also that:

\[\frac{d^2}{2^{c^*/\eps^6}\cdot \log m}=\frac{1024d}{\phi}\geq \eta.  \] 

We start by considering the case where
 $\frac{1}{\alpha}<\log n$. In this case, we  apply the algorithm from \Cref{thm: outer advanced path peeling} to graph $G$, the set $\mset$ of pairs of its vertices, and parameters $d,\eta,
\alpha$ and $\eps$. Assume first that the outcome of the algorithm is a cut  $(A,B)$ with $|E_G(A,B)|\leq \frac{1024d}{\eta}\cdot \min\set{|E_G(A)|,|E_G(B)|}=\phi\cdot \min\set{|E_G(A)|,|E_G(B)|}$, and $|T\cap A|, |T\cap B|\geq \frac{\alpha k} 8$. In this case, we return the cut $(A,B)$ as the outcome of the algorithm. Otherwise, the outcome of the algorithm from \Cref{thm: outer advanced path peeling} is a routing $\pset$ in $G$ of a subset $\mset'\subseteq \mset$ containing at least $k\cdot (1-\alpha)$ pairs of vertices, such that every path in $\pset$ has length at most $d\leq \frac{2^{O(1/\eps^6)}\cdot \log n}{\phi}$, and the total congestion caused by the paths in $\pset$ is at most $\frac{4\eta}{\alpha}\leq \frac{2^{O(1/\eps^5)}\cdot \log n}{\alpha\cdot \phi^2}$. We then return this set of paths as the outcome of the algorithm.

The running time of the algorithm from  \Cref{thm: outer advanced path peeling} is bounded by:

\[O\left (m^{1+O(\eps)}(d^2+\eta d)\right )\leq O\left (m^{1+O(\eps)}\left(\frac{2^{O(1/\eps^6)}\cdot \log^2 n}{\phi^2}+ \frac{2^{O(1/\eps^6)}\cdot \log^2 n}{\phi^3} \right )\right )\leq O\left (\frac{m^{1+O(\eps)}}{\phi^3}\right ). \]

Next, we consider the case where $\frac{1}{\alpha}>\log n$. In this case, we perform at most $\log n$ iterations. At the beginning of iteration $i$, we are given a collection $\mset_i\subseteq \mset$ of pairs of vertices that have been already routed, together with their routing $\pset_i$ in graph $G$. At the beginning of the algorithm, $\mset_1=\emptyset$ and $\pset_1=\emptyset$. The iterations continue as long as $|\mset_i|<(1-\alpha)k$ holds. 

We now describe the execution of the $i$th iteration. Let $\mset'_i=\mset\setminus\mset_i$, and recall that $|\mset'_i|\geq \alpha k$ must hold. 
We  apply the algorithm from \Cref{thm: outer advanced path peeling} to graph $G$, the set $\mset'_i$ of pairs of its vertices, and parameters $d,\eta,
\alpha'=1/2$ and $\eps$. Assume first that the outcome of the algorithm is a cut  $(A,B)$ with $|E_G(A,B)|\leq \frac{1024d}{\eta}\cdot \min\set{|E_G(A)|,|E_G(B)|}=\phi\cdot \min\set{|E_G(A)|,|E_G(B)|}$, and $|T\cap A|, |T\cap B|\geq \frac{|\mset'_i|} {16}\geq \frac{\alpha k}{16}$. In this case, we terminate the algorithm and return the cut $(A,B)$ as the outcome of the algorithm. Otherwise, the outcome of the algorithm from \Cref{thm: outer advanced path peeling} is a routing $\tilde \pset_i$ in $G$ of a subset $\tilde \mset_i\subseteq \mset'_i$ containing at least $|\mset'_i|/2$ pairs of vertices, such that every path in $\tilde \pset_i$ has length at most $d\leq \frac{2^{O(1/\eps^6)}\cdot \log n}{\phi}$, and the total congestion caused by the paths in $\tilde \pset_i$ is at most $8\eta\leq \frac{2^{O(1/\eps^6)}\cdot \log n}{\phi^2}$. We then 
set $\mset_{i+1}=\mset_i\cup \tilde \mset_i$, $\pset_{i+1}=\pset_i\cup \tilde \pset_i$, and continue to the next iteration. 

Assume that the last iteration of the algorithm is iteration $i$. If the algorithm did not terminate with a cut, then $|\mset_{i+1}|\geq (1-\alpha)k$ must hold. We then return the set $\mset'=\mset_{i+1}$ of pairs of vertices, and their routing $\pset=\pset_{i+1}$. It is easy to verify that the cardinality of the set $\mset'_{i'}$ of pairs of vertices that remains to be routed decreases by at least factor $2$ in every iteration, and so the number of iterations in the algorithm is bounded by $\log k$. Since, for every iteration $i'$, the congestion caused by the set $\tilde \pset_{i'}$ of paths is at most $\frac{2^{O(1/\eps^6)}\cdot \log n}{\phi^2}$, the total congestion caused by the set $\pset$ of paths is at most $\frac{2^{O(1/\eps^6)}\cdot \log^2 n}{\phi^2}$.
The running time of a single iteration is bounded by  $O\left (\frac{m^{1+O(\eps)}}{\phi^3}\right )$ as before,  and, since the number of iterations is $O(\log n)$, the total running time of the algorithm remains bounded by $O\left (\frac{m^{1+O(\eps)}}{\phi^3}\right )$.

In the remainder of this section we prove \Cref{thm: outer advanced path peeling}. Following is a key lemma that we use in the proof.

\begin{lemma}\label{lem: inner advanced path peeling}
	There is a large enough constant $c^*$, and a deterministic algorithm, whose input consists of a connected $m$-edge graph $G$ with $|V(G)|\leq n$,  a collection $\mset=\set{(s_1,t_1),\ldots,(s_k,t_k)}$ of pairs of vertices of $G$, such that $\mset$ is a matching, and parameters $d,\eta>0$ and $\frac{2}{(\log n)^{1/24}}< \eps<\frac{1}{400}$, such that $1/\eps$ is an integer, 
	and $128d< \eta\leq \frac{d^2}{2^{c^*/\eps^6}\cdot \log m}$ holds. The algorithm computes one of the following:

	\begin{itemize}
		\item either a cut $(A,B)$ with $|E_G(A,B)|\leq \frac{64d}{\eta}\cdot \min\set{|E_G(A)|,|E_G(B)|}$, and each of $A$, $B$ contains at least $\frac{k^{1-\eps}}{16}$ vertices of set $T=\set{s_1,t_1,\ldots,s_k,t_k}$; or
		\item a routing $\pset$ in $G$ of a subset $\mset'\subseteq \mset$ containing at least $z=\frac{k^{1-22\eps}}{d}$ pairs of vertices, such that every path in $\pset$ has length at most $d$, and the total congestion caused by the paths in $\pset$ is at most $\eta$.
	\end{itemize} 

The running time of the algorithm is bounded by $O\left (m^{1+O(\eps)}(\eta+d\log n)\right )$.
\end{lemma}

We defer the proof of \Cref{lem: inner advanced path peeling} to Section \ref{subsec: proof of inner routing lemma}, after we complete the proof of \Cref{thm: outer advanced path peeling} using it. Our algorithm iteratively applies the algorithm from \Cref{lem: inner advanced path peeling}, while gradually constructing both a routing of some pairs from $\mset$, and a low-conductance cut in $G$.

Let $\eta'=\frac{4\eta}{\alpha}$.
Our algorithm consists of two stages. In the first stage, we either construct the desired routing $\pset$ of a large subset $\mset'\subseteq\mset$ of pairs of vertices, or compute a collection $\sset$ of disjoint subsets of vertices of $G$ with some useful properties. In the former case, we terminate the algorithm and return the resulting routing $\pset$, while in the latter case we continue to Stage 2, in which we exploit the collection $\sset$ of vertex subsets, in order to construct the desired low-conductance cut $(A,B)$. We now describe each of the two stages in turn.

\subsection*{Stage 1: Constructing a Routing}

 Let $q'=4d\cdot k^{22\eps}\cdot \log k$. Our algorithm in Stage 1 consists of at most $q'$ iterations. At the beginning of iteration $q$, we are given a subset $\mset'_q\subseteq \mset$ of pairs of vertices with $|\mset'_q|< (1-\alpha)\cdot k$, and a routing $\pset_q$ of the matching $\mset'_q$ in graph $G$. Additionally, we are given a collection $\sset_q$ of disjoint subsets of vertices of $G$, and we denote $A_q=\bigcup_{S\in \sset_q}S$. We will ensure that the following invariants hold:
 
 \begin{properties}{I}
 	\item every path in $\pset_q$ has length at most $d$; \label{inv: short path}
 	\item $|\pset_q|<(1-\alpha)k$; \label{inv: few paths}
 	\item the paths in $\pset_q$ cause congestion at most $\eta'$ in $G$; \label{inv: low cong}
 	\item if we denote by $E'_q$ the set of all edges that lie on at least $\eta'/2$ paths in $\pset_q$, then $\left |\left(\bigcup_{S\in \sset_q}\delta_G(S)\right )\setminus E'_q\right |\leq \frac{64d}{\eta}\sum_{S\in \sset_q}\left |E_G(S)\right|$; \label{inv: sparse cuts} and 
 	\item $|A_q\cap T|< \frac{\alpha k}{2}$.\label{inv: cuts in terminals}
 	\end{properties}


At the beginning of the algorithm, we set $\pset_1=\emptyset$ and $\sset_1=\emptyset$, so $A_1=\emptyset$ holds. Clearly, all invariants hold for this setting.
We now describe the execution of the $q$th iteration, for some $q\geq 1$.
We assume that we are given a set $\mset'_q\subseteq\mset$ of pairs of vertices, a set $\pset_q$ of paths routing the pairs in $\mset'_q$ in $G$, and a collection $\sset_q$ of disjoint subsets of vertices of $G$, for which invariants \ref{inv: short path}--\ref{inv: cuts in terminals} hold.

We let $\mset_q$ be a set of pairs of vertices of $G$, containing all pairs $(s_i,t_i)\in \mset\setminus\mset'_q$ with $s_i,t_i\in V(G)\setminus A_q$. We also let $G_q$ be the graph obtained from $G$, after we delete from it all vertices of $A_q$, and all edges $e\in E(G)$, such that $e$ that belongs to at least $\eta'/2$ paths of $\pset_q$. In other words, $G_q=(G\setminus A_q)\setminus E'_q$. 
Denote $k_q=|\mset_q|$. Since, from Invariant \ref{inv: cuts in terminals}, $|A_q\cap T|< \frac{\alpha k}{2}$, while from Invariant \ref{inv: few paths}, $|\mset\setminus\mset'_q|\geq k-|\pset_q|\geq \alpha k$, we get that $k_q>\frac{\alpha k}{2}$ must hold.
Notice that graph $G_q$ may not be connected. We assume first that there is some connected component $C_q$ of graph $G_q$, and a subset $\tilde \mset_q\subseteq \mset_q$ containing at least $k_q/2$ pairs, such that all vertices participating in the pairs in $\tilde \mset_q$ lie in $C_q$.

We apply the algorithm from  \Cref{lem: inner advanced path peeling} to graph $C_q$, the set $\tilde \mset_q$ of pairs of vertices, and parameters $d,\eta,\eps$ that remain unchanged.  We now consider two cases. The first case happens if the algorithm from \Cref{lem: inner advanced path peeling} returns a cut $(X_q,Y_q)$ in graph $C_q$, with 
 $|E_{G_q}(X_q,Y_q)|\leq \frac{64d}{\eta}\cdot \min\set{|E_{G_q}(X_q)|,|E_{G_q}(Y_q)|}\leq \frac{64d}{\eta}\cdot \min\set{|E_{G}(X_q)|,|E_{G}(Y_q)|}$. Recall that we are also guaranteed that each of $X_q$, $Y_q$ contains at least $\frac{k_q^{1-\eps}}{32}$ vertices of set $T'=\set{s_i,t_i\mid (s_i,t_i)\in \tilde \mset_q}$. We say that iteration $q$ is a \emph{type-1 iteration}. We assume w.l.o.g. that $|X_q\cap T'|\leq |Y_q\cap T'|$. We set $\sset_{q+1}=\sset_q\cup \set{X_q}$, and we let $\mset'_{q+1}=\mset'_q$ and $\pset_{q+1}=\pset_q$. It is immediate to verify that Invariants \ref{inv: short path}--\ref{inv: low cong} continue to hold for $\pset_{q+1}$. 
 It is also immediate to verify that $E'_{q+1}=E'_q$. Therefore, $\left |\left(\bigcup_{S\in \sset_q}\delta_G(S)\right )\setminus E'_{q+1}\right |\leq \frac{64d}{\eta}\sum_{S\in \sset_q} \left |E_G(S)\right |$.
Let $E^*_q=\left ( \left(\bigcup_{S\in \sset_{q+1}}\delta_G(S)\right )\setminus \left(\bigcup_{S\in \sset_{q}}\delta_G(S)\right )\right )\setminus  E'_{q}$. Then $E^*_q=E_{G_q}(X_q,Y_q)$, and we are guaranteed that $|E^*_q|\leq \frac{64d}{\eta}\cdot|E_G(X_q)|$. Therefore, we get that:

\[
\begin{split}
\left |\left(\bigcup_{S\in \sset_{q+1}}\delta_G(S)\right )\setminus E'_{q+1}\right |&\leq \left |\left(\bigcup_{S\in \sset_q}\delta_G(S)\right )\setminus E'_{q+1}\right |+|E^*_q|\\
&\leq \frac{64d}{\eta} \sum_{S\in \sset_q}|E_G(S)  |+\frac{64d}{\eta}\cdot|E_G(X_q)|\\
&\leq  \frac{64d}{\eta}\sum_{S\in \sset_{q+1}} |E_G(S) |. \end{split}\]
 
 Therefore, Invariant \ref{inv: sparse cuts} continues to hold for $\sset_{q+1}$. If Invariant \ref{inv: cuts in terminals} continues to hold as well, then we continue to the next iteration. Otherwise, we terminate the first stage, and continue to the second stage.
 
 Consider now the second case, when the algorithm from \Cref{lem: inner advanced path peeling} returns 
 a routing $\rset_q$ in $C_q$ of a subset $\mset^*_q\subseteq \tilde \mset_q$ containing at least $\frac{k_q^{1-22\eps}}{2d}$ pairs of vertices, such that every path in $\rset_q$ has length at most $d$, and the total congestion caused by the paths in $\rset_q$ is at most $\eta$. In this case, we say that iteration $q$ is a \emph{type-2 iteration}. We set $\pset_{q+1}=\pset_q\cup \rset_q$ and $\mset'_{q+1}=\mset'_q\cup \mset^*_q$.
 Clearly, $\pset_{q+1}$ is a routing of the pairs in $\mset'_{q+1}$ in graph $G$, and every path in $\pset_{q+1}$ has length at most $d$. Furthermore, since the edges of $G$ that participate in at least $\eta'/2$ paths in $\pset_q$ do not lie in graph $G_q$, and since $\eta<\eta'/2$, we get that the total congestion that the paths of $\pset_{q+1}$ cause in graph $G$ is at most $\eta'$.
 We also set $\sset_{q+1}=\sset_q$. Since $E'_{q}\subseteq E'_{q+1}$, it is easy to verify that Invariant \ref{inv: sparse cuts} continues to hold for $\sset_{q+1}$, and it is immediate to verify that Invariant \ref{inv: cuts in terminals} holds as well. If $|\pset_{q+1}|\geq (1-\alpha)k$, then we terminate the algorithm and return the set $\mset'=\mset_{q+1}$ of pairs of vertices and the set $\pset=\pset_{q+1}$ of paths. From the above discussion, the paths in $\pset$ are a routing of the pairs in $\mset'$; every path in $\pset$ has length at most $d$, and the paths in $\pset$ cause congestion at most $\eta'$. Otherwise, 
 if $|\pset_{q+1}|< (1-\alpha)k$, then
 from the above discussion, all invariants hold for $\pset_{q+1}$ and $\sset_{q+1}$, and we continue to the next iteration. 
 
It remains to consider the case where for every connected component $C$ of graph $G_q$, the number of demand pairs $(s_i,t_i)\in \mset_q$ with $s_i,t_i\in V(C)$ is less than $k_q/2$. Let $\cset$ denote the set of all connected components of $G_q$, and let $T'=\set{s_i,t_i\mid (s_i,t_i)\in \mset_q}$, so $|T'|=2k_q$. For each component $C\in \cset$, let $n_C=|V(C)\cap T'|$. Then for all $C\in \cset$, $n_C\leq 1.5k_q$ must hold. 
Let $(A',B')$ be a partition of $V(G_q)$, that is computed as follows. We denote $\cset=\set{C_1,C_2,\ldots,C_r}$, where the components are indexed so that $n_{C_1}\geq n_{C_2}\geq\cdots\geq n_{C_r}$. We start with $A'=B'=\emptyset$, and consider the components of $\cset$ in the order of their indices. When component $C_i$ is processed, if $|A'\cap T'|\leq |B'\cap T'|$, then we add the vertices of $C_i$ to $A'$, and otherwise we add the vertices of $C_i$ to $B'$. Consider the partition $(A',B')$ of the vertices of $V(G_q)$ that we obtain at the end of the algorithm, and assume w.l.o.g. that $|A'\cap T'|\geq |B'\cap T'|$. It is easy to verify that $|A'\cap T'|-|B'\cap T'|\leq \max_i\set{n_{C_i}}\leq n_{C_1}\leq 1.5k_q$. Since $|T'|=2k_q$, we then get that $|A'\cap T'|,|B'\cap T'|\geq \frac{k_q} 4\geq \frac{\alpha k}8$. We obtain a cut $(A,B)$ in graph $G$ by letting $A=A'\cup A_q$ and $B=B'$. Clearly, $|A\cap T|,|B\cap T|\geq \frac{\alpha k}{8}$. Next, we show that $|E_G(A,B)|\leq \frac{1024d}{\eta}\cdot \min\set{|E_G(A)|,|E_G(B)|}$. Indeed, it is immediate to verify that $E_G(A,B)\subseteq E'_q$. Since the paths in $\pset_q$ have length at most $d$ each, and since $\eta'=\frac{4\eta}{\alpha}$, we get that $|E'_q|\leq \frac{2|\pset_q|\cdot d}{\eta'}\leq \frac{\alpha kd}{2\eta}$. On the other hand, since graph $G$ is connected, and since $|A|\geq \frac{\alpha k}{8}$, we get that $|E_G(A)|+|E_G(A,B)|\geq \frac{\alpha k}{16}$. Similarly,  $|E_G(B)|+|E_G(A,B)|\geq \frac{\alpha k}{16}$. 
Altogether, we get that: $|E'_q|\leq \frac{\alpha kd}{2\eta}\leq \frac{8d}{\eta}\cdot\min\set{|E_G(A)|,|E_G(B)|} +\frac{8d}{\eta}|E_G(A,B)|$.
Since $\eta\geq 128d$, we get that:

\[|E_G(A,B)|\leq |E'_q|\leq \frac{8d}{\eta}\cdot \min\set{|E_G(A)|,|E_G(B)|}+\frac{|E_G(A,B)|}{2}, \]

and so:

\[|E_G(A,B)|\leq \frac{1024d}{\eta}\cdot \min\set{|E_G(A)|,|E_G(B)|}. \]
 In this case, we terminate the algorithm and return the cut $(A,B)$. We say that the current iteration is a type-1 iteration.

 This completes the description of the first stage of the algorithm. 
 We now show that the number of iteration in this stage is bounded by $q'$, and bound the running time of the algorithm.

 \begin{observation}\label{obs: bound num of phases}
 	The number of iterations in the algorithm is bounded by  $q'=4d\cdot k^{22\eps}\cdot \log k$.
 \end{observation}

\begin{proof}
We partition the execution of the algorithm into phases. For all $i\geq 1$, the $i$th phase includes all iterations $q$, for which $\frac{k}{2^i}<|\mset_q|\leq \frac{k}{2^{i-1}}$.  Clearly, the number of phases is bounded by $\log k$. Next, we bound the number of iterations in a single phase.
	
	Consider some integer $i$, and denote by $n_i=\frac{k}{2^i}$. If iteration $q$ is a type-2 iteration that belongs to phase $i$, then $k_q=|\mset_q|\geq n_i$, and, from \Cref{lem: inner advanced path peeling}, at least $\frac{n_i^{1-22\eps}}{2d}$ pairs of vertices are routed in iteration $q$. Therefore, after $2d\cdot n_i^{22\eps}$ type-2 iterations, the number of pairs that remain to be routed must decrease by at least factor $2$. We conclude that the $i$th phase may contain at most  $2d\cdot n_i^{22\eps}\leq 2d\cdot k^{22\eps}$ type-2 iterations.
	
	If iteration $q$ is a type-2 iteration in phase $i$, then we are guaranteed that set $X_q$ contains at least $\frac{n_i^{1-\eps}}{32}$ terminals that participate in pairs in $\mset_q$. Therefore, after $32n_i^{\eps} \leq 32k^{\eps}$ type-1 iterations, the number of terminals in $V(G)\setminus A_q$ that remain to be routed (that is, the terminals of $\mset_q$), must decrease by at least factor $4$. We conclude that a single phase may contain at most $32k^{\eps}$ type-1 iterations.
	
	Overall, a single phase may contain at most $2d\cdot k^{22\eps}+32k^{\eps}\leq 
	3d\cdot k^{22\eps}$ iterations, and the total number of iterations in the algorithm is bounded by $3dk^{22\eps}\log k\leq q'$.
\end{proof}

 Since the running time of the algorithm from \Cref{lem: inner advanced path peeling} is bounded by 
$O\left (m^{1+O(\eps)}(\eta+d\log n)\right )$, the total running time of the first stage of the algorithm algorithm is bounded by  $O\left (m^{1+O(\eps)}(\eta d+d^2\log n)\right )\leq O\left (m^{1+O(\eps)}(\eta d+d^2)\right ) $,
since $\eps>\frac{2}{(\log n)^{1/24}}$, so $n^{\eps}>n^{2/(\log n)^{1/24}}>2^{(\log n)^{3/24}}>\log n$.

\subsection*{Stage 2: Computing the Cut}

Assume that the last iteration of the algorithm was iteration $q$. 
	Let $E'$ be the set of all edges $e\in E(G)$, such that $e$ belongs to at least $\eta'/2$ paths of $\pset_q$. Since the paths in $\pset_q$ have length at most $d$ each, and since $\eta'=\frac{4\eta}{\alpha}$, we get that $|E'|\leq \frac{2|\pset_q|\cdot d}{\eta'}\leq \frac{\alpha kd}{2\eta}$.
Let $G'=G\setminus E'$. In the second stage, we will compute a cut $(A,B)$ in graph $G'$, with $|E_{G'}(A,B)|\leq \frac{512d}{\eta}\cdot \min\set{|E_G(A)|,|E_G(B)|}$, so that each of $A,B$ contains at least $\frac{\alpha k}{8}$ vertices of $T$. 
Assume w.l.o.g. that $|E_G(A)|\leq |E_G(B)|$.
Since graph $G$ is connected, and since $|A|\geq \frac{\alpha k}{8}$, we get that $|E_G(A)|+|E_G(A,B)|\geq \frac{\alpha k}{16}$. Therefore, $|E'|\leq \frac{\alpha kd}{2\eta}\leq \frac{8d}{\eta}\left(|E_G(A)|+|E_G(A,B)|\right )$. We then get that:

\[\begin{split}
|E_G(A,B)|&\leq |E_{G'}(A,B)|+|E'|\\
&\leq  \frac{512d}{\eta}|E_G(A)|+ \frac{8d}{\eta}\left(|E_G(A)|+|E_G(A,B)|\right )\\& \leq \frac{520d}{\eta}|E_G(A)|+\frac{8d}{\eta}|E_G(A,B)|.
\end{split}\]

Since $\eta\geq 256d$, we get that $|E_G(A,B)|\leq \frac{1024d}{\eta}|E_G(A)|=\frac{1024d}{\eta}\cdot\min\set{|E_G(A)|,|E_G(B)|}$.
In the remainder of the algorithm, it is enough to compute a cut $(A,B)$ in graph $G'$, with $|E_{G'}(A,B)|\leq \frac{512d}{\eta}\cdot \min\set{|E_G(A)|,|E_G(B)|}$, so that each of $A,B$ contains at least $\frac{\alpha k}{8}$ vertices of $T$. 

Recall that the last iteration of the algorithm was iteration $q$, and in iteration $q$ we have computed a cut $(X_q,Y_q)$ of the connected component $C_q$ of the corresponding graph $G_q$. We denote $Y'_q=V(G_q)\setminus X_q$, so $Y_q\subseteq Y'_q$. We have also defined a set $\sset_{q+1}$ of disjoint subsets of vertices, with $X_q\in \sset_{q+1}$. Let $\sset'=\sset_{q+1}\cup \set{Y'_q}$. From the description of the algorithm, it is immediate to verify that the subsets of vertices in $\sset'$ are all disjoint, and they partition $V(G)$. Since the algorithm terminated at iteration $q$, we are guaranteed that $|A_{q+1}\cap T|\geq \frac{\alpha k}{2}$. 
Additionally, if $T'$ denotes the set of all terminals participating in the demand pairs in $\tilde \mset_q$, then at least $k_q$ terminals from $T'$ lie in $C_q$. Since we have assumed that $|X_q\cap T'|\leq |Y_q\cap T'|$, we get that $|Y'_q\cap T'|\geq \frac{k_q}2\geq \frac{\alpha k}{4}$. 


Let $E''=\bigcup_{S\in \sset_{q+1}}\delta_{G'}(S)=\left (\bigcup_{S\in \sset_{q+1}}\delta_{G}(S)\right )\setminus E'=\left (\bigcup_{S\in \sset_{q+1}}\delta_{G}(S)\right )\setminus E'_{q+1}$. Recall that we have established that Invariant \ref{inv: sparse cuts} holds for $\sset_{q+1}$, so $|E''|\leq \frac{64d}{\eta}\sum_{S\in \sset_{q+1}}|E_G(S)|$.

We now consider two cases. The first case happens if $|E_G(Y'_q)|\geq \frac{\eta}{256d}|E''|$. In this case, we consider the cut $(A,B)$ in graph $G'$, where $B=Y'_q$ and $A=V(G)\setminus Y'_q$. From the above discussion $|B\cap T|\geq \frac{\alpha k}{4}$, and $|A\cap T|=|A_{q+1}\cap T|\geq \frac{\alpha k}2$. Moreover, $|E_{G'}(A,B)|\leq |E''|\leq \frac{64d}{\eta}\sum_{S\in \sset_{q+1}}|E_G(S)|\leq \frac{64d}{\eta}|E_G(A)|$. Altogether, we get that $|E_{G'}(A,B)|\leq |E''|\leq \frac{256d}{\eta}\cdot\min\set{|E_G(A)|,|E_G(B)|}$, as required.

From now on we consider the second case, where $|E_G(Y'_q)|< \frac{\eta}{256d}|E''|$. We compute a partition $(A',B')$ of $V(G)\setminus Y'_q$ as follows. Assume w.l.o.g. that $\sset_{q+1}=\set{S_1,S_2,\ldots,S_r}$, where the sets are indexed so that $|E_{G}(S_1)|\geq |E_{G}(S_2)|\geq\cdots\geq |E_{G}(S_r)|$. We start with $A'=B'=\emptyset$, and then consider the sets $S_1,\ldots,S_r$ in this order. When set $S_i$ is considered, if $|E_{G}(A')|\leq |E_{G}(B')|$, then we add the vertices of $S_i$ to $A'$, and otherwise we add the vertices of $S_i$ to $B'$. Consider the  partition $(A',B')$ of $V(G)\setminus Y'_q$ that we obtain at the end of this algorithm. If $|A'\cap T|<|B'\cap T|$, then we set $A=A'\cup Y'_q$ and $B=B'$. Otherwise, we set $A=A'$ and $B=B'\cup Y'_q$. We now show that cut $(A,B)$ has all required properties. Assume w.l.o.g. that $|A'\cap T|<|B'\cap T|$ (the other case is symmetric). Then $|B\cap T|=|B'\cap T|\geq \frac{|A_{q+1}\cap T|}2\geq \frac{\alpha k}4$. Also, $|A\cap T|\geq |Y'_q\cap T|\geq \frac{\alpha k}2$. We next show that  $|E''|\leq \frac{512d}{\eta}\cdot\min\set{|E_G(A)|,|E_G(B)|}$ in the following claim.

\begin{claim}\label{claim: sparse cut after greedy}
\[|E''|\leq \frac{512d}{\eta}\cdot\min\set{|E_G(A)|,|E_G(B)|}.\]
\end{claim}

\begin{proof}
%
	Recall that we have denoted $\sset_{q+1}=\set{S_1,\ldots,S_r}$, where the sets of vertices $S_i$ are indexed in the non-increasing order of the cardinalities of the corresponding sets of edges $E_G(S_i)$. We use the following observation.
	
	\begin{observation}\label{obs: lots of edges in all but one}
		$\sum_{i=2}^r|E_G(S_i)|\geq \frac{\eta}{512d}|E''|$.
	\end{observation}

	We prove \Cref{obs: lots of edges in all but one} below, after we complete the proof of \Cref{claim: sparse cut after greedy} using it.
	Since $A'\subseteq A$ and $B'\subseteq B$, it is enough to prove that $|E''|\leq \frac{512d}{\eta}\cdot\min\set{|E_G(A')|,|E_G(B')|}$.
	Denote $M=\frac{\eta}{512d}|E''|$.
	 We consider two cases. The first case happens if $|E_G(S_1)|\geq M$. Our algorithm then adds the vertices of $S_1$ to $A'$, and it will keep adding vertices from sets $S_i$ to $B'$ until $|E_{G}(B')|\geq |E_{G}(A')|\geq M$ holds. Therefore, we are guaranteed that, at the end of the algorithm, $|E_G(A')|,|E_G(B')|\geq M=\frac{\eta}{512d}|E''|$, and so $|E''|\leq \frac{512d}{\eta}\cdot\min\set{|E_G(A')|,|E_G(B')|}$.
	 
	 Consider now the second case, where $|E_G(S_1)|<M$. Assume for contradiction that $|E''|> \frac{512d}{\eta}\cdot\min\set{|E_G(A')|,|E_G(B')|}$, and assume w.l.o.g. that $|E_G(A')|\leq |E_G(B')|$, so $|E_G(A')|<\frac{\eta}{512d}|E''|$. Recall that, from  
	Invariant \ref{inv: sparse cuts}, $\sum_{S\in \sset_{q+1}}|E_G(S)|\geq \frac{\eta}{64d}|E''|$. Since $|E_G(A')|<\frac{\eta}{512d}|E''|$, it must be the case that $\sum_{\stackrel{S\in \sset_{q+1}:}{S\subseteq B'}}|E_G(S)|> \frac{\eta}{128d}|E''|$. Let $S_i\in \sset_{q+1}$ be the set whose vertices were added to $B'$ last. 
	Then at the time when the vertices of $S_i$ were added to $B'$, $|E_G(A')|<\frac{\eta}{512d}|E''|$ held. Therefore, from our algorithm, at the same time $|E_G(B')|<\frac{\eta}{512d}|E''|$ held.
	Therefore, at the end of the algorithm, $\sum_{\stackrel{S\in \sset_{q+1}:}{S\subseteq B'}}|E_G(S)|\leq \frac{\eta}{512d}|E''|+|E_G(S_i)|$. But since $|E_G(S_i)|\leq |E_G(S_1)|< M=\frac{\eta}{512d}|E''|$, we get that 
	$\sum_{\stackrel{S\in \sset_{q+1}:}{S\subseteq B'}}|E_G(S)|< \frac{\eta}{256d}|E''|$ holds at the end of the algorithm, a contradiction.

In order to complete the proof of \Cref{claim: sparse cut after greedy}, it is now enough to prove \Cref{obs: lots of edges in all but one}.

		\begin{proofof}{\Cref{obs: lots of edges in all but one}}
		Let $1\leq i_1<i_2<\cdots<i_r=q$ be the indices of type-2 iterations.  Recall that for each $1\leq j\leq r$, in iteration $i_j$ we computed a cut $(X_{i_j},Y_{i_j})$ of the connected component $C_{i_j}$ of graph $G_{i_j}$, with  $|E_{G_{i_j}}(X_{i_j},Y_{i_j})|\leq \frac{64d}{\eta}\cdot \min\set{|E_{G}(X_{i_j})|,|E_{G}(Y_{i_j})|}$.
		
		We can then denote $\sset_{q+1}=\set{X_{i_1},X_{i_2},\ldots,X_{i_r}}$. For all $1\leq j\leq r$, we define a subset $E_j\subseteq E''$ of edges that the set $X_{i_j}$ of vertices is responsible for. We let $E_1=\delta_{G'}(X_{i_1})$, and for $1<j\leq r$, we let $E_j=\delta_{G'}(X_{i_j})\setminus (E_1\cup \cdots\cup E_{j-1})$. It is easy to verify that $(E_1,\ldots,E_r)$ is a partition of $E''$.
		
		Consider now some index $1\leq j\leq r$, and the cut $(X_{i_j},Y_{i_j})$ that our algorithm computed in iteration $i_j$. It is easy to verify that $Y_{i_j}\subseteq X_{i_j+1}\cup\cdots\cup X_{i_r}\cup Y'_{q}$.
		Moreover, $E_j=E_{G_{i_j}}(X_{i_j},Y_{i_j})\setminus E'$, and so, from the above discussion:
		
		\[|E_j|\leq  \frac{64d}{\eta}\cdot \min\set{|E_{G}(X_{i_j})|,|E_{G}(Y_{i_j})|}.\]
		
		Assume w.l.o.g. that $S_1=X_{i_{j^*}}$. We partition the edges of $E''$ into three subsets: set $\hat E_1=\bigcup_{1\leq j<j^*}E_{i_j}$; set $\hat E_2=E_{i_{j^*}}$; and set $\hat E_3=\bigcup_{j^*< j\leq r}E_{i_j}$.
		From the discussion so far, we get that:
		
		\[ |\hat E_1|\leq \frac{64d}{\eta}\cdot\sum_{1\leq j<j^*}|E_G(X_{i_j})|, \]
		
		\[ |\hat E_3|\leq \frac{64d}{\eta}\cdot\sum_{j^*< j\leq q}|E_G(X_{i_j})|, \]
		
		and
		
		\[ |\hat E_2|\leq \frac{64d}{\eta}\cdot |E_{G}(Y_{i_{j^*}})|.\]
		
		Recall that $Y_{i_{j^*}}\subseteq X_{i_{j^*}+1}\cup\cdots\cup X_{i_r}\cup Y'_{q}$. Therefore, for every edge $e\in E_{G}(Y_{i_{j^*}})$, either both endpoints of $e$ lie in one of the sets $X_{i_j^*+1},\cdots, X_{i^*_r}, Y'_{q}$ (in which case $e\in E_G(X_{i_{j^*}+1})\cup \cdots\cup E_G(X_{i_r})\cup E_G(Y'_{q})$), or the endpoints of $e$ lie in different sets of $\set{X_{i_{j^*}+1},\cdots, X_{i_r}, Y'_{q} }$. In the latter case, $e\in E_{i_{j^*}+1}\cup\cdots E_r\cup E'=\hat E_3\cup E'$ must hold.
		Altogether, we get that:

		\[
		\begin{split}
		|E_{G}(Y_{i_{j^*}})|&\leq \sum_{j^*<j\leq r}|E_G(X_{i_j})|+|E_G(Y'_q)|+|\hat E_3|+|E'|\\
		&\leq \left(1+\frac{64d}{\eta}\right )\cdot \sum_{j^*<j\leq r}|E_G(X_{i_j})|+|E_G(Y'_q)|+|E'|\\& \leq \frac{5}{4}\cdot \sum_{j^*<j\leq r}|E_G(X_{i_j})|+|E_G(Y'_q)|+|E'|.
		\end{split}
		\]
		
			(since $\eta>256d$).
			
		Recall that $|E'|\leq \frac{\alpha kd}{2\eta}\leq \frac{\alpha k}{512}$, while $|Y_q|\geq |T\cap Y_q|\geq \frac{\alpha k}{4}$. Since $Y_q$ is a set of vertices of a connected component of $G_q$, $|E_{G_q}(Y_q)|+|E_{G_q}(X_q,Y_q)|\ge \frac{|Y_q|}{2}\geq \frac{\alpha k}{8}$. Moreover, we are guaranteed that $|E_{G_q}(X_q,Y_q)|\leq \frac{64d}{\eta}\cdot |E_{G_q}(Y_q)|\leq \frac{|E_{G_q}(Y_q)|}{4}$. Therefore, $|E_G(Y'_q)|\geq |E_{G}(Y_q)|\geq |E_{G_q}(Y_q)|\geq \frac{\alpha k}{10}$, and so $|E'| \leq \frac{\alpha k}{512}\leq \frac{|E_G(Y'_q)|}{50}$. Overall, we get that:
		
		\[ |E_{G}(Y_{i_{j^*}})|\leq \frac{5}{4}\cdot \sum_{j^*<j\leq r}|E_G(X_{i_j})|+\frac {51 |E_G(Y'_q)|}{50}.\]

		Altogether:
		
			\[ |\hat E_2|\leq \frac{80d}{\eta} \sum_{j^*<j\leq r}|E_G(X_{i_j})|+\frac{80d}{\eta}|E_G(Y'_q)|,\]
		
		and:

		\[ |E''|=|\hat E_1|+|\hat E_2|+|\hat E_3|\leq \frac{144d}{\eta}\sum_{a=2}^r|E_G(S_a)|+ \frac{80d}{\eta}|E_G(Y'_{q})|.\]	
		
		Recall that we have assumed that $|E_G(Y'_q)|< \frac{\eta}{256d}|E''|$. Therefore, we get that $|E''|\leq \frac{144d}{\eta}\sum_{a=2}^r|E_G(S_a)|+\frac{|E''|}{3}$, and so $\sum_{a=2}^r|E_G(S_a)|\geq \frac{\eta}{512d}|E''|$.
	\end{proofof}
\end{proof}

Since $E_{G'}(A,B)\subseteq E''$, we get that, in the second case, $|E_{G'}(A,B)|\leq |E''|\leq \frac{512d}{\eta}\cdot\min\set{|E_G(A)|,|E_G(B)|}$ holds. To conclude, we have computed a cut $(A,B)$ in graph $G$, with $|A\cap T|,|B\cap T|\geq \frac{\alpha k}{8}$ and $|E_G(A,B)|\leq \frac{1024d}{\eta}\cdot \min\set{|E_G(A)|,|E_G(B)|}$.

Recall that the running time of the first stage of the algorithm is at most
$O\left (m^{1+O(\eps)}(d^2+\eta d)\right )$, while the running time of the second stage can be bounded by $O(m)$. Therefore, the total running time of the algorithm is $O\left (m^{1+O(\eps)}(d^2+\eta d)\right )$. In order to prove 
\Cref{thm: outer advanced path peeling} it now remains to prove \Cref{lem: inner advanced path peeling}, which we do next.


\subsection{Proof of \Cref{lem: inner advanced path peeling}}
\label{subsec: proof of inner routing lemma}
	We denote $T=\set{s_1,t_1,\ldots,s_k,t_k}$, and for convenience we denote by $\tilde k=|T|=2k$. 
	
The algorithm consists of two stages. In the first stage, we apply the algorithm from \Cref{cor: HSS witness} to graph $G$ and the set $T$ of terminals. If the algorithm returns two subsets $T_1,T_2\subseteq T$ of terminals and a set $E'$ of edges (essentially defining a distancing $(T_1,T_2,E')$), then we will use the algorithm from \Cref{lem: distancing to sparse cut} in order to convert this distancing into a low-conductance cut $(A,B)$ as required. Otherwise, the algorithm must embed a large graph $H$ into $G$, and construct a level-$(1/\eps)$ hierarchical  support structure for $H$, so that $H$ is $(\eta',\td)$-well-linked with respect to the set $S(H)$ of vertices defined by the hierarchical support structure. In this latter case, we continue to Stage 2. In Stage 2, we will initialize an \EST data structure, rooted at the set $S(H)$ of vertices in graph $G$. Additionally, we will maintain a data structure from \Cref{thm: APSP in HSS} in order to support approximate short-path queries in graph $H$, as it undergoes edge deletions. We will use these data structure to iteratively identify pairs $(s_i,t_i)\in \mset$ of vertices, that can be connected via a short path $P_i$ in $G$. The corresponding path $P_i$ is then added to the routing $\pset$ that we are constructing, and every edge of $G$ that currently appears on $\eta$ paths in $\pset$ is deleted from $G$. Edges of $H$ whose embedding paths have thus been eliminated will be deleted from $H$. The algorithm terminates when we can no longer route the remaining paths of $\pset$ via short paths. If, by that time $|\pset|\geq z$, then we return the routing in $\pset$. Otherwise, we will obtain a distancing in graph $G$, that can again be converted into a low-conductance cut.
Before we describe each of the two stages, we need to consider an easier special case where $k\leq n^{\eps}$.

\subsubsection{Special Case: $k\leq n^{\eps}$}

For all $1\leq i\leq k$, let $A_i=\set{s_i}$ and $B_i=\set{t_i}$.
We apply Procedure \procpathpeel from \Cref{lem: path peel} to graph $G$ and the collections  $A_1,B_1,\ldots,A_k,B_k$ of subsets of its vertices, and parameters $d$ and $\eta$. Let $\pset_1,\ldots,\pset_k$ be the collection of paths that the algorithm outputs. Then for all $1\leq i\leq k$, either $\pset_i=\emptyset$, or $\pset_i$ contains a single path $P_i$, connecting $s_i$ to $t_i$ in $G$. Let $\pset=\bigcup_{i=1}^k\pset_i$, and let $\mset'\subseteq \mset$ contain all pairs $(s_i,t_i)\in \mset$, such that some path in $\pset$ connects $s_i$ to $t_i$. Clearly, $\pset$ is a routing of the pairs in $\mset'$, and we are guaranteed that every path in $\pset$ has length at most $d$, and the paths in $\pset$ cause congestion at most $\eta$. 
Let $E'$ be the set of all edges of $G$ that participate in exactly $\eta$ paths in $\pset$, and let $\mset''=\mset\setminus\mset'$. From Property \ref{pp: distancing}, for every pair $(s_i,t_i)\in \mset''$ of vertices, $\dist_{G\setminus E'}(s_i,t_i)>d$. Notice also that $|E'|\leq \frac{\sum_{P\in \pset}|E(P)|}{\eta}\leq \frac{d\cdot |\pset|}{\eta}$.
Recall that the running time of the algorithm is $O(m\eta+mdk\log n)\leq O(mn^{O(\eps)}d+m\eta)\leq O\left (m^{1+O(\eps)}(\eta+d\log n)\right )$. 

We now consider two cases. The first case happens if $|\pset|\geq z$. In this case, we return the routing $\pset$ of the set $\mset'$ of pairs of vertices.

Consider now the second case, where $|\pset|<z$. In this case, $|\mset''|\geq k-z\geq k/2$. Let $T'\subseteq T$ be the set of all vertices that participate in the pairs in $\mset''$, and let $\tilde G=G\setminus E'$. Since, for every pair $(s_i,t_i)\in \mset''$ of vertices, $\dist_{\tilde G}(s_i,t_i)>d$, we get that for every vertex $x\in T'$, $|B_{\tilde G}(x,d/2)\cap T'|<|T'|/2$. 
Let $\Delta=\frac{64}{\eps}$ and $\hat d=\frac{d}{2\Delta}=\frac{d\eps}{128}$. We apply Procedure \procsep from \Cref{lem: find separation of terminals} to graph $\tilde G$, the set $T'$ of terminals, 
parameters $\Delta$, $\alpha=2/3$, and
distance parameter $\hat d$ replacing $d$. Note that the algorithm may not return a terminal  $t\in T'$ with $|B_{\tilde G}(t,\Delta\cdot \hat d)\cap T'|=|B_{\tilde G}(t, d/2)\cap T'|>\alpha |T'|$, since, as observed above, for each such terminal $t$,  $|B_{\tilde G}(t,d/2)\cap T'|<|T'|/2$. 

Therefore, the algorithm must return two subsets $T_1,T_2$ of terminals, with $|T_1|=|T_2|$, such that $|T_1|\geq \frac{|T'|^{1-64/\Delta}}{3}\geq \frac{k^{1-\eps}}{16}$. Moreover, for every pair $t\in T_1$, $t'\in T_2$ of terminals, $\dist_{\tilde G}(t,t')\geq \hat d=\frac{d\eps}{128}$.
Recall that the running time of Procedure \procsep is bounded by $O(m\cdot |V(G)|^{64/\Delta})\leq O(m^{1+O(\eps)})$.

Recall that $|E'|\leq \frac{d|\pset|}{\eta}\leq \frac{dz}{\eta}=\frac{ k^{1-22\eps}}{\eta}\leq \frac{16d|T_1|}{\eta}$. Clearly, $(T_1,T_2,E')$ is a $(\delta,\hat d)$-distancing in graph $G$, for some parameter $0<\delta<1$. Let $\phi=\frac{64d}{\eta}$, so that $|E'|\leq \frac{\phi |T_1|}{4}$. Since $\eta>128d$, we get that $\phi<1/2$.
Notice also that:

\[\hat d=\frac{d\eps}{256}\geq \frac{\eta\log m}{2d}=\frac{32 \log m}{\phi},  \]

since $\eta\leq   \frac{d^2}{2^{c^*/\eps^6}\cdot \log m}\leq \frac{d^2\eps}{256\log m}$ from the statement of \Cref{lem: inner advanced path peeling}.

We can now use the algorithm from \Cref{lem: distancing to sparse cut} 
to compute a cut $(A,B)$ in graph $G$, with $T_1\subseteq A$, $T_2\subseteq B$, such that $|E_G(A,B)|\leq \phi\cdot\min\set{|E_G(A)|,|E_G(B)|}\leq \frac{64d}{\eta}\cdot \min\set{|E_G(A)|,|E_G(B)|}$. We return the cut $(A,B)$ as the outcome of the algorithm. From the above discussion, $|A\cap T|,|B\cap T|\geq |T_1|\geq \frac{k^{1-\eps}}{16}$.

The running time of the algorithm from \Cref{lem: distancing to sparse cut} is bounded by $O(m)$, and so the total running time in Case 1 is bounded by $O\left (m^{1+O(\eps)}(\eta+d\log n)\right )$.

From now on, we assume that $k\geq n^{\eps}$. The remainder of the algorithm consists of two stages, that we now describe.

\subsubsection{Stage 1: Embedding a Well-Connected Graph}
In this stage, we will apply the algorithm from \Cref{cor: HSS witness} to graph $G$ and the set $T$ of terminals. In order to be able to do so, we need to ensure that $\eps>\frac{2}{(\log \tk)^{1/12}}$ holds.
Recall that $\tk=2k\geq 2n^{\eps}$ from our assumption. Therefore:

\[\frac{2^{12}}{\log \tk}\leq \frac{2^{12}}{\log(2n^{\eps})}\leq \frac{2^{12}}{\eps\cdot \log n}.  \]

From our assumption that $\eps>\frac{2}{(\log n)^{1/24}}$, we get that $\log n>\frac{2^{24}}{\eps^{24}}$. Therefore:

\[\frac{2^{12}}{\log \tk}\leq \frac{\eps^{23}}{2^{12}}.  \]

We conclude that:

\begin{equation}\label{eq: main bound on eps}
\eps> \frac{2}{(\log \tilde k)^{1/23}}\geq  \frac{2}{(\log \tilde k)^{1/12}}.
\end{equation}

Let $d'=\frac{8d}{2^{c^*/\eps^6}}$ and  $\eta'=\frac{8\eta}{2^{c^*/\eps^6}}$.
Since $128d<\eta\leq \frac{d^2}{2^{c^*/\eps^6}\cdot\log m}$, $d',\eta'>1$.
We apply the algorithm from \Cref{cor: HSS witness} to graph $G$, set $T$ of terminals, parameters $d'$, $\eta'$, and parameter $\eps$ that remains unchanged. 
Recall that the running time of the algorithm is:  

\[O\left (k^{1+O(\eps)}+m\cdot k^{O(\eps^3)}\cdot(\eta'+d'\log m)\right )\leq O\left (m^{1+O(\eps)}(d'+\eta')\right )\leq O\left (m^{1+O(\eps)}(d+\eta)\right ). \]

We now consider two cases.

\paragraph{Case 1.} The first case happens if the algorithm from \Cref{cor: HSS witness} returns 
	 a pair $T_1,T_2\subseteq T$ of disjoint subsets of terminals, and a set $E'$ of edges of $G$, such that $|T_1|=|T_2|$ and $|T_1|\geq \frac{\tk^{1-4\eps^3}}{4}\geq \frac{k^{1-4\eps^3}}{4}$. Recall that the algorithm also guarantees 
	 for every pair $t\in T_1,t'\in T_2$ of terminals, $\dist_{G\setminus E'}(t,t')>d'$.
	 
	 Moreover, we are guaranteed that: $|E'|\leq \frac{d'\cdot |T_1|}{\eta'}=\frac{d\cdot |T_1|}{\eta}$.
	Denote $\phi=\frac{4d}{\eta}$, so that $|E'|\leq \frac{\phi|T_1|}{4}$ holds. Since $\eta>128d$, we get that $\phi<1/2$. 
	Since  $\eta\leq \frac{d^2}{2^{c^*/\eps^6}\cdot \log m}$, we get that:
	
	\[\frac{32\log m}{\phi}=\frac{8\eta\log m}{d}\leq \frac{8d}{2^{c^*/\eps^6}}= d'.  \]

We can now apply the algorithm from \Cref{lem: distancing to sparse cut}, to compute a cut $(A,B)$ in graph $G$, with $T_1\subseteq A$, $T_2\subseteq B$, such that $|E_G(A,B)|\leq \phi\cdot \min\set{|E_G(A)|,|E_G(B)|}= \frac{4d}{\eta}\cdot \min\set{|E_G(A)|,|E_G(B)|}$.
The running time of this algorithm is bounded by $O(m)$.	
We return the cut $(A,B)$ as the output of the algorithm. Clearly, $|A\cap T|,|B\cap T|\geq |T_1|\geq 	\frac{k^{1-4\eps^3}}{4}\geq \frac{k^{1-\eps}}{16}$.

\paragraph{Case 2.} In the second case, the algorithm from \Cref{cor: HSS witness} must return  a graph $H$ with $V(H)\subseteq T$, $|V(H)|=N^{1/\eps}\geq \tk-\tk^{1-\eps/2}$, where $N=\floor{\tk^{\eps}}=\floor{(2k)^{\eps}}$,  so that the maximum vertex degree in $H$ is at most   $\tk^{32\eps^3}$.

Recall that the algorithm must also return an embedding $\pset^*$ of $H$ into $G$ via paths of length at most $d'$, that cause congestion at most $\eta'\cdot \tk^{32\eps^3}$, and a level-$(1/\eps)$ hierarchical support structure for $H$, such that $H$ is $(\tilde \eta,\td)$-well-connected with respect to the set $S(H)$ of vertices defined by the support structure, where $\tilde \eta=N^{6+256\eps}$, and $\td=2^{c/ \eps^5}$, with $c$ being the constant used in the definition of the Hierarchical Support Structure.
In this case, we continue to Stage 2 of the algorithm.

This completes the description of the first stage of the algorithm. From the above analysis, the running time of this stage is bounded by  $O\left (m\cdot n^{O(\eps)}(d+\eta)\right )$.

\subsubsection{Stage 2: Computing the Routing}

In this stage, we start with $\pset=\emptyset$, and then gradually add paths to $\pset$. We also maintain a graph $\tilde G$, where initially $\tilde G=G$. As the algorithm progresses, every edge $e\in E(\tilde G)$ that participates in $\eta$ paths of $\pset$ is deleted from $\tilde G$. Therefore, we can think of graph $\tilde G$ as a dynamic graph, with edges deleted from $\tilde G$ over time. Whenever an edge $e$ is deleted from graph $\tilde G$, for every edge $e'\in E(H)$ whose embedding path $P(e')\in \pset^*$ contains $e$, we also delete edge $e'$ from graph $H$. Therefore, graph $H$ can also be viewed as a dynamic graph. We will maintain the following data structures throughout the algorithm.

The first data structure, that we denote by $\dset(H)$, will be used in order to support approximate shortest-path queries in graph $H$. The data structure is maintained using the algorithm from \Cref{thm: APSP in HSS}.  Recall that $N=\floor{\tk^{\eps}}$, and that we have established in Inequality \ref{eq: main bound on eps} that $\frac{2}{(\log \tk)^{1/12}}<\eps<1/400$. In order to be able to use \Cref{thm: APSP in HSS}, we need to verify that $\frac{N^{\eps^4}}{\log N}\geq 2^{128/\eps^6}$ holds.
Indeed observe first that, since  $\eps>\frac{2}{(\log \tk)^{1/12}}$, we get that 

\[
\tk^{\eps^8}>\tk^{(2/(\log \tk)^{1/12})^8}= \tk^{256/(\log \tk)^{2/3}} >2^{(\log \tk)^{1/3}}>\log \tk.
\]

Additionally, from the inequality $\eps>\frac{2}{(\log \tk)^{1/12}}$, we get that $\log \tk\geq (2/\eps)^{12}$, and:

\begin{equation}\label{eq: bound on k'} 
\tk\geq 2^{(2/\eps)^{12}}
\end{equation} 

Lastly,
since $N=\floor{\tk^{\eps}}$ and $\eps<1/400$, we get that:

\[\frac{N^{\eps^4}}{\log N}\geq \frac{\tk^{\eps^6}}{\eps\cdot\log \tk}\geq \tk^{\eps^6/2}\geq 2^{128/\eps^6}. \]

We let $q=1/\eps$. We maintain the data structure from \Cref{thm: APSP in HSS} in graph $H$, with parameters $j=q$, and parameters $N$ and $\eps$ that remain unchanged. Recall that we are given a level-$q$ hierarchical support structure for $H$, such that $H$ is $(\tilde \eta,\td)$-well-connected with respect to the set $S(H)$ of vertices defined by the support structure, where $\tilde \eta=N^{6+256\eps}=N^{6+256q\eps^2}=\eta_q$, and $\td=2^{c/ \eps^5}=2^{cq/\eps^4}=\td_q$, where $\eta_q$ and $\td_q$ are the parameters that are used in the definition of the Hierarchical Support Structure.

We denote the data structure that the algorithm from \Cref{thm: APSP in HSS} maintains by $\dset(H)$. Recall that this data structure can withstand the deletion of up to $\Lambda_q=N^{q-8-300q\eps^2}=N^{q-8-300\eps}$ edge deletions from graph $H$. As long as fewer than $\Lambda_q$ edges are deleted from $H$, the algorithm maintains a decremental set $S'(H)\subseteq V(H)$ of vertices, with $|S'(H)|\geq \frac{N^q}{16^q}=\frac{|V(H)|}{16^q}\geq \frac{k}{2^{8/\eps}}$.

The algorithm supports short-path queries between vertices of $S'(H)$: given a pair $x,y\in S'(H)$ of vertices, return a path $P$ connecting $x$ to $y$ in the current graph $H$, whose length is at most $d^*_q=2^{O(q/\eps^5)}=2^{O(1/\eps^6)}$, in time $O(|E(P)|)$. 

If we denote by $m'$ the number of edges in $H$ at the beginning of the algorithm, then the total update time of the algorithm is bounded by: 
\[O\left(qN^{q+3}\cdot 2^{O(1/\eps^6)}+m'\cdot N^2\cdot 2^{O(1/\eps^6)}\right ).\]

Recall that $N\leq \tk^{\eps}\leq n^{\eps}$, and maximum vertex degree in $H$ is at most $n^{O(\eps^3)}$. Recall also that $N^q\leq n$, and $q=1/\eps$. Furthermore, as established in Inequality \ref{eq: bound on k'}, $n\geq \tk\geq 2^{(2/\eps)^{12}}$, and so 
$2^{O(1/\eps^6)}\leq n^{O(\eps)}$. 
  It is then easy to verify that the running time of the algorithm is bounded by $O\left (n^{1+O(\eps)}  \right )$.

The second data structure is, intuitively, an \EST in $G$ that is rooted at the set $S'(H)$ of vertices, and whose depth parameter is $d/2$. Specifically, we maintain a graph $G'$, that is defined as follows. Initially, we obtain graph $G'$ from $G$, by adding a source vertex $s$, that connects to every vertex $v\in S(H)$ with an edge. As the algorithm progresses and new paths are added to the set $\pset$ of paths that we construct, whenever some edge $e\in E(G)$ appears in $\eta$ paths of $\pset$, we delete $e$ from $G'$. For every edge $e'\in E(H)$, whose embedding path $P(e')\in \pset^*$ contains edge $e$, we also delete $e'$ from $H$, and update data structure $\dset(H)$ accordingly. If, as the result of this update, some vertex $x$ is deleted from set $S'(H)$, then we also delete edge $(s,x)$ from graph $G'$.
We maintain an \EST data structure in graph $G'$, rooted at vertex $s$, with depth bound $d/2+1$. We denote the data structure, and the corresponding tree, by $\tau$. 
The total update time that is needed in order to maintain the \EST data structure $\tau$ is bounded by $O(md\log n)$. 

Lastly, for every edge $e\in E(G)$, we maintain a list $L(e)$ of all edges $e'\in E(H)$, such that the embedding path $P(e')\in \pset^*$ contains $e$. We also maintain a pointer from $e$ to every edge in $L(e)$ and back. Recall that the embedding $\pset^*$ causes congestion at most $\eta'\cdot \tk^{32\eps^3}$, so the length of each such list is bounded by $\eta'\cdot \tk^{32\eps^3}$. Moreover, the deletion of an edge $e\in E(G)$ from graph $G'$ may trigger the deletion of at most $\eta'\cdot \tk^{32\eps^3}$ edges from graph $H$. We will ensure that $|\pset|\leq z$ holds throughout the algorithm, and that every path in $\pset$ has length at most $d$. Therefore, if we denote by $E'\subseteq E(G)$ the collection of all edges that participate in $\eta$ paths in $\pset$, the we are guaranteed that throughout the algorithm:

\[|E'|\leq \frac{\sum_{P\in \pset}|E(P)|}{\eta}\leq \frac{|\pset|\cdot d}{\eta}\leq \frac{zd}{\eta}.  \]
 
Therefore, the total number of edges that may be deleted from graph $H$ over the course of the algorithm is bounded by:

\[\begin{split} 
\frac{zd}{\eta}\cdot \eta'\cdot \tk^{32\eps^3}&\leq
\tk^{1-22\eps}\cdot \tk^{32\eps^3}
\\&\leq \tk^{1-21\eps}\\
&\leq \frac{N^q\cdot 2^q}{N^{21\eps q}}\\
&<N^{q-9}<\Lambda_q\end{split}\]

(For the first inequality, we have used the fact that $z=\frac{k^{1-22\eps}}{d}$, $k\leq \tk$,  and $\eta'=\frac{8\eta}{2^{c^*/\eps^6}}\leq \eta$. For the third inequality we used the fact that $N=\floor{\tk^{\eps}}$, and $q=1/\eps$, so $N^q\leq \tk\leq N^q\cdot 2^q$ holds. The fourth inequality follows since $2^q=2^{1/\eps}<\left(\frac{\tk^{\eps}}{2}\right)^9\leq N^9$ from Inequality \ref{eq: bound on k'}.)

If $(s_i,t_i)$ is a pair of vertices in $\mset$, then we say that $s_i$ is a \emph{mate} of $t_i$, and $t_i$ is a mate of $s_i$. Throughout the algorithm, we will also maintain a subset $S''(H)\subseteq S'(H)$ containing all vertices $x\in S'(H)$, such that no path in $\pset$ has $x$ as its endpoint, and the mate of $x$ still lies in the tree $\tau$.
We also maintain, for every edge $e\in E(G)$, a counter $n(e)$, counting the number of paths in $\pset$ that contain $e$.

We are now ready to describe our algorithm. The algorithm consists of at most $z$ iterations, and they are performed as long as $|\pset|<z$ and $S''(H)\neq \emptyset$ hold. 

In order to perform a single iteration, we let $x$ be any vertex in set $S''(H)$. Assume w.l.o.g. that $x=s_i$, and that its mate is $t_i$. Since $x\in S''(H)$, vertex $t_i$ currently lies in the tree $\tau$. Using the tree, we can compute a path $Q$ in graph $G'$, that connects $t_i$ to some vertex $y\in S(H')$, so that the length of the path is bounded by $d/2$, and the path does not contain any vertices of $S'(H)\setminus \set{y}$. This can be done in time $O(|E(Q)|)$. Next, using data structure $\dset(H)$, we compute a path $\tilde Q'$ connecting $y$ to $s_i$ in graph $H$, such that the length of the path is bounded by $2^{O(1/\eps^6)}$. This can be done in time $O(|E(\tilde Q')|)$. Lastly, by replacing every edge $e'$ on path $\tilde Q'$ with its embedding path $P(e')\in \pset^*$, we obtain a path $Q'$ in graph $G$, connecting $y$ to $s_i$, whose length is bounded by:

\[2^{O(1/\eps^6)}\cdot d'\leq d/2,\]
 
 since $d'=\frac{8d}{2^{c^*/\eps^6}}$, and we can assume that $c^*$ is a large enough constant.
 By combining the paths $Q$ and $Q'$, we obtain a path $P$, connecting $s_i$ to $t_i$ in graph $G$, whose length is at most $d$. If path $P$ is not simple, then we convert it into a simple path, in time $O(d)$. This can be done, for example, by traversing the vertices of $P$ in the order of their appearance on the path, and marking every vertex that has been traversed in an array of length $n$. Whenever we attempt to mark a vertex $v$ that was already marked before, we recognize that we closed a simple cycle $C\subseteq P$. We can then retrace this cycle and un-mark all its vertices except for $v$. We  delete all vertices of $C\setminus \set{v}$ from $P$, and continue the traversal of $P$ starting from $v$. Since every vertex on $P$ may be traversed at most twice (once when we visit it for the first time, and the second time when we remove a cycle on which that vertex lies), this algorithm takes time $O(d)$, assuming that we are provided an empty array of length $n$, that can be used in order to mark and un-mark the vertices of $P$. At the end of this procedure, we un-mark every vertex of $P$, so the array can be reused in the next iteration. We denote the resulting simple path $P$ by $P_i$, and we add $P_i$ to the set $\pset$ of paths that we maintain.
 
 Next, we consider every edge $e\in E(P_i)$ one by one. For each such edge $e$, we increase the counter $n(e)$ by $1$. If $n(e)=\eta$ holds, then we delete $e$ from graph $G'$, and update data structure $\tau$ accordingly. Additionally, for every edge $e'\in L(e)$, we delete $e'$ from graph $H$, updating the data structure $\dset(H)$ accordingly, and we update all lists $L(e'')$ of edges $e''\in E(G)$ with $e'\in L(e'')$, by deleting $e'$ from  $L(e'')$.
 
 As the result of these updates to data structure $\dset(H)$, we may have deleted some vertices from set $S'(H)$. Let $Y$ be the set of all such vertices. For every vertex $v\in S'(H)$, we delete the edge $(s,v)$ from graph $G'$, and update the data structure $\tau$ accordingly. We also delete $v$ from set $S''(H)$ if it belongs to this set. 
 
 Lastly, whenever a vertex $u$ leaves tree $\tau$, if either $u$ or its mate $u'$ lie in $S''(H)$, then we delete both vertices from $S''(H)$. We also delete $s_i$ and $t_i$ from $S''(H)$.
 This completes the description of an iteration. Besides the time that is needed in order to maintain data structures $\dset(H),\tau, S''(H)$, and $\set{L(e),n(e)}_{e\in E(G)}$, the additional time that is needed in order to execute an iteration is bounded by $O(d)$.

We now consider two cases. In the first case, the algorithm terminates with $|\pset|\geq z$. In this case, we return the set $\pset$ of paths as the algorithm's outcome, together with the collection $\mset'\subseteq \mset$ of pairs of vertices, containing every pair $(s_i,t_i)\in \mset$ for which some path connecting $s_i$ to $t_i$ lies in $\pset$. It is easy to verify that the set $\pset$ of paths has all required properties.

Consider now the second case, when $|\pset|<z=\frac{k^{1-22\eps}}{d}<\frac{k}{2^{16/\eps}}$ (from Inequality \ref{eq: bound on k'}). 
In this case, the algorithm must have terminated because $S''(H)=\emptyset$ holds. Since the number of vertices that serve as endpoints of paths in $\pset$ is bounded by $2z<\frac{2k}{2^{16/\eps}}$, while we are guaranteed that
$|S'(H)|\geq \frac{k}{2^{8/\eps}}$, there must be a set $X\subseteq S'(H)$ of at least $\frac{|S'(H)|}{2}\geq \frac{k}{2^{16/\eps}}$ vertices, such that, for every vertex $x\in X$, no path in $\pset$ has $x$ as its endpoint, and the mate of $x$ does not belong to the tree $\tau$. 
Therefore, if we denote by $X'$ the set of vertices that are mates of the vertices of $X$, then, in the current graph $G'$, $\dist_{G'}(X,X')>d/2$. 
Clearly, $X'\cap X=\emptyset$, and $|X'|=|X|$.
Let $E'$ be the set of all edges of graph $G$ that belong to $\eta$ paths of $\pset$. Then $|E'|\leq \frac{|\pset|\cdot d}{\eta}\leq \frac{d}{\eta}\cdot |X|$.
Moreover, $\dist_{G\setminus E'}(X,X')>d/2$. Therefore, $(X,X',E')$ is a $(\delta,d/2)$-distancing in graph $G$, for some parameter $0<\delta<1$.

We denote $\phi=\frac{4d}{\eta}$, so that $|E'|\leq \frac{\phi|X|} 4$ holds.
Notice that:
$\frac{32\log m}{\phi}=\frac{8\eta\log m}{d}<\frac{d}{2}$, since $\eta\leq \frac{d^2}{2^{c^*/\eps^6}\cdot \log m}$ from the statement of \Cref{lem: inner advanced path peeling}.
We can now apply the algorithm from \Cref{lem: distancing to sparse cut}, to compute a cut $(A,B)$ in graph $G$, with $X\subseteq A$, $X'\subseteq B$, and $|E_G(A,B)|\leq \phi\cdot \min\set{|E_G(A)|,|E_G(B)|}=\frac{4d}{\eta}\cdot \min\set{|E_G(A)|,|E_G(B)|} $. Recall that $|X|,|X'|\geq \frac{k}{2^{16/\eps}}\geq \frac{k^{1-\eps}}{16}$, since $2^{16/\eps}\leq 16k^{\eps}$ from Inequality \ref{eq: bound on k'}. We then return the cut $(A,B)$ as the output of the algorithm.
The running time of the algorithm from \Cref{lem: distancing to sparse cut} is $O(m)$.

It now remains to analyze the running time of the algorithm. 
The  running time of Stage 1 is bounded by $O\left (m^{1+O(\eps)}(d+\eta)\right )$, as shown already.

In order to analyze the running time of Stage 2, recall that the total update time for maintaining data structure $\dset(H)$, as established already, is bounded by $O(n^{1+O(\eps)})$, and the total update time for maintaining data structure $\tau$ is bounded by $O(md\log n)$. Since the paths in $\pset^*$ cause congestion at most $\eta'\cdot k^{O(\eps^3)}$, maintaining the lists $\set{L(e)}_{e\in E(G)}$ takes total time $O(m\cdot \eta'\cdot k^{O(\eps^3)})\leq O(m\eta k^{O(\eps^3)})$.
The time required for additional computation in every iteration (that is, computing the path $P$, converting it into a simple path, and reducing the counters $n(e)$ for all edges $e\in E(P)$) is bounded by $O(d)$. The number of iterations is bounded by $O(k)$.
The running time required for the remaining calculations that the algorithm performs (such as, for example, applying the algorithm from  \Cref{lem: distancing to sparse cut} to compute a low-conductance cut at the end of the algorithm if $|\pset|<z$) is subsumed by the above running times. 
The total running time is then bounded by:

\[O\left (m^{1+O(\eps)}(d+\eta)\right )+O(n^{1+O(\eps)})+O(md\log n)+O(m\cdot \eta\cdot k^{O(\eps^3)})\leq O\left (m^{1+O(\eps)}(d+\eta)\right ). \]

\section{An Algorithm for the Cut Player in the Cut-Matching Game -- Proof of \Cref{thm: new cut player}}
\label{sec: cut player}

In this section we provide an algorithm for the Cut Player in the Cut-Matching Game, proving \Cref{thm: new cut player}.
The algorithm consists of a number of phases. At the beginning of phase $q$, we are given a partition $(X_q,Y_q)$ of $V(G)$, with $|X_q|\geq \frac{3n}4$, and $|E_{G}(X_q,Y_q)|\leq \frac{|Y_q|}{100}$. At the beginning of the algorithm, we use the partition $(X_1,Y_1)$ of $V(G)$, with $X_1=V(G)$ and $Y_1=\emptyset$. We now describe the execution of Phase $q$.

Let $G_q=G[X_q]$. Assume first that, for every connected component $C$ of graph $G_q$, $|V(C)|\leq 5n/8$. Let $C$ be the largest connected component of $G_q$. If $|V(C)|\geq n/4$, then we return a cut $(A,B)$ in graph $G$ with $A=V(C)$ and $B=V(G)\setminus V(C)$. Clearly, $|A|\ge n/4$, and, since $|A|\leq 5n/8$, we get that $|B|\geq n/4$ as well. Moreover, $E_G(A,B)\subseteq E_G(X_q,Y_q)$, and so $|E_G(A,B)|\leq \frac{|Y_q|}{100}\leq \frac{n}{100}$. Otherwise, if $|V(C)|<n/4$, then we compute a partition $(A,B)$ of $V(G)$ as follows. We start with $A=\emptyset$ and $B=Y_q$, and then process every connected component $C'$ of $G_q$ one by one. When a component $C'$ is processed, if $|A|\leq |B|$ holds, then we add the vertices of $C'$ to $A$, and otherwise we add them to $B$. Since every connected component of $G_q$ contains at most $n/4$ vertices, and $|Y_q|\leq n/4$, it is easy to verify that at the end of the algorithm, $|A|,|B|\ge n/4$ holds. As before, $E_G(A,B)\subseteq E_G(X_q,Y_q)$, and so $|E_G(A,B)|\leq \frac{|Y_q|}{100}\leq \frac{n}{100}$. We terminate the algorithm and return the cut $(A,B)$.

We assume from now on that there is a connected component $G'_q$ of graph $G_q$ with $|V(G'_q)|\geq 5n/8$. We denote $n_q=|V(G'_q)|$.  We use algorithm \constructexpander from 
\Cref{thm:explicit expander} to construct a graph $H_q$, with $|V(H_q)|=n_q$, such that $H_q$ is an $\alpha_0$-expander, and maximum vertex degree in $H_q$ is at most $9$. The running time of this algorithm is $O(n_q)\leq O(n)$. It will be convenient for us to identify the vertices of $H_q$ with the vertices of $G'_q$. In other words, we assume that $V(H_q)=V(G'_q)$, by arbitrarily mapping every vertex of $H_q$ to a distinct vertex of $G'_q$.

Using a simple standard greedy algorithm, we compute, in time $O(n)$, a partition $\mset_1,\mset_2,\ldots,\mset_{19}$ of the set $E(H_q)$ of edges, so that, for each $1\leq i\leq 19$, $\mset_i$ is a matching. We use a parameter $\rho=\frac{n}{2^{\hat c/\eps^6}\cdot \Delta^3\cdot \log^2n}$, where $\hat c$ is a large enough constant. 

Next, we perform $19$ iterations. 
For $1\leq i\leq 19$, in the $i$th iteration, we use the algorithm for advanced path peeling from \Cref{thm: main main advanced path peeling} in order to embed the edges of $\mset_i$ into graph $G'_q$ with low congestion. 
If we successfully embed all but at most $\rho$ edges of $\mset_i$, then we partition the set $\mset_i$ of edges into two subsets: set $\mset'_i$ containing all edges that we managed to embed, and set $F_i$ containing all remaining edges. We say that the $i$th iteration ended with a routing, and continue to the next iteration. If we failed to embed a large enough subset of $\mset_i$ into $G'_q$, then we will compute a sparse cut in graph $G'_q$, that will allow us to update the partition $(A_q,B_q)$ of $V(G)$. We then say that the $i$th iteration ended with a cut, and terminate the current phase. 
If every one of the $19$ iterations ended with a routing, then, by  letting $G''_q$ be the graph that is obtained from $G'_q$ by adding to it a set $\bigcup_{i=1}^{19}F_i$ of fake edges, we obtain an embedding of $H_q$ into $G''_q$ with low congestion. We can then use Algorithm \extractexpander from \Cref{lem: embedding expander w fake edges gives expander}, to extract a large enough expander graph $G^*\subseteq G'_q$.  
We now describe the execution of a single iteration.

Consider an index $1\leq i\leq 19$. If $|\mset_i|\leq \rho$, then we let $\mset'_i=\emptyset$, $\pset_i=\emptyset$, $F_i=\mset_i$, and continue to the next iteration. From now on we assume that $|\mset_i|>\rho$.

We apply the algorithm from \Cref{thm: main main advanced path peeling} to graph $G'_q$ and the set $\mset_i$ of pairs of its vertices, with parameters $\eps$, $\phi=\frac{1}{100\Delta}$, and $\alpha=\frac{\rho}{2|\mset_i|}$, so $0<\alpha\leq \half$ holds. Recall that  $\frac{2}{(\log n)^{1/25}}< \eps<\frac{1}{400}$, and so $\log n>\left(\frac{2}{\eps}\right )^{25}$. Since $n_q>n/2$, we get that $\log(n_q)\geq \log n-1\geq \left(\frac{2}{\eps}\right )^{24}$. Therefore, the condition of \Cref{thm: main main advanced path peeling} that $\frac{2}{(\log n)^{1/24}}< \eps<\frac{1}{400}$ holds. We denote by $T_i$ the set of all vertices of $G'_q$ that serve as endpoints of the edges of $\mset_i$, so $|T_i|=2|\mset_i|\geq 2\rho$.

We now consider two cases. In the first case, the algorithm from \Cref{thm: main main advanced path peeling} returns a cut $(A_q,B_q)$ in $G'_q$, with $|E_{G'_q}(A_q,B_q)|\leq \phi\cdot \min\set{|E_G(A_q)|,|E_G(B_q)|}$, such that $|A_q\cap T_i|,|B_q\cap T_i|\geq \frac{\alpha\cdot |\mset_i|}{16}\geq \frac{\rho}{32}$. In this case we say that iteration $i$ ended with a cut. Since maximum vertex degree in $G$ is at most $\Delta$, we get that $|E_G(A_q)|\leq \Delta\cdot |A_q|$ and $|E_G(B_q)|\leq \Delta\cdot |B_q|$, and since $\phi=\frac{1}{100\Delta}$, we get that  $|E_{G'_q}(A_q,B_q)|\leq \frac{1}{100\Delta} \min\set{|E_G(A_q)|,|E_G(B_q)|}\leq \frac{1}{100}\cdot\min\set{|A_q|,|B_q|}$.  Assume w.l.o.g. that $|A_q|\geq |B_q|$. Since we have assumed that $|V(G_q')|\geq \frac{5n}{8}$, we get that $|A_q|\geq \frac{5n}{16}\geq \frac{n}{4}$. We set $X_{q+1}=A_q$ and $Y_{q+1}=V(G)\setminus A_q$. Clearly, $(X_{q+1},Y_{q+1})$ is a partition of $V(G)$, and, from our discussion, $|X_{q+1}|\geq \frac{n}{4}$. Moreover, we can think of the cut $(X_{q+1},Y_{q+1})$ as obtained from cut $(X_q,Y_q)$, by moving all vertices of $V(G_q)\setminus A_q$ from $X_q$ to $Y_q$. Therefore, $|E_G(X_{q+1},Y_{q+1})|\leq |E_G(X_q,Y_q)|+|E_G(A_q,B_q)|\leq \frac{1}{100}|Y_q|+\frac{1}{100}|B_q|\leq \frac{1}{100}|Y_{q+1}|$. If $|X_{q+1}|\geq \frac{3n}{4}$, then we terminate the current phase and continue to Phase $(q+1)$. Otherwise, we return the cut $(A,B)=(X_{q+1},Y_{q+1})$. In the latter case, we are guaranteed that $|A|,|B|\geq \frac{n}{4}$, and $|E_G(A,B)|\leq \frac{1}{100}|Y_{q+1}|\leq \frac{n}{100}$.

In the second case, the algorithm from \Cref{thm: main main advanced path peeling} 
computes a routing $\pset_i$ in $G'_q$ of a subset $\mset'_i\subseteq \mset_i$ containing at least $(1-\alpha)|\mset_i|=|\mset_i|-\frac{\rho}2$ pairs of vertices, such that the total congestion caused by the paths in $\pset_i$ is at most $\frac{2^{O(1/\eps^6)}\cdot \log^2 n}{\phi^2}\leq 2^{O(1/\eps^6)}\cdot \Delta^2\cdot \log^2 n$. In this case, we say that iteration $i$ ended with a routing. For every edge $e\in \mset'_i$, we let $P(e)\in \pset_i$ be the embedding path of edge $e$. We set $F_i=\mset_i\setminus\mset'_i$, so $|F_i|\leq \frac{\rho}{2}$, and we continue to the next iteration. 
This concludes the description of the $i$th iteration. We now complete the description of Phase $q$.

If any iteration in Phase $q$ ended with a cut, then the phase is terminated as described above. We assume therefore from now on that every iteration in Phase $q$ ended with a routing. Let $F=\bigcup_{i=1}^{19}F_i$. Recall that, for all $1\leq i\leq 19$, $|F_i|\leq \rho$, so $|F|\leq 19\rho$. Consider the graph $G''_q=G'_q\cup F$, where we treat the edges of $F$ as fake edges. For every fake edge $e\in F$, we let $P(e)=\set{e}$ be a path that embeds the edge $e$ into itself in graph $G''_q$. Note that the maximum vertex degree in $G''_q$ is at most $\Delta+19\leq 19\Delta$. We have also now obtained an embedding $\pset''=\set{P(e)\mid e\in E(H_q)}$ of $H_q$ into $G''_q$, such that the paths in $\pset''$ cause congestion $\eta$, where $\eta\leq  2^{O(1/\eps^6)}\cdot \Delta^2\cdot \log^2 n$.

For convenience, we denote the maximum vertex degree of $G''_q$ by $\Delta_G\leq 19\Delta$, the maximum vertex degree of $H_q$ by $\Delta_H\leq 9$, and $\psi=\alpha_0$.
Notice that:

\[\frac{\psi\cdot n_q}{32\Delta_G\eta}\geq \frac{\alpha_0\cdot n}{ 2^{O(1/\eps^6)}\cdot \Delta^3\cdot  \log^2 n} \geq \frac{ n}{ 2^{O(1/\eps^6)}\cdot\Delta^3\cdot \log^2 n} \geq 20\rho\geq |F|,  \]

since
$\rho=\frac{n}{2^{\hat c/\eps^6}\cdot \Delta^3\cdot \log^2n}$, and $\hat c$ is a large enough constant. We can then use the algorithm \extractexpander from \Cref{lem: embedding expander w fake edges gives expander} to compute a subgraph $G^*\subseteq G'_q$, such that $G^*$ is a $\psi'$-expander, for $\psi'\geq\frac{\psi}{6\Delta_G\cdot \eta}\geq \frac{1}{2^{O(1/\eps^6)}\cdot \Delta^3\cdot \log^2 n }$. Moreover, since $|F|\leq \frac{\psi\cdot n_q}{32\Delta_G\eta}$, we get that $\frac{4|F|\eta}{\psi}\leq \frac{n_q}{8\Delta_G}$, and so $|V(G^*)|\geq n_q-\frac{4|F|\eta}{\psi}\geq \frac{15n_q}{16}$. Since $n_q\geq\frac{5n}{8}$, we get that $|V(G^*)|\geq \frac{n}{2}$. We return the set $S=V(G^*)$ of vertices and terminate the algorithm.

We now bound the running time of a single phase. A phase has at most $19$ iterations, and in every iteration we use the algorithm from \Cref{thm: main main advanced path peeling}, whose running time is bounded by $O\left (\frac{m^{1+O(\eps)}}{\phi^3}\right )\leq O\left (m^{1+O(\eps)}\cdot \Delta^3\right )$. Additionally, the running time of the algorithm from \Cref{lem: embedding expander w fake edges gives expander} is bounded by 
$\tilde O(|E(G)|\Delta_G\cdot\cong/\psi)\leq \tilde O\left (n\cdot \Delta^4\cdot 2^{O(1/\eps^6)}\right )\leq O\left (n^{1+O(\eps)}\cdot \Delta^4\right )$ (since $\eps>\frac{2}{(\log n)^{1/25}}$, so $n\geq 2^{(2/\eps)^{25}}$). Overall, the running time of a single phase is bounded by $O\left (m^{1+O(\eps)}\cdot \Delta^4\right )$.

Next, we bound the number of phases. Note that in every phase, the cardinality of the set $Y_q$ of vertices grows by at least $\frac{\rho}{32}\geq \frac{n}{2^{\Theta(1/\eps^6)}\cdot \Delta^3\cdot \log^2n}$. Therefore, the number of phases is bounded by $2^{O(1/\eps^6)}\cdot \Delta^3\cdot \log^2n$, and so the total running time of the algorithm is bounded by:

\[O\left (m^{1+O(\eps)}\cdot \Delta^7\cdot 2^{O(1/\eps^6)}\right )\leq O\left(m^{1+O(\eps)}\cdot \Delta^7\right ). \]

\section{Further Applications}
\label{sec: applications}

In this section we provide our improved deterministic approximation algorithms for \sparsest, \lowcond, \mbc, and Most-Balanced Sparse Cut. We also provide a new algorithm for expander decompositions. Several (but not all) of these results are obtained in the same manner as their weaker counterparts from \cite{detbalanced}, by plugging in our stronger algorithm for the Cut Player in the \CMG from \Cref{thm: new cut player} instead of its weaker analogue from \cite{detbalanced}. Some of the proofs in this section are therefore essentially identical to the proofs from \cite{detbalanced}, and are only provided here for completeness. We point out explicitly when this is the case. We start by introducing some technical tools that will be useful for us.

\subsection{Main Technical Tools}
\label{subsec: technical tools}

In this subsection we introduce two main technical tools that will be used in order to obtain improved algorithms for \mbc, \sparsest, \lowcond, and Expander Decomposition. Both these tools - degree reduction and a faster algorithm for basic path peeling - appeared in \cite{detbalanced}, and we do not make any changes to them. 

\subsubsection{Degree Reduction}
\label{subsec:constant degree}

Some of our algorithms are easier to describe, and provide better guarantees, when the maximum vertex degree of the input graph is low. However, in general, an input graph may have an arbitrarily large maximum vertex degree. We describe here a standard algorithm for transforming a general graph into a low-degree graph, by replacing every vertex of the input graph with an expander of appropriate size. The algorithm is identical to that from \cite{detbalanced}, and similar algorithms have been used in the past extensively.

We now turn to describe a deterministic algorithm, that we call \reducedegree. The algorithm is given as input an arbitrary graph $G=(V,E)$, and transforms it into a bounded-degree graph $\hat G$. Throughout, we denote $|V|=n$ and $|E|=m$.  
For convenience, we denote $V=\set{v_1,\ldots,v_n}$. For every vertex $v_i\in V$, we denote by $d(v_i)$ the degree of $v_i$ in $G$, and we let $\set{e_1(v_i),\ldots,e_{d(v_i)}(v_i)}$ be the set of edges incident to $v$, indexed in an arbitrary order. For every vertex $v_i\in V$, we use Algorithm \constructexpander from
\ref{thm:explicit expander} to construct a graph $H_i$, whose vertex set $V_i=\set{u_1(v_i),\ldots,u_{d(v_i)}(v_i)}$ contains $d(v_i)$ vertices, such that $H_i$ is an $\alpha_0$-expander, and the maximum vertex degree in $H_i$ is at most $9$. Recall that the running time of the algorithm for constructing $H_i$ is bounded by $O(d(v_i))$. 

In order to obtain the final graph $\hat G$, we start with a disjoint union of all graphs in $\set{H_i\mid v_i\in V}$. All edges lying in such graphs $H_i$ are called \emph{type-1 edges}. 
Additionally, we add to $\hat G$ a collection of type-2 edges, defined as follows. Consider any edge $e=(v,v')\in E$, and assume that $e=e_j(v)=e_{j'}(v')$ (that is, $e$ is the $j$th edge incident to $v$ and it is the $j'$th edge incident to $v'$). We then let $\hat e$ be the edge $(u_j(v),u_{j'}(v'))$. For every edge $e\in E$, we add the corresponding new edge $\hat e$ to graph $\hat G$ as a type-2 edge. This concludes the construction of the graph $\hat G$, that we denote by $\hat G=(\hat V,\hat E)$. Note that the maximum vertex degree in $\hat G$ is at most $10$, and $|\hat V|=2m$. Moreover, the running time of the algorithm for constructing the graph $\hat G$ is $O(m)$.

We say that a set $S\subseteq \hat V$ of vertices of $\hat G$ is \emph{canonical} if, for every vertex $v_i\in V$, either $V_i\subseteq S$, or $V_i\cap S=\emptyset$. Similarly, we say that a cut $(X,Y)$ in a subgraph of $\hat G$ is canonical, if each of $X,Y$ is a canonical subset of $\hat V$.
The following lemma, that was proved in \cite{detbalanced} allows us to convert an arbitrary sparse balanced cuts in a subgraph of $\hat G$ into a canonical one.

\begin{lemma}[Lemma 5.4 from \cite{detbalanced}]\label{lem: degree reduction balanced cut case}
	Let $\alpha_0 >0$ be the constant from \Cref{thm:explicit expander}.
	There is a deterministic algorithm, that we call \makecanonical, that, given a subgraph $\hat G'\subseteq \hat G$, where $V(\hat G')$ is a canonical vertex set, together with a cut $( A, B)$ in $\hat G'$, computes a canonical cut $(A',B')$ in $\hat G'$, such that $|A'|\geq |A|/2$, $|B'|\geq |B|/2$, and moreover, if $|E_{\hat G}(A,B)|\leq \frac{\alpha_0}{2}\cdot  \min\set{|A|,|B|}$, then  $|E_{\hat G}(A',B')|\leq O(|E_{\hat G}(A,B)|)$. The running time of the algorithm is $O(m)$,
\end{lemma}

\subsubsection{Faster Basic Path Peeling}
The following theorem provides a more efficient algorithm for basic Path Peeling. The theorem was proved in \cite{detbalanced}; a similar result appeared in \cite{NanongkaiS17} (see Lemma
B.18).

\begin{theorem}[Theorem 7.1 from \cite{detbalanced}]\label{thm:efficient matching player}
	There is a deterministic algorithm, that we call \matchorcut, whose input consists of an $m$-edge graph $G=(V,E)$, two disjoint subsets $A,B$ of its vertices with $|A|\leq |B|$, and parameters $z\geq 0$ and $0<\phi<1/2$. The algorithm computes one of the following:
	
	\begin{itemize}
		\item either a matching $\mset \subseteq A\times B$ with $|\mset|> |A|-z$, such that there exists a set $\pset=\set{P(a,b)\mid (a,b)\in \mset}$ of paths in $G$, where for each pair $(a,b)\in \mset$, path $P(a,b)$ connects $a$ to $b$, and the paths in $\pset$ cause congestion at most $O\left (\frac{ \log n}{\phi}\right )$; or
		
		\item a cut $(X,Y)$ in $G$, with $|X|,|Y|\geq z/2$, and $|E_G(X,Y)| \leq \phi\cdot\min\set{|X|,|Y|}$.
	\end{itemize}
	
	The running time of the algorithm is $O\left (m^{1+o(1)}\right)$.
\end{theorem}

We note that the algorithm from \Cref{thm:efficient matching player} does not compute the set $\pset$ of paths explicitly, as even listing all paths in the set may take time that is greater than $m^{1+o(1)}$, if parameter $\phi$ is sufficiently small. It only guarantees that set $\pset$ of paths with the above properties exists.

\subsection{Most-Balanced Sparse Cut}
\label{sec:most-bal cut}

Recall that, given a cut $(X,Y)$ in a graph $G$, the \emph{sparsity} of the cut is $\frac{|E_G(X,Y)|}{\min\set{|X|,|Y|}}$. We sometimes refer to $\min\set{|X|,|Y|}$ as the \emph{size} of the cut $(X,Y)$. 
	
In the \MBSC problem, the input is an $n$-vertex graph $G$, and a parameter $0<\phi\leq 1$. The goal is to compute a cut $(X,Y)$ in $G$ of sparsity at most $\phi$, while maximizing the size $\min\set{|X|,|Y|}$ of the cut. 
An $(\alpha,\beta)$-bicriteria approximation algorithm for the problem, given parameters $0<\phi<1$ and $z\geq 1$, must either compute a cut $(X,Y)$ in $G$ of sparsity at most $\phi$ and size at least $z$; or correctly establish that every cut $(X',Y')$ whose sparsity is at most $\phi/\alpha$ has size at most $\beta\cdot z$.

In \cite{detbalanced} (see Lemma 7.3), an $(\alpha,\beta)$-bicriteria deterministic approximation algorithm was obtained for the \MBSC problem, with $\alpha=(\log n)^{O(1/\eps)}$ and $\beta=(\log n)^{O(1/\eps)}$, in time $O\left(m^{1+O(\eps)+o(1)}\cdot (\log n)^{O(1/\eps^2)}\right )$ for any $\frac 1 {c\log n}\leq \eps\leq 1$, for some fixed constant $c$. 
 By using our algorithm for the Cut Player from \Cref{thm: new cut player}, 
 we immediately obtain the following stronger bicriteria approximation algorithm for the problem. 	The proof is essentially identical to the proof of Lemma 7.3 in \cite{detbalanced}; the only difference is that we use the stronger algorithm for the Cut Player from  \Cref{thm: new cut player}. We provide the proof here for completeness.

\begin{theorem}\label{thm: new sparse edge cut of large profit or witness}
	There is a constant $c_0$, and a deterministic algorithm, that, given an $n$-vertex and $m$-edge graph $G=(V,E)$ and parameters $0<\phi\leq 1$, $0<z\leq n$,  and parameter $\frac{4}{(\log n)^{1/25}}< \eps<\frac{1}{400}$:
	
	\begin{itemize}
		\item either returns a cut $(X,Y)$ in $G$ with $|E_G(X,Y)|\leq \phi\cdot\min\set{|X|,|Y|}$ and $|X|,|Y|\geq z$; 
		
		\item or correctly establishes that, for every cut $(X',Y')$ in $G$ with $|E_G(X',Y')|\leq \frac{\phi}{\alpha}\cdot \min\set{|X'|,|Y'|}$,  $\min\set{|X'|,|Y'|}<\alpha'\cdot z$ holds, for $\alpha=2^{c_0/\eps^6}\cdot \log^7 n$ and $\alpha'= 2^{c_0/\eps^6}\cdot \log^6 n$.
	\end{itemize}
	
	The running time of the algorithm is $O\left (m^{1+O(\eps)+o(1)}\right)$.
\end{theorem}

\begin{proof}
The algorithm employs the \CMG. We will maintain a set $F$ of fake edges that are added to graph $G$. Initially, $F=\emptyset$.
	We assume that $n$ is an even integer; otherwise we add a new isolated vertex $v_0$ to $G$, and we add a fake edge connecting $v_0$ to an arbitrary vertex of $G$ to $F$. We also maintain a graph $H$, that initially contains the set $V$ of vertices and no edges. We then perform a number of iterations, that correspond to the \CMG. In every iteration $i$, we will add a matching $M_i$ to graph $H$. We will ensure that the number of iterations is bounded by $O(\log n)$, so the maximum vertex degree in $H$ is always bounded by $\Delta_H\leq O(\log n)$. At the beginning of the algorithm, graph $H$ contains the set $V$ of vertices and no edges. We now describe the execution of the $i$th iteration.

	In order to execute the $i$th iteration, we apply the algorithm  from \Cref{thm: new cut player} to the current graph $H$, with parameter $\eps$ remaining unchanged. 
	Assume first that the output of the algorithm from \Cref{thm: new cut player} is a cut $(A_i,B_i)$ in $H$ with $|A_i|,|B_i|\geq n/4$ and $|E_H(A,B)|\leq n/100$. We treat this partition as the move of the Cut Player. Assume w.l.o.g. that $|A_i|\leq |B_i|$. Next, we compute an arbitrary partition $(A'_i,B'_i)$ of $V(G)$ with $|A'_i|=|B'_i|$, such that $A_i\subseteq A'_i$.  We apply Algorithm \matchorcut from \Cref{thm:efficient matching player} to the sets $A'_i,B'_i$ of vertices, a sparsity parameter $\phi'=\phi/2$ and parameter $z'=4z$. 
	If the algorithm returns a cut $(X,Y)$ in $G$, with $|X|,|Y|\geq z'/2\geq 2z$, and $|E_G(X,Y)|\leq \phi'\cdot\min\set{|X|,|Y|}$, then we terminate the algorithm and return the cut $(X,Y)$, after we delete the extra vertex $v_0$ from it (if it exists). It is easy to verify that $|X|,|Y|\geq z$ and that $|E_G(X,Y)|\leq \phi\cdot\min\set{|X|,|Y|}$.
	Otherwise, the algorithm from \Cref{thm:efficient matching player} computes a matching $\mset'_i\subseteq A'_i\times B'_i$ with $|\mset'_i|\geq |A'_i|-4z$, such that there exists a set $\pset'_i=\set{P(a,b)\mid (a,b)\in \mset'_i}$ of paths in $G$, where for each pair $(a,b)\in \mset'_i$, path $P(a,b)$ connects $a$ to $b$, and the paths in $\pset'_i$ cause congestion at most $O\left (\frac{ \log n}{\phi}\right )$.
	We let $A''_i\subseteq A'_i$, $B''_i\subseteq B'_i$ be the sets of vertices that do not participate in the matching $\mset'_i$, and we let $\mset''_i$ be an arbitrary perfect matching between these vertices. We define a set $F_i$ of fake edges, containing the edges of $M''_i$, and an embedding $\pset''_i=\set{P(e)\mid e\in F_i}$ of the edges in $\mset''_i$, where each fake edge is embedded into itself. Lastly, we set $\mset_i=\mset'_i\cup \mset''_i$. We view the matching $\mset_i$ as the response of the matching player in the \CMG. We add the edges of $M_i$ to $H$, and continue to the next iteration.
	Notice that $|F_i|\leq 4z$.

	We perform the iterations as described above, until the algorithm from \Cref{thm: new cut player} returns a subset $S\subseteq V$ of at least $n/2$ vertices, such that graph $H[S]$ is $\phi^*$-expander, for $\phi^*\geq \Omega\left ( \frac{1}{2^{O(1/\eps^6)}\cdot \Delta_H^3\cdot \log^2 n }\right )\geq \Omega\left ( \frac{1}{2^{O(1/\eps^6)}\cdot \log^5 n }\right ) $. Recall that \Cref{thm:KKOV-new} guarantees that this must happen after at most $O(\log n)$ iterations. We then perform one last iteration, whose index we denote by $q$.
	
	We let $B_q=S$ and $A_q=V(G)\setminus S$, and apply Algorithm  \matchorcut from \Cref{thm:efficient matching player} to the sets $A_q,B_q$ of vertices,  a sparsity parameter $\phi'=\phi/2$ and parameter $z'=4z$. As before, if the algorithm returns a cut $(X,Y)$ in $G$, with $|X|,|Y|\geq z'/2\geq 2z$ and $|E_G(X,Y)|\leq \phi'\cdot\min\set{|X|,|Y|}$, then we terminate the algorithm and return the cut $(X,Y)$, after we delete the extra vertex $v_0$ from it (if it exists). As before, we get that $|X|,|Y|\geq z$ and $|E_G(X,Y)|\leq \phi\cdot\min\set{|X|,|Y|}$. 
	Otherwise, the algorithm from \Cref{thm:efficient matching player} computes a matching $\mset'_q\subseteq A'_q\times B'_q$ with $|\mset'_q|\geq |A_q|-4z$, such that there exists a set $\pset'_q=\set{P(a,b)\mid (a,b)\in \mset'_q}$ of paths in $G$, where for each pair $(a,b)\in \mset'_q$, path $P(a,b)$ connects $a$ to $b$, and the paths in $\pset'_q$ cause congestion at most $O\left (\frac{ \log n}{\phi}\right )$.
	We let $A'_q\subseteq A_q$, $B'_q\subseteq B_q$ be the sets of vertices that do not participate in the matching $\mset'_q$, and we let $\mset''_q$ be an arbitrary matching that connects every vertex of $A'_q$ to a distinct vertex of $B'_q$ (such a matching must exist since $|A_q|\leq |B_q|$). As before, we define a set $F_q$ of fake edges, containing the edges of $\mset''_q$, and an embedding $\pset''_q=\set{P(e)\mid e\in F_q}$ of the edges in $\mset''_q$, where each fake edge is embedded into itself. Lastly, we set $\mset_q=\mset'_q\cup \mset''_q$, and we add the edges of $\mset_q$ to graph $H$.
	
	From now on we assume that the algorithm never terminated with a cut $(X,Y)$ with $|X|,|Y|\geq z$ and $|E_G(X,Y)|\le \phi\cdot\min\set{|X|,|Y|}$. Note that, from \Cref{obs: exp plus matching is exp}, the final graph $H$ is a $\psi$-expander, for $\psi\geq \frac{\phi^*}{2}\geq \Omega\left ( \frac{1}{2^{O(1/\eps^6)}\cdot \log^5 n }\right ) $. Moreover, we are guaranteed that there is an embedding of $H$ into $G+F$ with congestion $O\left (\frac{ \log^2 n}{\phi}\right )$, where $F=\bigcup_{i=1}^rF_i$ is a set of $O(z\log n)$ fake edges. Notice that, in the embedding that we constructed, every edge of $H$ is either embedded into a path consisting of a single fake edge, or it is embedded into a path in the graph $G$; every fake edge in $F$ serves as an embedding of exactly one edge of $H$.
	
	We now claim that there is a large enough universal constant $c_0$, such that, if we let $\alpha=2^{c_0/\eps^6}\cdot \log^7 n$ and $\alpha'= 2^{c_0/\eps^6}\cdot \log^6 n$, then for every cut $(X',Y')$ in $G$ with $|E_G(X',Y')\leq \frac{\phi}{\alpha}\cdot \min\set{|X'|,|Y'|}$,  $\min\set{|X'|,|Y'|}<\alpha'\cdot z$ holds.
	
	Indeed, consider any cut $(X',Y')$ in $G$ with $|X'|,|Y'|\geq \alpha'\cdot z$.  We assume w.l.o.g. that $|X'|\leq |Y'|$. It is enough to show that $|E_G(X',Y')|>\frac{\phi\cdot |X'|}{\alpha}$.

	
	Notice that $(X',Y')$ also defines a cut in graph $H$, and, since $H$ is a $\psi$-expander, $|E_{H}(X',Y')|\geq \psi\cdot |X'|\geq \psi\cdot \alpha'\cdot z\geq \psi\cdot z\cdot 2^{c_0/\eps^6}\cdot \log^6 n$.
	Since 	$\psi\geq  \frac{1}{2^{O(1/\eps^6)}\cdot \log^5 n }$, assuming that $c_0$ is a large enough constant, we get that $|E_H(X',Y')|\geq c_0z\log n$.

	 We partition the set  $E_{H}(X',Y')$ of edges into two subsets. The first subset, $E_1$, is the set of edges corresponding to the fake edges (so each edge $e\in E_1$ is embedded into a path $P(e)=\set{e}$ in $G+F$), and $E_2$ contains all remaining edges (each of which is embedded into a path of $G$). Recall that the total number of the fake edges, $|F|\leq O(z\log n)$, while $|E_H(X',Y')|\geq c_0z\log n$.
	Therefore, by letting $c_0$ be a large enough constant, we can ensure that $|E_1|\leq |E_H(X',Y')|/2$.

	The embedding of $H$ into $G+F$ defines, for every edge $e\in E_2$ a corresponding path $P(e)$ in $G$, that must contribute at least one edge to the cut $E_G(X',Y')$. Since the embedding causes congestion $O\left (\frac{ \log^2 n}{\phi}\right )$, we get that:

	\[
	\begin{split}
	|E_G(X',Y')|&\geq \Omega \left (\frac{|E_H(X',Y')|\cdot \phi}{\log^2 n}\right )\\
	& \geq \Omega \left (\frac{\phi \cdot \psi\cdot |X'|}{\log^2 n}\right )\\
	&\geq \Omega  \left (\frac{\phi\cdot |X'|}{ 2^{O(1/\eps^6)}\cdot \log^7 n }\right )\\
	&> \frac{\phi\cdot |X'|}{ 2^{c_0/\eps^6}\cdot \log^7 n }\\
	&=\frac{\phi |X'|}{\alpha},
	\end{split}
	\]
	
	(we have used the fact that $c_0$ is a large enough constant). We note that we have ignored the extra vertex $v_0$ that we have added to $G$ if $|V(G)|$ is odd, but the removal of this vertex can only change the cut sparsity and the cardinalities of $X'$ and $Y'$ by a small constant factor that can be absorbed in $c_0$.
	
	Lastly, we bound the running time of the algorithm. The algorithm consists of $O(\log n)$ iterations. Every iteration employs the algorithm from 
	\Cref{thm: new cut player}, whose running time is  $O\left(|E(H)|^{1+O(\eps)}\cdot \Delta_H^7\right )\leq O\left(n^{1+O(\eps)}\right ) $, since $\Delta_H\leq O(\log n)$, and $\log^8n<n^{4\eps}$ (the latter follows from the assumption that $\eps>\frac{2}{(\log n)^{1/25}}$, and since, from the inequality $\frac{2}{(\log n)^{1/25}}< \eps<\frac{1}{400}$, $n$ must be large enough).
Additionally, in every iteration we use Algorithm \matchorcut from \Cref{thm:efficient matching player}, whose running time is $O\left (m^{1+o(1)}\right)$. Therefore, the total running time is $O\left (m^{1+O(\eps)+o(1)}\right)$.
\end{proof}

We also immediately obtain the following analogue of Lemma 7.4 from \cite{detbalanced}.

\begin{theorem}\label{thm: sparse edge cut or expander}
	There is a constant $c_0$ and a deterministic algorithm, that, given an $n$-vertex and $m$-edge graph $G=(V,E)$ and parameters $0<\phi\leq 1$  and  $\frac{4}{(\log n)^{1/25}}< \eps<\frac{1}{400}$:	
	
	\begin{itemize}
		\item either returns a cut $(X,Y)$ in $G$ with $|E_G(X,Y)|\leq \phi\cdot\min\set{|X|,|Y|}$; 
		
		\item or correctly establishes that $G$ is a $\phi'$-expander, for $\phi'= \frac{\phi}{2^{c_0/\eps^6}\cdot \log^7 n}$.
	\end{itemize}
	
	The running time of the algorithm is $O\left (m^{1+O(\eps)+o(1)}\right )$.
\end{theorem} 

\begin{proof}
	The proof is almost identical to the proof of \Cref{thm: new sparse edge cut of large profit or witness}. The only difference is that we set the parameter $z$ that is used in the calls to Algorithm \matchorcut from \ref{thm:efficient matching player} to $1$. This ensures that no fake edges are introduced.
\end{proof}

\subsection{Sparsest Cut and Lowest-Conductance Cut -- Proof of \Cref{thm: sparsest and lowest cond}}
\label{subsec: sparsest and low-cong}

In this section we prove \Cref{thm: sparsest and lowest cond}.
 \Cref{thm: sparse edge cut or expander} immediately gives a deterministic $\left (2^{O(1/\eps^6)}\cdot \log^7n\right )$-approximation algorithm for the
Sparsest Cut problem with running time $O\left (m^{1+O(\eps)+o(1)}\right )$, for all $\eps\geq \frac{4}{(\log n)^{1/25}}$. Indeed, let $G$ be the input $m$-edge graph. For  $1\leq i\leq \ceil{\log m}$, let $\phi_i=1/2^i$. For $1\leq i\leq \ceil{\log m}$, we apply the algorithm from   \Cref{thm: sparse edge cut or expander} to graph $G$, with the parameter $\phi_i$. Let $i$ be the smallest integer, for which the algorithm returned a cut $(X,Y)$ with $|E_G(X,Y)|<\phi_i\cdot \min\set{|X|,|Y|}$. Then, when applied to $G$ with parameter $\phi_{i+1}=\phi_i/2$, the algorithm correctly established that $G$ is a $\phi'$-expander, for $\phi'=\frac{\phi_i}{2^{O(1/\eps^6)}\cdot \log^7n}$. In other words, the sparsity of the sparsest cut in $G$ is at least $\phi'$. Therefore, $(X,Y)$ is an $O(2^{O(1/\eps^6)}\cdot \log^7n)$-approximate sparsest cut. The running time of the algorithm remains $O\left (m^{1+O(\eps)+o(1)}\right )$. By setting $\eps=(1/c\log\log\log n)^{1/6}$, for a large enough constant $c$, we obtain a factor-$O(\log^7n \log \log n)$-approximation, in time  $O\left (m^{1+o(1)}\right )$.

We now show that we can obtain an algorithm with similar guarantees for the \lowcond problem. The algorithm follows easily from the algorithm for the \sparsest problem, and is identical to that of \cite{detbalanced}. The only difference is that we are using the stronger algorithm for \sparsest that we obtained. Let $G=(V,E)$ be an input to the \lowcond problem, with $|V|=n$ and $|E|=m$. Let $\frac{4}{(\log n)^{1/25}}< \eps<\frac{1}{400}$ be a parameter. We start by obtaining a factor-$\left (2^{O(1/\eps^6)}\cdot \log^7n\right )$-approximation algorithm with running time $O\left (m^{1+O(\eps)+o(1)}\right )$.

Denote by $\psi$ the conductance of the lowest-conductance cut in $G$. We can assume without loss of generality that $\psi<\frac 1 {2^{c/\eps^6}\cdot \log^7n}$ for some large enough constant $c$, since otherwise we can let $v$ be a lowest-degree vertex in $G$, and return the cut $(\set{v},V\setminus\set{v})$, whose conductance is $1$.
We use Algorithm \reducedegree from \Cref{subsec:constant degree}, in order to construct, in time $O(m)$, a graph $\hG$, whose maximum vertex degree is bounded by $10$, and $|V(\hG)|=2m$.
 
Note that, if we denote $\phi$ the value of the sparsest cut in $\hat G$, then $\phi\leq \psi$ must hold. This is since every cut $(A,B)$ in $G$ naturally defines a cut $(A',B')$ in $\hat G$, with $|A'|=\vol_G(A), |B'|=\vol_G(B)$, and $|E_{\hat G}(A',B')|=|E_G(A,B)|$. We use our approximation algorithm for the Sparsest Cut problem in graph $\hat G$, to obtain a cut $(X',Y')$ of $\hat G$, whose sparsity is at most $\left (2^{O(1/\eps^6)}\cdot \log^7n\right )\cdot \phi$, in time $O\left (m^{1+O(\eps)+o(1)}\right )$.

Using  Algorithm \makecanonical from \Cref{lem: degree reduction balanced cut case}, we obtain a cut $(X'',Y'')$ of $\hat G$, with $|X''|\geq |X'|/2$, $|Y''|\geq |Y'|/2$, and $|E_{\hat G}(X'',Y'')|\leq O(|E_{\hat G}(X',Y')|)\leq \left (2^{O(1/\eps^6)}\cdot \log^7n\right )\cdot \phi\cdot\min\set{|X'|,|Y'|}\leq \left (2^{O(1/\eps^6)}\cdot \log^7n\right )\cdot \psi\cdot \min\set{|X''|,|Y''|}$, such that both $X''$ and $Y''$ are canonical vertex sets. This cut naturally defines a cut $(X,Y)$ in $G$, with $\vol_G(X)=|X''|$, $\vol_G(Y)=|Y''|$, and $|E_G(X,Y)|=|E_{\hat G}(X'',Y'')|$. Therefore: 

\[
\begin{split}
|E_G(X,Y)|&=|E_{\hat G}(X'',Y'')|\\
&\leq  \left(2^{O(1/\eps^6)}\cdot \log^7n\right )\cdot \psi\cdot \min\set{|X''|,|Y''|}\\
&\leq \left (2^{O(1/\eps^6)}\cdot \log^7n\right )\cdot \psi\cdot\min\set{\vol_G(X),\vol_G(Y)}.\end{split} \]

We conclude that cut $(X,Y)$ is a factor $(2^{O(1/\eps^6)}\cdot \log^7n)$-approximate solution to instance $G$ of \lowcond.
Since the running time of Algorithm \makecanonical is $O(m)$, the running time of the whole algorithm remains 
bounded by $O\left (m^{1+O(\eps)+o(1)}\right )$. As before, by setting $\eps=(1/c\log\log\log n)^{1/6}$, for a large enough constant $c$, we obtain a factor-$O(\log^7n \log \log n)$-approximation for \lowcond, in time  $O\left (m^{1+o(1)}\right )$.

\subsection{Minimum Balanced Cut -- Proof of Theorems  \ref{thm: balanced cut high cond} and \ref{thm: balanced cut low cond}}

In this section we prove  \Cref{thm: balanced cut high cond} and \Cref{thm: balanced cut low cond}. Our proof is somewhat more involved than that of \cite{detbalanced}, who iteratively used the algorithm for \MBSC from (the weaker version of) \Cref{thm: new sparse edge cut of large profit or witness}.  The reason is that, while our bounds for the parameters $\alpha$ and $\alpha'$ are better than those obtained by \Cref{thm: new sparse edge cut of large profit or witness}, they are still super-logarithmic. Therefore, if we follow the framework of \cite{detbalanced}, who apply the algorithm from \Cref{thm: new sparse edge cut of large profit or witness} over the course of $O(1/\eps)$ iterations, we will still accumulate an approximation factor that is at least as high as $(\log n)^{\Theta(1/\eps)}$. 

The proofs of \Cref{thm: balanced cut high cond} and \Cref{thm: balanced cut low cond}  are very similar to each other, and we start with presenting the part that is common to both proofs.
Recall that we are given as input a graph $G$ with $|V(G)|=n$, $|E(G)|=m$ and a parameter $0<\psi\leq 1$.
We can assume w.l.o.g. that both $m$ and $n$ are greater than a large enough constant, since otherwise we can use the algorithm of \cite{detbalanced}, whose running time can now be bounded by $O(m)$. 

We use a parameter $\eps=1/(\log\log\log m)^{1/25}$. It is easy to verify that $m^{\eps}=2^{O(\log n/(\log\log \log n)^{1/25})}>\log^4 n$ must hold. 
We can also assume that $\psi<1/(2^{c/\eps^6}\log^8n)$ for a large enough constant $c$, since otherwise we can compute an arbitrary partition $(A,B)$ of $V(G)$ with $\vol_G(A),\vol_G(A)\geq \vol(G)/3$ (since for every vertex $v\in V(G)$, $\deg_G(v)<\vol(G)/2$, such a partition can be computed by a simple greedy algorithm, that iteratively adds each vertex to a set of $\set{A,B}$, whose current volume is smaller). Clearly, $|E_G(A,B)|\leq m\leq \frac{\vol(G)}{2}\leq \psi\cdot 2^{c/\eps^6} \cdot (\log^8n)\cdot \vol(G)\leq  \psi\cdot (\log n)^{8+o(1)}\cdot \vol(G)$. Therefore, from now on we assume that $\psi<1/(2^{c/\eps^6}\log^8n)$ for a large enough constant $c$.

As in the algorithm of \cite{detbalanced},
we start by applying Algorithm \reducedegree from \Cref{subsec:constant degree} to graph $G$, in order to construct, in time $O(m)$, a graph $\hG$ whose maximum vertex degree is bounded by $10$, and $|V(\hG)|=2m$. Denote $V(G)=\set{v_1,\ldots,v_n}$. Recall that graph $\hG$ is constructed from graph $G$ by replacing each vertex $v_i$ with an $\alpha_0$-expander $H(v_i)$ on $\deg_G(v_i)$ vertices, where $\alpha_0=\Theta(1)$. For convenience, we denote the set of vertices of $H(v_i)$ by $V_i$. Therefore, $V(\hG)=V_1\cup V_2\cup \cdots\cup V_n$.  Denote $|V(\hat G)|=\hat n$. Consider now some subset $S$ of vertices of $\hG$. As before, we say that $S$ is a \emph{canonical} set of vertices if, for all $1\leq i\leq n$, either $V_i\subseteq S$ or $V_i\cap S=\emptyset$ holds. 
Our starting point is the following lemma, that is an easy application of the \CMG, combined with Algorithm \makecanonical from \Cref{lem: degree reduction balanced cut case}. The proof is included in Section \ref{sec: cutting lemma} of Appendix.

\begin{lemma}\label{lem: single stage new}
	There is a deterministic algorithm, whose input consists of a canonical set $V'\subseteq V(\hat G)$ of vertices of $\hat G$, with $|V'|\geq 2\hat n/3$, 
	and parameters $0<\phi<1$ and $\rho\geq \log n$. The algorithm
	computes one of the following:
		
	\begin{itemize}
		\item either a partition $(X,Y)$ of $V'$, where both $X,Y$ are canonical subsets of $V(\hat G)$, $|X|,|Y|\geq \rho$, and $|E_{\hG}(X,Y)|\leq \phi \cdot \min\set{|X|,|Y|}$;

	\item or a $\phi^*$-expander graph $H$ with $V(H)=V'$ and maximum vertex degree $O(\log n)$, where $\phi^*\geq \Omega\left ( \frac{1}{2^{O(1/\eps^6)}\cdot \log^5 n }\right )$, together with a set $F$ of at most $O(\rho\cdot \log n)$ edges of $H$, such that there exists an embedding $\pset$ of $H\setminus F$ into $\hat G[V']$ with congestion at most $O\left (\frac{ \log^2 n}{\phi}\right )$. 
		\end{itemize}
	
	The running time of the algorithm is $O\left(m^{1+O(\eps)+o(1)}\right )$.
\end{lemma}

We emphasize that the algorithm from the above lemma does not compute the embedding $\pset$ explicitly, and instead it only guarantees its existence.
We obtain the following immediate corollary of \Cref{lem: single stage new}.
The corollary uses the parameter $\phi^*\geq \Omega\left ( \frac{1}{2^{O(1/\eps^6)}\cdot \log^5 n }\right )$ from \Cref{lem: single stage new}.




\begin{corollary}\label{cor: balancd cut or near-expander}
	There is a deterministic algorithm, that, given a parameter $\rho\geq \log n$, and a parameter $c'>1$, computes a cut $(X^*,Y^*)$ in graph $\hat G$, such that both $X^*$ and $Y^*$ are canonical sets of vertices, $|X^*|\geq |Y^*|$, and $|E_{\hat G}(X^*,Y^*)|\leq \frac{c'\psi\log^3n}{\phi^*}\cdot \hat n$. Moreover, if $|Y^*|<\hat n/3$, then the algorithm also computes a $\phi^*$-expander graph $H$ with $V(H)=X^*$ and maximum vertex degree $O(\log n)$, together with a set $F$ of at most $O(\rho \log n)$ edges of $H$, such that there exists an embedding $\pset$ of $H\setminus F$ into $\hat G[X^*]$ with congestion $\eta\leq O\left(\frac{\phi^*}{c'\psi\log n}\right )$.
	The running time of the algorithm is $O\left(\frac{m^{2+O(\eps)+o(1)}\cdot }{\rho}\right )$.
\end{corollary}

\begin{proof}
	The proof easily follows by iteratively applying the algorithm from  \Cref{lem: single stage new}.
	Let  $\phi= \frac{c'\psi\log^3n}{\phi^*}$. Our algorithm consists of at most $r=\frac{\hat n}{\rho}$ iterations. For all $1\leq i\leq r$, at the beginning of iteration $i$, we are given  a partition $(X_i,Y_i)$ of $V(\hat G)$, such that both $X_i,Y_i$ are canonical sets, $|X_i|\geq \frac{2\hat n}3$, and $|E_{\hat G}(X_i,Y_i)|\leq \phi\cdot |Y_i|$. At the beginning of the algorithm, we use the partition $(X_1,Y_1)$ of $V(\hat G)$, with $X_1=V(\hat G)$ and $Y_1=\emptyset$. We now describe the execution of the $i$th iteration.
	
	In order to execute the $i$th iteration, we apply the algorithm from \Cref{lem: single stage new} to set $V'=X_i$ of vertices of $\hat G$, and parameters $\rho$ and $\phi$, that remain unchanged. We now consider two cases. 
	
	Assume first, that the algorithm returned a partition $(X',Y')$ of $X_i$, with $|X'|,|Y'|\geq \rho$, and $|E_{\hG}(X',Y')|\leq \phi \cdot \min\set{|X'|,|Y'|}$, such that both $X'$ and $Y'$ are canonical sets of vertices. Assume w.l.o.g. that $|X'|\geq |Y'|$. We then construct a new cut $(X_{i+1},Y_{i+1})$ in $\hat G$, by letting $X_{i+1}=X'$ and $Y_{i+1}=V(\hat G)\setminus X'=Y'\cup Y_i$. Note that cut  $(X_{i+1},Y_{i+1})$ can be obtained from cut  $(X_{i},Y_{i})$ by moving the vertices of $Y'$ from $X_{i}$ to $Y_i$. Therefore,  $E_{\hat G}(X_{i+1},Y_{i+1})\subseteq E_{\hat G}(X_{i},Y_{i})\cup E_{\hat G}(X',Y')$, and:
	
	\[|E_{\hat G}(X_{i+1},Y_{i+1})|\leq  |E_{\hat G}(X_{i},Y_{i})|+|E_{\hat G}(X',Y')|\leq \phi\cdot |Y_i|+\phi\cdot |Y'|\leq \phi \cdot |Y_{i+1}|. \]

	If $|X_{i+1}|\geq 2\hat n/3$ holds, then we continue to the next iteration. Otherwise, since $|X_i|\geq 2\hat n/3$, we get that $|X_{i+1}|\geq |X_i|/2\geq \hat n/3$. Therefore, we obtain a partition $(X_{i+1},Y_{i+1})$ of $V(\hat G)$ with $|X_{i+1}|,|Y_{i+1}|\geq \hat n/3$, and $|E_{\hat G}(X_{i+1},Y_{i+1})|\leq \phi \cdot |Y_{i+1}|\leq \phi\cdot \hat n =\frac{c'\psi\log^3n}{\phi^*}\cdot\hat n$. We then return cut $(X^*,Y^*)=(X_{i+1},Y_{i+1})$ and terminate the algorithm. 


Next, we assume that the algorithm from \Cref{lem: single stage new} returned 
a $\phi^*$-expander graph $H$ with $V(H)=X_i$ and maximum vertex degree $O(\log n)$, together with a set $F$ of at most $O(\rho\cdot \log n)$ edges of $H$, such that there exists an embedding $\pset$ of $H\setminus F$ into $\hat G[X_i]$ with congestion at most $O\left (\frac{ \log^2 n}{\phi}\right )\leq  O\left(\frac{\phi^*}{c'\psi\log n}\right )$. 
In this case, we return the cut $(X^*,Y^*)=(X_i,Y_i)$, graph $H$, and set $F$ of edges. 

It now remains to bound the running time of the algorithm. Notice that in every iteration, the cardinality of the set $Y_i$ of vertices grows by at least  $\rho$, and, since $|V(\hat G)|=2m$, the number of iterations is bounded by $2m/\rho$. In each iteration we use the algorithm from \Cref{lem: single stage new}, whose running time is at most $O\left(m^{1+O(\eps)+o(1)}\right )$. Overall, the running time of the algorithm is bounded by $O\left(\frac{m^{2+O(\eps)+o(1)}}{\rho}\right )$.
\end{proof}

We are now ready to complete the proofs of \Cref{thm: balanced cut high cond} and \Cref{thm: balanced cut low cond}.

\subsubsection{Proof of \Cref{thm: balanced cut high cond}}
We first consider a special case, where $\psi<\frac{\log^4m}{m}$. In this case, our algorithm simply repeatedly cuts off low-conductance cuts from graph $G$. Specifically, we perform a number of iterations, and we maintain a vertex-induced subgraph $G'\subseteq G$. We also maintain a cut $(A,B)$ in $G$, with $A=V(G')$ and $B=V(G)\setminus V(G')$. Initially, we set $G'=G$, $A=V(G)$, and $B=\emptyset$. The algorithm continues as long as $\vol_G(A)\geq 2\vol(G)/3$. 
Let $\alpha=O(\log^7n\log\log n)$ be the approximation factor that the algorithm for the \lowcond problem from \Cref{thm: sparsest and lowest cond} achieves.
Throughout the algorithm, we will ensure that $|E_G(A,B)|\leq \alpha\cdot \psi\cdot \vol_G(B)$ holds.

In order to execute a single iteration, we apply the approximation algorithm for the \lowcond problem from \Cref{thm: sparsest and lowest cond} to the current graph $G'$. Let $(X,Y)$ be the cut that the algorithm returns, and let $\psi'=\frac{|E_{G'}(X,Y)|}{\min\set{\vol_{G'}(X),\vol_{G'}(Y)}}$ be the conductance of the cut. Assume first that $\psi'> \alpha\psi$. Then we are guaranteed that graph $G'$ has conductance at least $\psi$. We return the cut $(A,B)$ and terminate the algorithm. Since $|E_G(A,B)|\leq \alpha\cdot \psi\cdot \vol_G(B)\leq \psi\cdot (\log n)^{7+o(1)}\cdot \vol(G)$, this is a valid output for the algorithm.

Assume now that $\psi'\leq \alpha\psi$, and assume w.l.o.g. that $\vol_G(X)\geq\vol_G(Y)$. Then $|E_{G'}(X,Y)|\leq \psi'\cdot \vol_{G'}(Y)\leq \alpha\psi\vol_G(Y)$. Consider a new cut $(A',B')$ in $G$, where $A'=X$ and $B'=V(G)\setminus X=B\cup Y$. It is then easy to verify that:

\[|E_G(A',B')|\leq |E_G(A,B)|+|E_{G'}(X,Y)|\le \alpha\psi\vol_G(A)+\alpha\psi\vol_G(Y)\leq \alpha\psi\vol_G(B').\]

If $\vol_G(A')=\vol_G(X)\geq 2\vol(G)/3$, then we replace cut $(A,B)$ with cut $(A',B')$, and continue to the next iteration. Otherwise, $\vol_G(A')\geq \vol_G(A)/2\geq \vol(G)/3$ must hold, and $\vol_G(B')\geq \vol(G)-\vol_G(A')\geq \vol(G)/3$ holds as well. We return cut $(A',B')$ and terminate the algorithm. As before, since $|E_G(A',B')|\leq \alpha\cdot \psi\cdot \vol_G(B')\leq \psi\cdot (\log n)^{7+o(1)}\cdot \vol(G)$, this is a valid output of the algorithm.

It now remains to analyze the running time of the algorithm. The number of iterations in the algorithm is bounded by $O(m)$, and the running time of a single iteration is dominated by the running time of the algorithm from 
\Cref{thm: sparsest and lowest cond}, which is bounded by $O\left (m^{1+o(1)}\right )$. Therefore, the running time of the algorithm is bounded by $O\left (m^{2+o(1)}\right )\leq O\left (m^{1+o(1)}/\psi\right )$, since we have assumed that $\psi<\frac{\log^4m}{m}$. 

In the remainder of the proof, we assume that $\psi\geq \frac{\log^4m}{m}$.
We start by computing a graph $\hat G$ exactly as described above, and then applying the algorithm from \Cref{cor: balancd cut or near-expander} to it, with parameter $\rho=\frac{\psi m}{\tilde c^2\log n}$, where $\tilde c$ is a large constant, whose value we set later, and parameter $c'$, that is a large enough constant, whose value we set later. Since $\psi\geq \frac{\log^4m}{m}$, and $m$ is large enough, we get that $\rho>\log n$. Let $(X^*,Y^*)$ be the outcome of the algorithm.

Recall that both $X^*$ and $Y^*$ are canonical sets of vertices, $|X^*|\geq |Y^*|$, and  $|E_{\hat G}(X^*,Y^*)|\leq \frac{c'\psi\log^3n}{\phi^*}\cdot \hat n$. Notice that cut $(X^*,Y^*)$ in $\hat G$ naturally defines a cut $(A,B)$ in $G$: for every vertex $v_i\in V(G)$, we add $v_i$ to $A$ if $V_i\subseteq X^*$, and we add it to $B$ otherwise. It is immediate to verify that $\vol_G(A)=|X^*|$, $\vol_G(B)=|Y^*|$, and $|E_G(A,B)|=|E_{\hat G}(X^*,Y^*)|\leq \frac{c'\psi\log^3n}{\phi^*}\cdot \hat n\leq \psi\cdot 2^{O(1/\eps^6)} \cdot (\log^8n)\cdot \vol(G)\leq \psi \cdot (\log n)^{8+o(1)}\cdot \vol(G)$, since $\phi^*\geq \Omega\left ( \frac{1}{2^{O(1/\eps^6)}\cdot \log^5 n }\right )$ and $\eps=1/(\log\log\log m)^{1/25}$.

Consider first the case that $|Y^*|\geq \hat n/3$. Then, from the above discussion, $\vol_G(A),\vol_G(B)\geq \hat n/3=\vol(G)/3$. In this case,
we return cut $(A,B)$ as the outcome of the algorithm. 

We assume from now on that $|Y^*|<\hat n/3$, and so the algorithm from \Cref{cor: balancd cut or near-expander} computed a $\phi^*$-expander graph $H$ with $V(H)=X^*$ and maximum vertex degree $O(\log n)$, together with a set $F$ of at most $O(\rho\log n)$ edges of $H$, such that 
there exists an embedding $\pset$ of $H\setminus F$ into $\hat G[X^*]$ with congestion at most $\eta$, where $\eta\leq O\left(\frac{\phi^*}{c'\psi\log n}\right )$. 
Since $\rho=\frac{\psi m}{\tilde c^2\log n}$, and $\tilde c$ is a sufficiently large constant, we get that $|F|\leq  \frac{\psi m}{\tilde c}$.
Let $\hat G'$ be the graph that is obtained from $\hat G[X^*]$, by adding the edges of $F$ to it. We need the following simple observation.

\begin{observation}\label{obs: we get expander}
	Graph $\hat G'$ is a $\phi'$-expander, for $\phi'= \Omega(c'\psi \log n)$.
\end{observation}

\begin{proof}
	Consider any partition $(A',B')$ of $V(\hat G')$, with $|A'|\leq |B'|$. It is enough to prove that $|E_{\hat G'}(A',B')|\geq \phi'\cdot |A'|$. Recall that the set $\pset$ of paths defines an embedding of $H\setminus F$ into $\hat G[X^*]$ with congestion at most $\eta$. By embedding every edge of $F$ into itself, we can augment the set $\pset$ of paths to obtain an embedding of $H$ into $\hat G'=\hat G[X^*]\cup F$, with congestion at most $\eta$. 
	
	Notice that cut $(A',B')$ in $\hat G'$ also defines a cut in graph $H$. Since graph $H$ is a $\phi^*$-expander, if we denote by $E'=E_{H}(A',B')$, then $|E'|\geq \phi^*\cdot |A'|$. Consider now the set $\pset'=\set{P(e)\mid e\in E'}$ of paths, where for each edge $e\in E'$, $P(e)$ is the path of $\pset$ that serves as an embedding of the edge $e$. Then every path in $\pset'$ must contain an edge of $E_{\hat G'}(A',B')$, and the paths in $\pset'$ cause congestion at most $\eta$. Therefore, $|E_{\hat G'}(A',B')|\geq \frac{|E'|}{\eta}\geq \frac{\phi^*\cdot |A'|}{\eta}\geq \Omega(c'\psi |A'|\log n)=\phi'\cdot |A'|$, since $\eta \leq O\left(\frac{\phi^*}{c'\psi\log n}\right )$.
%
\end{proof}

Recall that we have used the cut $(X^*,Y^*)$ in $\hat G$ to define a cut $(A,B)$ in $G$, with $\vol_G(A)=|X^*|$, $\vol_G(B)=|Y^*|$, and  $|E_G(A,B)|=|E_{\hat G}(X^*,Y^*)|\leq \psi\cdot (\log n)^{8+o(1)}\cdot \vol(G)$. Consider now a graph $G'$, that, intuitively, is obtained from $G[A]$, by adding the edges of $F$ to it. Formally, in order to obtain graph $G'$, we start from graph $G[A]$, and the consider the edges $e\in F$ one by one. Let $e=(x,y)$ be any such edge, and let $v_i,v_j\in A$ be the vertices with $x\in V_i$ and $y\in V_j$. If $i\neq j$, then we add edge $e'=(v_i,v_j)$ to graph $G'$, and we think of $e'$ as a \emph{copy of edge $e$}. Notice that graph $G'$ can be equivalently obtained from graph $\hat G'$ by contracting, for every vertex $v_i\in A$, all vertices of $V_i$ into a single node. 
Next, we use the following easy observation:

\begin{observation}\label{obs: contracted graph high conductance}
 	Graph $G'$ has conductance at least $24\psi$.
 \end{observation}
\begin{proof}
	Consider any cut $(X,Y)$ in $G'$. We can naturally define a corresponding cut $(\hat X,\hat Y)$ in graph $\hat G'$: for every vertex $v_i\in X$, we add all vertices of $V_i$ to $\hat X$, and for every vertex $v_i\in Y$, we add all vertices of $V_i$ to $\hat Y$. It is easy to verify that $|\hat X|=\vol_G(X)$. Since every vertex of $X^*$ is incident to at most $c^*\log n$ edges of $F$, for some constant $c^*$ (as maximum vertex degree in $H$ is $O(\log n)$), it is easy to verify that, for every vertex $v_i\in A$, $\deg_{G'}(v_i)\leq (c^*\log n)\deg_G(v_i)$. Therefore, $\vol_{G'}(X)\leq (c^*\log n)\vol_G(X)$, and so $|\hat X|\geq \frac{\vol_{G'}(X)}{c^*\log n}$. Using similar reasoning, $|\hat Y|\geq \frac{\vol_{G'}(Y)}{c^*\log n}$. Lastly, since, from \Cref{obs: we get expander},
	graph $\hat G'$ is a $\phi'$-expander, for $\phi'=\Omega(c'\psi \log n)$, we get that $|E_{G'}(X,Y)|=|E_{\hat G'}(\hat X,\hat Y)|\geq \phi'\cdot \min\set{|\hat X|,|\hat Y|}\geq  \Omega\left(\frac{c'\psi}{c^*}\right)\cdot \min\set{\vol_{G'}(X),\vol_{G'}(Y)}$. Since $c^*$ is a fixed constant, and we can set $c'$ to be a large enough constant, we get that $|E_{G'}(X,Y)|\geq 24\psi\cdot\min\set{\vol_{G'}(X),\vol_{G'}(Y)}$.  We conclude that the conductance of graph $G'$ is at least $24\psi$.
\end{proof}

Next, we use the following pruning theorem of \cite{expander-pruning}, that works with graph conductance instead of expansion.

\begin{theorem}[Restatement of Theorem 1.3 from~\cite{expander-pruning}]\label{thm: expander pruning2}
	There is a deterministic algorithm, that, given access to adjacency list of a graph $G=(V,E)$ that has conductance $\psi$, for some $0<\psi\leq 1$, and a collection $E'\subseteq E$ of $ {k\leq \psi |E|/10}$ edges of $G$, computes a subgraph $G'\subseteq G\setminus E'$, that has conductance at least $ \psi/6$. Moreover, if we denote $A=V(G')$ and $B=V(G)\setminus A$, then $|E_G(A,B)| \leq 4k$, and $\vol_G(B)\leq 8k/\psi$.
	The running time of the algorithm is $ {\tilde O(k/\psi^2)}$.
\end{theorem}

We apply the algorithm from \Cref{thm: expander pruning2} to graph $G'$, conductance parameter $24\psi$, and the set $E'=F$ of edges. 
Since we have assumed that $|Y^*|<\hat n/3=\vol(G)/3$, we get that $|X^*|\geq 2\hat n/3\geq 2\vol(G)/3$. At the same time, $|E_G(X^*,Y^*)|<\psi\cdot 2^{O(1/\eps^6)} \cdot (\log^8n)\cdot \vol(G)<\vol(G)/10$, since we have assumed that $\psi<\frac{1}{2^{c/\eps^6} \cdot (\log^8n)}$ for a large enough constant $c$. Therefore, $|E(G')|\geq \frac{\vol_G(X^*)-|E_G(X^*,Y^*)|}{2}\geq \frac{\vol(G)}{4}\geq \frac{m}{2}$.
 Recall that $|F|\leq \frac{\psi m}{\tilde c}$ holds, where $\tilde c$ is a large enough constant. Therefore, $|F|\leq  \frac{24\psi |E(G')|}{10}$.

 The algorithm must then return a partition $(Z,Z')$ of $V(G')=A$, such that graph $G[Z]$ has conductance at least $\frac{24\psi}{6}\geq \psi$, and $\vol_{G'}(Z')\leq \frac{8|F|}{24\psi}\leq \frac{m}{3\tilde c}$, while $|E_{G}(Z,Z')|\leq 4|F|\leq \frac{4\psi m}{\tilde c}\leq O(\psi)\cdot \vol(G)$.

 We construct a cut $(A^*,B^*)$ in graph $G$, by letting $A^*=Z$ and $B^*=Z'\cup B$. 
Notice that $|E_G(A^*,B^*)|\subseteq |E_G(A,B)|+|E_G(Z,Z')|\leq \psi \cdot (\log n)^{8+o(1)}\cdot \vol(G)$. Additionally, we get that:

\[\vol_G(A^*)=\vol_G(Z)\geq \vol_G(A)-\vol_G(Z')\geq \frac{2\vol(G)}{3}-\frac{m}{3\tilde c}\geq \frac{7\vol(G)}{12}.  \]

If $\vol_G(B^*)\geq \vol(G)/3$, then we get that $\vol_G(A^*),\vol_G(B^*)\geq \vol(G)/3$. Otherwise, we are guaranteed that $\vol_G(A^*)\geq 2\vol(G)/3$, and that the conductance of graph $G[A^*]$ is at least $\psi$.

It now remains to bound the running time of the algorithm. The running time of the algorithm from \Cref{cor: balancd cut or near-expander} is bounded by 
$O\left(\frac{m^{2+O(\eps)+o(1)}}{\rho}\right )\leq  O\left(\frac{m^{1+o(1)}}{\psi}\right )$, since $\rho=\frac{\psi m}{\tilde c^2\log n}$ and $\eps=1/(\log\log\log m)^{1/25}$. The running time of the algorithm from \Cref{thm: expander pruning2} is $\tilde O\left(\frac{|F|}{\psi^2}\right )\leq \tilde O\left(\frac{m}{\psi}\right )$. Therefore, the total running time of the algorithm is bounded by $O\left(\frac{m^{1+o(1)}}{\psi}\right )$.

\subsubsection{Proof of \Cref{thm: balanced cut low cond}}


The proof is very similar to the proof of  \Cref{thm: balanced cut high cond}. The main difference is that we set $\rho=m^{1-2\eps}$, and that we do not employ
\Cref{thm: expander pruning2}.

We start by computing a graph $\hat G$ exactly as before, and then applying the algorithm from \Cref{cor: balancd cut or near-expander} to it, with parameter $\rho=m^{1-2\eps}$, and $c'$ a large constant whose value we set later. Let $(X^*,Y^*)$ be the outcome of the algorithm.
Recall that both $X^*$ and $Y^*$ are canonical sets of vertices, $|X^*|\geq |Y^*|$, and $|E_{\hat G}(X^*,Y^*)|\leq \frac{c'\psi\log^3n}{\phi^*}\cdot \hat n\leq (\log n)^{8+o(1)}\cdot \vol(G)$, since $\phi^*\geq \Omega\left ( \frac{1}{2^{O(1/\eps^6)}\cdot \log^5 n }\right )$ and $\eps=\frac{1}{(\log\log\log n)^{1/25}}$. We use cut $(X^*,Y^*)$ in $\hat G$ in order to define a cut $(A,B)$ in $G$ exactly as before. As before, $\vol_G(A)=|X^*|$, $\vol_G(B)=|Y^*|$, and $|E_G(A,B)|=|E_{\hat G}(X^*,Y^*)|\leq  \psi \cdot (\log n)^{8+o(1)}\cdot \vol(G)$.
As before, if $|Y^*|\geq \hat n/3$, then  $\vol_G(A),\vol_G(B)\geq \hat n/3=\vol(G)/3$, and
we return cut $(A,B)$ as the outcome of the algorithm.

We assume from now on that $|Y^*|<\hat n/3$, and so the algorithm from \Cref{cor: balancd cut or near-expander} computed a $\phi^*$-expander graph $H$ with $V(H)=X^*$ and maximum vertex degree $O(\log n)$, together with a set $F$ of at most $O(\rho\log n)$ edges of $H$, such that 
there exists an embedding $\pset$ of $H\setminus F$ into $\hat G[X^*]$ with congestion at most $\eta$, where $\eta\leq O\left(\frac{\phi^*}{c'\psi\log n}\right )$. 
Notice that in this case, $\vol_G(A)\geq |X^*|\geq 2\hat n/3=2\vol(G)/3$ holds.

Consider any partition $(Z,Z')$ of $A$, with   $|E_G(Z,Z')|< \psi\cdot \vol(G)$, and assume w.l.o.g. that $\vol_G(Z)\leq \vol_G(Z')$. We claim that $\vol_G(Z)< \vol(G)/100$ holds. Indeed, assume otherwise. Consider a cut $(\hat Z,\hat Z')$ in graph $H$, obtained as follows: for every vertex $v_i\in Z$, we add all vertices of $V_i$ to $\hat Z$, and for every vertex $v_j\in Z'$, we add all vertices of $V_j$ to $\hat Z'$. Clearly, $|\hat Z|=\vol_G(Z)\geq \vol(G)/100$, and similarly $|\hat Z'|\geq\vol(G)/100$. Since graph $H$ is a $\phi^*$-expander, $|E_H(\hat Z,\hat Z')|\geq \phi^*\cdot\min\set{|\hat Z|,|\hat Z'|}\geq \frac{\phi^*\cdot \vol(G)}{100}$. Denote $E'=E_H(\hat Z,\hat Z')$.
Recall that $\phi^*\geq \Omega\left ( \frac{1}{2^{O(1/\eps^6)}\cdot \log^5 n }\right )$, and so $|E'|\geq \Omega\left ( \frac{m}{2^{O(1/\eps^6)}\cdot \log^5 n }\right )$. On the other hand, $|F|\leq O(\rho\log n)\leq O(m^{1-2\eps}\cdot \log n)$. From our choice of $\eps=\frac{1}{(\log\log\log n)^{1/25}}$, we get that $|F|<|E'|/2$, and so $|E'\setminus F|\geq   \frac{\phi^*\cdot \vol(G)}{200}$. Since there exists an 
 embedding of $H\setminus F$ into $\hat G[X^*]$ with congestion at most $\eta$, and $\eta\leq O\left(\frac{\phi^*}{c'\psi\log n}\right )$, we get that: 
 
 \[|E_{\hat G}(\hat Z,\hat Z')|\geq \frac{|E'|}{2\eta}\geq \Omega(c'\psi \vol(G)\log n).\]

Since $|E_G(Z,Z')|=|E_{\hat G}(\hat Z,\hat Z')|$, we reach a contradiction to our assumption that $|E_G(Z,Z')|< \psi\cdot \vol(G)$. We conclude that for every partition $(Z,Z')$ of $A$ with $\vol_G(Z),\vol_G(Z')\geq \vol(G)/100$, $|E_G(Z,Z')|\geq \psi\cdot \vol(G)$ must hold.

The running time of the algorithm is asymptotically bounded by the running time of the algorithm from \Cref{cor: balancd cut or near-expander}, which is in turn bounded by 
$O\left(\frac{m^{2+O(\eps)+o(1)}}{\rho}\right )\leq  O\left(m^{1+O(\eps)+o(1)}\right )\leq O(m^{1+o(1)})$, since $\rho=m^{1-2\eps}$ and $\eps=1/(\log\log\log m)^{1/25}$.

\subsection{Expander Decomposition -- Proof of \Cref{thm:expander decomp}}

In this section we prove \Cref{thm:expander decomp}. The proof is essentially identical to the proof of Corollary 6.1 from \cite{detbalanced}. The only difference is that we use our algorithm from \Cref{thm: balanced cut high cond} instead of the algorithm of \cite{detbalanced}.

We maintain a collection $\hset$ of disjoint vertex-induced subgraphs of $G$ that we call \emph{clusters}, which is partitioned into two subsets, set $\hset^A$ of \emph{active clusters}, and set $\hset^I$ of \emph{inactive clusters}. We ensure that for every inactive cluster $H\in \hset^I$, the conductance of $H$ is at least $\psi$. We also maintain a set $E'$ of ``deleted'' edges, that are not contained in any cluster in $\hset$. At the beginning of the algorithm, we let $\hset=\hset^A=\set{G}$, $\hset^I=\emptyset$, and $E'=\emptyset$. The algorithm proceeds as long $\hset^A\neq \emptyset$, and consists of iterations. 
	For convenience, we denote $\alpha=(\log n)^{8+o(1)}$, so that the algorithm from \Cref{thm: balanced cut high cond}, when applied to an $n$-vertex graph $G$ and some parameter $\psi'$, is guaranteed to return a cut $(A,B)$ in $G$ with  $|E_G(A,B)|\leq \psi'\cdot \alpha \cdot \vol(G)$.
	We set $\psi=\frac{\delta}{c\alpha\cdot \log n}$, where $c$ is the constant from the theorem statement. Clearly,  $\psi=\Omega\left (\frac{\delta}{(\log n)^{9+o(1)}}\right )$.
	
	In every iteration, we apply the algorithm from \Cref{thm: balanced cut high cond} to every graph $H\in \hset^A$, with the parameter $\psi$. Consider the cut $(A,B)$ in $H$ that the algorithm returns, with $|E_H(A,B)|\leq \alpha \psi\cdot \vol(H)\leq \frac{\delta \cdot \vol(H)}{c \log n}$. We add the edges of $E_H(A,B)$ to set $E'$. If $\vol_{H}(A),\vol_H(B)\ge \vol(H)/3$, then we replace $H$ with $H[A]$ and $H[B]$ in $\hset$ and in $\hset^A$. Otherwise, we are guaranteed that 
	$\vol_H(A)\geq 2\vol(H)/3$, and graph $H[A]$ has conductance at least $\psi$. Then we remove $H$ from $\hset$ and $\hset^A$, add $H[A]$ to $\hset$ and $\hset^I$, and add $H[B]$ to $\hset$ and $\hset^A$.

	When the algorithm terminates, $\hset^A=\emptyset$, and so every graph in $\hset$ has conductance at least $\psi$. Notice that in every iteration, the maximum volume of a graph in $\hset^A$ must decrease by a constant factor. Therefore, the number of iterations is bounded by $O(\log m)$. It is easy to verify that the number of edges added to set $E'$ in every iteration is at most $\alpha\cdot \psi\cdot \vol(G)\leq \frac{\delta\cdot \vol(G)}{c\log m}$. Therefore, by letting $c$ be a large enough constant, we can ensure that $|E'|\leq \delta\cdot \vol(G)$. The output of the algorithm is the partition $\Pi=\set{V(H)\mid H\in \hset}$ of $V$. From the above discussion, we obtain a valid $(\delta, \psi)$-expander decomposition, for $\psi=\Omega\left (\frac{\delta}{(\log m)^{9+o(1)}}\right )$.

	It remains to analyze the running time of the algorithm.  The running time of a single iteration is bounded by $O(m^{1+o(1)}/\psi)$, and, since the number of iterations is $O(\log m)$, the total running time of the algorithm is bounded by $O(m^{1+o(1)}/\psi)\leq O(m^{1+o(1)}/\delta) $.

\newpage
\appendix

\section{Proof of \Cref{lem: distancing to sparse cut}}
\label{sec:ball growning}

The proof uses a standard ball-growing technique. 
Let $H=G\setminus E'$.
Let $S$ be any set of vertices of $H$, such that not all vertices of $S$ are isolated in $H$. We define $L_0=S$, and, for integers $i>0$, we define $L_i=B_H(S,i)$. We refer to the sets $L_0,L_1,\ldots$ of vertices as \emph{BFS layers} defined with respect to $S$. For an index $1\leq i<\floor{d/2}$, we say that layer $L_i$ is \emph{acceptable} if $|\delta_H(L_i)|<\frac{\phi}{2}\cdot |E_H(L_i)|$. We use the following simple claim:

\begin{claim}\label{claim: ball growing calculation}
	Let $B=B_H(S,d)$, and assume that $|E_H(B)|<|E(H)|$. Then there is an index $1\leq i<\floor{d/2}-1$, such that layer $L_i$ is acceptable.
\end{claim}

\begin{proof}
	Assume otherwise. Then for all $1\leq i<\floor{d/2}-1$, layer $L_i$ is not acceptable, and so $|\delta_H(L_i)|\geq \frac{\phi}{2}\cdot |E_H(L_i)|$. Therefore, $|E_H(L_{i+1})|\geq \left (1+\frac{\phi}{2}\right )|E_H(L_i)|$. Since not all vertices of $S$ are isolated in $H$, we get that $|E_H(L_1)|\geq 1$. Overall, we get that:

	\[|E_H(L_{\floor{d/2}-1})|\geq \left (1+\frac{\phi}{2}\right )^{\floor{d/2}-2}\geq  \left (1+\frac{\phi}{2}\right )^{(8\log m)/\phi}\geq e^{2\log m}>m,  \] 

(we have used the fact that for all $k>1$, $\left(1+\frac 1 k\right )^{k+1}>e$). This is impossible, so there must be an index $1\leq i<\floor{d/2}-1$ for which layer $L_i$ is acceptable.
\end{proof}

We are now ready to complete the proof of \Cref{lem: distancing to sparse cut}. 
We denote the BFS layers in graph $H$ defined with respect to $X$ by $L_0',L_1',\ldots$, and BFS layers defined with respect to $Y$ by $L_0'',L_1'',\ldots$. We 
run two algorithms in parallel. The first algorithm performs a BFS search in $H$ starting from $X$ to compute layers $L_0',L_1',\ldots$ one by one. When a layer $L_i'$ is computed, the algorithm checks whether it is acceptable. The second algorithm similarly performs a BFS search starting from $Y$ to compute layers $L_0'',L_1'',\ldots$ one by one, and, when a layer $L_i''$ is computed, the algorithm checks whether it is acceptable. The two algorithms run in parallel, so that at every time point both algorithms have explored the same number of edges of $H$. The moment one of the two algorithms finds an acceptable layer, we terminate both algorithms.

Assume w.l.o.g. that the first acceptable layer that was computed by either algorithm is $L_i'$. We then set $X'=L_i'$ and $Y'=V(G)\setminus X'$. 
Note that, from \Cref{claim: ball growing calculation}, $i<d/2$ must hold, so $X'\cap Y=\emptyset$ (as $\dist_H(X,Y)\geq d$).  Clearly, $X\subseteq X'$ and $Y\subseteq Y'$.
From the definition of an acceptable layer, we are guaranteed that $|E_H(X',Y')|=|\delta_H(X')|<\frac{\phi}{2}\cdot |E_H(L'_i)|=\frac{\phi}{2}\cdot |E_H(X')|$. Since we ensured that the number of edges that the two BFS searches explore at every time point is the same, we are guaranteed that $|E_H(X')|\leq |E(Y')|$. The running time of the algorithm is bounded by $O(|E_H(X')|+|\delta_H(X')|+n)\leq O(|E_H(X')|+n)\leq O(|E_G(X')|+n)$.

Lastly, observe that $|E_G(X')|\geq |E_H(X')|$ and $|E_G(Y')|\geq |E_H(Y')|$. Furthermore:

\[|E_G(X',Y')|\leq |E'|+|E_H(X',Y')|\leq \frac{\phi}{4}|X|+\frac{\phi}{2}\cdot \min\set{|E_H(X')|,|E_H(Y')|}\leq \phi\cdot \min\set{|E_G(X')|,|E_G(Y')|},
\]

since $|X|=|Y|$, and graph $G$ is connected, so $|X|\leq 2 \min\set{|E_G(X')|,|E_G(Y')|}$ .

\section{Proof of \Cref{lem: single stage new}}
\label{sec: cutting lemma}
	The proof of the lemma is a simple application of the \CMG. 
	Let $G'=\hat G[V']$, and let $n'=|V(G')|$. If $n'$ is an odd integer, then we add an extra vertex $v_0$ to $G'$, and connect it with an edge to an arbitrary vertex of $G'$. 
	We let $H$ be a graph with $V(H)=V(G')$ and $E(H)=\emptyset$. 
	
	We will execute the \CMG on graph $H$, while simultaneously computing an embedding of $H$ into $G'$. Some of the edges of $H$ will be designated as fake edges and added to the set $F$ of fake edges. These edges do not need to be embedded into $G'$. Initially, $F=\emptyset$.

	We  perform a number of iterations, that correspond to the \CMG. In every iteration $i$, we will add a matching $\mset_i$ to graph $H$, and a set $F_i\subseteq \mset_i$ of fake edges to set $F$. We will also implicitly maintain embedding $\pset_i$ of the set $\mset_i\setminus F_i$ of edges into $G'$ (in other words, the paths in $\pset_i$ are not computed explicitly, but are only guaranteed to exist). We will ensure that the number of iterations is bounded by $O(\log n')\leq O(\log n)$, so the maximum vertex degree in $H$ is always bounded by $\Delta_H\leq O(\log n)$. At the beginning of the algorithm, graph $H$ contains the set $V[G']$ of vertices and no edges. We now describe the execution of the $i$th iteration.

	In order to execute the $i$th iteration, we apply Algorithm  from \Cref{thm: new cut player} to the current graph $H$, with parameter $\eps$ remaining unchanged. 
	Notice that, since $\eps=\frac{1}{(\log\log\log n)^{1/25}}$, and $m$ is large enough, $\eps>\frac{8}{(\log m)^{1/25}}$ holds. Since $|V(G')|\geq \frac{2\hat n}{3}\geq m$, we are guaranteed that $\eps>\frac{2}{(\log n')^{1/25}}$, so that condition of \Cref{thm: new cut player}  holds.

	Assume first that the output of the algorithm from \Cref{thm: new cut player} is a cut $(A_i,B_i)$ in $H$ with $|A_i|,|B_i|\geq n'/4$ and $|E_H(A,B)|\leq n'/100$. We treat this partition as the move of the Cut Player. Assume w.l.o.g. that $|A_i|\leq |B_i|$. Next, we compute an arbitrary partition $(A'_i,B'_i)$ of $V(G')$ with $|A'_i|=|B'_i|$ and $A'_i\subseteq A_i$. 
	We apply Algorithm \matchorcut from \Cref{thm:efficient matching player} to graph $G'$, the sets $A'_i,B'_i$ of vertices, a sparsity parameter $\phi'=\phi/c$, where $c$ is a large enough constant, and parameter $z=8\rho$. Next, we consider two cases. The first case happens if the algorithm returns a cut $(X,Y)$ in $G'$, with $|X|,|Y|\geq z/2\geq 4\rho$, and $|E_{G'}(X,Y)|\leq \phi'\cdot\min\set{|X|,|Y|}=\frac{\phi}{c}\cdot \min\set{|X|,|Y|}$. Once we delete the extra vertex $v_0$ (if it exists), we obtain a cut $(X',Y')$ in the original graph $G'$, with $|X'|,|Y'|\geq 2\rho$ and $|E_{G'}(X',Y')|\leq \frac{2\phi}{c}\cdot \min\set{|X'|,|Y'|}$. Next, we apply Algorithm \makecanonical from \Cref{lem: degree reduction balanced cut case}, to compute a canonical cut $(X'',Y'')$ in $G'$, such that:
	 $|X''|\geq |X'|/2\geq \rho$, $|Y''|\geq |Y'|/2\geq \rho$, and: 
	 
\[\begin{split}
|E_{\hat G}(X'',Y'')|& \leq O(|E_{\hat G}(X',Y')|)\\
&=O(|E_{G'}(X',Y')|)\\
&\leq O\left(\frac{\phi}{c}\cdot \min\set{|X'|,|Y'|}\right )\\
&\leq
\phi\cdot \min\set{|X''|,|Y''|},
\end{split}\]

if $c$ is sufficiently large. We then terminate the algorithm and return the partition $(X'',Y'')$ of $V'$. From the above discussion, it has all required properties.

Consider now the second case, where the algorithm from \Cref{thm:efficient matching player} computes a matching $\mset'_i\subseteq A'_i\times B'_i$ with $|\mset'_i|\geq |A'_i|-z=|A'_i|-8\rho$, such that there exists a set $\pset'_i=\set{P(a,b)\mid (a,b)\in \mset'_i}$ of paths in $G'$, where for each pair $(a,b)\in \mset'_i$, path $P(a,b)$ connects $a$ to $b$, and the paths in $\pset'_i$ cause congestion at most $O\left (\frac{ \log n}{\phi}\right )$.
	We let $A''_i\subseteq A'_i$, $B''_i\subseteq B'_i$ be the sets of vertices that do not participate in the matching $\mset'_i$, and we let $F_i$ be an arbitrary perfect matching between these vertices. Lastly, we set $\mset_i=\mset'_i\cup F_i$. We view the matching $\mset_i$ as the response of the matching player in the \CMG. We add the edges of $M_i$ to $H$, and continue to the next iteration.
	Notice that $|F_i|\leq 8\rho$.

	We perform the iterations as described above, until the algorithm from \Cref{thm: new cut player} returns a subset $S\subseteq V(G')$ of at least $|V(G')|/2$ vertices, such that graph $H[S]$ is $\phi^*$-expander, for $\phi^*\geq \Omega\left ( \frac{1}{2^{O(1/\eps^6)}\cdot \Delta_H^3\cdot \log^2 n }\right )\geq \Omega\left ( \frac{1}{2^{O(1/\eps^6)}\cdot \log^5 n }\right ) $. Recall that \Cref{thm:KKOV-new} guarantees that this must happen after at most $O(\log n)$ iterations. We then perform one last iteration, whose index we denote by $q$.

	We let $B_q=S$ and $A_q=V(G)\setminus S$, and apply Algorithm  \matchorcut from \Cref{thm:efficient matching player} to the sets $A_q,B_q$ of vertices,  a sparsity parameter $\phi'=\phi/c$ and parameter $z=8\rho$. As before, we consider two cases. The first case happens if the algorithm returns a cut $(X,Y)$ in $G$, with $|X|,|Y|\geq z/2\geq 4\rho$ and $|E_{G'}(X,Y)|\leq \phi'\cdot\min\set{|X|,|Y|}$. In this case, we compute a partition $(X'',Y'')$ of $V(G')\setminus\set{v_0}$ exactly as before, so that both $X'',Y''$ are canonical sets of vertices of cardinality at least $\rho$ each, and $|E_{\hat G}(X'',Y'')|\leq
	\phi\cdot \min\set{|X''|,|Y''|}$. We return the cut $(X'',Y'')$ and terminate the algorithm.  
	In the second case, the algorithm from \Cref{thm:efficient matching player} computes a matching $\mset'_q\subseteq A'_q\times B'_q$ with $|\mset'_q|\geq |A_q|-z=|A_q|-8\rho$, such that there exists a set $\pset'_q=\set{P(a,b)\mid (a,b)\in \mset'_q}$ of paths in $G'$, where for each pair $(a,b)\in \mset'_q$, path $P(a,b)$ connects $a$ to $b$, and the paths in $\pset'_q$ cause congestion at most $O\left (\frac{ \log n}{\phi}\right )$. As before, we let $A'_q\subseteq A_q$, $B'_q\subseteq B_q$ be the sets of vertices that do not participate in the matching $\mset'_q$, and we let $F_q$ be an arbitrary matching that connects every vertex of $A'_q$ to a distinct vertex of $B'_q$ (such a matching must exist since $|A_q|\leq |B_q|$). We then set $\mset_q=\mset'_q\cup \mset''_q$, and we add the edges of $\mset_q$ to graph $H$.

	From now on we assume that the algorithm never terminated with a partition $(X'',Y'')$ of $V(G')\setminus\set{v_0}$, where both $X'',Y''$ are canonical sets of vertices of cardinality at least $\rho$ each, and $|E_{\hat G}(X'',Y'')|\leq
	\phi\cdot \min\set{|X''|,|Y''|}$. Note that, from \Cref{obs: exp plus matching is exp}, the final graph $H$ is a $\phi^*/2$-expander, for $\phi^*\geq \Omega\left ( \frac{1}{2^{O(1/\eps^6)}\cdot \log^5 n }\right ) $. 
	Let $F=\bigcup_iF_i$. Since, for all $i$, $|F_i|\leq 8\rho$, and since, from \Cref{thm:KKOV-new}, the number of iterations is bounded by $O(\log n)$, we get that $|F|\leq O(\rho\log n)$. 
	Lastly, consider the set $\pset=\bigcup_i\pset_i$ of paths in graph $G'$. It is immediately to verify that the paths in $\pset$ embed graph $H\setminus F$ into $G'$. Since every set $\pset_i$ of paths causes congestion at most $O\left (\frac{ \log n}{\phi}\right )$, the paths in $\pset$ cause congestion at most $O\left (\frac{ \log^2 n}{\phi}\right )$ in $G'$.
	
	One remaining subtlety is that graph $H$, as well as current graph $G'$ may contain the extra vertex $v_0$, that needs to be removed from both graphs. Recall that the degree of $v_0$ in graph $H$ is at most $O(\log n)$. Let $u_1,\ldots,u_r$ denote the neighbor vertices of $v_0$ in $H$. Let $H'$ be obtained from graph $H$ by deleting vertex $v_0$ from it, and adding, for every pair $u_j,u_{j'}$ of neighbor vertices of $v_0$, and edge $(u_j,u_{j'})$ connecting them. Each such new edge is added to the set $F$ of fake edges. It is easy to verify that $H'$ remains a $\phi^*/2$-expander. Since $\rho\geq \log  n$, while the degree of $v_0$ in $H$ is at most $O(\log n)$, $|F|\leq O(\rho\log n)$ continues to hold, and all vertex degrees in $H'$ are at most $O(\log n)$. 
	Since vertex $v_0$ has degree $1$ in $G'$, we can assume that it does not lie on any path in $\set{P(e)\mid e\in E(H')\setminus F}$, and so $v_0$ can be safely deleted from $G'$ as well. The output of the algorithm in this case is graph $H'$ and set $F$ of its edges.
	
	Lastly, we bound the running time of the algorithm. The algorithm consists of $O(\log n)$ iterations. Every iteration employs the algorithm from 
	\Cref{thm: new cut player}, whose running time is  $O\left(|E(H)|^{1+O(\eps)}\cdot \Delta_H^7\right )\leq O\left(n^{1+O(\eps)}\right ) $, since $\Delta_H\leq O(\log n)$, and $\log^8n<n^{4\eps}$ (since, as we have observed, $n^{2\eps}\geq m^{\eps}>\log^4 n$).
	Additionally, in every iteration we use Algorithm \matchorcut from \Cref{thm:efficient matching player}, whose running time is $O\left (m^{1+o(1)}\right)$, and Algorithm \makecanonical from \Cref{lem: degree reduction balanced cut case}, whose running time is $O(m)$. Therefore, the total running time is $O\left (m^{1+O(\eps)+o(1)}\right)$.

\bibliographystyle{alpha}

\bibliography{APSP-expanders}

\end{document}